\documentclass[letterpaper,11pt]{article}  
\usepackage{amsfonts}
\usepackage[pdftex]{graphicx}
\usepackage{amsmath,amssymb,amsthm}
\usepackage{natbib}
\usepackage{hyperref}
\usepackage{bm}
\usepackage{fancyhdr}
\usepackage[titletoc,title]{appendix}
\usepackage{sectsty}
\usepackage{stmaryrd}
\usepackage{setspace}                      
\usepackage{booktabs}    
\usepackage{pdfsync}        
\usepackage[backgroundcolor=blue!40,linecolor=blue!40]{todonotes}
\usepackage{fancybox}
\usepackage[skip=2pt, font=small]{caption}
\usepackage{tikz}

\usepackage{subfigure}
\usepackage[pdftex]{graphicx}
 \usepackage{multirow}
\usepackage{pdflscape}
\usepackage{graphicx}
\usepackage{booktabs,caption,fixltx2e}
\usepackage[flushleft]{threeparttable}
\usepackage{geometry}

\usepackage{titlesec}

\titleclass{\subsubsubsection}{straight}[\subsection]

\newcounter{subsubsubsection}[subsubsection]
\renewcommand\thesubsubsubsection{\thesubsubsection.\arabic{subsubsubsection}}

\titleformat{\subsubsubsection}
  {\normalfont\normalsize\bfseries}{\thesubsubsubsection}{1em}{}
\titlespacing*{\subsubsubsection}
{0pt}{3.25ex plus 1ex minus .2ex}{1.5ex plus .2ex}

\makeatletter
\renewcommand\paragraph{\@startsection{paragraph}{5}{\z@}%
  {3.25ex \@plus1ex \@minus.2ex}%
  {-1em}%
  {\normalfont\normalsize\bfseries}}
\renewcommand\subparagraph{\@startsection{subparagraph}{6}{\parindent}%
  {3.25ex \@plus1ex \@minus .2ex}%
  {-1em}%
  {\normalfont\normalsize\bfseries}}
\def\toclevel@subsubsubsection{4}
\def\toclevel@paragraph{5}
\def\toclevel@paragraph{6}
\def\l@subsubsubsection{\@dottedtocline{4}{7em}{4em}}
\def\l@paragraph{\@dottedtocline{5}{10em}{5em}}
\def\l@subparagraph{\@dottedtocline{6}{14em}{6em}}
\makeatother

\usetikzlibrary{decorations.pathreplacing}
\usetikzlibrary{patterns}
\usetikzlibrary{math}

\definecolor{light-gray}{gray}{0.8}

\newtheoremstyle{myplain}
  {9pt}
  {9pt}
  {\itshape}
  {\parindent}
  {\scshape}
  {:}
  {.5em}
  {}
\newtheoremstyle{mydefinition}
  {9pt}
  {9pt}
  {\itshape}
  {\parindent}
  {\scshape}
  {:}
  {.5em}
  {}
\newtheoremstyle{myremark}
  {9pt}
  {9pt}
  {}
  {\parindent}
  {\scshape}
  {:}
  {.5em}
  {}
\theoremstyle{plain}
\newtheorem{theorem}{Theorem}[section]

\theoremstyle{mydefinition}

\newtheorem{definition}{Definition}[section]

\newtheorem{SIR}{Theorem SIR-\ignorespaces}[section]
\newtheorem{OR}{Theorem OR-\ignorespaces}[section]
\newtheorem{BI}{Key Insight}[section]

\theoremstyle{myremark}
\newtheorem{IP}{Identification Problem}[section]

\renewcommand{\cite}{\citet}
\newcommand{\edis}{\stackrel{d}{=}}

\def\argmin{\mathop{\rm arg\,min}}

\newcommand{\M}{\mathbb{M}}

\newcommand{\T}{\mathbb{T}}
\newcommand{\R}{\mathbb{R}}
\newcommand{\E}{\mathbb{E}}
    
\newcommand{\bU}{\mathfrak{U}}    
\newcommand{\bu}{\mathfrak{u}}

\newcommand{\cB}{\mathcal{B}}   
\newcommand{\cC}{\mathcal{C}} 
 
\newcommand{\cE}{\mathcal{E}} 
\newcommand{\cF}{\mathcal{F}}
\newcommand{\cG}{\mathcal{G}}

\newcommand{\cI}{\mathcal{I}}
\newcommand{\cJ}{\mathcal{J}}
\newcommand{\cK}{\mathcal{K}}  
  
\newcommand{\cM}{\mathcal{M}}

\newcommand{\cP}{\mathcal{P}}

\newcommand{\cT}{\mathcal{T}}

\newcommand{\cV}{\mathcal{V}}
\newcommand{\cW}{\mathcal{W}}  
\newcommand{\cX}{\mathcal{X}}  
\newcommand{\cY}{\mathcal{Y}}  
\newcommand{\cZ}{\mathcal{Z}}  
\newcommand{\sF}{\mathsf{F}}   
\newcommand{\sM}{\mathsf{M}}   
\newcommand{\sG}{\mathsf{G}}  
\newcommand{\sT}{\mathsf{T}}  
\newcommand{\sB}{\mathsf{B}}  
\newcommand{\sC}{\mathsf{C}}   
\newcommand{\sP}{\mathsf{P}}  
\newcommand{\sQ}{\mathsf{Q}}  
\newcommand{\sq}{\mathsf{q}}  
\newcommand{\sR}{\mathsf{R}}  
 
\newcommand{\sd}{\mathsf{d}}  
\newcommand{\cp}{\mathsf{p}}

\newcommand{\cf}{\mathsf{f}}

\newcommand{\eU}{{\boldsymbol{U}}}
\newcommand{\eb}{{\boldsymbol{b}}}
\newcommand{\ed}{{\boldsymbol{d}}}
\newcommand{\eu}{{\boldsymbol{u}}}
\newcommand{\ew}{{\boldsymbol{w}}}

\newcommand{\eX}{{\boldsymbol{X}}}
\newcommand{\ex}{{\boldsymbol{x}}}
\newcommand{\eY}{{\boldsymbol{Y}}}
\newcommand{\eB}{{\boldsymbol{B}}}
\newcommand{\eC}{{\boldsymbol{C}}}

\newcommand{\eQ}{{\boldsymbol{Q}}}
\newcommand{\eS}{{\boldsymbol{S}}}
\newcommand{\eT}{{\boldsymbol{T}}}

\newcommand{\eH}{{\boldsymbol{H}}}

\newcommand{\ey}{{\boldsymbol{y}}}
\newcommand{\eZ}{{\boldsymbol{Z}}}
\newcommand{\eG}{{\boldsymbol{G}}}
\newcommand{\ez}{{\boldsymbol{z}}}
\newcommand{\es}{{\boldsymbol{s}}}
\newcommand{\et}{{\boldsymbol{t}}}
\newcommand{\ev}{{\boldsymbol{v}}}

\newcommand{\eq}{{\boldsymbol{q}}}

\newcommand{\eps}{\varepsilon}
\newcommand{\Eps}{\mathcal{E}}
\newcommand{\carrier}{{\mathfrak{X}}}
\newcommand{\Ball}{{\mathbb{B}}^{d}}
\newcommand{\Sphere}{{\mathbb{S}}^{d-1}}

\newcommand{\salg}{\mathfrak{F}}
\newcommand{\ssalg}{\mathfrak{B}}
\newcommand{\one}{\mathbf{1}}
\renewcommand{\P}{\mathbb{P}}
\newcommand{\Prob}[1]{\P\{#1\}}
\newcommand{\yL}{\ey_{\mathrm{L}}}
\newcommand{\yU}{\ey_{\mathrm{U}}}
\newcommand{\yLi}{\ey_{\mathrm{L}i}}
\newcommand{\yUi}{\ey_{\mathrm{U}i}}
\newcommand{\xL}{\ex_{\mathrm{L}}}
\newcommand{\xU}{\ex_{\mathrm{U}}}

\newcommand{\dist}{\mathbf{d}}
\newcommand{\rhoH}{\dist_{\mathrm{H}}}

\newcommand{\crit}{q}
\newcommand{\CS}{CS_n}
\newcommand{\CI}{CI_n}

\DeclareMathOperator{\conv}{conv}
\DeclareMathOperator{\cov}{Cov}

\DeclareMathOperator{\Sel}{Sel}

\DeclareMathOperator{\cl}{cl}

\newcommand{\idr}[1]{\mathcal{H}_\sP[#1]}
\newcommand{\outr}[1]{\mathcal{O}_\sP[#1]}
\newcommand{\idrn}[1]{\hat{\mathcal{H}}_{\sP_n}[#1]}

\hypersetup{colorlinks=true, linkcolor=blue, urlcolor  = blue, citecolor=blue}                          
\numberwithin{equation}{section}
\numberwithin{figure}{section}
\numberwithin{table}{section}

\newcommand{\possessivecite}[1]{\citeauthor{#1}'s \citeyear{#1}}

\newcommand\xqed[1]{%
  \leavevmode\unskip\penalty9999 \hbox{}\nobreak\hfill
  \quad\hbox{#1}}
\newcommand\qedex{\xqed{$\triangle$}}

\usetikzlibrary{shapes.geometric, arrows,positioning}
\usetikzlibrary{arrows}

\tikzstyle{lemma} = [rectangle, minimum width=3cm, minimum height=1cm,text centered, draw=black, fill=white!100]
\tikzstyle{comment} = [rectangle, rounded corners, minimum width=3cm, minimum height=1cm,text centered, draw=black, fill=gray!30]
\tikzstyle{line} = [draw, -latex']

\newtheoremstyle{exampstyle}
  {\topsep} 
  {\topsep} 
  {} 
  {} 
  {\bfseries} 
  {.} 
  {.5em} 
  {} 

\theoremstyle{exampstyle} \newtheorem{examp}[theorem]{Example}

\newcommand\independent{\protect\mathpalette{\protect\independenT}{\perp}}
\def\independenT#1#2{\mathrel{\rlap{$#1#2$}\mkern2mu{#1#2}}}

\begin{document}

\title{Microeconometrics with Partial Identification}
\date{March 12, 2020}
\author{%
\begin{tabular}{c}  Francesca Molinari \\ Cornell University \\ Department of Economics \\ \url{fm72@cornell.edu}\thanks{
This manuscript was prepared for the \href{https://www.elsevier.com/books/handbook-of-econometrics/durlauf/978-0-444-63649-2}{Handbook of Econometrics, Volume 7A} \textcopyright North Holland, 2019.
I thank Don Andrews, Isaiah Andrews, Levon Barseghyan, Federico Bugni, Ivan Canay, Joachim Freyberger, Hiroaki Kaido, Toru Kitagawa, Chuck Manski, Rosa Matzkin, Ilya Molchanov, \'{A}ureo de Paula, Jack Porter, Seth Richards-Shubik, Adam Rosen, Shuyang Sheng, J\"{o}rg Stoye, Elie Tamer, Matthew Thirkettle, and participants to the 2017 Handbook of Econometrics Conference, for helpful comments, and the National Science Foundation for financial support through grants SES-1824375 and SES-1824448.
I am grateful to Louis Liu and Yibo Sun for research assistance supported by the Robert S. Hatfield Fund for Economic Education at Cornell University.
Part of this research was carried out during my sabbatical leave at the Department of Economics at Duke University, whose hospitality
I gratefully acknowledge.} \\ \\ 
\end{tabular} }

\maketitle
\begin{abstract}
This chapter reviews the microeconometrics literature on \emph{partial identification}, focusing on the developments of the last thirty years. 
The topics presented illustrate that the available data combined with credible maintained assumptions may yield much information about a parameter of interest, even if they do not reveal it exactly.
Special attention is devoted to discussing the challenges associated with, and some of the solutions put forward to,
(1) obtain a tractable characterization of the values for the parameters of interest which are observationally equivalent, given the available data and maintained assumptions; (2) estimate this set of values; (3) conduct test of hypotheses and make confidence statements. 
The chapter reviews advances in partial identification analysis both as applied to learning (functionals of) probability distributions that are well-defined in the absence of models, as well as to learning parameters that are well-defined only in the context of particular models.
A simple organizing principle is highlighted: the source of the identification problem can often be traced to a collection of random variables that are consistent with the available data and maintained assumptions.
This collection may be part of the observed data or be a model implication.
In either case, it can be formalized as a \emph{random set}.
Random set theory is then used as a mathematical framework to unify a number of special results and produce a general methodology to carry out partial identification analysis.
\end{abstract}
\vfill
\thispagestyle{empty}
\pagebreak
\onehalfspacing

\pagenumbering{arabic}
\tableofcontents
\pagebreak

\section{Introduction}
\label{sec:intro}
\subsection{Why Partial Identification?}
Knowing the population distribution that data are drawn from, what can one learn about a parameter of interest?
It has long been understood that assumptions about the data generating process (DGP) play a crucial role in answering this \emph{identification question} at the core of all empirical research.
Inevitably, assumptions brought to bear enjoy a varying degree of credibility.
Some are rooted in economic theory (e.g., optimizing behavior) or in information available to the researcher on the DGP (e.g., randomization mechanisms).
These assumptions can be argued to be highly credible.
Others are driven by concerns for tractability and the desire to answer the identification question with a certain level of precision (e.g., functional form and distributional assumptions).
These are arguably less credible.

%
Early on, \cite{koo:rei50} highlighted the importance of imposing restrictions based on prior knowledge of the phenomenon under analysis and some criteria of simplicity, but not for the purpose of identifiability of a parameter that the researcher happens to be interested in, stating (p.~169): 
``One might regard problems of identifiability as a necessary part of the specification problem.
We would consider such a classification acceptable, provided the temptation to specify models in such a way as to produce identifiability of relevant characteristics is resisted."

Much work, spanning multiple fields, has been devoted to putting forward strategies to carry out empirical research while relaxing distributional, functional form, or behavioral assumptions. 
One example, embodied in the research program on semiparameteric and nonparametric methods, is to characterize sufficient sets of assumptions, that exclude many suspect ones --sometimes as many as possible-- to guarantee that point identification of specific economically interesting parameters attains.
This literature is reviewed in, e.g., \cite{mat07,mat13}, and is not discussed here.

Another example, embodied in the research program on Bayesian model uncertainty, is to specify multiple models (i.e., multiple sets of assumptions), put a prior on the parameters of each model and on each model, embed the various separate models within one large hierarchical mixture model, and obtain model posterior probabilities which can be used for a variety of inferences and decisions.
This literature is reviewed in, e.g., \cite{was00} and \cite{cly:geo04}, and is not discussed here.

The approach considered here fixes a set of assumptions and a parameter of interest a priori, in the spirit of \cite{koo:rei50}, and asks what can be learned about that parameter given the available data, recognizing that even partial information can be illuminating for empirical research, while enjoying wider credibility thanks to the weaker assumptions imposed.
The bounding methods at the core of this approach appeared in the literature nearly a century ago.
Arguably, the first exemplar that leverages economic reasoning is given by the work of \cite{mar:and44}.
They provided bounds on Cobb-Douglas production functions in models of supply and demand, building on optimization principles and restrictions from microeconomic theory.
\cite{lea81} revisited their analysis to obtain bounds on the elasticities of demand and supply in a linear simultaneous equations system with uncorrelated errors.
The first exemplars that do not rely on specific economic models appear in \cite{gin21}, \cite{fri34}, and \cite{rei41}, who bounded the coefficient of a simple linear regression in the presence of measurement error.
These results were extended to the general linear regression model with errors in all variables by \cite{kle:lea84} and \cite{lea87}.

This chapter surveys some of the methods proposed over the last thirty years in the microeconometrics literature to further this approach.
These methods belong to the systematic program on \emph{partial identification} analysis started with \cite{man89,man90,man95,man03,man07a,man13book} and developed by several authors since the early 1990s. 
Within this program, the focus shifts from points to sets: the researcher aims to learn what is the set of values for the parameters of interest that can generate the same distribution of observables as the one in the data, for some DGP consistent with the maintained assumptions.
In other words, the focus is on the \emph{set of observationally equivalent} values, which henceforth I refer to as the parameters'  \emph{sharp identification region}.
In the partial identification paradigm, empirical analysis begins with characterizing this set using the data alone.
This is a nonparametric approach that dispenses with all assumptions, except basic restrictions on the sampling process such that the distribution of the observable variables can be learned as data accumulate.
In subsequent steps, one incorporates additional assumptions into the analysis, reporting how each assumption (or set of assumptions) affects what one can learn about the parameters of interest, i.e., how it modifies and possibly shrinks the sharp identification region.
Point identification may result from the process of increasingly strengthening the maintained assumptions, but it is not the goal in itself.
Rather, the objective is to make transparent the relative role played by the data and the assumptions in shaping the inference that one draws.\medskip

There are several strands of independent, but thematically related literatures that are not discussed in this chapter.
As a consequence, many relevant contributions are left out of the presentation and the references.
One example is the literature in finance.
\cite{han:jag91} developed nonparametric bounds for the admissible set for means and standard deviations of intertemporal marginal rates of substitution (IMRS) of consumers.  
The bounds were developed exploiting the condition, satisfied in many finance models, that the equilibrium price of any traded security equals the expectation (conditioned on current information) of the product's future payoff and the IMRS of any consumer.\footnote{
\cite{han:jag91} deduce a duality relation with the mean variance theory of \cite{mar52} and \cite{fam96}, but the relation does not apply to the sharp bounds they derive.  
In the Arbitrage Pricing Model \citep{ros76}, bounds on extensions of existing pricing functions, consistent with the absence of arbitrage opportunities, were considered by \cite{har:kre79} and \cite{kre81}.
}  
\cite{lut96} extended the analysis to economies with frictions.  
\cite{han:hea:lut95} developed econometric tools to estimate the regions, to assess asset pricing models, and to provide nonparametric characterizations of asset pricing anomalies.  
Earlier on, the existence of volatility bounds on IMRSs were noted by \cite{shi82} and \cite{han82comment}.  
The bounding arguments that build on the minimum-volatility frontier for stochastic discount factors proposed by \cite{han:jag91} have become a litmus test to detect anomalies in asset pricing models \citep[see, e.g.][p. 89]{shi03}. 
I refer to the textbook presentations in \cite[Chapter 13]{lju:sar04} and \cite[Chapters 5 and 21]{coc05}, and the review articles by \cite{fer03} and \cite{cam14}, for a careful presentation of this literature.


In macroeconomics, \cite{fau98}, \cite{can:den02}, and \cite{uhl05} proposed bounds for impulse response functions in sign-restricted structural vector autoregression models, and carried out Bayesian inference with a non-informative prior for the non-identified parameters. 
I refer to \cite[Chapter 13]{kil:lut17} for a careful presentation of this literature. 

In microeconomic theory, bounds were derived from inequalities resulting as necessary and sufficient conditions that data on an individual's choice need to satisfy in order to be consistent with optimizing behavior, as in the research pioneered by \cite{sam38} and advanced early on by \cite{hou50} and \cite{ric66}.
\cite{afr67} and \cite{var82} extended this research program to revealed preference extrapolation.
Notably, in this work no stochastic terms enter the analysis. 
\cite{blo:mar60}, \cite{mar60}, \cite{hal73}, \cite{mcf75}, \cite{fal78}, and \cite{mcf:ric91}, extended revealed preference arguments to random utility models, and obtained bounds on the distributions of preferences.
I refer to the survey articles by \cite{cra:der14} and \cite[Chapter XXX in this Volume]{blu19} for a careful presentation of this literature.

A complementary approach to partial identification is given by sensitivity analysis, advocated for in different ways by, e.g., \cite{gil:lea83}, \cite{ros:rub83}, \cite{lea85}, \cite{ros95}, \cite{imb03}, and others.
Within this approach, the analysis begins with a fully parametric model that point identifies the parameter of interest.
One then reports the set of values for this parameter that result when the more suspicious assumptions are relaxed.

Related literatures, not discussed in this chapter, abound also outside Economics.
For example, in probability theory, \cite{hoe40} and \cite{fre51} put forward bounds on the joint distributions of random variables, and \cite{mak81}, \cite{rus82}, and \cite{fra:nel:sch87} on the sum of random variables, when only marginal distributions are observed. 
The literature on probability bounds is discussed in the textbook by \cite[Appendix A]{sho:wel09}.
Addressing problems faced in economics, sociology, epidemiology, geography, history, political science, and more, \cite{dun:dav53} derived bounds on correlations among variables measured at the individual level based on observable correlations among variables measured at the aggregate level.
The so called ecological inference problem they studied, and the associated literature, is discussed in the survey article by \cite{cho:man09} and references therein.

\subsection{Goals and Structure of this Chapter}
To carry out econometric analysis with partial identification, one needs: (1) computationally feasible characterizations of the parameters' sharp identification region; (2) methods to estimate this region; and (3) methods to test hypotheses and construct confidence sets.
The goal of this chapter is to provide insights into the challenges posed by each of these desiderata, and into some of their solutions.
In order to discuss the partial identification literature in microeconometrics with some level of detail while keeping this chapter to a manageable length, I focus on a selection of papers and not on a complete survey of the literature.
As a consequence, many relevant contributions are left out of the presentation and the references.
I also do not discuss the important but
separate topic of statistical decisions in the presence of partial identification, for which I refer to the textbook treatments in \cite{man05,man07a} and to the review by \cite[Chapter XXX in this Volume]{hir:por19}.

The presumption in identification analysis that the distribution from which the data are drawn is known allows one to keep separate the identification question from the distinct question of statistical inference from a finite sample.
I use the same separation in this chapter.
I assume solid knowledge of the topics covered in first year Economics PhD courses in econometrics and microeconomic theory.

I begin in Section \ref{sec:prob:distr} with the analysis of what can be learned about features of probability distributions that are well defined in the absence of an economic model, such as moments, quantiles, cumulative distribution functions, etc., when one faces measurement problems.
Specifically, I focus on cases where the \emph{data is incomplete}, either due to selectively observed data or to interval measurements.
I refer to \cite{man95,man03,man07a} for textbook treatments of many other cases.
I lay out formally the maintained assumptions for several examples, and then discuss in detail what is the source of the identification problem.
I conclude with providing tractable characterizations of what can be learned about the parameters of interest, with formal proofs.
I show that even in simple problems, great care may be needed to obtain the sharp identification region.
It is often easier to characterize an \emph{outer region}, i.e., a collection of values for the parameter of interest that contains the sharp one but may contain also additional values. 
Outer regions are useful because of their simplicity and because in certain applications they may suffice to answer questions of great interest, e.g., whether a policy intervention has a nonnegative effect.
However, compared to the sharp identification region they may afford the researcher less useful predictions, and a lower ability to test for misspecification, because they do not harness all the information in the observed data and maintained assumptions.

In Section \ref{sec:structural} I use the same approach to study what can be learned about features of parameters of structural econometric models when the \emph{model is incomplete} \citep{tam03,hai:tam03,cil:tam09}.
Specifically, I discuss single agent discrete choice models under a variety of challenging situations (interval measured as well as endogenous explanatory variables; unobserved as well as counterfactual choice sets); finite discrete games with multiple equilibria; auction models under weak assumptions on bidding behavior; and network formation models.
Again I formally derive sharp identification regions for several examples.

I conclude each of these sections with a brief discussion of further theoretical advances and empirical applications that is meant to give a sense of the breadth of the approach, but not to be exhaustive.
I refer to the recent survey by \cite{ho:ros17} for a thorough discussion of empirical applications of partial identification methods.

In Section \ref{sec:inference} I discuss finite sample inference.
I limit myself to highlighting the challenges that one faces for consistent estimation when the identified object is a set, and several coverage notions and requirements that have been proposed over the last 20 years.
I refer to the recent survey by \cite{can:sha17} for a thorough discussion of methods to tests hypotheses and build confidence sets in moment inequality models.

In Section \ref{sec:misspec} I discuss the distinction between refutable and non-refutable assumptions, and how model misspecification may be detectable in the presence of the former, even within the partial identification paradigm.
I then highlight certain challenges that model misspecification presents for the interpretation of sharp identification (as well as outer) regions, and for the construction of confidence sets.

In Section \ref{sec:computations} I highlight that while most of the sharp identification regions characterized in Section \ref{sec:prob:distr} can be easily computed, many of the ones in Section \ref{sec:structural} are more challenging.
This is because the latter are obtained as level sets of criterion functions in moderately dimensional spaces, and tracing out these level sets or their boundaries is a non-trivial computational problem. 
In Section \ref{sec:future} I conclude providing some considerations on what I view as open questions for future research.

I refer to \cite{tam10} for an earlier review of this literature, and to \cite{lew18} for a careful presentation of the many notions of identification that are used across the econometrics literature, including an important historical account of how these notions developed over time.

\subsection{Random Set Theory as a Tool for Partial Identification Analysis}
Throughout Sections \ref{sec:prob:distr} and \ref{sec:structural}, a simple organizing principle for much of partial identification analysis emerges.
The cause of the identification problems discussed can be traced back to a collection of random variables that are consistent with the available data and maintained assumptions.
For the problems studied in Section \ref{sec:prob:distr}, this set is often a simple function of the observed variables.
The incompleteness of the data stems from the fact that instead of observing the singleton variables of interest, one observes set-valued variables to which these belong, but one has no information on their exact value within the sets.
For the problems studied in Section \ref{sec:structural}, the collection of random variables consistent with the maintained assumptions comprises what the model predicts for the endogenous variable(s).
The incompleteness of the model stems from the fact that instead of making a singleton prediction for the variable(s) of interest, the model makes multiple predictions but does not specify how one is chosen.

The central role of set-valued objects, both stochastic and nonstochastic, in partial identification renders \emph{random set theory} a natural toolkit to aid the analysis.\footnote{The first idea of a general random set in the form of a region that depends on chance appears in \cite{kol50}, originally published in 1933. 
For another early example where confidence regions are explicitly described as random sets, see \cite[p. 67]{haa44}.
The role of random sets in this chapter is different.}
This theory originates in the seminal contributions of \cite{cho53}, \cite{aum65}, and \cite{deb67}, with the first self contained treatment of the theory given by \cite{mat75}.  
I refer to \cite{mo1} for a textbook presentation, and to \cite{mol:mol14,mol:mol18} for a treatment focusing on its applications in econometrics.

\cite{ber:mol08} introduce the use of random set theory in econometrics to carry out identification analysis and statistical inference with incomplete data.
\cite{ber:mol:mol11,ber:mol:mol12} propose it to characterize sharp identification regions both with incomplete data and with incomplete models.
\cite{gal:hen11} propose the use of optimal transportation methods that in some applications deliver the same characterizations as the random set methods.
I do not discuss optimal transportation methods in this chapter, but refer to \cite{gal16} for a thorough treatment.

Over the last ten years, random set methods have been used to unify a number of specific results in partial identification, and to produce a general methodology for identification analysis that dispenses completely with case-by-case distinctions.
In particular, as I show throughout the chapter, the methods allow for simple and tractable characterizations of sharp identification regions.
The collection of these results establishes that indeed this is a useful tool to carry out econometrics with partial identification, as exemplified by its prominent role both in this chapter and in Chapter XXX in this Volume by \cite{che:ros19}, which focuses on general classes of instrumental variable models.
The random sets approach complements the more traditional one, based on mathematical tools for (single valued) random vectors, that proved extremely productive since the beginning of the research program in partial identification.

This chapter shows that to fruitfully apply random set theory for identification and inference, the econometrician needs to carry out three fundamental steps. 
First, she needs to define the random closed set that is relevant for the problem under consideration using all information given by the available data and maintained assumptions. 
This is a delicate task, but one that is typically carried out in identification analysis regardless of whether random set theory is applied.
Indeed, throughout the chapter I highlight how relevant random closed sets were characterized in partial identification analysis since the early 1990s, albeit the connection to the theory of random sets was not made.
As a second step, the econometrician needs to determine how the observable random variables relate to the random closed set. 
Often, one of two cases occurs: either the observable variables determine a random set to which the unobservable variable of interest belongs with probability one, as in incomplete data scenarios; or the (expectation of the) (un)observable variable belongs to (the expectation of) a random set determined by the model, as in incomplete model scenarios. 
Finally, the econometrician needs to determine which tool from random set theory should be utilized. 
To date, new applications of random set theory to econometrics have fruitfully exploited (Aumann) expectations and their support functions, (Choquet) capacity functionals, and laws of large numbers and central limit theorems for random sets.
Appendix \ref{app:RCS} reports basic definitions from random set theory of these concepts, as well as some useful theorems.
The chapter explains in detail through applications to important identification problems how these steps can be carried out.

\begin{table}
\begin{center}
\caption{Notation Used} \label{tab:notation}
\resizebox{\textwidth}{!}{\begin{tabular}{ll}
\hline 
$(\Omega,\salg,\P)$ & Nonatomic probability space\\
$\R^d,\|\cdot\|$ & Euclidean space equipped with the Euclidean norm \\
$\cF,\cG,\cK$ & Collection of closed, open, and compact subsets of $\R^d$ (respectively) \\ 
$\Sphere = \{x \in \R^d: \|x\| = 1\}$ & Unit sphere in $\R^d$ \\ 
$\Ball = \{x \in \R^d: \|x\| \leq 1\}$ & Unit ball in $\R^d$ \\
$\conv(A),\cl(A),|B|$ & Convex hull and closure of a set $A\subset\R^d$ (respectively), and cardinality of a finite set $B\subset\R^d$ \\
\hline 
$\ex,\ey,\ez,\dots$ & Random vectors \\ 
$x,y,z,\dots$ & Realizations of random vectors or deterministic vectors \\ 
$\eX,\eY,\eZ,\dots$ & Random sets \\ 
$X,Y,Z,\dots$ & Realizations of random sets or deterministic sets\\ 
$\epsilon,\eps,\nu,\zeta$ & Unobserved random variables (heterogeneity) \\ 
$\Theta,\theta,\vartheta$ & Parameter space, data generating value for the  parameter vector, and a generic element of $\Theta$  \\ 
\hline 
$\sR$ & Joint distribution of all variables (observable and unobservable) \\
$\sP$ & Joint distribution of the observable variables \\
$\sQ$ & Joint distribution whose features one wants to learn \\
$\sM$ & A joint distribution of observed variables implied by the model \\
$\sq_{\tau}(\alpha)$ & Quantile function at level $\alpha \in (0,1)$ for a random variable distributed $\tau\in\{\sR,\sP,\sQ\}$ \\
$\E_\tau$ & Expectation operator associated with distribution $\tau\in\{\sR,\sP,\sQ\}$ \\
$\sT_\eX(K)=\Prob{\eX\cap K\neq\emptyset},\, K\in\cK$ & Capacity functional of random set $\eX$ \\
$\sC_\eX(F)=\Prob{\eX\subset F},\, F\in\cF$ & Containment functional of random set $\eX$ \\
\hline
$\stackrel{p}{\rightarrow},~\stackrel{\text{a.s.}}{\rightarrow},~\Rightarrow$ & Convergence in probability, convergence almost surely, and weak convergence (respectively) \\
$\ex\edis\ey$ & $\ex$ and $\ey$ have the same distribution \\
$\ex\independent\ey$ & Statistical independence between random variables $\ex$ and $\ey$ \\
$x^\top y$ & Inner product between vectors $x$ and $y$, $x,y\in\R^d$ \\
\hline
$\bU,\bu$ & Family of utility functions and one of its elements \\
$\crit_\sP$ & Criterion function that aggregates violations of the population moment inequalities \\
$\crit_n$ & Criterion function that aggregates violations of the sample moment inequalities \\
\hline
$\idr{\cdot}$ & Sharp identification region of the functional in square brackets (a function of $\sP$)\\
$\outr{\cdot}$ & An outer region of the functional in square brackets (a function of $\sP$) \\
\hline
\hline
\end{tabular}}
\end{center} 
\end{table}

\subsection{Notation}
\label{subsec:notation}
This chapter employs consistent notation that is summarized in Table \ref{tab:notation}. 
Some important conventions are as follows: 
$\ey$ denotes outcome variables, $(\ex,\ew)$ denote explanatory variables, and $\ez$ denotes instrumental variables (i.e., variables that satisfy some form of independence with the outcome or with the unobservable variables, possibly conditional on $\ex,\ew$).

I denote by $\sP$ the joint distribution of all observable variables.
Identification analysis is carried out using the information contained in this distribution, and finite sample inference is carried out under the presumption that one draws a random sample of size $n$ from $\sP$.
I denote by $\sQ$ the joint distribution whose features the researcher wants to learn.
If $\sQ$ were identified given the observed data (e.g., if it were a marginal of $\sP$), point identification of the parameter or functional of interest would attain.
I denote by $\sR$ the joint distribution of all variables, observable and unobservable ones; both $\sP$ and $\sQ$ can be obtained from it.
In the context of structural models, I denote by $\sM$ a distribution for the observable variables that is consistent with the model.
I note that model incompleteness typically implies that $\sM$ is not unique.
I let $\idr{\cdot}$ denote the sharp identification region of the functional in square brackets, and $\outr{\cdot}$ an outer region.
In both cases, the regions are indexed by $\sP$, because they depend on the distribution of the observed data.

\section{Partial Identification of Probability Distributions}
\label{sec:prob:distr}
The literature reviewed in this chapter starts with the analysis of what can be learned about functionals of probability distributions that are well-defined in the absence of a model.
The approach is nonparametric, and it is typically \emph{constructive}, in the sense that it leads to ``plug-in" formulae for the bounds on the functionals of interest.

\subsection{Selectively Observed Data}
\label{subsec:missing_data}
As in \cite{man89}, suppose that a researcher is interested in learning the probability that an individual who is homeless at a given date has a home six months later.
Here the population of interest is the people who are homeless at the initial date, and the outcome of interest $\ey$ is an indicator of whether the individual has a home six months later (so that $\ey=1$) or remains homeless (so that $\ey=0$).
A random sample of homeless individuals is interviewed at the initial date, so that individual background attributes $\ex$ are observed, but six months later only a subset of the individuals originally sampled can be located.
In other words, attrition from the sample creates a \emph{selection problem} whereby $\ey$ is observed only for a subset of the population.
Let $\ed$ be an indicator of whether the individual can be located (hence $\ed=1$) or not (hence $\ed=0$).
The question is what can the researcher learn about $\E_\sQ(\ey|\ex=x)$, with $\sQ$ the distribution of $(\ey,\ex)$?
\cite{man89} showed that $\E_\sQ(\ey|\ex=x)$ is not point identified in the absence of additional assumptions, but informative nonparametric bounds on this quantity can be obtained.
In this section I review his approach, and discuss several important extensions of his original idea.\medskip

Throughout the chapter, I formally state the structure of the problem under study as an ``Identification Problem", and then provide a solution, either in the form of a sharp identification region, or of an outer region.
To set the stage, and at the cost of some repetition, I do the same here, slightly generalizing the question stated in the previous paragraph.
\begin{IP}[Conditional Expectation of Selectively Observed Data] 
\label{IP:bounds:mean:md}
Let $\ey \in \mathcal{Y}\subset \R$ and $\ex \in \mathcal{X}\subset \R^d$ be, respectively, an outcome variable and a vector of covariates with support $\cY$ and $\cX$ respectively, with $\cY$ a compact set. 
Let $\ed \in \{0,1\}$. 
Suppose that the researcher observes a random sample of realizations of $(\ex,\ed)$ and, in addition, observes the realization of $\ey$ when $\ed=1$.
Hence, the observed data is $(\ey\ed,\ed,\ex)\sim \sP$.
Let $g:\cY\mapsto\R$ be a measurable function that attains its lower and upper bounds $g_0=\min_{y\in\cY}g(y)$ and $g_1=\max_{y\in\cY}g(y)$, and assume that $-\infty < g_0 < g_1 < \infty$.
Let $y_{j}\in\cY$ be such that $g(y_j)=g_j$, $j=0,1$.\footnote{The bounds $g_0,g_1$ and the values $y_0,y_1$ at which they are attained may differ for different functions $g(\cdot)$.}
In the absence of additional information, what can the researcher learn about $\E_\sQ(g(\ey)|\ex=x)$, with $\sQ$ the distribution of $(\ey,\ex)$?
	\qedex
\end{IP}
\citeauthor{man89}'s analysis of this problem begins with a simple application of the law of total probability, that yields
\begin{align}
\sQ(\ey|\ex=x) = \sP(\ey|\ex=x,\ed=1)\sP(\ed=1|\ex=x)+\sR(\ey|\ex=x,\ed=0)\sP(\ed=0|\ex=x).\label{eq:LTP_md}
\end{align}
Equation \eqref{eq:LTP_md} lends a simple but powerful anatomy of the selection problem. 
While $\sP(\ey|\ex=x,\ed=1)$ and $\sP(\ed|\ex=x)$ can be learned from the observable distribution $\sP(\ey\ed,\ed,\ex)$, under the maintained assumptions the sampling process reveals nothing about $\sR(\ey|\ex=x,\ed=0)$. 
Hence, $\sQ(\ey|\ex=x)$ is not point identified.

If one were to assume \emph{exogenous selection} (or data missing at random conditional on $\ex$), i.e., $\sR(\ey|\ex,\ed=0)=\sP(\ey|\ex,\ed=1)$, point identification would obtain. 
However, that assumption is non-refutable and it is well known that it may fail in applications.\footnote{Section \ref{sec:misspec} discusses the consequences of model misspecification (with respect to refutable assumptions).}
Let $\cT$ denote the space of all probability measures with support in $\cY$. The unknown functional vector is $\{\tau(x),\upsilon(x)\}\equiv \{\sQ(\ey|\ex=x),\sR(\ey|\ex=x,\ed=0)\}$.
What the researcher can learn, in the absence of additional restrictions on $\sR(\ey|\ex=x,\ed=0)$, is the region of \emph{observationally equivalent} distributions for $\ey|\ex=x$, and the associated set of expectations taken with respect to these distributions. 
\begin{SIR}[Conditional Expectations of Selectively Observed Data]
\label{SIR:prob:E:md}
Under the assumptions in Identification Problem \ref{IP:bounds:mean:md},
		\begin{multline}
	\idr{\E_\sQ(g(\ey)|\ex=x)} = \Big[\E_\sP(g(\ey)|\ex=x,\ed=1)\sP(\ed=1|\ex=x)+ g_0P(\ed=0|\ex=x),\\
	\E_\sP(g(\ey)|\ex=x,\ed=1)\sP(\ed=1|\ex=x)+ g_1\sP(\ed=0|\ex=x)\Big]\label{eq:bounds:mean:md}
	\end{multline}
is the sharp identification region for $\E_\sQ(g(\ey)|\ex=x)$.
\end{SIR}
\begin{proof}
Due to the discussion following equation \eqref{eq:LTP_md}, the collection of observationally equivalent distribution functions for $\ey|\ex=x$ is
	\begin{multline}
	\idr{\sQ(\ey|\ex=x)}=\Big\{ \tau(x) \in \cT: \tau(x) = \sP(\ey|\ex=x,\ed=1)\sP(\ed=1|\ex=x)\\
	+\upsilon(x)\sP(\ed=0|\ex=x),~\text{for some } \upsilon(x)\in\cT\Big\}.\label{eq:Tau_md}
	\end{multline}
Next, observe that the lower bound in equation \eqref{eq:bounds:mean:md} is achieved by integrating $g(\ey)$ against the distribution $\tau(x)$ that results when $\upsilon(x)$ places probability one on $y_0$. The upper bound is achieved by integrating $g(\ey)$ against the distribution $\tau(x)$ that results when $\upsilon(x)$ places probability one on $y_1$.
	Both are contained in the set $\idr{\sQ(\ey|\ex=x)}$ in equation \eqref{eq:Tau_md}. 
\end{proof}	
These are the \emph{worst case bounds}, so called because assumptions free and therefore representing the widest possible range of values for the parameter of interest that are consistent with the observed data.
A simple ``plug-in" estimator for $\idr{\E_\sQ(g(\ey)|\ex=x)}$ replaces all unknown quantities in \eqref{eq:bounds:mean:md} with consistent estimators, obtained, e.g., by kernel or sieve regression. 
I return to consistent estimation of partially identified parameters in Section \ref{sec:inference}.
Here I emphasize that identification problems are fundamentally distinct from finite sample inference problems.
The latter are typically reduced as sample size increase (because, e.g., the variance of the estimator becomes smaller).
The former do not improve, unless a different and better type of data is collected, e.g. with a smaller prevalence of missing data \citep[see][for a discussion]{dom:man17}.

	\cite[Section 1.3]{man03} shows that the proof of Theorem SIR-\ref{SIR:prob:E:md} can be extended to obtain the smallest and largest points in the sharp identification region of any parameter that respects stochastic dominance.\footnote{
Recall that a probability distribution $\sF\in\cT$ stochastically dominates $\sF^\prime\in\cT$ if $\sF(-\infty,t]\le \sF^\prime(-\infty,t]$ for all $t\in\R$. A real-valued functional $\sd:\cT\to\R$ respects stochastic dominance if $\sd(\sF)\ge \sd(\sF^\prime)$ whenever $\sF$ stochastically dominates $\sF^\prime$.}
This is especially useful to bound the quantiles of $\ey|\ex=x$. 
For any given $\alpha \in (0,1)$, let $\sq_{\sP}^{g(\ey)}(\alpha,1,x)\equiv \left\{\min t:\sP(g(\ey)\le t|\ed=1,\ex=x)\ge \alpha\right\}$.
Then the smallest and largest  admissible values for the $\alpha$-quantile of $g(\ey)|\ex=x$ are, respectively,
\begin{align*}
r(\alpha,x)&\equiv 
\begin{cases}
\sq_{\sP}^{g(\ey)}\left(\left[1-\frac{(1-\alpha)}{\sP(\ed=1|\ex=x)}\right],1,x\right)
 & \text{if } \sP(\ed=1|\ex=x)>1-\alpha,\\
g_0 & \text{otherwise};
\end{cases}\\
s(\alpha,x)&\equiv 
\begin{cases}
\sq_{\sP}^{g(\ey)}\left(\left[\frac{\alpha}{\sP(\ed=1|\ex=x)}\right],1,x\right) & \text{if } \sP(\ed=1|\ex=x)\ge\alpha,\\
g_1 & \text{otherwise}.
\end{cases}
\end{align*}
The lower bound on $\E_\sQ(g(\ey)|\ex=x)$ is informative only if $g_0>-\infty$, and the upper bound is informative only if $g_1<\infty$. 
By comparison, for any value of $\alpha$, $r(\alpha,x)$ and $s(\alpha,x)$ are
generically informative if, respectively, $\sP(\ed=1|\ex=x) > 1-\alpha$ and $\sP(\ed=1|\ex=x) \ge \alpha$, regardless of the range of $g$.

\cite{sto10} further extends partial identification analysis to the study of spread parameters in the presence of missing data (as well as interval data, data combinations, and other applications).
These parameters include ones that respect second order stochastic dominance, such as the variance, the Gini coefficient, and other inequality measures, as well as other measures of dispersion which do not respect second order stochastic dominance, such as interquartile range and ratio.\footnote{
Earlier related work includes, e.g., \cite{gas72} and \cite{cow91}, who obtain worst case bounds on the sample Gini coefficient under the assumption that one knows the income bracket but not the exact income of every household.}
\citeauthor{sto10} shows that the sharp identification region for these parameters can be obtained by fixing the mean or quantile of the variable of interest at a specific value within its sharp identification region, and deriving a distribution consistent with this value which is ``compressed" with respect to the ones which bound the cumulative distribution function (CDF) of the variable of interest, and one which is ``dispersed" with respect to them.  
Heuristically, the compressed distribution minimizes spread, while the dispersed one maximizes it (the sense in which this optimization occurs is formally defined in the paper).  
The intuition for this is that a compressed CDF is first below and then above any non-compressed one; a dispersed CDF is first above and then below any non-dispersed one.  
Second-stage optimization over the possible values of the mean or the quantile delivers unconstrained bounds.  
The main results of the paper are sharp identification regions for the expectation and variance, for the median and interquartile ratio, and for many other combinations of parameters.

\begin{BI}[Identification is not a binary event]
\label{big_idea:id_not_binary}
Identification Problem \ref{IP:bounds:mean:md} is mathematically simple, but it puts forward a new approach to empirical research.
The traditional approach aims at finding a sufficient (possibly minimal) set of assumptions guaranteeing point identification of parameters, viewing identification as an ``all or nothing" notion, where either the functional of interest can be learned exactly or nothing of value can be learned. 
The partial identification approach pioneered by \cite{man89} points out that much can be learned from combination of data and assumptions that restrict the functionals of interest to a set of observationally equivalent values, even if this set is not a singleton. 
Along the way, \cite{man89} points out that in Identification Problem \ref{IP:bounds:mean:md} the observed outcome is the singleton $\ey$ when $\ed=1$, and the set $\cY$ when $\ed=0$.
This is a random closed set, see Definition \ref{def:rcs}. 
I return to this connection in Section \ref{subsec:interval_data}.
\end{BI}

Despite how transparent the framework in Identification Problem  \ref{IP:bounds:mean:md} is, important subtleties arise even in this seemingly simple context.
For a given $t\in\R$, consider the function $g(\ey)=\one(\ey\le t)$, with $\one(A)$ the indicator function taking the value one if the logical condition in parentheses holds and zero otherwise.
Then equation \eqref{eq:bounds:mean:md} yields \emph{pointwise-sharp} bounds on the CDF of $\ey$ at any fixed $t\in\R$:
\begin{multline}
\label{eq:pointwise_bounds_F_md}
	\idr{\sQ(\ey\le t|\ex=x)} = \left[\sP(\ey\le t|\ex=x,\ed=1)\sP(\ed=1|\ex=x)\right.,\\
	\left.\sP(\ey\le t|\ex=x,\ed=1)\sP(\ed=1|\ex=x)+ \sP(\ed=0|\ex=x)\right].
\end{multline}
Yet, the collection of CDFs that belong to the band defined by \eqref{eq:pointwise_bounds_F_md} is \emph{not} the sharp identification region for the CDF of $\ey|\ex=x$. Rather, it constitutes an \emph{outer region}, as originally pointed out by \cite[p. 149 and note 2]{man94}.
\begin{OR}[Cumulative Distribution Function of Selectively Observed Data]
\label{OR:CDF_md}
Let $\cC$ denote the collection of cumulative distribution functions on $\cY$. 
Then, under the assumptions in Identification Problem \ref{IP:bounds:mean:md},
\begin{multline}
\label{eq:outer_cdf_md}
\outr{\sF(\ey|\ex=x)}=\left\{\sF\in\cC:~\sP(\ey\le t|\ex=x,\ed=1)\sP(\ed=1|\ex=x)\le\sF(t|x)\le \right. \\
	\left.\sP(\ey\le t|\ex=x,\ed=1)\sP(\ed=1|\ex=x)+ \sP(\ed=0|\ex=x)~\forall t\in\R \right\}
\end{multline} 
is an outer region for the CDF of $\ey|\ex=x$.
\end{OR}
\begin{proof}
Any admissible CDF for $\ey|\ex=x$ belongs to the family of functions in equation \eqref{eq:outer_cdf_md}. However, the bound in equation \eqref{eq:outer_cdf_md} does not impose the restriction that for any $t_0\le t_1$,
\begin{align}
\label{eq:CDF_md_Kinterval}
\sQ(t_0\le\ey\le t_1|\ex=x)\ge \sP(t_0\le\ey\le t_1|\ex=x,\ed=1)\sP(\ed=1|\ex=x).
\end{align}
This restriction is implied by the maintained assumptions, but is not necessarily satisfied by all CDFs in $\outr{\sF(\ey|\ex=x)}$, as illustrated in the following simple example.
\end{proof}
\begin{figure}[tp]
\centering
\includegraphics[scale=1]{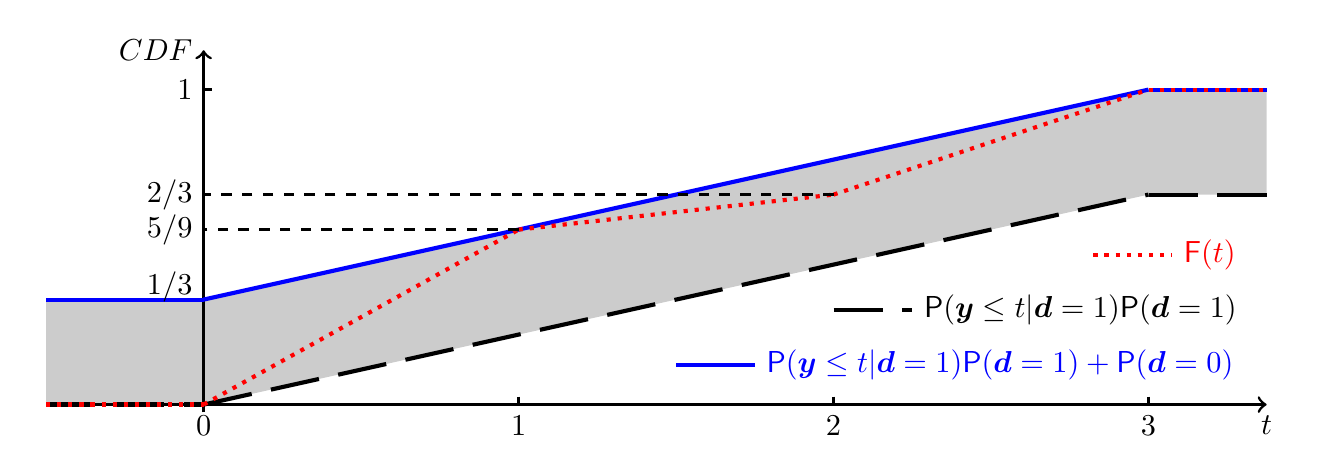}
\caption{\small{The tube defined by inequalities \eqref{eq:pointwise_bounds_F_md} in the set-up of Example \ref{example:CDF_md}, and the CDF in \eqref{eq:CDF_counterexample_md}.}}
	\label{fig:boundsCDF:md}
		\end{figure}

\begin{examp}
\label{example:CDF_md}
Omit $\ex$ for simplicity, let $\sP(\ed=1)=\frac{2}{3}$, and let
\begin{align*}
\sP(\ey\le t|\ed=1)\left\{
\begin{tabular}{ll}
$0$ & if $t<0$,\\
$\frac{1}{3}t$ & if $0\le t<3$,\\
$1$ & if $t\ge 3$.
\end{tabular}
\right.
\end{align*}
The bounding functions and associated tube from the inequalities in \eqref{eq:pointwise_bounds_F_md} are depicted in Figure \ref{fig:boundsCDF:md}.
Consider the cumulative distribution function
\begin{align}
\label{eq:CDF_counterexample_md}
\sF(t)=
\left\{
\begin{tabular}{lll}
$0$ & if & $t<0$,\\
$\frac{5}{9}t$ & if & $0\le t <1$,\\
$\frac{1}{9}t+\frac{4}{9}$ & if & $1\le t <2$,\\
$\frac{1}{3}t$ & if & $2\le t <3$,\\
$1$ & if & $t\ge 3$.
\end{tabular}
\right.
\end{align}
For each $t\in\R$, $\sF(t)$ lies in the tube defined by equation \eqref{eq:pointwise_bounds_F_md}. 
However, it cannot be the CDF of $\ey$, because $\sF(2)-\sF(1)=\frac{1}{9}<\sP(1\le\ey\le 2|\ed=1)\sP(\ed=1)$, directly contradicting equation \eqref{eq:CDF_md_Kinterval}.
	\qedex
	\end{examp}

How can one characterize the sharp identification region for the CDF of $\ey|\ex=x$ under the assumptions in Identification Problem \ref{IP:bounds:mean:md}? 
In general, there is not a single answer to this question: different methodologies can be used. 
Here I use results in \cite[Corollary 1.3.1]{man03} and \cite[Theorem 2.25]{mol:mol18}, which yield an alternative characterization of $\idr{\sQ(\ey|\ex=x)}$ that translates directly into a characterization of $\idr{\sF(\ey|\ex=x)}$.\footnote{Whereas \cite{man94} is  very clear that the collection of CDFs in \eqref{eq:pointwise_bounds_F_md} is an outer region for the CDF of $\ey|\ex=x$, and  \cite{man03} provides the sharp characterization in \eqref{eq:sharp_id_P_md_Manski}, \cite[p. 39]{man07a} does not state all the requirements that characterize $\idr{\sF(\ey|\ex=x)}$.}
\begin{SIR}[Conditional Distribution and CDF of Selectively Observed Data]
\label{SIR:CDF_md}
Given $\tau\in\cT$, let $\tau_K(x)$ denote the probability that distribution $\tau$ assigns to set $K$ conditional on $\ex=x$, with $\tau_y(x)\equiv\tau_{\{y\}}(x)$.
Under the assumptions in Identification Problem \ref{IP:bounds:mean:md},
\begin{align}
	\idr{\sQ(\ey|\ex=x)}&=\Big\{\tau(x) \in \cT:  \tau_K(x) \ge \sP(\ey\in K|\ex=x,\ed=1)\sP(\ed=1|\ex=x),\,\forall K\subset \cY \Big\},\label{eq:sharp_id_P_md_Manski}
\end{align}
where $K$ is measurable.
If $\cY$ is countable,
\begin{align}
	\idr{\sQ(\ey|\ex=x)}&=\Big\{\tau(x) \in \cT:  \tau_y(x) \ge \sP(\ey=y|\ex=x,\ed=1)\sP(\ed=1|\ex=x),\,\forall y\in \cY \Big\}.\label{eq:sharp_id_P_md_discrete}
\end{align}
If $\cY$ is a bounded interval,
\begin{multline}
	\idr{\sQ(\ey|\ex=x)}=\Big\{\tau(x) \in \cT:  \tau_{[t_0,t_1]}(x) \ge \\
	\sP(t_0\le\ey\le t_1|\ex=x,\ed=1)\sP(\ed=1|\ex=x),\,\forall t_0\le t_1, t_0,t_1\in\cY \Big\}.\label{eq:sharp_id_P_md_interval}
\end{multline}
\end{SIR}
\begin{proof}
The characterization in \eqref{eq:sharp_id_P_md_Manski} follows from equation \eqref{eq:Tau_md}, observing that if $\tau(x)\in\idr{\sQ(\ey|\ex=x)}$ as defined in equation \eqref{eq:Tau_md}, then there exists a distribution $\upsilon(x)\in\cT$ such that $\tau(x) = \sP(\ey|\ex=x,\ed=1)\sP(\ed=1|\ex=x)+\upsilon(x)\sP(\ed=0|\ex=x)$.
Hence, by construction $\tau_K(x) \ge \sP(\ey\in K|\ex=x,\ed=1)\sP(\ed=1|\ex=x)$, $\forall K\subset \cY$. Conversely, if one has $\tau_K(x) \ge \sP(\ey\in K|\ex=x,\ed=1)\sP(\ed=1|\ex=x)$, $\forall K\subset \cY$, one can define $\upsilon(x)=\frac{\tau(x) - \sP(\ey|\ex=x,\ed=1)\sP(\ed=1|\ex=x)}{\sP(\ed=0|\ex=x)}$. 
The resulting $\upsilon(x)$ is a probability measure, and hence $\tau(x)\in\idr{\sQ(\ey|\ex=x)}$ as defined in equation \eqref{eq:Tau_md}.
When $\cY$ is countable, if $\tau_y(x) \ge \sP(\ey=y|\ex=x,\ed=1)\sP(\ed=1|\ex=x)$ it follows that for any $K\subset\cY$,
\begin{multline}
\tau_K(x)=\sum_{y\in K}\tau_y(x) \ge \sum_{y\in K}\sP(\ey=y|\ex=x,\ed=1)\sP(\ed=1|\ex=x)\\
=\sP(\ey\in K|\ex=x,\ed=1)\sP(\ed=1|\ex=x).\notag
\end{multline}
The result in equation \eqref{eq:sharp_id_P_md_interval} is proven in \cite[Theorem 2.25]{mol:mol18} using elements of random set theory, to which I return in Section \ref{subsec:interval_data}. 
Using elements of random set theory it is also possible to show that the characterization in \eqref{eq:sharp_id_P_md_Manski} requires only to check the inequalities for $K$ the compact subsets of $\cY$.
\end{proof}

This section provides sharp identification regions and outer regions for a variety of functionals of interest.
The computational complexity of these characterizations varies widely.
Sharp bounds on parameters that respect stochastic dominance only require computing the parameters with respect to two probability distributions.
An outer region on the CDF can be obtained by evaluating all tail probabilities of a certain distribution.
A sharp identification region on the CDF requires evaluating the probability  that a certain distribution assigns to all intervals.
I return to computational challenges in partial identification in Section \ref{sec:computations}.

\subsection{Treatment Effects with and without Instrumental Variables}
\label{subsec:programme:eval}
The discussion of partial identification of probability distributions of selectively observed data naturally leads to the question of its implications for program evaluation.
The literature on program evaluation is vast.
The purpose of this section is exclusively to show how the ideas presented in Section \ref{subsec:missing_data} can be applied to learn features of treatment effects of interest, when no assumptions are imposed on treatment selection and outcomes.
I also provide examples of assumptions that can be used to tighten the bounds. 
To keep this chapter to a manageable length, I discuss only partial identification of the average response to a treatment and of the average treatment effect (ATE).
There are many different parameters that received much interest in the literature.
Examples include the \emph{local average treatment effect} of \cite{imb:ang94} and the \emph{marginal treatment effect} of 
\cite{hec:vyt99,hec:vyt01,hec:vyt05}. 
For thorough discussions of the literature on program evaluation, I refer to the textbook treatments in \cite{man95,man03,man07a} and \cite{imb:rub15}, to the Handbook chapters by \cite{hec:vyt07I,hec:vyt07II} and \cite{abb:hec07}, and to the review articles by \cite{imb:woo09} and \cite{mog:tor18}.\medskip

Using standard notation \citep[e.g.,][]{ney23}, let $\ey:\T \mapsto \cY$ be an individual-specific response function, with $\T=\{0,1,\dots,T\}$ a finite set of mutually exclusive and exhaustive treatments, and let $\es$ denote the individual's received treatment (taking its realizations in $\T$).\footnote{Here the treatment response is a function only of the (scalar) treatment received by the given individual, an assumption known as \emph{stable unit treatment value assumption} \citep{rub78}.}
The researcher observes data $(\ey,\es,\ex)\sim\sP$, with $\ey\equiv\ey(\es)$ the outcome corresponding to the received treatment $\es$, and $\ex$ a vector of covariates. 
The outcome $\ey(t)$ for $\es\neq t$ is counterfactual, and hence can be conceptualized as missing. 
Therefore, we are in the framework of Identification Problem \ref{IP:bounds:mean:md} and all the results from Section \ref{subsec:missing_data} apply in this context too, subject to adjustments in notation.\footnote{\cite{ber:mol:mol12} and \cite[Section 2.5]{mol:mol18} provide a characterization of the sharp identification region for the joint distribution of $[\ey(t),t\in\T]$.}
For example, using Theorem SIR-\ref{SIR:prob:E:md},
\begin{align}
\idr{\E_\sQ(\ey(t)|\ex=x)}= \Big[\E_\sP&(\ey|\ex=x,\es=t)\sP(\es=t|\ex=x)+ y_0P(\es\neq t|\ex=x),\notag\\
	&\E_\sP(\ey|\ex=x,\es=t)\sP(\es=t|\ex=x)+ y_1P(\es\neq t|\ex=x)\Big],\label{eq:WCB:treat}
\end{align}
where $y_0\equiv\inf_{y\in\cY}y$, $y_1\equiv\sup_{y\in\cY}y$.
If $y_0<\infty$ and/or $y_1<\infty$, these \emph{worst case bounds} are informative.
When both are infinite, the data is uninformative in the absence of additional restrictions.

If the researcher is interested in an Average Treatment Effect (ATE), e.g.
\begin{multline*}
\E_\sQ(\ey(t_1)|\ex=x)-\E_\sQ(\ey(t_0)|\ex=x)=\\
\E_\sP(\ey|\ex=x,\es=t_1)\sP(\es=t_1|\ex=x)+\E_\sQ(\ey(t_1)|\ex=x,\es\neq t_1)\sP(\es\neq t_1|\ex=x)\\
-\E_\sP(\ey|\ex=x,\es=t_0)\sP(\es=t_0|\ex=x)-\E_\sQ(\ey(t_0)|\ex=x,\es\neq t_0)\sP(\es\neq t_0|\ex=x),
\end{multline*} 
with $t_0,t_1\in\T$, sharp worst case bounds on this quantity can be obtained as follows.
First, observe that the empirical evidence reveals $\E_\sP(\ey|\ex=x,\es=t_j)$ and $\sP(\es|\ex=x)$, but is uninformative about $\E_\sQ(\ey(t_j)|\ex=x,\es\neq t_j)$, $j=0,1$.
Each of the latter quantities (the expectations of $\ey(t_0)$ and $\ey(t_1)$ conditional on different realizations of $\es$ and $\ex=x$) can take any value in $[y_0,y_1]$.
Hence, the sharp lower bound on the ATE is obtained by subtracting the upper bound on $\E_\sQ(\ey(t_0)|\ex=x)$ from the lower bound on $\E_\sQ(\ey(t_1)|\ex=x)$.
The sharp upper bound on the ATE is obtained by subtracting the lower bound on $\E_\sQ(\ey(t_0)|\ex=x)$ from the upper bound on $\E_\sQ(\ey(t_1)|\ex=x)$.
The resulting bounds have width equal to $(y_1-y_0)[2-\sP(\es=t_1|\ex=x)-\sP(\es=t_0|\ex=x)]\in[(y_1-y_0),2(y_1-y_0)]$, and hence are informative only if both $y_0>-\infty$ and $y_1<\infty$.
As the largest logically possible value for the ATE (in the absence of information from data) cannot be larger than $(y_1-y_0)$, and the smallest cannot be smaller than $-(y_1-y_0)$, the sharp bounds on the ATE always cover zero.

\begin{BI}
How should one think about the finding on the size of the worst case bounds on the ATE?
On the one hand, if both $y_0<\infty$ and $y_1<\infty$ the bounds are informative, because they are a strict subset of the ATE's possible realizations.
On the other hand, they reveal that the data alone are silent on the sign of the ATE.
This means that assumptions play a crucial role in delivering stronger conclusions about this policy relevant parameter.
The partial identification approach to empirical research recommends that as assumptions are added to the analysis, one systematically reports how each contributes to shrinking the bounds, making transparent their role in shaping inference. 
\end{BI}
What assumptions may researchers bring to bear to learn more about treatment effects of interest?
The literature has provided a wide array of well motivated and useful restrictions.
Here I consider two examples.
The first one entails \emph{shape restrictions} on the treatment response function, leaving selection unrestricted.
\cite{man97:monotone} obtains bounds on treatment effects under the assumption that the response functions are monotone, semi-monotone, or concave-monotone.
These restrictions are motivated by economic theory, where it is commonly presumed, e.g., that demand functions are downward sloping and supply functions are upward sloping.
Let the set $\T$ be ordered in terms of degree of intensity. 
Then \citeauthor{man97:monotone}'s \emph{monotone treatment response} assumption requires that
\begin{align*}
t_1\ge t_0 \Rightarrow \sQ(\ey(t_1)\ge\ey(t_0))=1~~\forall t_0,t_1\in\T.
\end{align*}
Under this assumption, one has a sharp characterization of what can be learned about $\ey(t)$: 
\begin{align}
  \ey(t)\in
  \begin{cases}
    (-\infty,\ey]\cap\cY & \text{if } t<\es,\\
    \{\ey\} & \text{if } t=\es,\\
    [\ey,\infty)\cap\cY & \text{if } t>\es.
  \end{cases}\label{eq:RCS:MTR}
\end{align}
Hence, the sharp bounds on $\E_\sQ(\ey(t)|\ex=x)$ are \citep[Proposition M1]{man97:monotone}
\begin{align}
\idr{\E_\sQ(\ey(t)|\ex=x)}= \Big[\E_\sP&(\ey|\ex=x,\es\le t)\sP(\es\le t|\ex=x)+ y_0P(\es> t|\ex=x),\notag\\
	&\E_\sP(\ey|\ex=x,\es\ge t)\sP(\es \ge t|\ex=x)+ y_1P(\es< t|\ex=x)\Big].\label{eq:MTR:treat}
\end{align}
This finding highlights some important facts.
Under the monotone treatment response assumption, the bounds on $\E_\sQ(\ey(t)|\ex=x)$ are obtained using information from all $(\ey,\es)$ pairs (given $\ex=x$), while the bounds in \eqref{eq:WCB:treat} only use the information provided by $(\ey,\es)$ pairs for which $\es=t$ (given $\ex=x$).
As a consequence, the bounds in \eqref{eq:MTR:treat} are informative even if $\sP(\es= t|\ex=x)=0$, whereas the worst case bounds are not.

Concerning the ATE with $t_1>t_0$, under monotone treatment response its lower bound is zero, and its upper bound is obtained by subtracting the lower bound on $\E_\sQ(\ey(t_0)|\ex=x)$ from the upper bound on $\E_\sQ(\ey(t_1)|\ex=x)$, where both bounds are obtained as in \eqref{eq:MTR:treat} \citep[Proposition M2]{man97:monotone}.

The second example of assumptions used to tighten worst case bounds is that of \emph{exclusion restrictions}, as in, e.g., \cite{man90}.
Suppose the researcher observes a random variable $\ez$, taking its realizations in $\cZ$, such that\footnote{Stronger exclusion restrictions include statistical independence of the response function at each $t$ with $\ez$: $\sQ(\ey(t)|\ez,\ex)=\sQ(\ey(t)|\ex)~\forall t \in\T,~\ex$-a.s.; and statistical independence of the entire response function with $\ez$: $\sQ([\ey(t),t \in\T]|\ez,\ex)=\sQ([\ey(t),t \in\T]|\ex),~\ex$-a.s.
Examples of partial identification analysis under these conditions can be found in \cite{bal:pea97}, \cite{man03}, \cite{kit09}, \cite{ber:mol:mol12}, \cite{mac:sha:vyt18}, and many others.}
\begin{align}
\E_\sQ(\ey(t)|\ez,\ex)=\E_\sQ(\ey(t)|\ex)~~\forall t \in\T,~\ex\text{-a.s.}.\label{eq:ass:MI}
\end{align}
This assumption is treatment-specific, and requires that the treatment response to $t$ is mean independent with $\ez$.
It is easy to show that under the assumption in \eqref{eq:ass:MI}, the bounds on $\E_\sQ(\ey(t)|\ex=x)$ become
\begin{multline}
\idr{\E_\sQ(\ey(t)|\ex=x)}=\Big[\mathrm{ess}\sup_\ez\E_\sP(\ey|\ex=x,\es=t,\ez)\sP(\es=t|\ex=x,\ez)+ y_0P(\es\neq t|\ex=x,\ez),\\
	\mathrm{ess}\inf_\ez\E_\sP(\ey|\ex=x,\es=t,\ez)\sP(\es=t|\ex=x,\ez)+ y_1P(\es\neq t|\ex=x,\ez)\Big].\label{eq:intersection:bounds}
\end{multline}
These are called \emph{intersection bounds} because they are obtained as follows.
Given $\ex$ and $\ez$, one uses \eqref{eq:WCB:treat} to obtain sharp bounds on $\E_\sQ(\ey(t)|\ez=z,\ex=x)$.
Due to the mean independence assumption in \eqref{eq:ass:MI}, $\E_\sQ(\ey(t)|\ex=x)$ must belong to each of these bounds $\ez$-a.s., hence to their intersection.
The expression in \eqref{eq:intersection:bounds} follows.
If the instrument affects the probability of being selected into treatment, or the average outcome for the subpopulation receiving treatment $t$, the bounds on $\E_\sQ(\ey(t)|\ex=x)$ shrink.
If the bounds are empty, the mean independence assumption can be refuted (see Section \ref{sec:misspec} for a discussion of misspecification in partial identification).
\cite{man:pep00,man:pep09} generalize the notion of instrumental variable to \emph{monotone} instrumental variable, and show how these can be used to obtain tighter bounds on treatment effect parameters.\footnote{See \cite[Chapter XXX in this Volume]{che:ros19} for further discussion.}
They also show how shape restrictions and exclusion restrictions can jointly further tighten the bounds.
\cite{man13social} generalizes these findings to the case where treatment response may have social interactions -- that is, each individual's outcome depends on the treatment received by all other individuals.

\subsection{Interval Data}
\label{subsec:interval_data}
Identification Problem \ref{IP:bounds:mean:md}, as well as the treatment evaluation problem in Section \ref{subsec:programme:eval}, is an instance of the more general question of what can be learned about (functionals of) probability distributions of interest, in the presence of interval valued outcome and/or covariate data.
Such data have become commonplace in Economics. 
For example, since the early 1990s the Health and Retirement Study collects income data from survey respondents in the form of brackets, with degenerate (singleton) intervals for individuals who opt to fully reveal their income \citep[see, e.g.,][]{jus:suz95}.
Due to concerns for privacy, public use tax data are recorded as the number of tax payers which belong to each of a finite number of cells \citep[see, e.g.,][]{pic05}.
The Occupational Employment Statistics (OES) program at the Bureau of Labor Statistics \citep{BLS} collects wage data from employers as intervals, and uses these data to construct estimates
for wage and salary workers in more than 800 detailed occupations.
\cite{man:mol10} and \cite{giu:man:mol19round} document the extensive prevalence of rounding in survey responses to probabilistic expectation questions, and propose to use a person's response pattern across different  questions to infer his rounding practice, the result being interpretation of reported numerical values as interval data.
Other instances abound.
Here I focus first on the case of interval outcome data.

\begin{IP}[Interval Outcome Data] 
\label{IP:interval_outcome} Assume that in addition to being compact,  either $\cY$ is countable or $\cY=[y_0,y_1]$, with $y_0=\min_{y\in\cY}y$ and $y_1=\max_{y\in\cY}y$. 
Let $(\yL,\yU,\ex)\sim\sP$ be observable random variables and $\ey$ be an unobservable random variable whose distribution (or features thereof) is of interest, with $\yL,\yU,\ey\in\cY$. 
Suppose that $(\yL,\yU,\ey)$ are such that $\sR(\yL\le\ey\le\yU)=1$.\footnote{\label{fn:missing_special_case_interval}In Identification Problem \ref{IP:bounds:mean:md} the observable variables are $(\ey\ed,\ed,\ex)$, and $(\yL,\yU)$ are determined as follows: $\yL=\ey\ed+y_0(1-\ed)$, $\yU=\ey\ed+y_1(1-\ed)$. For the analysis in Section \ref{subsec:programme:eval}, the data is $(\ey,\es,\ex)$ and $\yL=\ey\one(\es=t)+y_0\one(\es\ne t)$, $\yU=\ey\one(\es=t)+y_1\one(\es\ne t)$.
Hence, $\sP(\yL\le\ey\le\yU)=1$ by construction.}
In the absence of additional information, what can the researcher learn about features of $\sQ(\ey|\ex=x)$, the conditional distribution of $\ey$ given $\ex=x$?
	\qedex
\end{IP}
It is immediate to obtain the sharp identification region
\begin{align*}
\idr{\E_\sQ(\ey|\ex=x)} = \left[\E_\sP(\yL|\ex=x),\E_\sP(\yU|\ex=x)\right].
\end{align*}
As in the previous section, it is also easy to obtain sharp bounds on parameters that respect stochastic dominance, and pointwise-sharp bounds on the CDF of $\ey$ at any fixed $t\in\R$:
\begin{align}
\label{eq:pointwise_bounds_F}
\sP(\yU\le t|\ex=x)\le\sQ(\ey\le t|\ex=x)\le\sP(\yL\le t|\ex=x).
\end{align}
In this case too, however, as in Theorem OR-\ref{OR:CDF_md}, the tube of CDFs satisfying equation \eqref{eq:pointwise_bounds_F} for all $t\in\R$ is an outer region for the CDF of $\ey|\ex=x$, rather than its sharp identification region.
Indeed, also in this context it is easy to construct examples similar to Example \ref{example:CDF_md}.

How can one characterize the sharp identification region for the probability distribution of $\ey|\ex$ when one observes $(\yL,\yU,\ex)$ and assumes $\sR(\yL\le\ey\le\yU)=1$? 
Again, there is not a single answer to this question.
Depending on the specific problem at hand, e.g., the specifics of the interval data and whether $\ey$ is assumed discrete or continuous, different methods can be applied. 
I use \emph{random set theory} to provide a characterization of $\idr{\sQ(\ey|\ex=x)}$. 
Let
\begin{align*}
\eY\equiv [\yL,\yU]\cap\cY.
\end{align*}
Then $\eY$ is a random closed set according to Definition \ref{def:rcs}.\footnote{For a proof of this statement, see \cite[Example 1.11]{mol:mol18}.} The requirement $\sR(\yL\le\ey\le\yU)=1$ can be equivalently expressed as
\begin{align}
\label{eq:y_in_Y}
\ey\in\eY~~\text{almost surely.}
\end{align}
Equation \eqref{eq:y_in_Y}, together with knowledge of $\sP$, exhausts all the information in the data and maintained assumptions.
In order to harness such information to characterize the set of observationally equivalent probability distributions for $\ey$, one can leverage a result due to \cite{art83} \citep[and][]{nor92}, reported in Theorem \ref{thr:artstein} in Appendix \ref{app:RCS}, which allows one to translate \eqref{eq:y_in_Y} into a collection of conditional moment inequalities. 
Specifically, let $\cT$ denote the space of all probability measures with support in $\cY$. 
\begin{SIR}[Conditional Distribution of Interval-Observed Outcome Data]
\label{SIR:CDF_id}
Given $\tau\in\cT$, let $\tau_K(x)$ denote the probability that distribution $\tau$ assigns to set $K$ conditional on $\ex=x$.
Under the assumptions in Identification Problem \ref{IP:interval_outcome}, the sharp identification region for $\sQ(\ey|\ex=x)$ is
	\begin{align}
	\idr{\sQ(\ey|\ex=x)}=\Big\{\tau(x) \in \cT: \tau_K(x) \ge \sP(\eY\subset K|\ex=x),\,\forall K\subset \cY,\,K\text{ compact} \Big\}\label{eq:sharp_id_P_interval_1}
	\end{align}
	When $\cY=[y_0,y_1]$, equation \eqref{eq:sharp_id_P_interval_1} becomes
	\begin{align}
	\idr{\sQ(\ey|\ex=x)}=\Big\{\tau(x) \in \cT: \tau_{[t_0,t_1]}(x) \ge \sP(\yL\ge t_0,\yU\le t_1|\ex=x),\,\forall t_0\le t_1,\, t_0,t_1 \in \cY \Big\}.\label{eq:sharp_id_P_interval_2}
	\end{align}
	\end{SIR}
	\begin{proof}
Theorem \ref{thr:artstein} yields \eqref{eq:sharp_id_P_interval_1}.
	If $\cY=[y_0,y_1]$, \cite[Theorem 2.25]{mol:mol18} show that it suffices to verify the inequalities in \eqref{eq:sharp_id_P_interval_2} for sets $K$ that are intervals.
	\end{proof}
Compare equation \eqref{eq:sharp_id_P_interval_1} with equation \eqref{eq:sharp_id_P_md_Manski}. 
Under the set-up of Identification Problem \ref{IP:bounds:mean:md}, when $\ed=1$ we have $\eY=\{\ey\}$ and when $\ed=0$ we have $\eY=\cY$.
Hence, for any $K \subsetneq \cY$, $\sP(\eY \subset K|\ex=x)=\sP(\ey\in K|\ex=x,\ed=1)\sP(\ed=1)$.\footnote{For $K = \cY$, both \eqref{eq:sharp_id_P_interval_1} and \eqref{eq:sharp_id_P_md_Manski} hold trivially.}  
It follows that the characterizations in \eqref{eq:sharp_id_P_interval_1} and \eqref{eq:sharp_id_P_md_Manski} are equivalent. 
If $\cY$ is countable, it is easy to show that \eqref{eq:sharp_id_P_interval_1} simplifies to \eqref{eq:sharp_id_P_md_Manski} \citep[see, e.g.,][Proposition 2.2]{ber:mol:mol12}. 
\begin{BI}[Random set theory and partial identification]
\label{big_idea:pi_and_rs}
The mathematical framework for the analysis of random closed sets embodied in random set theory is naturally suited to conduct identification analysis and statistical inference in partially identified models.
This is because, as argued by \cite{ber:mol08} and \cite{ber:mol:mol11,ber:mol:mol12}, lack of point identification can often be traced back to a collection of random variables that are consistent with the available data and maintained assumptions. 
In turn, this collection of random variables is equal to the family of selections of a properly specified random closed set, so that random set theory applies. 
The interval data case is a simple example that illustrates this point.
More examples are given throughout this chapter. 
As mentioned in the Introduction, the exercise of defining the random closed set that is relevant for the problem under consideration is routinely carried out in partial identification analysis, even when random set theory is not applied.
For example, in the case of treatment effect analysis with monotone response function, \cite{man97:monotone} derived the set in the right-hand-side of \eqref{eq:RCS:MTR}, which satisfies Definition \eqref{def:rcs}.
\end{BI}
An attractive feature of the characterization in \eqref{eq:sharp_id_P_interval_1} is that it holds regardless of the specific assumptions on $\yL,\,\yU$, and $\cY$.
Later sections in this chapter illustrate how Theorem \ref{thr:artstein} delivers the sharp identification region in other more complex instances of partial identification of probability distributions, as well as in structural models.
In Chapter \textbf{XXX} in this Volume, \cite{che:ros19} apply Theorem \ref{thr:artstein} to obtain sharp identification regions for functionals of interest in the important class of \emph{generalized instrumental variable models}. To avoid repetitions, I do not systematically discuss that class of models in this chapter.

When addressing questions about features of $\sQ(\ey|\ex=x)$ in the presence of interval outcome data, an alternative approach \citep[e.g.][]{tam10,pon:tam11} looks at all (random) mixtures of $\yL,\yU$.
The approach is based on a random variable $\eu$ (a \emph{selection mechanism} that picks an element of $\eY$) with values in $[0,1]$, whose distribution conditional on $\yL,\yU$ is left completely unspecified. 
Using this random variable, one defines
  \begin{align}
    \ey_\eu=\eu\yL+(1-\eu)\yU.\label{eq:y_s}
  \end{align}
The sharp identification region in Theorem SIR-\ref{SIR:CDF_id} can be characterized as the collection of conditional distributions of all possible random variables $\ey_\eu$ as defined in \eqref{eq:y_s}, given $\ex=x$.
This is because each $\ey_\eu$ is a (stochastic) convex combination of $\yL,\yU$, hence each of these random variables satisfies $\sR(\yL\le\ey_\eu\le\yU)=1$.
While such characterization is sharp, it can be of difficult implementation in practice, because it requires working with all possible random variables $\ey_\eu$ built using all possible random variables $\eu$ with support in $[0,1]$.
Theorem \ref{thr:artstein} allows one to bypass the use of $\eu$, and obtain directly a characterization of the sharp identification region for $\sQ(\ey|\ex=x)$ based on conditional moment inequalities.\footnote{It can be shown that the collection of random variables $\ey_\eu$ equals the collection of \emph{measurable selections} of the random closed set $\eY\equiv [\yL,\yU]$ (see Definition \ref{def:selection}); see \cite[Lemma 2.1]{ber:mol:mol11}.
Theorem \ref{thr:artstein} provides a characterization of the distribution of any $\ey_\eu$ that satisfies $\ey_\eu \in \eY$ a.s., based on a dominance condition that relates the distribution of $\ey_\eu$ to the distribution of the random set $\eY$. 
Such dominance condition is given by the inequalities in \eqref{eq:sharp_id_P_interval_1}.
}

\cite{hor:man98,hor:man00} study nonparametric conditional prediction problems with missing outcome and/or missing covariate data.
Their analysis shows that this problem is considerably more pernicious than the case where only outcome data are missing.
For the case of interval covariate data, \cite{man:tam02} provide a set of sufficient conditions under which simple and elegant sharp bounds on functionals of $\sQ(\ey|\ex)$ can be obtained, even in this substantially harder identification problem.
Their assumptions are listed in Identification Problem \ref{IP:interval_covariate}, and their result (with proof) in Theorem SIR-\ref{SIR:man:tam:nonpar}.
\begin{IP}[Interval Covariate Data] 
\label{IP:interval_covariate} 
Let $(\ey,\xL,\xU)\sim\sP$ be observable random variables in $\R\times\R\times\R$ and $\ex\in\R$ be an unobservable random variable.
Suppose that $\sR$, the joint distribution of $(\ey,\ex,\xL,\xU)$, is such that: (I) $\sR(\xL\le\ex\le\xU)=1$; (M) $\E_\sQ(\ey|\ex=x)$ is weakly increasing in $x$; and (MI) $\E_{\sR}(\ey|\ex,\xL,\xU)=\E_\sQ(\ey|\ex)$.
In the absence of additional information, what can the researcher learn about $\E_\sQ(\ey|\ex=x)$ for given $x\in\cX$?
	\qedex
\end{IP}
Compared to the earlier discussion for the interval outcome case, here there are two additional assumptions. 
The monotonicity condition (M) is a simple shape restrictions, which however requires some prior knowledge about the joint distribution of $(\ey,\ex)$.
The mean independence restriction (MI) requires that if $\ex$ were observed, knowledge of $(\xL,\xU)$ would not affect the conditional expectation of $\ey|\ex$.
The assumption is not innocuous, as pointed out by the authors.
For example, it may fail if censoring is endogenous.\footnote{\label{foot:auc:bug:hot17}For the case of missing covariate data, which is a special case of interval covariate data similarly to arguments in footnote \ref{fn:missing_special_case_interval}, \cite{auc:bug:hot17} show that the MI restriction implies the assumption that data is missing at random.}
\begin{SIR}[Conditional Expectation with Interval-Observed Covariate Data]
\label{SIR:man:tam:nonpar}
Under the assumptions of Identification Problem \ref{IP:interval_covariate}, the sharp identification region for $\E_\sQ(\ey|\ex=x)$ for given $x\in\cX$ is
\begin{align}
\label{eq:man:tam:nonpar}
\idr{\E_\sQ(\ey|\ex=x)}=\left[\sup_{\xU\le x}\E_\sP(\ey|\xL,\xU),\inf_{\xL \ge x}\E_\sP(\ey|\xL,\xU)\right].
\end{align}
\end{SIR}
\begin{proof}
The law of iterated expectations and the independence assumption yield $\E_\sP(\ey|\xL,\xU)=\int \E_\sQ(\ey|\ex)d\sR(\ex|\xL,\xU)$.
For all $\underline{x}\le \bar{x}$, the monotonicity assumption and the fact that $\ex\in[\xL,\xU]$-a.s. yield $\E_\sQ(\ey|\ex=\underline{x})\le \int \E_\sQ(\ey|\ex)d\sR(\ex|\xL=\underline{x},\xU=\bar{x}) \le \E_\sQ(\ey|\ex=\bar{x})$.
Putting this together with the previous result, $\E_\sQ(\ey|\ex=\underline{x})\le \E_\sP(\ey|\xL=\underline{x},\xU=\bar{x}) \le \E_\sQ(\ey|\ex=\bar{x})$.
Then (using again the monotonicity assumption) for any $x\ge \bar{x}$, $\E_{\sP}(\ey|\xL=\underline{x},\xU=\bar{x}) \le \E_\sQ(\ey|\ex=x)$ so that the lower bound holds.
The bound is weakly increasing as a function of $x$, so that the monotonicity assumption on $\E_\sQ(\ey|\ex=x)$ holds and the bound is sharp.
The argument for the upper bound can be concluded similarly.
\end{proof}

Learning about functionals of $\sQ(\ey|\ex=x)$ naturally implies learning about predictors of $\ey|\ex=x$.
For example, $\idr{\E_\sQ(\ey|\ex=x)}$ yields the collection of values for the best predictor under square loss;
$\idr{\M_\sQ(\ey|\ex=x)}$, with $\M_\sQ$ the median with respect to distribution $\sQ$, yields the collection of values for the best predictor under absolute loss.
And so on.
A related but distinct problem is that of \emph{parametric} conditional prediction.
Often researchers specify not only a loss function for the prediction problem, but also a parametric family of predictor functions, and wish to learn the member of this family that minimizes expected loss.
To avoid confusion, let me clarify that here I am not referring to a parametric assumption on the best predictor, e.g., that $\E_\sQ(\ey|\ex)$ is a linear function of $\ex$.
I return to such assumptions at the end of this section.
For now, in the example of linearity and square loss, I am referring to best linear prediction, i.e., best linear approximation to $\E_\sQ(\ey|\ex)$.
\cite[pp. 56-58]{man03} discusses what can be learned about the best linear predictor of $\ey$ conditional on $\ex$, when only interval data on $(\ey,\ex)$ is available.

I treat first the case of interval outcome and perfectly observed covariates.
\begin{IP}[Parametric Prediction with Interval Outcome Data] 
\label{IP:param_pred_interval}
Maintain the same assumptions as in Identification Problem \ref{IP:interval_outcome}.
Let $(\yL,\yU,\ex)\sim\sP$ be observable random variables and $\ey$ be an unobservable random variable, with $\sR(\yL\le\ey\le\yU)=1$.
In the absence of additional information, what can the researcher learn about the best linear predictor of $\ey$ given $\ex=x$?
	\qedex
\end{IP}
For simplicity suppose that $\ex$ is a scalar, and let $\theta=[\theta_0~\theta_1]^\top\in\Theta\subset\R^2$ denote the parameter vector of the best linear predictor of $\ey|\ex$. 
Assume that $Var(\ex)>0$.
Combining the definition of best linear predictor with a characterization of the sharp identification region for the joint distribution of $(\ey,\ex)$, we have that
\begin{equation}
\idr{\theta}=\left\{ \vartheta
=\arg \min \int \left( y-\theta_0-\theta_1x\right)^2 d\eta
,~\eta \in \idr{\sQ(\ey,\ex)}\right\} ,  \label{eq:manski_blp}
\end{equation}
where, using an argument similar to the one in Theorem SIR-\ref{SIR:CDF_id},
\begin{multline}
\idr{\sQ(\ey,\ex)}= \Big\{\eta : \eta_{([t_0,t_1],(-\infty,s])}\ge 
\sP(\yL\ge t_0,\yU\le t_1,\ex \le s)\\
\forall t_0\le t_1,t_0,t_1\in\R,\forall s\in \R\Big\}.
\label{eq:Qyx}
\end{multline}
\cite[Proposition 4.1]{ber:mol08} show that \eqref{eq:manski_blp} can be re-written in an intuitive way that generalizes the well-known formula for the best linear predictor that arises when $\ey$ is perfectly observed.
Define the random segment $\eG$ and the matrix $\Sigma_\sP$ as
\begin{align}
  \eG=\left\{
    \begin{pmatrix}
      \ey\\ \ey\ex 
    \end{pmatrix}
    :\; \ey \in \Sel(\eY)\right\}\subset\R^2,
   ~~\text{and}~~  \Sigma_\sP=\E_\sP 
  \begin{pmatrix}
    1 & \ex\\ \ex & \ex^2
  \end{pmatrix},\label{eq:G_and_Sigma}
\end{align}
where $\Sel(\eY)$ is the set of all measurable selections from $\eY$, see Definition \ref{def:selection}. Then,
\begin{SIR}[Best Linear Predictor with Interval Outcome Data]
\label{SIR:BLP_intervalY}
Under the assumptions of Identification Problem \ref{IP:param_pred_interval}, the sharp identification region for the parameters of the best linear predictor of $\ey|\ex$ is
\begin{equation}
  \label{eq:ThetaI_BLP}
  \idr{\theta}=  \Sigma_\sP^{-1} \E_\sP\eG,
\end{equation}
with $\E_\sP\eG$ the Aumann (or selection) expectation of $\eG$ as in Definition \ref{def:sel-exp}.
\end{SIR}
\begin{proof}
By Theorem \ref{thr:artstein}, $(\tilde\ey,\tilde\ex)\in(\eY\times\ex)$ (up to an ordered coupling as discussed in Appendix \ref{app:RCS}), if and only if the distribution of $(\tilde\ey,\tilde\ex)$ belongs to $\idr{\sQ(\ey,\ex)}$.
The result follows.
\end{proof}
In either representation \eqref{eq:manski_blp} or \eqref{eq:ThetaI_BLP}, $\idr{\theta}$ is the collection of best linear predictors for each selection of $\eY$.\footnote{Under our assumption that $\cY$ is a bounded interval, all the selections of $\eY$ are integrable. \cite{ber:mol08} consider the more general case where $\cY$ is not required to be bounded.}
Why should one bother with the representation in \eqref{eq:ThetaI_BLP}?
The reason is that $\idr{\theta}$ is a convex set, as it can be evinced from representation \eqref{eq:ThetaI_BLP}: $\eG$ has almost surely convex realizations that are segments and the Aumann expectation of a convex set is convex.\footnote{In $\R^2$ in our example, in $\R^d$ if $\ex$ is a $d-1$ vector and the predictor includes an intercept.}
Hence, it can be equivalently represented through its \emph{support function} $h_{\idr{\theta}}$, see Definition \ref{def:sup-fun} and equation \eqref{eq:rocka}. In particular, in this example,
\begin{align}
\label{eq:supfun:BLP}
h_{\idr{\theta}}(u)=\E_\sP[(\yL\one(f(\ex,u)<0)+\yU\one(f(\ex,u)\ge 0))f(\ex,u)],~~u\in\mathbb{S},
\end{align}
where $f(\ex,u)\equiv [1~\ex]\Sigma_\sP^{-1}u$.\footnote{See \cite[p. 808]{ber:mol08} and \cite[p. 1136]{bon:mag:mau12}.}
The characterization in \eqref{eq:supfun:BLP} results from Theorem \ref{thr:exp-supp}, which yields $h_{\idr{\theta}}(u)=h_{\Sigma_\sP^{-1} \E_\sP\eG}(u)=\E_\sP h_{\Sigma_\sP^{-1} \eG}(u)$, and the fact that $\E_\sP h_{\Sigma_\sP^{-1} \eG}(u)$ equals the expression in \eqref{eq:supfun:BLP}.
As I discuss in Section \ref{sec:inference} below, because the support function fully characterizes the boundary of $\idr{\theta}$, \eqref{eq:supfun:BLP} allows for a simple sample analog estimator, and for inference procedures with desirable properties.
It also immediately yields sharp bounds on linear combinations of $\theta$ by judicious choice of $u$.\footnote{For example, in the case that $\ex$ is a scalar, sharp bounds on $\theta_1$ can be obtained by choosing $u=[0~1]^\top$ and $u=[0~-1]^\top$, which yield $\theta_1\in[\theta_{1L},\theta_{1U}]$ with $\theta_{1L}=\min_{\ey\in[\yL,\yU]}\frac{Cov(\ex,\ey)}{Var(\ex)}=\frac{\E_\sP[(\ex-\E_\sP\ex)(\yL\one(\ex >\E_\sP\ex)+\yU\one(\ex\le\E\ex))]}{\E_\sP\ex^2-(\E_\sP\ex)^2}$ and $\theta_{1U}=\max_{\ey\in[\yL,\yU]}\frac{Cov(\ex,\ey)}{Var(\ex)}=\frac{\E_\sP[(\ex-\E_\sP\ex)(\yL\one(\ex <\E_\sP\ex)+\yU\one(\ex\ge\E\ex))]}{\E_\sP\ex^2-(\E_\sP\ex)^2}$.}
\cite{sto07} and \cite{mag:mau08} provide the same characterization as in \eqref{eq:supfun:BLP} using, respectively, direct optimization and the Frisch-Waugh-Lovell theorem. \medskip

A natural generalization of Identification Problem \ref{IP:param_pred_interval} allows for both outcome and covariate data to be interval valued.
\begin{IP}[Parametric Prediction with Interval Outcome and Covariate Data]
\label{IP:param_pred_interval_out_cov}
Maintain the same assumptions as in Identification Problem \ref{IP:param_pred_interval}, but with $\ex\in\cX\subset\R$ unobservable.
Let the researcher observe $(\yL,\yU,\xL,\xU)$ such that $\sR(\yL \leq \ey \leq \yU , \xL \leq \ex \leq \xU)=1$. 
Let $\eX\equiv [\xL,\xU]$ and let $\cX$ be bounded.
In the absence of additional information, what can the researcher learn about the best linear predictor of $\ey$ given $\ex=x$?
	\qedex
\end{IP}
Abstractly, $\idr{\theta}$ is as given in \eqref{eq:manski_blp}, with 
\begin{align*}
\idr{\sQ(\ey,\ex)}= \left\{\eta : \eta_K\ge 
\sP((\eY\times\eX)\subset K) ~\forall \text{ compact } K\subset \cY\times\cX\right\}
\end{align*}
replacing \eqref{eq:Qyx} by an application of Theorem \ref{thr:artstein}.
While this characterization is sharp, it is cumbersome to apply in practice, see \cite{hor:man:pon:sto03}.

On the other hand, when both $\ey$ and $\ex$ are perfectly observed, the best linear predictor is simply equal to the parameter vector that yields a mean zero prediction error that is uncorrelated with $\ex$.
How can this basic observation help in the case of interval data?
The idea is that one can use the same insight applied to the set-valued data, and obtain $\idr{\theta}$ as the collection of $\theta$'s for which there exists a selection $(\tilde{\ey},\tilde{\ex}) \in \Sel(\eY \times \eX)$, and associated prediction error $\eps_\theta=\tilde{\ey}-\theta_0-\theta_1 \tilde{\ex}$, satisfying $\E_\sP \eps_\theta=0$ and $\E_\sP (\eps_\theta \tilde{\ex})=0$ \citep[as shown by][]{ber:mol:mol11}.\footnote{Here for simplicity I suppose that both $\xL$ and $\xU$ have bounded support.
\cite{ber:mol:mol11} do not make this simplifying assumption.} 
To obtain the formal result, define the $\theta$-dependent set\footnote{Note that while $\eG$ is a convex set, $\Eps_\theta$ is not.}
  \[\Eps_\theta = \left\lbrace \begin{pmatrix}
  \tilde{\ey}-\theta_0-\theta_1 \tilde{\ex} \\ 
  (\tilde{\ey}-\theta_0-\theta_1 \tilde{\ex})\tilde{\ex}
  \end{pmatrix} \: : \, (\tilde{\ey},\tilde{\ex}) \in \Sel(\eY \times\eX) \right\rbrace. \]
\begin{SIR}[Best Linear Predictor with Interval Outcome and Covariate Data]
\label{SIR:blp_intervalYX}
Under the assumptions of Identification Problem \ref{IP:param_pred_interval_out_cov}, the sharp identification region for the parameters of the best linear predictor of $\ey|\ex$ is
\begin{align}
  \label{eq:ThetaI:BLP}
  \idr{\theta} = \{\theta\in\Theta:\mathbf{0}\in \E_\sP\Eps_\theta\}
  = \left\{\theta\in\Theta:~
    \min_{u \in \Ball}\E_\sP h_{\Eps_\theta}(u) = 0 \right\},
\end{align}
where $h_{\Eps_\theta}(u) = \max_{y\in\eY,x\in\eX}  [u_1(y-\theta_0-\theta_1 x)+ u_2(yx-\theta_0 x-\theta_1 x^2)]$ is the support function of the set $\Eps_\theta$ in direction $u\in\Sphere$, see Definition \ref{def:sup-fun}.
\end{SIR}
\begin{proof}
By Theorem \ref{thr:artstein}, $(\tilde\ey,\tilde\ex)\in(\eY\times\eX)$ (up to an ordered coupling as discussed in Appendix \ref{app:RCS}), if and only if the distribution of $(\tilde\ey,\tilde\ex)$ belongs to $\idr{\sQ(\ey,\ex)}$.
For given $\theta$, one can find $(\tilde\ey,\tilde\ex)\in(\eY\times\eX)$ such that $\E_\sP \eps_\theta=0$ and $\E_\sP (\eps_\theta \tilde{\ex})=0$ with $\eps_\theta\in\Eps_\theta$ if and only if the zero vector belongs to $\E_\sP \Eps_\theta$. 
By Theorem \ref{thr:exp-supp}, $\E_\sP \Eps_\theta$ is a convex set and by \eqref{eq:dom_Aumann}, $\mathbf{0} \in \E_\sP \Eps_\theta$ if and only if $0 \leq h_{\E_\sP \Eps_\theta}(u) \,\forall \, u \in \Ball$. 
The final characterization follows from \eqref{eq:supf}.
\end{proof}
The support function $h_{\Eps_\theta}(u)$ is an easy to calculate convex sublinear function of $u$, regardless of whether the variables involved are continuous or discrete.
The optimization problem in (\ref{eq:ThetaI:BLP}), determining whether $\theta \in \idr{\theta}$, is a convex program, hence easy to solve. 
See for example the CVX software by \cite{gra:boy10}. 
It should be noted, however, that the set $\idr{\theta}$ itself is not necessarily convex. 
Hence, tracing out its boundary is non-trivial.
I discuss computational challenges in partial identification in Section \ref{sec:computations}.

I conclude this section by discussing parametric regression.
\cite{man:tam02} study identification of parametric regression models under the assumptions in Identification Problem \ref{IP:man:tam02_param}; Theorem SIR-\ref{SIR:man:tam02_param} below reports the result.
The proof is omitted because it follows immediately from the proof of Theorem SIR-\ref{SIR:man:tam:nonpar}.
\begin{IP}[Parametric Regression with Interval Covariate Data]
\label{IP:man:tam02_param}
Let $(\ey,\xL,\xU,\ew)\sim\sP$ be observable random variables in $\R\times\R\times\R\times\R^d$, $d<\infty$, and let $\ex\in\R$ be an unobservable random variable.
Assume that the joint distribution $\sR$ of $(\ey,\ex,\xL,\xU)$ is such that $\sR(\xL\le\ex\le\xU)=1$ and $\E_{\sR}(\ey|\ew,\ex,\xL,\xU)=\E_\sQ(\ey|\ew,\ex)$. 
Suppose that $\E_\sQ(\ey|\ew,\ex)=f(\ew,\ex;\theta)$, with $f:\R^d\times\R\times\Theta \mapsto \R$ a known function such that for each $w\in\R$ and $\theta\in\Theta$, $f(w,x;\theta)$ is weakly increasing in $x$.
In the absence of additional information, what can the researcher learn about $\theta$?
	\qedex
\end{IP}
\begin{SIR}[Parametric Regression with Interval Covariate Data]
\label{SIR:man:tam02_param}
	Under the Assumptions of Identification Problem \ref{IP:man:tam02_param}, the sharp identification region for $\theta$ is
	\begin{multline}
	\idr{\theta}=\big\{\vartheta\in \Theta: f(\ew,\xL;\vartheta)\le \E_\sP(\ey|\ew,\xL,\xU) \le f(\ew,\xU;\vartheta),~(\ew,\xL,\xU)\text{-a.s.} \big\}.\label{eq:ThetaI_man:tam02_param}
	\end{multline}
\end{SIR}
\cite{auc:bug:hot17} study Identification Problem \ref{IP:man:tam02_param} for the case of missing covariate data \emph{without} imposing the mean independence restriction of \cite{man:tam02} (Assumption MI in Identification Problem \ref{IP:interval_covariate}).
As discussed in footnote \ref{foot:auc:bug:hot17}, restriction MI is undesirable in this context because it implies the assumption that data are missing at random.
\cite{auc:bug:hot17} characterize $\idr{\theta}$ under the weaker assumptions, but face the problem that this characterization is usually too  complex to compute or to use for inference.  
They therefore provide outer regions that are easier to compute, and they show that these regions are informative and relatively easy to use.\medskip

\subsection{Measurement Error and Data Combination}
\label{subsec:meas_error}
One of the first examples of bounding analysis appears in \cite{fri34}, to assess the impact in linear regression of covariate measurement error. 
This analysis was substantially extended in \cite{gil:lea83}, \cite{kle:lea84}, and \cite{lea87}.
The more recent literature in partial identification has provided important advances to learn features of probability distributions when the observed variables are error-ridden measures of the variables of interest.
Here I briefly mention some of the papers in this literature, and refer to Chapter \textbf{XXX} in this Volume by \cite{sch19} for a thorough treatment of identification and inference with mismeasured and unobserved variables.
In an influential paper, \cite{hor:man95} study what can be learned about features of the distribution of $\ey|\ex$ in the presence of contaminated or corrupted outcome data.  
Whereas a contaminated sampling model assumes that data errors are statistically independent of sample realizations from the population of interest, the corrupted sampling model does not.
These models are regularly used in the important literature on robust estimation \citep[e.g.,][]{hub64,hub04,ham:ron:rou:sta11}.
However, the goal of that literature is to characterize how point estimators of population parameters behave when data errors are generated in specified ways.
As such, the inference problem is approached ex-ante: before collecting the data, one looks for point estimators that are not greatly affected by error.
The question addressed by \cite{hor:man95} is conceptually distinct.
It asks what can be learned about specific population parameters ex-post, that is, after the data has been collected.
For example, whereas the mean is well known not to be a robust estimator in the presence of contaminated data, \cite{hor:man95} show that it can be (non-trivially) bounded provided the probability of contamination is strictly less than one.
\cite{dom:she04,dom:she05} and \cite{kre:pep07,kre:pep08} extend the results of \citeauthor{hor:man95} to allow for (partial) verification of the distribution from which the data are drawn.
They apply the resulting sharp bounds to learn about school performance when the observed test scores may not be valid for all students. 
\cite{mol08} provides sharp bounds on the distribution of a misclassified outcome variable under an array of different assumptions on the extent and type of misclassification.

A completely different problem is that of data combination.
Applied economists often face the problem that no single data set contains all the variables that are necessary to conduct inference on a population of interest. 
When this is the case, they need to integrate the information contained in different samples; for example, they might need to combine survey data with administrative data \citep[see][for a
survey of the econometrics of data combination]{rid:mof07}.
From a methodological perspective, the problem is that while the samples being combined might contain some common variables, other variables belong only to one of the samples. 
When the data is collected at the same aggregation level (e.g., individual level, household level, etc.), if the common variables include a unique and correctly recorded identifier of the units constituting each sample, and there is a substantial overlap of units across all samples, then exact matching of the data sets is relatively straightforward, and the combined data set provides all the relevant information to identify features of the population of interest. 
However, it is rather common that there is a limited overlap in the units constituting each sample, or that variables that allow identification of units are not available in one or more of the input files, or that one sample provides information at the individual or household level (e.g., survey data) while the second sample provides information at a more aggregate level (e.g., administrative data providing information at the precinct or district level).
Formally, the problem is that one observes data that identify the joint distributions $\sP(\ey,\ex)$ and $\sP(\ex,\ew)$, but not data that identifies the joint distribution $\sQ(\ey,\ex,\ew)$ whose features one wants to learn.
The literature on \textit{statistical matching} has aimed at using the
common variable(s) $\ex$ as a bridge to create synthetic records containing $(\ey,\ex,\ew)$ \citep[see, e.g.,][for an early contribution]{okn72}. 
As \cite{sim72} points out, the inherent assumption at the base of
statistical matching is that conditional on $\ex$, $\ey$ and $\ew$ are
independent. 
This conditional independence assumption is strong and
untestable. 
While it does guarantee point identification of features of the
conditional distributions $\sQ(\ey|\ex,\ew)$, it often finds very little justification in practice. 
Early on, \cite{dun:dav53} provided numerical illustrations on how one can bound the object of interest, when both $\ey$ and $\ew$ are binary variables.
\cite{cro:man02} provide a general analysis of the problem.
They obtain bounds on the long regression $\E_\sQ(\ey|\ex,\ew)$, under the assumption that $\ew$ has finite support.
They show that sharp bounds on $\E_\sQ(\ey|\ex,\ew=w)$ can be obtained using the results in \cite{hor:man95}, thereby establishing a connection with the analysis of contaminated data.
They then derive sharp identification regions for $[\E_\sQ(\ey|\ex=x,\ew=w),x\in\cX,w\in\cW]$.  
They show that these bounds are sharp when $\ey$ has finite support, and \cite{mol:pes06} establish sharpness without this restriction.
\cite{fan:she:shu14} address the question of what can be learned about counterfactual distributions and treatment effects under the data scenario just described, but with $\ex$ replaced by $\es$, a binary indicator for the received treatment (using the notation of the previous section).
In this case, the exogenous selection assumption (conditional on $\ew$) does not suffice for point identification of the objects of interest.
The authors derive, however, sharp bounds on these quantities using monotone rearrangement inequalities.
\cite{pac17} provides partial identification results for the coefficients in the linear projection of $\ey$ on $(\ex,\ew)$.
 
\subsection{Further Theoretical Advances and Empirical Applications}
\label{subsec:applications_PIPD}
In order to discuss the partial identification approach to learning features of probability distributions in some level of detail while keeping this chapter to a manageable length, I have focused on a selection of papers.
In this section I briefly mention several other excellent theoretical contributions that could be discussed more closely, as well as several papers that have applied partial identification analysis to answer important empirical questions.\medskip

While selectively observed data are commonplace in observational studies, in randomized experiments subjects are randomly placed in designated treatment groups conditional on $\ex$, so that the assumption of exogenous selection is satisfied with respect to the assigned treatment.
Yet, identification of some highly policy relevant parameters can remain elusive in the absence of strong assumptions.
One challenge results from noncompliance, where individuals' received treatments differs from the randomly assigned ones.
\cite{bal:pea97} derive sharp bounds on the ATE in this context, when $\cY=\T=\{0,1\}$.
Even if one is interested in the intention-to-treat parameter, selectively observed data may continue to be a problem.
For example, \cite{lee09} studies the wage effects of the Job Corps training program, which randomly assigns eligibility to participate in the program.
Individuals randomized to be eligible were not compelled to receive treatment, hence \cite{lee09} focuses on the intention-to-treat effect.
Because wages are only observable when individuals are employed, a selection problem persists despite the random assignment of eligibility to treatment, as employment status may be affected by the training program.
\citeauthor{lee09} obtains sharp bounds on the intention-to-treat effect, through a trimming procedure that leverages results in \cite{hor:man95}.
\cite{mol08MT} analyzes the problem of identification of the ATE and other treatment effects, when the received treatment is unobserved for a subset of the population. 
Missing treatment data may be due to item or survey nonresponse in observational studies, or noncompliance with randomly assigned treatments that are not directly monitored. 
She derives sharp worst case bounds leveraging results in \cite{hor:man95}, and she shows that these are a function of the available prior information on the distribution of missing treatments. 
If the response function is assumed monotone as in \eqref{eq:MTR:treat}, she obtains informative bounds without restrictions on the distribution of missing treatments. 

Even randomly assigned treatments and perfect compliance with no missing data may not suffice for point identification of all policy relevant parameters.
Important examples are given by \cite{hec:smi:cle97} and \cite{man97:mixing}.
\citeauthor{hec:smi:cle97} show that features of the joint distribution of the potential outcomes of treatment and control, including the distribution of treatment effects impacts, cannot be point identified in the absence of strong restrictions.
This is because although subjects are randomized to treatment and control, nobody's outcome is observed under both states.
Nonetheless, the authors obtain bounds for the functionals of interest.
\cite{mul18} derives related bounds on the probability that the potential outcome of one treatment is larger than that of the other treatment, and applies these results to health economics problems.
\citeauthor{man97:mixing} shows that features of outcome distributions under treatment rules in which treatment may vary within groups cannot be point identified in the absence of strong restrictions.
This is because data resulting from randomized experiments with perfect compliance allow for point identification of the outcome distributions under treatment rules that assign all persons with the same $\ex$ to the same treatment group.
However, such data only allow for partial identification of outcome distributions under rules in which treatment may vary within groups.
\citeauthor{man97:mixing} derives sharp bounds for functionals of these distributions. 

Analyses of data resulting from natural experiments also face identification challenges.
\cite{hot:mul:san97} study what can be learned about treatment effects when one uses a contaminated instrumental variable, i.e. when a mean-independence assumption holds in a population of interest, but the observed population is a mixture of the population of interest and one in which the assumption doesn't hold.
They extend the results of \cite{hor:man95} to learn about the causal effect of teenage childbearing on a teen mother's subsequent outcomes, using the natural experiment of miscarriages to form an instrumental variable for teen births. 
This instrument is contaminated because miscarriges may not occur randomly for a subset of the population (e.g., higher miscarriage rates are associated with smoking and drinking, and these behaviors may be correlated with the outcomes of interest).

Of course, analyses of selectively observed data present many challenges, including but not limited to the ones described in Section \ref{subsec:missing_data}.
\cite{ath:imb06} generalize the difference-in-difference (DID) design to a \emph{changes-in-changes} (CIC) model, where the distribution of the unobservables is allowed to vary across groups, but not overtime within groups, and the additivity and linearity assumptions of the DID are dispensed with.
For the case that the outcomes have a continuous distribution, \citeauthor{ath:imb06} provide conditions for point identification of the entire counterfactual distribution of effects of the treatment on the treatment group as well as the distribution of effects of the treatment on the control group, without restricting how these distributions differ from each other.
For the case that the outcome variables are discrete, they provide partial identification results, as well as additional conditions compared to their baseline model under which point identification attains.

Motivated by the question of whether the age-adjusted mortality rate from cancer in 2000 was lower than that in the early 1970s, \cite{hon:lle06} study partial identification of competing risk models \citep[see][for earlier partial identification results]{pet76}.
To answer this question, they need to contend with the fact that mortality rate from cardiovascular disease declined substantially over the same period of time, so that individuals that in the early 1970s might have died from cardiovascular disease before being diagnosed with cancer, do not in 2000.
In this context, it is important to carry out the analysis without assuming that the underlying risks are independent.
\citeauthor{hon:lle06} show that bounds for the parameters of interest can be obtained as the solution to linear programming problems.
The estimated bounds suggest much larger improvements in cancer mortality rates than previously estimated.

\cite{blu:gos:ich:meg07} use UK data to study changes over time in the distribution of male and female wages, and in wage inequality.
Because the composition of the workforce changes over time, it is difficult to disentangle that effect from changes in the distribution of wages, given that the latter are observed only for people in the workforce.
\citeauthor{blu:gos:ich:meg07} begin their empirical analysis by reporting worst case bounds \citep[as in][]{man94} on the CDF of wages conditional on covariates.
They then consider various restrictions on treatment selection, e.g., a first order stochastic dominance assumption according to which  people with higher wages are more likely to work, and derive tighter bounds under this assumption (and under weaker ones).
Finally, they bring to bear shape restrictions.
At each step of the analysis, they report the resulting bounds, thereby illuminating the role played by each assumption in shaping the inference.
\cite{cha:che:mol:sch18} provide best linear approximations to the identification region for the quantile gender wage gap using Current Population Survey repeated cross-sections data from 1975-2001, using treatment selection assumptions in the spirit of \cite{blu:gos:ich:meg07} as well as exclusion restrictions.

\cite{bha:sha:vyt12} study the effect of Swan-Ganz catheterization on subsequent mortality.\footnote{The Swan-Ganz catheter is a device placed in patients in the intensive care unit to guide therapy.}
Previous research had shown, using propensity score matching (assuming that there are no unobserved differences between catheterized and non catheterized patients) that Swan-Ganz catheterization increases the probability that patients die within 180 days from admission to the intensive care unit.
\citeauthor{bha:sha:vyt12} re-analyze the data using (and extending) bounds results obtained by \cite{sha:vyt11}.
These results are based on exclusion restrictions combined with a threshold crossing structure for both the treatment and the outcome variables in problems where $\cY=\cT=\{0,1\}$.
\citeauthor{bha:sha:vyt12} use as instrument for Swan-Ganz catheterization the day of the week that the patient was admitted to the intensive care unit.
The reasoning is that patients are less likely to be catheterized on the weekend, but the admission day to the intensive care unit is plausibly uncorrelated with subsequent mortality.
Their results confirm that for some diagnoses, Swan-Ganz catheterization increases mortality at 30 days after catheterization and beyond.

\cite{man:pep18} use data from Maryland, Virginia and Illinois to learn about the impact of laws allowing individuals to carry concealed handguns (right-to-carry laws) on violent and property crimes.
Point identification of these treatment effects is possible under invariance assumptions that certain features of treatment response are constant across states and years.
\citeauthor{man:pep18} propose the use of weaker but more credible restrictions according to which these features exhibit bounded variation -- the invariance case being the limit where the bound on variation equals zero.
They carry out their analysis under different combinations of the bounded variation assumptions, and at each step they report the resulting bounds, thereby illuminating the role played by each assumption in shaping the inference. 

\cite{mou:hen:mea18} provide sharp bounds on the joint distribution of potential (binary) outcomes in a Roy model with sector specific unobserved heterogeneity and self selection based on potential outcomes.
The key maintained assumption is that the researcher has access to data that includes a stochastically monotone instrumental variable.
This is a selection shifter that is restricted to affect potential outcomes monotonically.
An example is parental education, which may not be independent from potential wages, but plausibly does not negatively affect future wages.
Under this assumption, \citeauthor{mou:hen:mea18} show that all observable implications of the model are embodied in the stochastic monotonicity of observed outcomes in the instrument, hence Roy selection behavior can be tested by checking this stochastic monotonicity.
They apply the method to estimate a Roy model of college major choice in Canada and Germany, with special interest in the under-representation of women in STEM.

\cite{mog:san:tor18} provide a general method to obtain sharp bounds on a certain class of treatment effects parameters.
This class is comprised of parameters that can be expressed as weighted averages of marginal treatment effects \citep{hec:vyt99,hec:vyt01,hec:vyt05}.
\cite{tor19pies} provides a general method, based on copulas, to obtain sharp bounds on treatment effect parameters in semiparametric binary models.
A notable feature of both \cite{mog:san:tor18} and \cite{tor19pies} is that the bounds are obtained as solutions to convex (even linear) optimization problems, rendering them computationally attractive.
\cite{fre:hor14} provide partial identification results and inference methods for a linear functional $\ell(g)$ when $g:\cX\mapsto\R$ is such that $\ey=g(\ex)+\epsilon$ and $\E(\ey|\ez)=0$.
The instrumental variable $\ez$ and regressor $\ex$ have discrete distributions, and $\ez$ has fewer points of support than $\ex$, so that $\ell(g)$ can only be partially identified.
They impose shape restrictions on $g$ (e.g., monotonicity or convexity) to achieve interval identification of $\ell(g)$, and they show that the lower and upper points of the interval can be obtained by solving linear programming problems. 
They also show that the bootstrap can be used to carry out inference.

\section{Partial Identification of Structural Models}
\label{sec:structural}
In this section I focus on the literature concerned with learning features of \emph{structural econometric models}.
These are models where economic theory is used to postulate relationships among observable outcomes $\ey$, observable covariates $\ex$, and unobservable variables $\nu$.
For example, economic theory may guide assumptions on economic behavior (e.g., utility maximization) and equilibrium that yield a mapping from $(\ex,\nu)$ to $\ey$.
The researcher is interested in learning features of these relationships (e.g., utility function, distribution of preferences), and to this end may supplement the data and economic theory with functional form assumptions on the mapping of interest and distributional assumptions on the observable and unobservable variables.

The earlier literature on partial identification of features of structural models includes important examples of nonparametric analysis of random utility models and revealed preference extrapolation, e.g. \cite{blo:mar60}, \cite{mar60}, \cite{hal73}, \cite{mcf75}, \cite{fal78}, \cite{mcf:ric91}, and others.
The earlier literature also addresses semiparametric analysis, where the underlying models are specified up to parameters that are finite dimensional (e.g., preference parameters) and parameters that are infinite dimensional (e.g., distribution functions); important examples include \cite{mar:and44}, \cite{mar52}, \cite[Section 2.10]{fis66}, \cite{har:kre79}, \cite{kre81}, \cite{lea81}, \cite{man88}, \cite{jov89}, \cite{phi89}, \cite{han:jag91}, \cite{han:hea:lut95}, \cite{lut96}, and others.
Contrary to the nonparametric bounds results discussed in Section \ref{sec:prob:distr}, and especially in the case of semiparametric models, structural partial identification often yields an identification region that is \emph{not} constructive.\footnote{Of course, this is not always the case, as exemplified by the bounds in \cite{han:jag91}.}
Indeed, the boundary of the set is not obtained in closed form as a functional of the distribution of the observable data.
Rather, the identification region can often be characterized as a \emph{level set} of a properly specified criterion function.

The recent spark of interest in partial identification of structural microeconometric models was fueled by the work of \cite{man:tam02}, \cite{tam03} and \cite{cil:tam09}, and \cite{hai:tam03}.
Each of these papers has advanced the literature in fundamental ways, studying conceptually very distinct problems.
\cite{man:tam02} are concerned with partial identification of the decision process yielding binary outcomes in a semiparametric model, when one of the explanatory variables is interval valued.
Hence, the root cause of the identification problem they study is that the \emph{data is incomplete}.\footnote{\cite{man:tam02} study also partial identification (and estimation) of nonparametric, semiparametric, and parametric conditional expectation functions that are well defined in the absence of a structural model, when one of the conditioning variables is interval valued. I refer to Section \ref{sec:prob:distr} for a discussion.}

\cite{tam03} and \cite{cil:tam09} are concerned with identification (and estimation) of simultaneous equation models with dummy endogeneous variables which are representations of two-player entry games with multiple equilibria.\footnote{\cite{cil:tam09} consider more general multi-player entry games.}
\cite{hai:tam03} are concerned with nonparametric identification and estimation of the distribution of valuations in a model of English auctions under weak assumptions on bidders' behavior.
In both cases, the root cause of the identification problem is that the \emph{structural model is incomplete}.
This is because the model makes multiple predictions for the observed outcome variables (respectively: the players' actions; and the bidders' bids), but does not specify how one of them is selected to yield the observed data.
 
\emph{Set-valued predictions} for the observable outcome (endogenous variables) are a key feature of partially identified structural models.
The goal of this section is to explain how they result in a wide array of theoretical frameworks, and how sharp identification regions can be characterized using a unified approach based on random set theory.
Although the work of \cite{man:tam02}, \cite{tam03} and \cite{cil:tam09}, and \cite{hai:tam03} has spurred many of the developments discussed in this section, for pedagogical reasons I organize the presentation based on application topic rather than chronologically.
The work of \cite{pak10} and \cite{pak:por:ho:ish15} further stimulated a large empirical literature that applies partial identification methods to a wide array of questions of substantive economic importance, to which I return in Section \ref{subsec:applications:struct}.

\subsection{Discrete Choice in Single Agent Random Utility Models}
\label{subsec:single:ag:RUM}
Let $\cI$ denote a population of decision makers and $\cY=\{c_1,\dots,c_{|\cY|}\}$ a finite universe of potential alternatives (\emph{feasible set} henceforth).
Let $\bU$ be a family of real valued functions defined over the elements of $\cY$. 
Let $\in^* $ denote ``is chosen from."
Then observed choice is consistent with a \emph{random utility model} if there exists a function $\bu_i$ drawn from $\bU$ according to some probability distribution, such that $\P(c \in^* C)=\P(\bu_i(c) \ge \bu_i(b)~\forall b \in C)$ for all $c\in C$, all non empty sets $C \subset \cY$, and all $i\in\cI$ \citep{blo:mar60}.
See \cite[Chapter 13]{man07a} for a textbook presentation of this class of models, and \cite{mat07} for a review of sufficient conditions for point identification of nonparametric and semiparametric limited dependent variables models.

As in the seminal work of \cite{mcf73}, assume that the decision makers and alternatives are characterized by observable and unobservable vectors of real valued attributes.
Denote the observable attributes by $\ex_i \equiv \{\ex_i^1,(\ex_{ic}^2,c\in\cY)\},i\in\cI$.
These include attribute vectors $\ex_i^1$ that are specific to the decision maker, as well as attribute vectors $\ex_{ic}^2$ that include components that are specific to the alternative and components that are indexed by both.
Denote the unobservable attributes (preferences) by $\nu_i\equiv(\zeta_i,\{\epsilon_{ic},~c\in\cY\}),i\in\cI$.
These are idiosyncratic to the decision maker and similarly may include alternative and decision maker specific terms. 
Denote $\cX,\cV$ the supports of $\ex,\nu$, respectively. 

In what follows, I label ``standard" a random utility model that maintains some form of exogeneity for $\ex_i$ (e.g., mean or quantile or statistical independence with $\nu_i$) and presupposes observation of data that include $\{(\eC_i,\ey_i,\ex_i):\ey_i \in^* \eC_i\}, i=1,\dots,n$, with $\eC_i$ the choice set faced by decision maker $i$ and $|\eC_i|\ge 2$ \citep[e.g.,][Assumption 1]{man75}.
Often it is also assumed that all members of the population face the same choice set, $\eC_i=D$ for all $i\in\cI$ and some known $D\subseteq\cY$, although this requirement is not critical to identification analysis.

\subsubsection{Semiparametric Binary Choice Models with Interval Valued Covariates} 
\label{subsubsec:man:tam02}
\cite{man:tam02} provide inference methods for nonparametric, semiparametric, and parametric conditional expectation functions when one of the conditioning variables is interval valued.
I have discussed their nonparametric and parametric sharp bounds on conditional expectations with interval valued covariates in Identification Problems \ref{IP:interval_covariate} and \ref{IP:man:tam02_param}, and Theorems SIR-\ref{SIR:man:tam:nonpar} and SIR-\ref{SIR:man:tam02_param}, respectively.
Here I focus on their analysis of semiparametric binary choice models.
Compared to the generic notation set forth at the beginning of Section \ref{subsec:single:ag:RUM}, I let $\eC_i=\cY=\{0,1\}$ for all $i\in\cI$, and with some abuse of notation I denote the vector of observed covariates $(\xL,\xU,\ew)$.
\begin{IP}[Semiparametric Binary Regression with Interval Covariate Data]
\label{IP:man:tam02_binary}
Let $(\ey,\xL,\xU,\ew)\sim\sP$ be observable random variables in $\{0,1\}\times\R\times\R\times\R^d$, $d<\infty$, and let $\ex\in\R$ be an unobservable random variable.
Let $\ey=\one(\ew\theta + \delta\ex +\epsilon>0)$.
Assume $\delta>0$, and further normalize $\delta=1$ because the threshold-crossing condition is invariant to the scale of the parameters.
Here $\epsilon$ is an unobserved heterogeneity term with continuous distribution conditional on $(\ew,\ex,\xL,\xU)$, $(\ew,\ex,\xL,\xU)$-a.s., and $\theta\in\Theta\subset\R^d$ is a parameter vector representing decision makers' preferences, with compact parameter space $\Theta$.
Assume that $\sR$, the joint distribution of $(\ey,\ex,\xL,\xU,\ew,\epsilon)$, is such that $\sR(\xL\le\ex\le\xU)=1$; $
\sR(\epsilon |\ew,\ex,\xL,\xU)=\sR(\epsilon|\ew,\ex)$; and for a specified $\alpha \in (0,1)$, $\sq_{\sR}^\epsilon(\alpha,\ew,\ex)=0$ and $\sR(\epsilon \le 0|\ew,\ex)=\alpha$, $(\ew,\ex)$-a.s..
In the absence of additional information, what can the researcher learn about $\theta$?
	\qedex
\end{IP}
Compared to Identification Problem \ref{IP:interval_covariate} (see p.~\pageref{IP:interval_covariate}), here one continues to impose $\ex\in[\xL,\xU]$ a.s.
The sign restriction on $\delta$ replaces the monotonicity restriction (M) in Identification Problem \ref{IP:interval_covariate}, but does not imply it unless the distribution of $\epsilon$ is independent of $\ex$ conditional on $\ew$.
The quantile independence restriction is inspired by \cite{man85}.

For given $\theta\in\Theta$, this model yields set valued predictions because $\ey=1$ can occur whenever $\epsilon> -\ew\theta-\xU$, whereas $\ey=0$ can occur whenever $\epsilon\le -\ew\theta-\xL$, and $-\ew\theta-\xU \le -\ew\theta-\xL$.
Conversely, observation of $\ey=1$ allows one to conclude that $\epsilon\in(-\ew\theta-\xU,+\infty)$, whereas observation of $\ey=0$ allows one to conclude that $\epsilon\in(-\infty,-\ew\theta-\xL]$, and these regions of possible realizations of $\epsilon$ overlap.
In contrast, when $\ex$ is observed the prediction is unique because the value $-\ew\theta-\ex$ partitions the space of realizations of $\epsilon$ in two disjoint sets, one associated with $\ey=1$ and the other with $\ey=0$.
Figure \ref{fig:set_valued_pred:man:tam:binary} depicts the model's set-valued predictions for $\ey$ given $(\ew,\xL,\xU)$ as a function of $\epsilon$, and the model's set valued predictions for $\epsilon$ given $(\ew,\xL,\xU)$ as a function of $\ey$.\footnote{Figure \ref{fig:set_valued_pred:man:tam:binary} is based on Figure 1 in \cite{man:tam02}. See \cite[Chapter XXX in this Volume]{che:ros19} for an extensive discussion of the duality between the model's set valued predictions for $\ey$ as a function of $\epsilon$ and for $\epsilon$ as a function of $\ey$, in both cases given the observed covariates.}

Why does this set-valued prediction hinder point identification?
The reason is that the distribution of the observable data relates to the model structure in an \emph{incomplete} manner.
The model predicts $\sM(\ey=1|\ew,\xL,\xU)=\int \sR(\ey=1|\ew,\ex,\xL,\xU)d\sR(\ex|\ew,\xL,\xU)=\int \sR(\epsilon>-\ew\theta-\ex|\ew,\ex)d\sR(\ex|\ew,\xL,\xU),~(\ew,\xL,\xU)$-a.s. 
Because the distribution $\sR(\ex|\ew,\xL,\xU)$ is left completely unspecified, one can find multiple values for $(\theta,\sR(\ex|\ew,\xL,\xU),\sR(\epsilon|\ew,\ex))$, satisfying the assumptions in Identification Problem \ref{IP:man:tam02_binary}, such that $\sM(\ey=1|\ew,\xL,\xU)=\sP(\ey=1|\ew,\xL,\xU),~(\ew,\xL,\xU)$-a.s. 
Nonetheless, in general, not all values of $\theta\in\Theta$ can be paired with some $\sR(\ex|\ew,\xL,\xU)$ and $\sR(\epsilon|\ew,\ex)$ so that they are compatible with $\sP(\ey=1|\ew,\xL,\xU),~(\ew,\xL,\xU)$-a.s. and with the maintained assumptions.
Hence, $\theta$ can be partially identified using the information in the model and observed data.
\begin{figure}[tp]	
\centering
\includegraphics[scale=1]{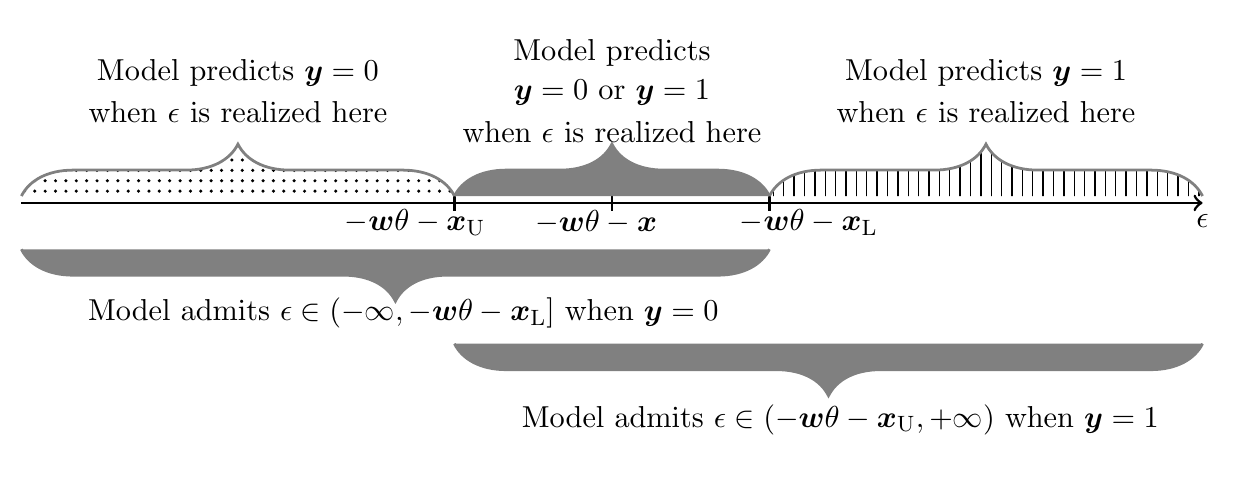}
\caption{\small{Predicted value of $\ey$ as a function of $\epsilon$, and admissible values of $\epsilon$ for each realization of $\ey$, in Identification Problem \ref{IP:man:tam02_binary}, conditional on $(\ew,\xL,\xU)$.}}
\label{fig:set_valued_pred:man:tam:binary}
\end{figure}
\begin{SIR}[Semiparametric Binary Regression with Interval Covariate Data]
\label{SIR:man:tam02_binary}
	Under the Assumptions of Identification Problem \ref{IP:man:tam02_binary}, the sharp identification region for $\theta$ is
	\begin{multline}
	\idr{\theta}=\Big\{\vartheta\in \Theta: \sP\Big((\ew,\xL,\xU):\, \{0\le\ew\vartheta+\xL\cap \sP(\ey=1|\ew,\xL,\xU)\le 1-\alpha\}\\
	\cup \{\ew\vartheta+\xU\le 0\cap \sP(\ey=1|\ew,\xL,\xU)\ge 1-\alpha\}\Big) = 0 \Big\}.\label{eq:ThetaI_man:tam02_binary}
	\end{multline}
\end{SIR}
\begin{proof}
For any $\vartheta\in\Theta$, define the set of possible values for the unobservable associated with the possible realizations of $(\ey,\ew,\xL,\xU)$, illustrated in Figure \ref{fig:set_valued_pred:man:tam:binary}, as\footnote{In the definition of $\Eps_\vartheta(1,\ew,\xL,\xU)$ I exploit the fact that under the maintained assumptions $\P(\epsilon=-\ew\vartheta-\xU|\ew,\ex,\xL,\xU)=0$ to enforce its closedness.}
\begin{align}
\Eps_\vartheta(\ey,\ew,\xL,\xU) =\left \{ 
\begin{tabular}{ll}
$(-\infty,-\ew\vartheta-\xL]$ & if $\ey=0$,\\
$[-\ew\vartheta-\xU,+\infty)$ & if $\ey=1$.
\end{tabular}
\right.\label{eq:def_Epsilon:man:tam}
\end{align}
Then $\Eps_\vartheta(\ey,\ew,\xL,\xU)$ is a random closed set as per Definition \ref{def:rcs}.
To simplify notation, let $\Eps_\vartheta(\ey)\equiv\Eps_\vartheta(\ey,\ew,\xL,\xU)$ suppressing the dependence on $(\ew,\xL,\xU)$.
Let $(\Eps_\vartheta(\ey),\ew,\xL,\xU)=\Eps_\vartheta(\ey)\times(\ew,\xL,\xU)=\{(\mathbf{e},\ew,\xL,\xU):\mathbf{e}\in\Eps_\vartheta(\ey)\}$.
If the model is correctly specified, for the data generating value $\theta$, $(\epsilon,\ew,\xL,\xU) \in (\Eps_\theta(\ey),\ew,\xL,\xU)$ a.s.
By Theorem \ref{thr:artstein} and Theorem 2.33 in \cite{mol:mol18}, this occurs if and only if
\begin{align}
\sR(\epsilon\in C|\ew,\xL,\xU)&\ge \sP(\Eps_\theta(\ey)\subset C|\ew,\xL,\xU),~(\ew,\xL,\xU)\text{-a.s.}~\forall C\in\cF,\label{eq:Artstein_on_man:tam}
\end{align}
where $\cF$ here denotes the collection of closed subsets of $\R$.

We then have that $\vartheta$ is observationally equivalent to $\theta$ if and only if \eqref{eq:Artstein_on_man:tam} holds for $\Eps_\vartheta(\ey)$ as defined in \eqref{eq:def_Epsilon:man:tam}.
The condition can be rewritten as
\begin{align*}
\int \sR(\epsilon\in C|\ew,\ex,\xL,\xU)d\sR(\ex|\ew,\xL,\xU)&\ge \sP(\Eps_\vartheta(\ey)\subset C|\ew,\xL,\xU),~(\ew,\xL,\xU)\text{-a.s.}~\forall C\in\cF.
\end{align*}
The assumption that $\sR(\epsilon|\ew,\ex,\xL,\xU)=\sR(\epsilon|\ew,\ex)$ yields that the above system of inequalities reduces to
\begin{align*}
\int \sR(\epsilon\in C|\ew,\ex)d\sR(\ex|\ew,\xL,\xU)&\ge \sP(\Eps_\vartheta(\ey)\subset C|\ew,\xL,\xU),~(\ew,\xL,\xU)\text{-a.s.}~\forall C\in\cF.
\end{align*}
Next, note that given the possible realizations of $\Eps_\vartheta(\ey)$, the above inequality is trivially satisfied unless $C=(-\infty,t]$ or $C=[t,\infty)$ for some $t\in\R$.
Finally, the only restriction on the distribution of $\epsilon$ is the quantile independence condition, hence it suffices to consider $t=0$.
To see why this is the case, let for example $t>0$ and fix a realization $(w,x_L,x_U)$ for $(\ew,\xL,\xU)$.\footnote{There are no $(\ew,\xL,\xU)$-cross restrictions.}
Then for the inequality not to be trivially satisfied it must be that either $w\vartheta+x_L\ge -t$ or $w\vartheta+x_U\le -t$ (both are not possible because $w\vartheta+x_L\le w\vartheta+x_U$).
If $w\vartheta+x_U\le -t$, it must be that $t\in(0,-w\vartheta-x_U]$ and $-w\vartheta-x_U>0$.
Then a distribution $\sR$ such that $\int \sR(\epsilon\in [0,t)|\ew=w,\ex)d\sR(\ex|\ew=w,\xL=x_L,\xU=x_U)=0$ is always feasible for $t\in(0,-w\vartheta-x_U]$. 
A similar argument holds if $w\vartheta+x_L\ge -t$; and also if $t<0$.
We then have that if the inequalities are satisfied for $t=0$, they are satisfied also for $t\neq 0$.

Finally, using the definition of $\Eps_\vartheta(\ey)$, for $t=0$ we have
\begin{align}
1-\alpha &\ge \sP(\ey=1|\ew,\xL,\xU)~~\text{for all}~(\ew,\xL,\xU)~\text{such that } \ew\vartheta+\xU\le 0,\label{eq:key_sharp:man:tam02_1}\\
1-\alpha &  \le \sP(\ey=1|\ew,\xL,\xU)~~\text{for all}~(\ew,\xL,\xU)~\text{such that } \ew\vartheta+\xL  \ge 0.\label{eq:key_sharp:man:tam02_2}
\end{align} 
Any given $\vartheta\in\Theta$, $\vartheta\neq\theta$, violates the above conditions if and only if $\sP\big((\ew,\xL,\xU):\, \{0\le\ew\vartheta+\xL\cap \sP(\ey=1|\ew,\xL,\xU)\le 1-\alpha\}\cup \{\ew\vartheta+\xU\le 0\cap \sP(\ey=1|\ew,\xL,\xU)\ge 1-\alpha\}\big) > 0$.
\end{proof}

\begin{BI}
The analysis in \cite{man:tam02} systematically studies what can be learned under increasingly strong sets of assumptions.
These include both assumptions that constrain the model from fully nonparametric to semiparametric to parametric, as well as assumptions that constrain the distribution of the observable covariates.
For example, \cite[Corollary to Proposition 2]{man:tam02} provide sufficient conditions on the joint distribution of $(\ew,\xL,\xU)$ that allow for identification of the sign of components of $\theta$, as well as for point identification of $\theta$.\footnote{This Corollary is related in spirit to the analysis in \cite{man88}.}
The careful analysis of the identifying power of increasingly stronger assumptions is the pillar of the partial identification approach to empirical research proposed by Manski, as illustrated in Section \ref{sec:prob:distr}.
The work of \cite{man:tam02} was the first example of this kind in semiparametric structural models.
\end{BI}
Revisiting \possessivecite{man:tam02} study of Identification Problem \ref{IP:man:tam02_binary} nearly 20 years later yields important insights on the differences between point and partial identification analysis.
It is instructive to take as a point of departure the analysis of \cite{man85}, which under the additional assumption that $(\ey,\ew,\ex)$
is observed yields
\begin{align*}
\ew\theta+\ex>0 \Leftrightarrow \sP(\ey=1|\ew,\ex)>1-\alpha.
\end{align*} 
In this case, $\theta$ is identified relative to $\vartheta\in\Theta$ if
\begin{align}
\sP\left((\ew,\ex):\, \{\ew\theta+\ex\le 0<\ew\vartheta+\ex\}
	\cup \{\ew\vartheta+\ex\le 0<\ew\theta+\ex\}\right) > 0.\label{eq:manski85}
	\end{align}
\citeauthor{man:tam02} extend this reasoning to the case that $\ex$ is unobserved, but known to satisfy $\ex\in [\xL,\xU]$ a.s.
The first part of their analysis, collected in their Proposition 2, characterizes the collection of values that cannot be distinguished from $\theta$ on the basis of $\sP(\ew,\xL,\xU)$ alone, through a clear generalization of \eqref{eq:manski85}:
\begin{align}
\{\vartheta\in \Theta: \sP\left((\ew,\xL,\xU):\, \{\ew\theta+\xU\le 0<\ew\vartheta+\xL\}
	\cup \{\ew\vartheta+\xU\le 0<\ew\theta+\xL\}\right) = 0\}.\label{eq:region:man:tam02:potential}
\end{align}
It is worth emphasizing that the characterization in \eqref{eq:region:man:tam02:potential} depends on $\theta$, and makes no use of the information in $\sP(\ey|\ew,\xL,\xU)$.
The Corollary to Proposition 2 yields conditions on $\sP(\ew,\xL,\xU)$ under which either the sign of components of $\theta$, or $\theta$ itself, can be identified, regardless of the distribution of $\ey|\ew,\xL,\xU$.

\cite[Lemma 1]{man:tam02} provide a second characterization, which presupposes knowledge of $\sP(\ey,\ew,\xL,\xU)$, yields a set smaller than the one in \eqref{eq:region:man:tam02:potential}, and coincides with the result in Theorem SIR-\ref{SIR:man:tam02_binary}. 
\cite{man:tam02} use the same notation for the two sets, although the sets are conceptually and mathematically distinct.\footnote{This was confirmed in personal communication with Chuck Manski and Elie Tamer.}
The result in Theorem SIR-\ref{SIR:man:tam02_binary} is due to \cite[Lemma 1]{man:tam02}, but the proof provided here is new, as is the use of random set theory in this application.\footnote{The proof closes a gap in the argument in \cite{man:tam02} connecting their Proposition 2 and Lemma 1, due to the fact that for a given $\vartheta$ the sets $\{(\ew,\xL,\xU):\, \{\ew\theta+\xU\le 0<\ew\vartheta+\xL\}	\cup \{\ew\vartheta+\xU\le 0<\ew\theta+\xL\}\}$ and $\{(\ew,\xL,\xU):\, \{0<\ew\vartheta+\xL\cap \sP(\ey=1|\ew,\xL,\xU)\le 1-\alpha\}\cup \{\ew\vartheta+\xU\le 0\cap \sP(\ey=1|\ew,\xL,\xU)> 1-\alpha\}\}$ need not coincide, with the former being a subset of the latter due to part (c) of the proof of Proposition 2 in \cite{man:tam02}.}
\begin{BI}\label{remark:man:tam02:che:ros}
The preceding discussion allows me to draw a novel connection between the two characterizations in \cite{man:tam02}, and the distinction put forward by \cite{che:ros17} and \cite[Chapter XXX in this Volume, Definition 2]{che:ros19} in partial identification between \emph{potential observational equivalence} and \emph{observational equivalence}.\footnote{This distinction echos the distinction drawn by \cite[Section 1.1.1]{man88book} between \emph{point identification} and \emph{uniform point identification}.
\citeauthor{man88book} considers a scenario where a parameter vector of interest $\theta$ is defined as the solution to an equation of the form $\crit_\sP(\theta)=0$ for some criterion function $\crit_\sP:\Theta\mapsto\R_+$.
Then $\theta$ is point identified relative to $(\sP,\Theta)$ if it is the unique solution to $\crit_\sP(\theta)=0$.
It is \emph{uniformly} point identified relative to $(\cP,\Theta)$, with $\cP$ a space of probability distributions to which $\sP$ belongs, if for every $\tilde\sP\in\cP$, $\crit_{\tilde\sP}(\vartheta)=0$ has a unique solution.}
Applying \citeauthor{che:ros17}'s definition, parameter vectors $\theta$ and $\vartheta$ are \emph{potentially} observationally equivalent if there exists \emph{some} distribution of $\ey|\ew,\xL,\xU$ for which conditions \eqref{eq:key_sharp:man:tam02_1}-\eqref{eq:key_sharp:man:tam02_2} hold.
Simple algebra confirms that this yields the region in \eqref{eq:region:man:tam02:potential}.
This notion of potential observational equivalence parallels one of the notions used to obtain sufficient conditions for point identification in the semiparametric literature \citep[as in, e.g.][]{man85}.
Both notions, as explained in \cite[Section 4.1]{che:ros19}, make no reference to the conditional distribution of outcomes given covariates delivered by the process being studied. 
To obtain that parameters $\theta$ and $\vartheta$ \emph{are} observationally equivalent one requires instead that conditions \eqref{eq:key_sharp:man:tam02_1}-\eqref{eq:key_sharp:man:tam02_2} hold for the \emph{observed} distribution $\sP(\ey=1|\ew,\xL,\xU)$ (as opposed to ``for some distribution" as in the case of potential observational equivalence).
This yields the sharp identification region in \eqref{eq:ThetaI_man:tam02_binary}.
\end{BI}
\cite{man10} studies random \emph{expected} utility models, where agents choose the alternative that maximizes their expected utility.
The core difference with standard models is that \citeauthor{man10} does not fully specify the subjective beliefs that agents use to form their expectations, but only a \emph{set} of such beliefs.
\citeauthor{man10} shows that the resulting, partially identified, discrete choice model can be formulated similarly to how \cite{man:tam02} treat interval valued covariates, and leverages their results to obtain bounds on preference parameters.\footnote{\cite[Supplementary Appendix F]{ber:mol:mol11} extend the analysis of \cite{man:tam02} to multinomial choice models with interval covariates.}\medskip

\cite{mag:mau08} consider a different but closely related model to the semiparametric binary response model studied by \citeauthor{man:tam02}.
They assume that an instrumental variable $\ez$ is available, that $\epsilon$ is independent of $\ex$ conditional on $(\ew,\ez)$, and that $Corr(\ez,\epsilon)=0$.
They assume that the distribution of $\ex$ is absolutely continuous with support $[v_1,v_k]$, and that $\ex$ is not a deterministic linear function of $(\ew,\ez)$. 
They consider the case that $\ex$ is unobserved but known to belong to one of the fixed (and known) intervals $[v_i,v_{i+1})$, $i=1,\dots,k-1$, with $\sR[\ex\in[v_i,v_{i+1})|\ew,\ez]>0$ almost surely for all $i$. 
Finally, they assume that $(-\ew\theta-\epsilon)\in [v_1,v_k]$ with probability one.
They do not, however, make quantile independence assumptions.

Their point of departure is the fact that under these conditions, if $\ex$ were observed, one could employ a transformation proposed by \cite{lew00} for the binary outcome $\ey$, such that $\theta$ can be identified through a simple linear moment condition. 
Specifically, let
\begin{align*}
\tilde{\ey}=\frac{\ey - \one_{\ex>0}}{f_\ex(\ex|\ew,\ez)},
\end{align*}
where $f_\ex(\cdot|\ew,\ez)$ is the conditional density function of $\ex$. 
Then, using the assumption that $\ez$ and $\epsilon$ are uncorrelated, one has
\begin{align}
\E_\sP(\ez \tilde{\ey})-\E_\sP(\ez \ew^\top) \theta = 0.\label{eq:sem-bin}
\end{align}

With interval valued $\ex$, \cite{mag:mau08} denote by $\ex^*$ the random variable that takes value $i\in\{1,\dots,k-1\}$ if $\ex\in[v_i,v_{i+1})$, so that the observed data are draws from the joint distribution of $(\ey,\ew,\ez,\ex^*)$.
They let $\delta(\ex^*)=v_{\ex^*+1}-v_{\ex^*}$ denote the length of the $\ex^*$-th interval, and define the transformed outcome variable:
\begin{displaymath}
  \ey^*=\frac{\delta(\ex^*)}{\sP(\ex^*=i|\ew,\ez)}\ey-v_k.
\end{displaymath}
The assumptions on $\ex$ yield that, given $\ez$ and $\ew$, $\epsilon$ does not depend on $\ex^*$. 
Moreover, $\sP(\ey=1|\ex^*,\ew,\ez)$ is non-decreasing in $\ex^*$ and $\sF_\epsilon(\cdot|\ez,\ew,\ex,\ex^*)=\sF_\epsilon(\cdot|\ez,\ew)$.
\cite{mag:mau08} show that the sharp identification region for $\theta$ is 
\begin{align}
  \idr{\theta}=\E_\sP(\ez \ew^\top)^{-1}\E_\sP(\ez \ey^* + \ez \eU),\label{eq:SIR:mag:mau}
\end{align}
where $\E_\sP(\ez \ey^* + \ez \eU)$ is the Aumann (or selection) expectation of the random interval $\ez \ey^* + \ez \eU$, see Definition \ref{def:sel-exp}, with
\begin{align*}
  \eU=\left[-\sum_{i=1}^{k-1}(r_i(\ew,\ez)-r_{i-1}(\ew,\ez))(v_{i+1}-v_i),
  \sum_{i=1}^{k-1}(r_{i+1}(\ew,\ez)-r_i(\ew,\ez))(v_{i+1}-v_i) \right].
\end{align*}
In this expression, $r_{\ex^*}(\ew,\ez)\equiv\sP(\ey=1|\ex^*,\ew,\ez)$ and by convention $r_0(\ew,\ez)=0$ and $r_K(\ew,\ez)=1$, see \cite[Theorem 4]{mag:mau08}. 
If $r_i(\ew,\ez),i=0,\dots,k$, were observed, this characterization would be very similar to the one provided by \cite{ber:mol08} for Identification Problem \ref{IP:param_pred_interval}, see equation \eqref{eq:ThetaI_BLP}.
However, these random functions need to be estimated. 
While the first-stage estimation of $r_i(\ew,\ez),i=0,\dots,k$, does not affect the identification arguments, it does complicate inference, see \cite{cha:che:mol:sch18} and the discussion in Section \ref{sec:inference}.

\subsubsection{Endogenous Explanatory Variables}
\label{subsubsec:CRS}
Whereas the standard random utility model presumes some form of exogeneity for $\ex$, in practice often some explanatory variables are endogenous.
This problem has been addressed in the literature to obtain point identification of the model through a combination of several assumptions, including large support conditions, special regressors, control function restrictions, and more \citep[see, e.g.,][]{mat93,ber:lev:pak95,lew00,pet:tra10}.
\cite{hon:tam03} analyze the distinct but related problem of identification in a censored regression model with endogeneous explanatory variables, and provide sufficient conditions for point identification.\footnote{The estimator that they propose extends the minimum distance estimator put forward by \cite{man:tam02}, see Section \ref{subsec:consistent}, so that if the conditions required for point identification do not hold, it estimates the parameter's identification region (under regularity conditions).
\cite{hon:tam03letters} carry out a similar analysis for the binary choice model with endogenous explanatory variables.} 

Here I discuss how to carry out identification analysis in the absence of such assumptions when instrumental variables $\ez$ are available, as proposed by \cite{che:ros:smo13}.
They consider a more general case than I do here, with utility function that is not parametrically specified and not restricted to be separable in the unobservables. 
Even in that more general case, the identification analysis follows through similar steps as reported here.
\begin{IP}[Discrete Choice with Endogenous Explanatory Variables]\label{IP:discrete:choice:endogenous}
Let $(\ey,\ex,\ez)\sim\sP$ be observable random variables in $\cY\times\cX\times\cZ$.
Let all members of the population face the same choice set $\cY$.
Suppose that each alternative has one unobservable attribute $\epsilon_c,c\in\cY$ and let $\nu\equiv(\epsilon_{c_1},\dots,\epsilon_{c_{|\cY|}})$.\footnote{Compared to the general model put forward in Section \ref{subsec:single:ag:RUM}, in this model there are no preference heterogeneity terms $\zeta$ (random coefficients) that vary only across decision makers.}
Let $\nu\sim\sQ$ and assume that $\nu\independent\ez$.
Suppose $\sQ$ belongs to a nonparametric family of distributions $\cT$, and that the conditional distribution of $\nu|\ex,\ez$, denoted $\sR(\nu|\ex,\ez)$, is absolutely continuous with respect to Lebesgue measure with everywhere positive density on its support, $(\ex,\ez)$-a.s.
Suppose utility is separable in unobservables and has a functional form known up to finite dimensional parameter vector $\theta\in\Theta\subset\R^m$, so that $\bu_i(c)=g(\ex_c;\theta)+\epsilon_c$, $(\ex_c,\epsilon_c)$-a.s., for all $c\in\cY$.
Maintain the normalizations $g(\ex_{c_{|\cY|}};\theta)=0$ for all $\theta\in\Theta$ and all $\ex\in\cX$, and $g(x_c^0;\theta)=\bar{g}$ for known $(x_c^0,\bar{g})$ for all $\theta\in\Theta$ and $c\in\cY$.\footnote{Of course, under these conditions one can work directly with utility differences. To try and economize on notation, I do not explicitly do so here.}
Given $(\ex,\ez,\nu)$, suppose $\ey$ is the utility maximizing choice in $\cY$.
In the absence of additional information, what can the researcher learn about $(\theta,\sQ)$?	
\qedex
\end{IP}
The key challenge to identification here results because the distribution of $\nu$ can vary across different values of $\ex$, both conditional and unconditional on $\ez$.
Why does this fact hinder point identification?
For a given $\vartheta\in\Theta$ and for any $c\in\cY$ and $x\in\cX$, the model yields that $c$ is optimal, and hence chosen, if and only if $\nu$ realizes in the set
\begin{align}
\cE_\vartheta(c,x)=\{e\in\cV:g(x_c;\vartheta)+e_c\ge g(x_d;\vartheta)+e_d~\forall d\in\cY\}.\label{eq:che:ros:E}
\end{align}
Figure \ref{fig:discrete:choice:endogenous} plots the set $\cE_\vartheta(\ey,\ex)$ in a stylized example with $\cY=\{1,2,3\}$ and $\cX=\{x^1,x^2\}$, as a function of $(\epsilon_1-\epsilon_3,\epsilon_2-\epsilon_3)$.\footnote{This figure is based on Figures 1-3 in \cite{che:ros:smo13}.}
Consider the model implied distribution, denoted $\sM$ below, of the optimal choice. 
Then, recalling the restriction $\ez\independent\nu$, we have
\begin{align}
\sM(c|\ex\in R_x,\ez;\vartheta)&=\int_{x\in R_x}\sR(\cE_\vartheta(c,\ex)|\ex=x,\ez)d\sP(x|z),~\forall~R_x\subseteq\cX,~\ez\text{-a.s.}\label{eq:che:ros:model:distrib}\\
\sQ(F)&=\int_{x\in\cX}\sR(F|\ex=x,\ez)d\sP(x|z),~~\forall~F\subseteq\cV,~\ez\text{-a.s.},\label{eq:che:ros:instrument}
\end{align}
Because the joint distribution of $(\ex,\nu)$ conditional on $\ez$ is left completely unrestricted (other than \eqref{eq:che:ros:instrument}), one can find multiple triplets $(\vartheta,\sQ,\sR(\nu|\ex,\ez))$ satisfying the maintained assumptions and with $\sM(c|\ex\in R_x,\ez;\vartheta)=\sP(c|\ex\in R_x,\ez)$ for all $c\in\cY$ and $R_x\subseteq\cX$, $\ez$-a.s.
\begin{figure}[tp]
\centering
\includegraphics[scale=1.1]{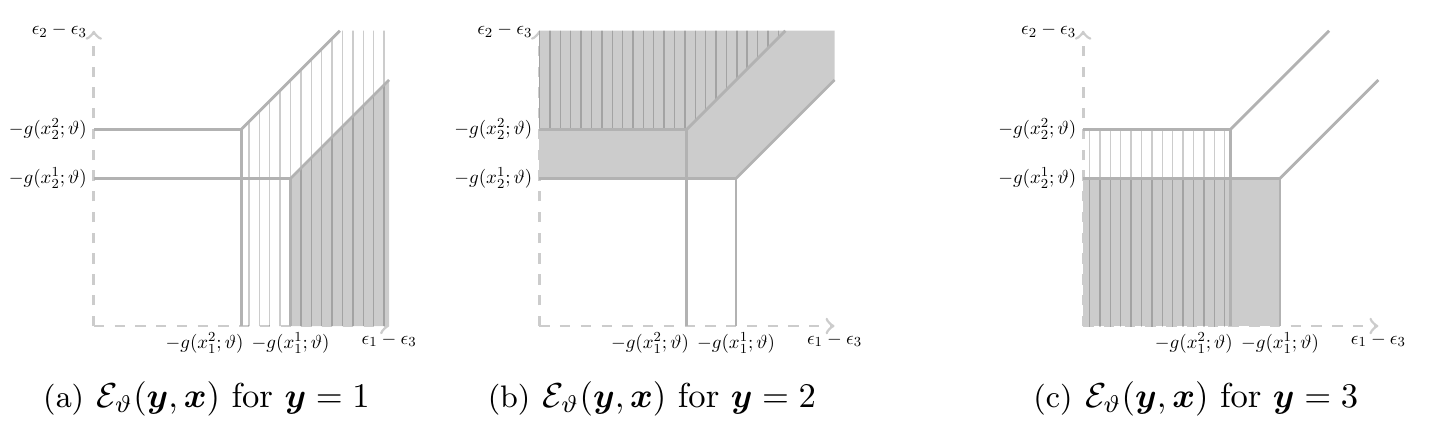}\smallskip
\caption{\small{The set $\cE_\vartheta$ in equation \eqref{eq:che:ros:E} and the corresponding admissible values for $(\ey,\ex)$ as a function of $(\epsilon_1-\epsilon_3,\epsilon_2-\epsilon_3)$ under the simplifying assumption that $\cX=\{x^1,x^2\}$ and $\cY=\{1,2,3\}$. 
The admissible values for $(\ey,\ex)$ are $\{(c,x^1)\}$ in the gray area, and $\{(c,x^2)\}$ in the area with vertical lines.
Because the two areas overlap, the model has set-valued predictions for $(\ey,\ex)$.}}
\label{fig:discrete:choice:endogenous}
\end{figure}

It is instructive to compare \eqref{eq:che:ros:model:distrib}-\eqref{eq:che:ros:instrument} with \possessivecite{mcf73} conditional logit.
Under the standard assumptions, $\ex\independent\nu$ so that no instrumental variables are needed.
This yields $\sQ(\nu)=\sR(\nu|\ex)$ $\ex$-a.s., and in addition $\sQ$ is typically known, with corresponding simplifications in \eqref{eq:che:ros:model:distrib}.
The resulting system of equalities can be inverted under standard order and rank conditions to yield point identification of $\theta$.

Further insights can be gained by looking at Figure \ref{fig:discrete:choice:endogenous}.
As the value of $\ex$ changes from $x^1$ to $x^2$, the region of values where, say, alternative 1 is optimal changes.
When $\ex$ is exogenous, say independent of $\nu$, this yields a system of equalities relating $(\theta,\sQ)$ to the observed distribution $\sP(\ey,\ex)$ which, as stated above, can be inverted to obtain point identification.
When $\ex$ is endogenous, this reasoning breaks down because the conditional distribution $\sR(\nu|\ex,\ez)$ may change across realizations of $\ex$.
Figure \ref{fig:discrete:choice:endogenous} also offers an instructive way to connect Identification Problem \ref{IP:discrete:choice:endogenous} with the identification problem studied in Section \ref{subsubsec:man:tam02} (as well as with those in Sections \ref{subsec:multiple:eq}-\ref{subsec:auctions} below).
In the latter, the model has set-valued predictions for the \emph{outcome variable} given realizations of the covariates and unobserved heterogeneity terms, which overlap across realizations of the unobserved heterogeneity terms.
In the problem studied here, the model has singleton-valued predictions for the outcome variable of interest $\ey$ as a function of the observable explanatory variables $\ex$ and unobservables $\nu$.
However, for given realization of $\nu$, the model admits \emph{sets} of values for the \emph{endogenous variables} $(\ey,\ex)$, which overlap across realizations of $\nu$.
Because the model is silent on the joint distribution of $(\ex,\nu)$ (except for requiring that the marginal distribution of $\nu$ does not depend on $\ez$), partial identification results.

It is possible to couple the maintained assumptions with the observed data to learn features of $(\theta,\sQ)$. 
Because the observed choice $\ey$ is assumed to maximize utility, for the data generating $(\theta,\sQ)$ the model yields
\begin{align}
\nu\in \cE_\theta(\ey,\ex)\text{-a.s.},\label{eq:che:ros:e_in_E}
\end{align}
with $\cE_\theta(\ey,\ex)$ a random closed set as per Definition \ref{def:rcs}.
Equation \eqref{eq:che:ros:e_in_E} exhausts the modeling content of Identification Problem \ref{IP:discrete:choice:endogenous}.
Theorem \ref{thr:artstein} (as expressed in \eqref{eq:dom-c}) can then be leveraged to extract its empirical content from the observed distribution $\sP(\ey,\ex,\ez)$.
As a preparation for doing so, note that for given $F\in\cF$ (with $\cF$ the collection of closed subsets of $\cV$) and $\vartheta\in\Theta$, we have
\begin{align*}
\sP(\cE_\vartheta(\ey,\ex)\subseteq F|\ez)=\int_{x\in\cX}\sum_{c\in\cY}\one(\cE_\vartheta(c,x)\subseteq F)\sP(\ey=c|\ex=x,\ez)d\sP(x|\ez),
\end{align*}
so that this probability can be learned from the observed data.\begin{SIR}[Discrete Choice with Endogenous Explanatory Variables]\label{SIR:discrete:choice:endogenous}
Under the assumptions of Identification Problem \ref{IP:discrete:choice:endogenous}, the sharp identification region for $(\theta,\sQ)$ is
\begin{align}
\idr{\theta,\sQ}=\left\{\vartheta\in\Theta,\tilde\sQ\in\cT:\tilde\sQ(F)\ge \sP(\cE_\vartheta(\ey,\ex)\subseteq F|\ez),~\forall F\in\cF,~\ez\text{-a.s.}\right\}.\label{eq:SIR:discrete:choice:endogenous}
\end{align}
\end{SIR}
\begin{proof}
To simplify notation, I write $\cE_\vartheta\equiv\cE_\vartheta(\ey,\ex)$.
Let $(\cE_\vartheta,\ex,\ez)=\{(\mathbf{e},\ex,\ez):\mathbf{e}\in\cE_\vartheta\}$.
If the model is correctly specified, $(\nu,\ex,\ez)\in(\cE_\theta,\ex,\ez)$-a.s. for the data generating value of $(\theta,\sQ)$.
Using Theorem \ref{thr:artstein} and Theorem 2.33 in \cite{mol:mol18}, it follows that $(\vartheta,\tilde\sQ)$ is observationally equivalent to $(\theta,\sQ)$ if and only if
\begin{align*}
\tilde\sQ(F|\ex,\ez)\ge \sP(\cE_\vartheta(\ey,\ex)\subseteq F|\ex,\ez),~\forall F\in\cF,~(\ex,\ez)\text{-a.s.}
\end{align*}
As the distribution of $\nu$ is only restricted so that $\nu\independent\ez$, one can integrate both sides of the inequality with respect to $\ex$.
The final result follows because $\tilde\sQ$ does not depend on $\ez$.
\end{proof}
While Theorem SIR-\ref{SIR:discrete:choice:endogenous} relies on checking inequality \eqref{eq:SIR:discrete:choice:endogenous} for all $F\in\cF$, the results in \cite[Theorem 2]{che:ros:smo13} and \cite[Chapter 2]{mol:mol18} can be used to obtain a smaller collection of sets over which to verify it.
In particular, if $\ex$ has a discrete distribution, it suffices to use a finite collection of sets.
For example, in the case depicted in Figure \ref{fig:discrete:choice:endogenous} with $\cX=\{x^1,x^2\}$, \cite[Section 3.3 of the 2011 CeMMAP working paper version CWP39/11]{che:ros:smo13} show that $\idr{\theta,\sQ}$ is obtained by checking at most twelve inequalities in \eqref{eq:SIR:discrete:choice:endogenous}.
The left hand side of these inequalities is a linear function of six values that the distribution $\tilde\sQ$ assigns to each of the component regions depicted in Figure \ref{fig:discrete:choice:endogenous} (the one where $\cE_\vartheta(1,x^1)\cap\cE_\vartheta(1,x^2)$ realizes; the one where $\cE_\vartheta(1,x^1)\cap\cE_\vartheta(3,x^2)$ realizes; etc.)
Hence, in this example, $(\vartheta,\tilde\sQ)\in\idr{\theta,\sQ}$ if and only if $\tilde\sQ$ assigns to these six regions a probability mass such that for $\vartheta$ the twelve inequalities characterized by \citeauthor{che:ros:smo13} hold.
\begin{BI}
A conceptual contribution of \cite{che:ros:smo13} is to show that one can frame models with endogenous explanatory variables as \emph{incomplete} models.
Incompleteness here results from the fact that the model does not specify how the endogenous variables $\ex$ are determined. 
One can then think of these as models with set-valued predictions for the endogeneous variables ($\ey$ and $\ex$ in this application), even though the outcome of the model ($\ey$) is uniquely predicted by the realization of the observed explanatory variables ($\ex$) and the unobserved heterogeneity terms ($\nu$).
Random set theory can again be leveraged to characterize sharp identification regions.
\end{BI}
\cite[Chapter XXX in this Volume]{che:ros19} discuss related generalized instrumental variables models where random set methods are used to obtain characterizations of sharp identification regions in the presence of endogenous explanatory variables.

\subsubsection{Unobserved Heterogeneity in Choice Sets and/or Consideration Sets}
\label{subsubsec:BCMT}
Compared to the general framework set forth at the beginning of Section \ref{subsec:single:ag:RUM}, as pointed out in \cite{man77}, often the researcher observes $(\ey_i,\ex_i)$ but not $\eC_i$, $i=1,\dots,n$.
Even when $\eC_i$ is observable, the researcher may be unaware of which of its elements the decision maker actually evaluates before selecting one.
In what follows, to shorten expressions, I refer to both the measurement problem of unobserved choice sets and the (cognitive) problem of limited consideration as ``unobserved heterogeneity in choice sets."

Learning features of preferences using discrete choice data in the presence of unobserved heterogeneity in choice sets is a formidable task. 
When a decision maker chooses an alternative, this may be because her choice set equals the feasible set and the chosen alternative is the one yielding the highest utility.
Then observed choice reveals  preferences.
But it can also be that the decision maker has access to/considers only the chosen alternative \citep[e.g.,][p. 99]{blo:mar60}.
Then observed choice is driven entirely by choice set composition, and is silent about preferences.
A plethora of scenarios between these extremes is possible, but the researcher does not know which has generated the observed data.
This fundamental identification problem calls either for restrictions on the random utility model and consideration set formation process, or for collection of richer data that eliminates unobserved heterogeneity in $\eC_i$ or allows for enhanced modeling of it \citep[see, e.g.,][]{cap16}.

A sizable literature spanning behavioral economics, econometrics, experimental economics, marketing, microeconomics, and psychology, has put forward different models to formalize the complex process that leads to the formation of the set of alternatives that the agent considers or can choose from \citep[see, e.g.,][for early contributions]{sim59, how63, tve72}.
\cite{man77} proposes both a general econometric model where decision makers draw choice sets from an unknown distribution, as well as a specific model of choice set formation, independent from preferences, and studies their implications for the distributional structure of random utility models.\footnote{The specific model in \cite[Section II-A]{man77} is often used in applications.
It posits that each alternative $c\in\cY$ enters the decision maker's choice set with probability $\phi_c$, independently of the other alternatives.
The probability $\phi_c$ may depend on observable individual characteristics, and $\phi_c=1$ for at least one option $c\in\cY$ (the ``default" good).}

However, assumptions about the choice set formation process are often rooted in a desire to achieve point identification rather than in information contained in the model or observed data.\footnote{These assumptions are akin to assumptions about selection mechanisms in models with multiple equilibria.
The latter are discussed further below in Section \ref{subsubsec:tam03:cil:tam09}, along with their criticisms.}
It is then important to ask what can be learned about decision maker's preferences under minimal assumptions on the choice set formation process.
Allowing for unrestricted dependence between choice sets and preferences, while challenging for identification analysis, is especially relevant.
Indeed, decision makers' unobserved attributes may determine both their preferences and which items in the feasible set they pay attention to or are available to them (e.g., through unobserved liquidity constraints, unobserved characteristics such as religious preferences in the context of school choice, or behavioral phenomena such as aversion to extremes, salience, etc.).
Here I use the framework put forward by \cite{bar:cou:mol:tei18} to study identification of discrete choice models with unobserved heterogeneity in choice sets and preferences.

\begin{IP}[Discrete Choice with Unobserved Heterogeneity in Choice Sets and Preferences]\label{IP:BCMT}
Let $(\ey,\ex)\sim \sP$ be observable random variables in $\cY\times\cX$.
Assume that there exists a real valued function $g$, which for simplicity I posit known up to parameter $\delta\in\Delta\subset\R^m$ and continuous in its second argument, such that $\bu_i(c)=g(\ex_{ic},\nu_i;)$, $(\ex_{ic},\nu_i)$-a.s., for all $c\in\cY,i\in\cI$, where $\ex_{ic}$ denotes the vectors of attributes relevant to alternative $c$, and includes attributes that are alternative invariant and ones that are alternative specific (respectively, $\ex_i^1$ and $\ex_{ic}^2$ in the general notation laid out in Section \ref{subsec:single:ag:RUM}).
Suppose that $\ey=\arg\max_{c\in \eC}g(\ex_c,\nu;\delta)$, where ties are assumed to occur with probability zero and $\eC$ is an unobservable choice set drawn from the subsets of $\cY$ according to some unknown probability distribution.
Suppose $\sR(|\eC|\ge\kappa)=1$ for some known constant $\kappa\ge 2$.
Let $\sQ$ denote the distribution of $\nu$, and assume that it is known up to a finite dimensional parameter $\gamma\in\Gamma\subset\R^k$.
For simplicity, assume that $\nu\independent\ex$.\footnote{This assumption can be relaxed as discussed in \cite{mat07}. The procedure proposed here can also be adapted to allow for endogenous explanatory variables as in Section \ref{subsubsec:CRS} by combining the results in \cite{bar:cou:mol:tei18} with those in \cite{che:ros:smo13}.}
In the absence of additional information, what can the researcher learn about $\theta\equiv[\delta;\gamma]$?
	\qedex
	\end{IP}
\begin{figure}[tp]
\begin{center}
\includegraphics[scale=1.1]{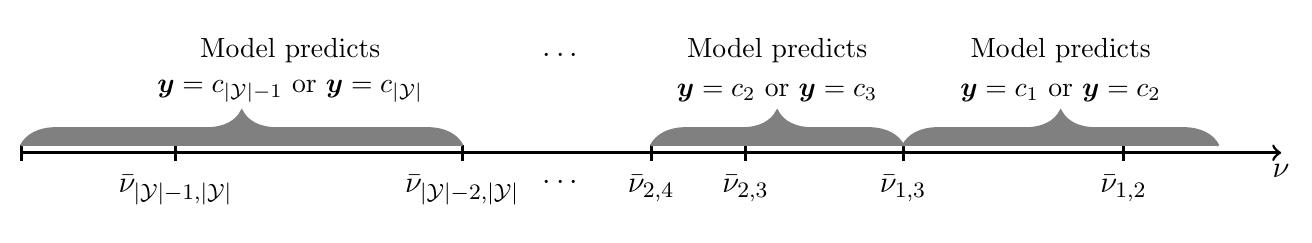}
\caption{\small{Predicted value of $\ey$ in Identification Problem \ref{IP:BCMT} as a function of $\nu$ for $\kappa=|\cY|-1$. In this case, $\eC=\cY\setminus\{c\}$ for some $c\in\cY$, and the model predicts either the first or the second best alternative in $\cY$.}}
	\label{fig:set_valued_BCMT}
\end{center}
\end{figure}

The model just laid out has set valued predictions for the decision maker's optimal choice, because different alternatives might be optimal depending on which choice set the decision maker draws.
Figure \ref{fig:set_valued_BCMT}, which is based on the analysis in \cite{bar:cou:mol:tei18}, illustrates the set valued predictions in a stylized example.
In the figure $\nu$ is assumed to be a scalar; $\bar{\nu}_{j,m}$ denotes the threshold value of $\nu$ above which $c_j$ yields higher utility than $c_m$ and below which $c_m$ yields higher utility than $c_j$ (the threshold's dependence on $(\ex;\delta)$ is suppressed for notational convenience). 
Consider the case that $\nu\in[\bar{\nu}_{2,3},\bar{\nu}_{1,2}]$, so that $c_2$ is the option yielding the highest utility among all options in $\cY$. 
When $\kappa=|\cY|-1$, the agent may draw a choice set that does not include one of the alternatives in $\cY$.
If the excluded alternative is not $c_2$ (or if $\eC$ realizes equal to $\cY$), the model predicts that the decision maker chooses $c_2$.
If $\eC$ realizes equal to $\cY\setminus\{c_2\}$, the model predicts that the decision maker chooses the second best: $c_1$ if $\nu\in[\bar{\nu}_{1,3},\bar{\nu}_{1,2}]$, and $c_3$ if $\nu\in[\bar{\nu}_{2,3},\bar{\nu}_{1,3}]$.
Conversely, observation of $\ey=c_1$ allows one to conclude that $\nu\ge\bar\nu_{1,3}$, and $\ey=c_2$ that $\nu\ge\bar\nu_{2,4}$, with $\bar\nu_{2,4}\le\bar\nu_{1,3}$, and these regions of possible realizations of $\nu$  overlap.

Why does this set valued prediction hinder point identification?
The reason is similar to the explanation given for Identification Problem \ref{IP:man:tam02_binary}:
the distribution of the observable data relates to the model structure in an \emph{incomplete} manner, because the distribution of the (unobserved) choice sets is left completely unspecified.
\cite{bar:cou:mol:tei18} show that one can find multiple candidate distributions for $\eC$ and parameter vectors $\vartheta$, such that together they yield a model implied distribution for $\ey|\ex$ that matches $\sP(\ey|\ex)$, $\ex$-a.s.

\citeauthor{bar:cou:mol:tei18} propose to work directly with the set of model implied optimal choices given $(\ex,\nu)$ associated with each possible realization of $\eC$, which is depicted in Figure \ref{fig:set_valued_BCMT} for a specific example.
The key idea is that, according to the model, the observed choice maximizes utility among the alternatives in $\eC$.
Hence, for the data generating value of $\theta$, it belongs to the set of model implied optimal choices.
With this, the authors are able to characterize $\idr{\theta}$ through Theorem \ref{thr:artstein} as the collection of parameter vectors that satisfy a finite number of conditional moment inequalities.
\begin{BI}
\cite{bar:cou:mol:tei18} show that working directly with the set of model implied optimal choices given $(\ex,\nu)$ allows one to dispense with considering all possible distributions of choice sets that are allowed for in Identification Problem \ref{IP:BCMT} to complete the model.
Such distributions may depend on $\nu$ even after conditioning on observables and may constitute an infinite dimensional nuisance parameter, which creates great difficulties for the computation of $\idr{\theta}$ and for inference.
\end{BI}

Identification Problem \ref{IP:BCMT} sets up a structure where preferences include idiosyncratic components $\nu$ that are decision maker specific and can depend on $\eC$, and where heterogeneity in $\eC$ can be driven either by a measurement problem, or by the decision maker's limited attention to the options available to her.
However, for computational and finite sample inference reasons, it restricts the family of utility functions to be known up to a finite dimensional parameter vector $\delta$.

A rich literature in decision theory has analyzed a different framework, where the decision maker's choice set is observable to the researcher, but the decision maker does not consider all alternatives in it \citep[for recent contributions see, e.g.,][]{mas:nak:ozb12,man:mar14}.
In this literature, the utility function is left completely unspecified, so that interest focuses on identification of preference orderings of the available options.
Unobserved heterogeneity in preferences is assumed away, so that heterogeneous choice is driven by randomness in consideration sets.
If the consideration set formation process is left unspecified or is subject only to weak restrictions, point identification of the preference orderings is not possible even if preferences are homogeneous and the researcher observes a representative agent facing multiple distinct choice problems with varying choice sets.
\cite{cat:ma:mas:sul17} propose a general model for the consideration set formation process where the only restriction is a weak and intuitive monotonicity condition: the probability that any particular consideration set is drawn does not decrease when the number of possible consideration sets decreases.
Within this framework, they provide revealed preference theory and testable implications for observable choice probabilities.
\begin{IP}[Homogeneous Preference Orderings in Random Attention Models]\label{IP:RAM}
Let $(\ey,\eC)\sim\sP$ be a pair of observable random variable and random set in $\cY\times\mathfrak{D}$, where $\mathfrak{D}=\{D:D\subseteq\cY\}\setminus\emptyset$.\footnote{
Here I omit observable covariates $\ex$ for simplicity.}
Let $\mu:\mathfrak{D}\times\mathfrak{D}\to[0,1]$ denote an \emph{attention rule} such that $\mu(A|G)\ge 0$ for all $A\subseteq G$, $\mu(A|G)=0$ for all $A\nsubseteq G$, and $\sum_{A\subset G}\mu(A|G)=1$, $A,G\in\mathfrak{D}$.
Assume that for any $b\in G\setminus A$,
\begin{align}
\label{eq:RAM:monotonicity}
\mu(A|G)\le\mu(A|G\setminus\{b\}),
\end{align}
and that the decision maker has a strict preference ordering $\succ$ on $\cY$.\footnote{
Specifically, $\succ$ is an asymmetric, transitive and complete binary relation.}
In the absence of additional information, what can the researcher learn about $\succ$?
	\qedex
\end{IP}
\cite{cat:ma:mas:sul17} posit that an observed distribution of choice $\sP(\ey|\eC)$ has a random attention representation, and hence they name it a \emph{random attention model}, if there exists a preference ordering $\succ$ over $\cY$ and a monotonic attention rule $\mu$ such that 
\begin{align}
\cp(c|G)\equiv\sP(\ey=c|\eC=G)=\sum_{A\subseteq G}\one(c\text{ is }\succ\text{-best in }A)\mu(A|G),~~\forall c\in G,~\forall G\in\mathfrak{D}.\label{eq:RAM}
\end{align}
The sharp identification region for the preference ordering, denoted $\idr{\succ}$ henceforth, is given by the collection of preference orderings for which one can find a monotonic attention rule to pair it with, so that \eqref{eq:RAM} holds.

Of course, an observed distribution of choice can be represented by multiple preference orderings and attention rules.
The authors, however, show in their Lemma 1 that if for \emph{some} $G\in\mathfrak{D}$ with $\{b,c\}\in G$,
\begin{align}
\cp(c|G)>\cp(c|G\setminus \{b\}),\label{eq:RAM_violation_reg}
\end{align}
then $c \succ b$ for any $\succ$ for which one can find a monotonic attention rule $\mu$ such that \eqref{eq:RAM} holds.
Because of preference transitivity, one can also learn $a\succ b$ if in addition to the above condition one has $\cp(a|G^\prime)>\cp(a|G^\prime\setminus \{c\})$ for some $c\in G^\prime$ and $G^\prime\in\mathfrak{D}$.
The authors further show in their Theorem 1 that the collection of preference relations associated with all possible instances of \eqref{eq:RAM_violation_reg} for all $c\in G$ and $G\in\mathfrak{D}$ yield all information about preferences given the observed choice probabilities.
This yields a system of linear inequalities in $\cp(c|G)$ that fully characterize $\idr{\succ}$.
Let $\vec{\cp}$ denote the vector with elements $[\cp(c|G):c\in G,G\in\mathfrak{D}]$ and $\Pi_\succ$ denote a conformable matrix collecting the constraints on $\sP(\ey|\eC)$ embodied in \eqref{eq:RAM_violation_reg} and its generalizations based on transitive closure. Then
\begin{align}
\idr{\succ}=\{\succ: \Pi_\succ \vec{\cp}\le 0\}.\label{eq:SIR:RAM}
\end{align}
The authors show that for any given preference ordering $\succ$, the matrix $\Pi_\succ$ characterizing whether $\succ \in \idr{\succ}$ through the system of linear inequalities in \eqref{eq:SIR:RAM} is unique, and they provide a simple algorithm to compute it. 
They also show that mild additional assumptions, such as, for example, that decision makers facing binary choice sets pay attention to both alternatives frequently enough, can substantially increase the informational content of the data (i.e., substantially tighten $\idr{\succ}$).
\begin{BI}
\cite{cat:ma:mas:sul17} show that learning features of preference orderings in Identification Problem \ref{IP:RAM} requires the existence in the data of choice problems where the choice probabilities satisfy \eqref{eq:RAM_violation_reg}.
The latter is a violation of the principle of ``regularity" \citep{luc:sup65} according to which the probability of choosing an alternative from any set is at least as large as the probability of choosing it from any of its supersets.
Regularity is a monotonicity property of choice probabilities, and it is implied by a wide array of models of decision making.
The monotonicity of attention rules in \eqref{eq:RAM:monotonicity} can be viewed as regularity of the process that chooses a consideration set from the subsets of the choice set.
\cite{cat:ma:mas:sul17} show that it is implied by various models of limited attention.
While the violation required in \eqref{eq:RAM_violation_reg} is weak in that it needs only to occur for some $G$, it sheds a different light on the severity of the identification problem described at the beginning of this section.
Regularity of choice probabilities and (partial) identification of preference orderings can co-exist only under restrictions on the consideration set formation process that are stronger than the regularity of attention rules in \eqref{eq:RAM:monotonicity}.
\end{BI}

%
\cite{aba:ada18} and \cite{bar:mol:thi19} provide different sets of sufficient conditions for point identification of models of limited consideration.
In both cases, the authors posit specific models of consideration set formation and provide sufficient conditions for point identification under exclusion and large support assumptions.
\cite{aba:ada18} assume that unobserved heterogeneity in preferences and in consideration sets are independent.
They exploit violations of Slutsky symmetry that result from inattention, assuming that for each alternative there is an observable characteristic with large support that does not affect the consideration probability of the other options.
\cite{bar:mol:thi19} provide a thorough analysis of the extent of dependency between consideration and preferences under which semi-nonparametric point identification of the distribution of preferences and consideration attains.
They exploit a requirement of standard economic theory --the Spence-Mirrlees single crossing property of utility functions-- coupled with a mild strengthening of the classic conditions for semi-nonparametric identification of discrete choice models with full consideration and identical choice sets \citep[see, e.g.,][]{mat07}, assuming that there is at least one decision maker-specific characteristic with large support that affects utility but not consideration.

\subsubsection{Prediction of Choice Behavior with Counterfactual Choice Sets}
\label{subsubsec:counterfactual:choice:set}
Building on \cite{mar60}, \cite{man07b} studies a question related but distinct from those in Identification Problems \ref{IP:BCMT}-\ref{IP:RAM}. 
He is concerned with prediction of choice behavior when decision makers face counterfactual choice sets.
\citeauthor{man07b} frames this question as one of predicting treatment response (see Section \ref{subsec:programme:eval}).
Here the collection of potential treatments is given by $\mathfrak{D}$, the nonempty subsets of the universe of feasible alternatives $\cY$, and the response function specifies the alternative chosen by a decision maker when facing choice set $G\in\mathfrak{D}$.
\citeauthor{man07b} assumes that the researcher observes realized choice sets and chosen alternatives, $(\ey,\eC)\sim\sP$.\footnote{Here I suppress covariates for simplicity.}
Under the standard assumptions laid out at the beginning of Section \ref{subsec:single:ag:RUM}, specifically if utility functions are (say) linear in $\epsilon_{ic}$ and the distribution of $\epsilon_{ic}$ is (say) Type I extreme value or multivariate normal, prediction of choice behavior with counterfactual choice sets is immediate (and point identified).
\citeauthor{man07b}, however, leaves utility functions completely unspecified, and in fact works directly with preference orderings, which he labels decision maker's \emph{types}.
He places no restriction on the distribution of preference types, except requiring that they are independent of the observed choice sets.
\citeauthor{man07b} shows that under these rather weak assumptions, the distribution of predicted choices from counterfactual choice sets can be partially identified, and characterized as the solution to linear programs.

Specifically, let $\ey^*(G)$ denote the decision maker's optimal choice when facing choice set $G\in\mathfrak{D}$.
Assume $\ey^*(\cdot)\independent\eC$, and let $y_k$ denote the choice function for a decision maker of type $k$ --that is, a decision maker with a specific preference ordering labeled $k$.
One example of such preference ordering might be $c_1\succ c_2\succ\dots\succ c_{|\cY|}$.
If a decision maker of this type faces, say, choice set $G=\{c_2,c_3,c_4\}$, then she chooses alternative $c_2$.
Let $K$ denote the set of logically possible types, and $\theta_k$ the probability that a decision maker in the population is of type $k$.
Suppose that the researcher posits a behavioral model specifying $K$, $\{y_k,k=1,\dots,K\}$, and restrictions that constrain $\theta$ to lie in some specified set of distributions.
Let $\Theta$ denote the values of $\vartheta$ that satisfy these requirements plus the conditions $\vartheta_k\ge 0$ for all $k\in K$ and $\sum_{k\in K}\vartheta_k=1$.
Then for any $c\in\cY$ and $\vartheta\in\Theta$, the model predicts
\begin{align*}
\sQ(\ey^*(G)=c)=\sum_{k\in K}\one(y_k(G)=c)\vartheta_k.
\end{align*}
How can one partially identify this probability based on the observed data?
Suppose $\eC$ is observed to take realizations $D_1,\dots,D_m$.
Then the data reveal
\begin{align*}
\sP(\ey(D_j)=d_j)=\sum_{k\in K}\one(y_k(D_j)=d_j)\theta_k~\forall d_j\in D_j,j=1,\dots,m.
\end{align*}
This yields that the sharp identification region for $\theta$ is
\begin{align*}
\idr{\theta}=\{\vartheta\in\Theta:~\sP(\ey(D_j)=d_j)=\sum_{k\in K}\one(y_k(D_j)=d_j)\vartheta_k~\forall d_j\in D_j,j=1,\dots,m\}.
\end{align*}
If the behavioral model is correctly specified, $\idr{\theta}$ is non-empty.
In turn, the sharp identification region for each choice probability is
\begin{align*}
\idr{\sQ(\ey^*(G)=c)}=\left\{\sum_{k\in K}\one(y_k(G)=c)\vartheta_k:\vartheta\in\idr{\theta}\right\},
\end{align*}
and its extreme points can be obtained by solving linear programs.

\cite{kit:sto19} provide closely related sharp bounds on features of counterfactual choices in the nonparametric random utility model of demand, where observable choices are repeated cross-sections and one allows for unrestricted, unobserved heterogeneity. 
Their approach builds on the work of \cite{kit:sto18}, who test weather agents' behavior is consistent with the Axiom of Revealed Stochastic Preference (SARP) in a random utility model in which the utility function of each consumer over commodity bundles is assumed to satisfy only the basic restriction that ``more is better" with no satiation.  
Because the testing exercise is to be carried out using repeated cross-sections data, the authors maintain the assumption that multiple populations of consumers who face distinct choice sets have the same distribution of preferences.  
With this structure in place, de facto the task is to test the full implications of rationality without functional form restrictions.
\citeauthor{kit:sto18}'s approach is based on several novel ideas.  
As a first step, they leverage an earlier insight of \cite{mcf05} to discretize the data without loss of information, so that they can define a large but finite set of rational preferences types.  
As a second step, they show that this implies that rationality can be tested by checking whether observed behavior lies in a cone corresponding to positive linear combinations of preference types.  
While the problem is discrete, its dimension is at first sight  prohibitive.  
Nonetheless, Kitamura and Stoye are able to develop novel computational methods that render the problem tractable.  
They apply their method to the U.K. Household Expenditure Survey, adapting to their framework results on nonparametric instrumental variable analysis by \cite{imb:new09} so that they can handle price endogeneity.

\cite{kam18} builds on \cite{man07b} to learn program effects when agents are randomly assigned to control or treatment.
The treatment group is provided access to the program, while the control group is not.
However, members of the control group may receive access to the program from outside the experiment, leading to noncompliance with the randomly assigned treatment.
The researcher wants to learn about the average effect of program access on the decision to participate in the program and on the subsequent outcome.
While sufficiently rich data may allow the researcher to learn these effects, \citeauthor{kam18} is concerned with the identification problem that arises when the researcher only observes the treatment assignment status, the program participation decision, and the outcome, but not the receipt of program access for every agent. 
\citeauthor{kam18} formalizes this problem as one where the received treatment is selected from a choice set that depends on the assigned treatment and is unobservable to the researcher, and the agents optimally choose whether to participate in the program by maximizing their utility function over their choice set.
Importantly, the utility functions are not subject to parametric restrictions, similarly to \cite{man07b}.
But while \citeauthor{man07b} assumed independence of choice sets and preference types, \citeauthor{kam18} allows them to be arbitrarily dependent on each other, as in \cite{bar:cou:mol:tei18}.
\possessivecite{kam18} approach leverages specific assumptions on random assignment of treatments and on compliance (or lack thereof) of participants to obtain nonparametric bounds on the treatment effects of interest that can be characterized using tractable linear programs.

\subsection{Static, Simultaneous-Move Finite Games with Multiple Equilibria}
\label{subsec:multiple:eq}

\subsubsection{An Inference Approach Robust to the Presence of Multiple Equilibria}
\label{subsubsec:tam03:cil:tam09}
\cite{tam03} and \cite{cil:tam09} substantially enlarge the scope of partial identification analysis of structural models by showing how to apply it to learn features of payoff functions in static, simultaneous-move finite games of complete information with multiple equilibria.
\cite{ber:tam06} extend the approach and considerations that follow to games of incomplete information.
To start, here I focus on two-player entry games with complete information.\footnote{Completeness of information is motivated by the idea that firms in the industry have settled in a long-run equilibrium, and have detailed knowledge of both their own and their rivals' profit functions.}
\begin{IP}[Complete Information Two Player Entry Game]
\label{IP:entry_game}
Let $(\ey_1,\ey_2,\ex_1,\ex_2)\sim\sP$ be observable random variables in $\{0,1\}\times\{0,1\}\times\R^d\times\R^d$, $d<\infty$.
Suppose that $(\ey_1,\ey_2)$ result from simultaneous move, pure strategy Nash play (PSNE) in a game where the payoffs are $\bu_j(\ey_j,\ey_{3-j},\ex_j;\beta_j,\delta_j)\equiv \ey_j(\ex_j\beta_j+\delta_j\ey_{3-j}+\eps_j)$, $j=1,2$ and the strategies are ``enter" ($\ey_j=1$) or ``stay out" ($\ey_j=0$).
Here $(\ex_1,\ex_2)$ are observable payoff shifters, $(\eps_1,\eps_2)$ are payoff shifters observable to the players but not to the econometrician, $\delta_1\le 0,\delta_2\le 0$ are interaction effect parameters, and $\beta_1,\beta_2$ are parameter vectors in $B\subset\R^d$ reflecting the effect of the observable covariates on payoffs. 
Each player enters the market if and only if entering yields non-negative payoff, so that $\ey_j=\one(\ex_j\beta_j+\delta_j\ey_{3-j}+\eps_j\ge 0)$.
For simplicity, assume that $\eps\equiv(\eps_1,\eps_2)$ is independent of $\ex\equiv(\ex_1,\ex_2)$ and has bivariate Normal distribution with mean vector zero, variances equal to one (a normalization required by the threshold crossing nature of the model), and correlation $\rho\in [-1,1]$.
In the absence of additional information, what can the researcher learn about $\theta=[\delta_1~\delta_2~\beta_1~\beta_2~\rho]$?
	\qedex
\end{IP}
From the econometric perspective, this is a generalization of a standard discrete choice model to a bivariate simultaneous response model which yields a stochastic representation of equilibria in a two player, two action game.
Generically, for a given value of $\theta$ and realization of the payoff shifters, the model just laid out admits multiple equilibria (existence of PSNE is guaranteed because the interaction parameters are non-positive).
In other words, it yields set valued predictions as depicted in Figure \ref{fig:set_valued_pred:tam03}.\footnote{This figure is based on Figure 1 in \cite{tam03}.}

Why does this set valued prediction hinder point identification?
Intuitively, the challenge can be traced back to the fact that for different values of $\theta\in\Theta$, one may find different ways to assign the probability mass in $[-\ex_1\beta_1,-\ex_1\beta_1-\delta_1)\times [-\ex_2\beta_2,-\ex_2\beta_2-\delta_2)$ to $(0,1)$ and $(1,0)$, so as to match the observed distribution $\sP(\ey_1,\ey_2|\ex_1,\ex_2)$.
More formally, for fixed $\vartheta\in\Theta$ and given $(\ex,\eps)$ and $(y_1,y_2)\in\{0,1\}\times\{0,1\}$, let
\begin{align*}
\cE_\vartheta[(1,0),(0,1);\ex]&\equiv[-\ex_1\beta_1,-\ex_1\beta_1-\delta_1)\times [-\ex_2\beta_2,-\ex_2\beta_2-\delta_2),\\
\cE_\vartheta[(y_1,y_2);\ex]&\equiv\{(\eps_1,\eps_2):(y_1,y_2)~\text{is the unique equilibrium}\},
\end{align*} 
so that in Figure \ref{fig:set_valued_pred:tam03} $\cE_\vartheta[(1,0),(0,1);\ex]$ is the gray region, $\cE_\vartheta[(0,1);\ex]$ is the dotted region, etc.
Let $\sR(y_1,y_2|\ex,\eps)$ be a \emph{selection mechanism} that assigns to each possible outcome of the game $(y_1,y_2)\in\{0,1\}\times\{0,1\}$ the probability that it is played conditional on observable and unobservable payoff shifters.
In order to be \emph{admissible}, $\sR(y_1,y_2|\ex,\eps)$ must be such that $\sR(y_1,y_2|\ex,\eps)\ge 0$ for all $(y_1,y_2)\in\{0,1\}\times\{0,1\}$, $\sum_{(y_1,y_2)\in\{0,1\}\times\{0,1\}}\sR(y_1,y_2|\ex,\eps)=1$, and
\begin{align}
\forall \eps\in\cE_\vartheta[(1,0),(0,1);\ex],~&\sR(0,0|\ex,\eps)=\sR(1,1|\ex,\eps)=0 \label{eq:games:sel:mec:1}\\
\forall\eps\in\cE_\vartheta[(y_1,y_2);\ex],~&\sR(\tilde y_1,\tilde y_2|\ex,\eps)=0 ~\forall(\tilde y_1,\tilde y_2)\in\{0,1\}\times\{0,1\}~\text{s.t. }(\tilde y_1,\tilde y_2)\neq(y_1,y_2).\label{eq:games:sel:mec:2}
\end{align} 
Let $\Phi_r$ denote the probability distribution of a bivariate Normal random variable with zero means, unit variances, and correlation $r\in[-1,1]$.
Let $\sM(y_1,y_2|\ex)$ denote the model predicted probability that the outcome of the game realizes equal to $(y_1,y_2)$.
Then the model yields
\begin{align}
\sM(y_1,y_2|\ex)&=\int\sR(y_1,y_2|\ex,\eps)d\Phi_r\notag\\
&=\int_{(\eps_1,\eps_2)\in\cE_\vartheta[(y_1,y_2);\ex]}d\Phi_r+\int_{\eps_1,\eps_2\in\cE_\vartheta[(1,0),(0,1);\ex]}\sR(y_1,y_2|\ex,\eps)d\Phi_r.\label{eq:games_model:pred}
\end{align}
Because $\sR(\cdot|\ex,\eps)$ is left completely unspecified, other than the basic restrictions listed above that render it an admissible selection mechanism, one can find multiple values for $(\vartheta,\sR(\cdot|\ex,\eps))$ such that $\sM(y_1,y_2|\ex)=\sP(y_1,y_2|\ex)$ for all $(y_1,y_2)\in\{0,1\}\times\{0,1\}$ $\ex$-a.s.
\medskip
\begin{figure}[tp]
\centering
\includegraphics[scale=1]{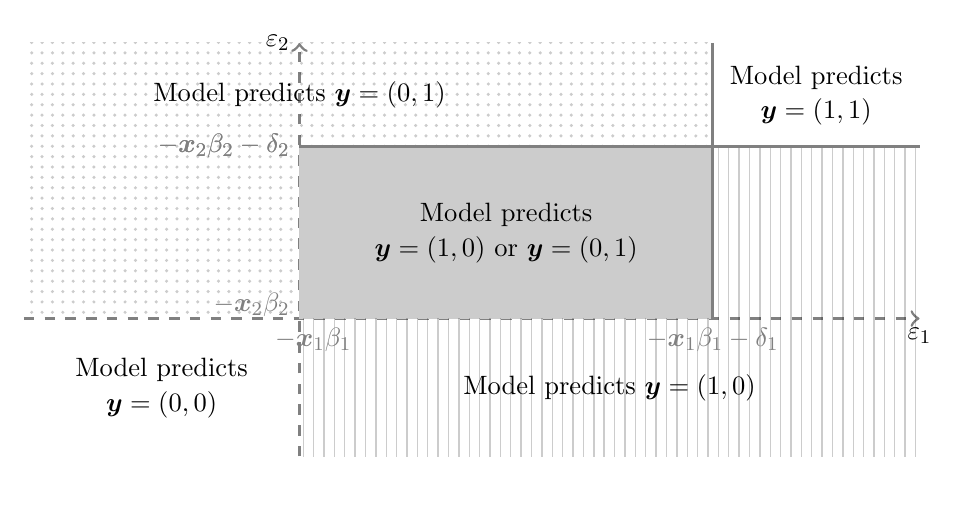}	
		\caption{\small{PSNE outcomes of the game in Identification Problem \ref{IP:entry_game} as a function of $(\eps_1,\eps_2)$.}}
\label{fig:set_valued_pred:tam03}
\end{figure}

Multiplicity of equilibria implies that the mapping from the model's exogenous variables $(\ex_1,\ex_2,\eps_1,\eps_2)$ to outcomes $(\ey_1,\ey_2)$ is a correspondence rather than a function. 
This violates the classical \textquotedblleft principal assumptions\textquotedblright\ or \textquotedblleft coherency conditions\textquotedblright\ for simultaneous discrete response models discussed extensively in the econometrics literature \citep[e.g.,][]{hec78,gou80,sch81,mad83,blu:smi94}.
Such coherency conditions require the existence of a unique reduced form, mapping the model's exogenous variables and parameters to a
unique realization of the endogenous variable; hence, they constrain the model to be recursive or triangular in nature.
As pointed out by \cite{bjo:vuo84}, however, the coherency conditions shut down exactly the social interaction effect of interest by requiring, e.g., that $\delta_1\delta_2=0$, so that at least one player's action has no impact on the other player's payoff.

The desire to learn about interaction effects coupled with the difficulties generated by multiplicity of equilibria prompted the earlier literature to provide at least two different ways to achieve point identification.
The first one relies on imposing simplifying assumptions that shift focus to outcome features that are common across equilibria. 
For example, \cite{bre:rei88,bre:rei90,bre:rei91} and \cite{ber92} study entry games where the number, though not the identities, of entrants is uniquely predicted by the model in equilibrium.
Unfortunately, however, these simplifying assumptions substantially constrain the amount of heterogeneity in player's payoffs that the model allows for.
The second approach relies on explicitly modeling a selection mechanism which specifies the equilibrium played in the regions of multiplicity.
For example, \cite{bjo:vuo84} assume it to be a constant; \cite{baj:hon:rya10} assume a more flexible, covariate dependent parametrization; and \cite{ber92} considers two possible selection mechanism specifications, one where the incumbent moves first, and the other where the most profitable player moves first.
Unfortunately, however, the chosen selection mechanism can have non-trivial effects on inference, and the data and theory might be silent on which is more appropriate. 
A nice example of this appears in \cite[Table VII]{ber92}.
\cite{ber:tam06} review and extend a number of results on the identification of entry models extensively used in the empirical
literature.
\cite{jov89} discusses the observable implications of models with multiple equilibria, and within the analysis of a model with homogeneous preferences shows that partial identification is possible \citep[see][p. 1435]{jov89}.
I refer to \cite{pau13} for a review of the literature on econometric analysis of games with multiple equilibria.

\cite{cil:tam09} show, on the other hand, that it is possible to partially identify entry models that allow for rich heterogeneity in payoffs and for any possible selection mechanism (even ones that are arbitrarily dependent on the unobservable payoff shifters after conditioning on the observed payoff shifters).
In addition, \cite{tam03} provides sufficient conditions for point identification based on exclusion restrictions and large support assumptions.
\cite{kli:tam12} analyze partial identification of nonparametric models of entry in a two-player model, drawing connections with the program evaluation literature.

\begin{BI}
An important conceptual contribution of \cite{tam03} is to clarify the distinction between a model which is \emph{incoherent}, so that no reduced form exists, and a model which is \emph{incomplete}, so that multiple reduced forms may exist. 
Models with multiple equilibria belong to the latter category.
Whereas the earlier literature in partial identification had been motivated by \emph{measurement problems}, e.g., missing or interval data, the work of \cite{tam03} and \cite{cil:tam09} is motivated by the fact that economic theory often does not specify how an equilibrium is selected in the regions of the exogenous variables which admit multiple equilibria.
This is a conceptually completely distinct identification problem.
\end{BI}

\cite{cil:tam09} propose to use simple and tractable implications of the model to learn features of the structural parameters of interest.
Specifically, they point out that the probability of observing any outcome of the game cannot be smaller than the model's implied probability that such outcome is the \emph{unique} equilibrium of the game, and cannot be larger than the model's implied probability that such outcome is \emph{one of the possible} equilibria of the game.
Looking at Figure \ref{fig:set_valued_pred:tam03} this means, for example, that the observed $\sP((\ey_1,\ey_2)=(0,1)|\ex_1,\ex_2)$ cannot be smaller than the probability that $(\eps_1,\eps_2)$ realizes in the dotted region, and cannot be larger than the probability that it realizes either in the dotted region or in the gray region.
Compared to the model predicted distribution in \eqref{eq:games_model:pred}, this means that $\sP((\ey_1,\ey_2)=(0,1)|\ex_1,\ex_2)$ cannot be smaller than the expression obtained setting, for $\eps\in\Eps_\vartheta[(1,0);(0,1);\ex]$, $\sR(0,1|\ex,\eps)=0$, and cannot be larger than that obtained with $\sR(0,1|\ex,\eps)=1$.
Denote by $\Phi(A_1,A_2;\rho)$ the probability that the bivariate normal with mean vector zero, variances equal to one, and correlation $\rho$ assigns to the event $\{\eps_1\in A_1,\eps_2\in A_2\}$.
Then \cite{cil:tam09} show that any $\vartheta=[d_1,d_2,b_1,b_2,r]$ that is observationally equivalent to the data generating value $\theta$ satisfies, $(\ex_1,\ex_2)$-a.s.,
\begin{align}
\sP((\ey_1,\ey_2)=(0,0)|\ex_1,\ex_2)&=\Phi((-\infty,-\ex_1b_1),(-\infty,-\ex_2b_2);r)\label{eq:CT_00}\\
\sP((\ey_1,\ey_2)=(1,1)|\ex_1,\ex_2)&=\Phi([-\ex_1b_1-d_1,\infty),[-\ex_2b_2-d_2,\infty);r)\label{eq:CT_11}\\
\sP((\ey_1,\ey_2)=(0,1)|\ex_1,\ex_2)&\le\Phi((-\infty,-\ex_1b_1-d_1),(-\ex_2b_2,\infty);r)\label{eq:CT_01U}\\
\sP((\ey_1,\ey_2)=(0,1)|\ex_1,\ex_2)&\ge\Big\{\Phi((-\infty,-\ex_1b_1-d_1),(-\ex_2b_2,\infty);r)\notag\\
&\quad\quad-\Phi((-\ex_1b_1,-\ex_1b_1-d_1),(-\ex_2b_2,-\ex_2b_2-d_2);r)\Big\}\label{eq:CT_01L}
\end{align}
While the approach of \cite{cil:tam09} is summarized here for a two player entry game, it extends without difficulty to any finite number of players and actions and to solution concepts other than pure strategy Nash equilibrium.

\cite{ara:tam08} build on the insights of \cite{cil:tam09} to study what is the identification power of equilibrium in games.
To do so, they compare the set-valued model predictions and what can be learned about $\theta$ when one assumes only level-$k$ rationality as opposed to Nash play.
In static entry games of complete information, they find that the model's predictions when $k\ge 2$ are similar to those obtained with Nash behavior and allowing for multiple equilibria and mixed strategies.
\cite{mol:ros08} extend the analysis of \cite{ara:tam08} to the class of supermodular games.
\medskip

The collections of parameter vectors satisfying (in)equalities \eqref{eq:CT_00}-\eqref{eq:CT_01L} yields the sharp identification region $\idr{\theta}$ in the case of two player entry games with pure strategy Nash equilibrium as solution concept, as shown by \cite[Supplementary Appendix D, Corollary D.4]{ber:mol:mol11}.
When there are more than two players or more than two actions \citep[or with different solutions concepts, such as, e.g., mixed strategy Nash equilibrium; correlated equilibrium; or rationality of level $k$ as in][]{ara:tam08}, the characterization in \cite{cil:tam09} obtained by extending the reasoning just laid out yields an outer region.
\cite{ber:mol:mol11} use elements of random set theory to provide a general and computationally tractable characterization of the identification region that is sharp, regardless of the number of players and actions, or the solution concept adopted.
For the case of PSNE with any finite number of players or actions, \cite{gal:hen11} provide a computationally tractable sharp characterization of the identification region using elements of optimal transportation theory.

\subsubsection{Characterization of Sharpness through Random Set Theory}
\label{subsubsec:sharp:games}
\cite{ber:mol:mol11} provide a general approach based on random set theory that delivers sharp identification regions on parameters of structural semiparametric models with set valued predictions.
Here I summarize it for the case of static, simultaneous move finite games of complete information, first with PSNE as solution concept and then with mixed strategy Nash equilibrium.
Then I discuss games of incomplete information.

For a given $\vartheta\in\Theta$, denote the set of pure strategy Nash equilibria (depicted in Figure \ref{fig:set_valued_pred:tam03}) as $\eY_\vartheta(\ex,\eps)$.
It is easy to show that $\eY_\vartheta(\ex,\eps)$ is a random closed set as in Definition \ref{def:rcs}.
Under the assumption in Identification Problem \ref{IP:entry_game} that $\ey$ results from simultaneous move, pure strategy Nash play, at the true DGP value of $\theta\in\Theta$, one has
\begin{align}
\ey\in\eY_\theta~\text{a.s.}\label{eq:y_in_Y_games}
\end{align}
Equation \eqref{eq:y_in_Y_games} exhausts the modeling content of Identification Problem \ref{IP:entry_game}.
Theorem \ref{thr:artstein} can be leveraged to extract its empirical content from the observed distribution $\sP(\ey,\ex)$.
For a given $\vartheta\in\Theta$ and $K\subset\cY$, let $\sT_{\eY_{\vartheta}(\ex,\eps)}(K;\Phi_r)$ denote the probability of the event $\{\eY_\vartheta(\ex,\eps)\cap K\neq \emptyset\}$ implied when $\eps\sim\Phi_r$, $\ex$-a.s.
\begin{SIR}[Structural Parameters in Static, Simultaneous Move Finite Games of Complete Information with PSNE]
\label{SIR:entry_game}
Under the assumptions of Identification Problem \ref{IP:entry_game}, the sharp identification region for $\theta$ is
\begin{align}
\idr{\theta}=\{\vartheta\in\Theta:\sP(\ey\in K|\ex)\le \sT_{\eY_{\vartheta}(\ex,\eps)}(K;\Phi_r)\,\forall K\subset\cY, \, \ex\text{-a.s.}\}.\label{eq:SIR:entry_game}
\end{align} 
\end{SIR}
\begin{proof}
To simplify notation, let $\eY_\vartheta\equiv \eY_\vartheta(\ex,\eps)$.
In order to establish sharpness, it suffices to show that $\vartheta\in \idr{\theta}$ if and only if one can complete the model with an admissible selection mechanism $\sR(y_1,y_2|\ex,\eps)$ such that $\sR(y_1,y_2|\ex,\eps)\ge 0$ for all $(y_1,y_2)\in\{0,1\}\times\{0,1\}$, $\sum_{(y_1,y_2)\in\{0,1\}\times\{0,1\}}\sR(y_1,y_2|\ex,\eps)=1$, and satisfying \eqref{eq:games:sel:mec:1}-\eqref{eq:games:sel:mec:2}, so that $\sM(y_1,y_2|\ex)=\sP(y_1,y_2|\ex)$ for all $(y_1,y_2)\in\{0,1\}\times\{0,1\}$ $\ex$-a.s., with $\sM(y_1,y_2|\ex)$ defined in \eqref{eq:games_model:pred}. 
Suppose first that $\vartheta$ is such that a selection mechanism with these properties is available. 
Then there exists a selection of $\eY_\vartheta$ which is equal to the prediction selected by the selection mechanism and whose conditional distribution is equal to $\sP(\ey|\ex)$, $\ex$-a.s., and therefore $\vartheta \in \idr{\theta}$.
Next take $\vartheta \in \idr{\theta}$. 
Then by Theorem~\ref{thr:artstein}, $\ey$ and $\eY_\vartheta$ can be realized on the same probability space as random elements $\ey'$ and $\eY'_\vartheta$, so that $\ey'$ and $\eY'_\vartheta$ have the same distributions, respectively, as $\ey$ and $\eY_\vartheta$, and $\ey' \in \Sel(\eY'_\vartheta)$, where $\Sel(\eY'_\vartheta)$ is the set of all measurable selections from $\eY'_\vartheta$, see Definition \ref{def:selection}. 
One can then complete the model with a selection mechanism that picks $\ey'$ with probability 1, and the result follows.
\end{proof}
The characterization provided in Theorem SIR-\ref{SIR:entry_game} for games with multiple PSNE, taken from \cite[Supplementary Appendix D]{ber:mol:mol11}, is equivalent to the one in \cite{gal:hen11}.
When $J=2$ and $\cY=\{0,1\}\times\{0,1\}$, the inequalities in \eqref{eq:SIR:entry_game} reduce to \eqref{eq:CT_00}-\eqref{eq:CT_01L}.
With more players and/or more actions, the inequalities in \eqref{eq:SIR:entry_game} are a superset of those in \eqref{eq:CT_00}-\eqref{eq:CT_01L}, with the latter comprised of the ones in \eqref{eq:SIR:entry_game} for $K=\{k\}$ and $k=\cY\setminus\{k\}$, for all $k\in\cY$.
Hence, the inequalities in \eqref{eq:SIR:entry_game} are more informative.
Of course, the computational cost incurred to characterize $\idr{\theta}$ may grow with the number of inequalities involved.
I discuss computational challenges in partial identification in Section \ref{sec:computations}.

\begin{BI}(Random set theory and partial identification -- continued)
In Identification Problem \ref{IP:entry_game} lack of point identification can be traced back to the set valued predictions delivered by the model, which in turn derive from the model incompleteness defined by \cite{tam03}.
As stated in the Introduction, constructing the (random) set of model predictions delivered by the maintained assumptions is an exercise typically carried out in identification analysis, regardless of whether random set theory is applied.
Indeed, for the problem studied in this section, \cite[Figure 1]{tam03} put forward the set of admissible outcomes of the game.
\cite{ber:mol:mol11} propose to work directly with this random set to characterize $\idr{\theta}$.
The fundamental advantage of this approach is that it dispenses with considering the possible selection mechanisms that may complete the model.
Selection mechanisms may depend on the model's unobservables even after conditioning on observables and may constitute an infinite dimensional nuisance parameter, which creates great difficulties for the computation of $\idr{\theta}$ and for inference.
\end{BI}

Next, I discuss the case that the outcome of the game results from simultaneous move, mixed strategy Nash play.\footnote{The same reasoning given here applies if instead of mixed strategy Nash the solution concept is correlated equilibrium, by replacing the set of MSNE below with the set of correlated equilibria.}
When mixed strategies are allowed for, the model predicts multiple mixed strategy Nash equilibria (MSNE).
But whereas when only pure strategies are allowed for, if the model is correctly specified, the observed outcome of the game is one of the predicted PSNE, with mixed strategy it is only the result of a random mixing draw from one of the predicted MSNE.
Hence, the identification problem is more complex, and in order to obtain a tractable characterization of $\theta$'s sharp identification region one needs to use different tools from random set theory.

To keep the treatment simple here I continue to consider the case of two players with two strategies, as in Identification Problem \ref{IP:entry_game}, with mixed strategies allowed for, and refer to \cite[Section 3.4]{mol:mol18} for the general case.
Fix $\vartheta\in\Theta$.
Let $\sigma_j:\{0,1\}\to [0,1]$ denote the probability that player $j$ enters the market, with $1-\sigma_j$ the probability that she stays out.
With some abuse of notation, let $\bu_j(\sigma_j,\sigma_{-j},\ex_j,\eps_j,\vartheta)$ denote the expected payoff associated with the mixed strategy profile $\sigma=(\sigma_1,\sigma_2)$.
For a given realization $(x,e)$ of $(\ex,\eps)$ and a given value of $\vartheta\in\Theta$, the set of mixed strategy Nash equilibria is
\begin{multline*}
  S_\vartheta(x,e) 
  =\bigg\{\sigma \in [0,1]^2:\;
  \bu_j(\sigma_j,\sigma_{-j},x_j,e_j;\vartheta)
  \geq \bu_j(\tilde{\sigma}_j,\sigma_{-j},x_j,e_j;\vartheta)\;\,
  \forall \tilde{\sigma}_j\in [0,1]\; j=1,2\bigg\}.
\end{multline*}
\cite{ber:mol:mol11} show that $\eS_\vartheta\equiv S_\vartheta(\ex,\eps)$ is a random closed set in $[0,1]^2$.
Its realizations are illustrated in Panel (a) of Figure \ref{fig:set_valued_pred:MSNE} as a function of $(\eps_1,\eps_2)$.\footnote{This figure is based on Figure 1 in \cite{ber:mol:mol11}.} 
\begin{figure}[tp]
\centering
\includegraphics[scale=1]{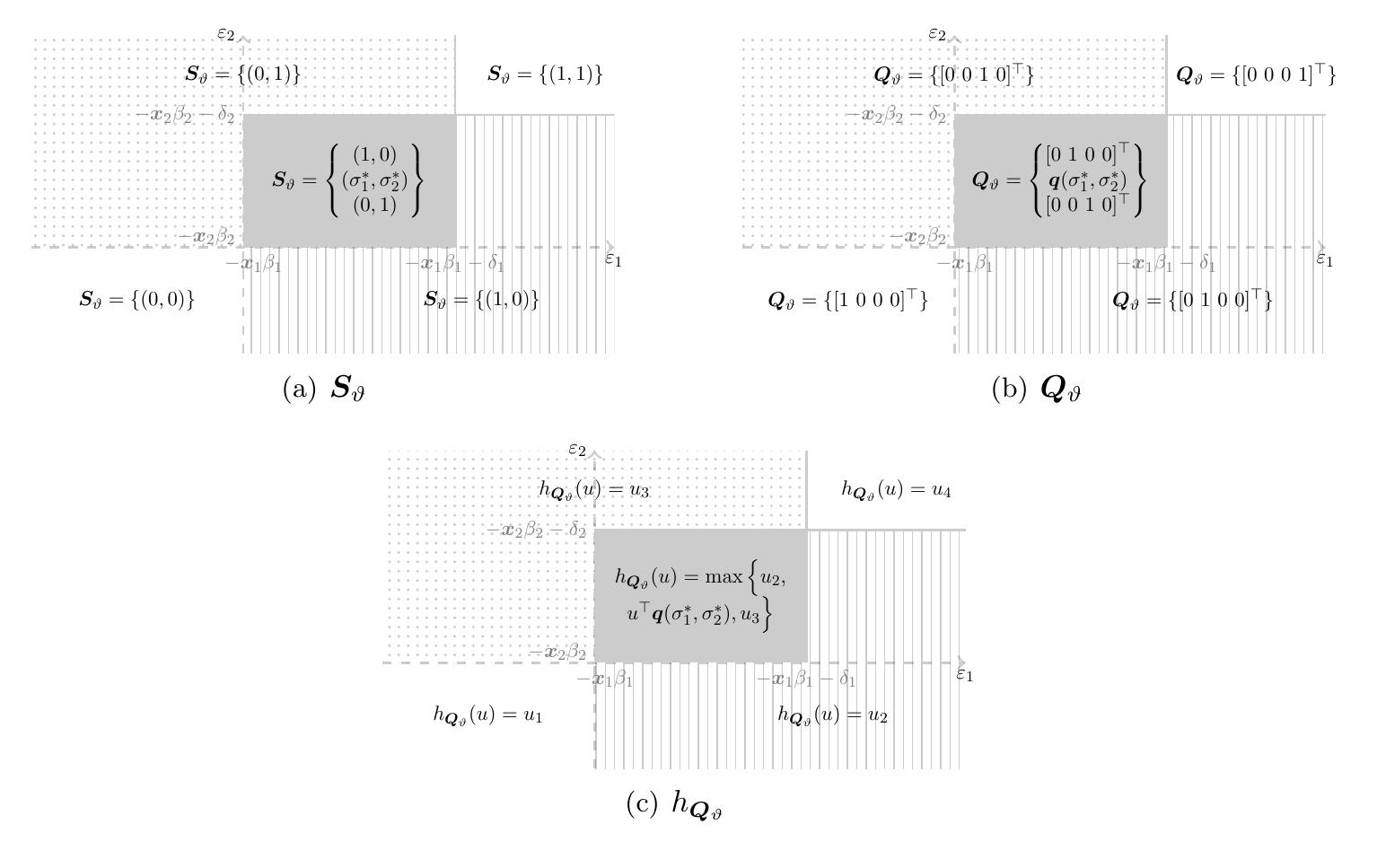}		
\caption{\small{MSNE strategies ($\eS_\vartheta$), set of multinomial distributions over outcomes of the game ($\eQ_\vartheta$), and its support function ($h_{\eQ_\vartheta}$), as a function of $(\eps_1,\eps_2)$, where $\sigma_1^*\equiv\frac{-\eps_2-\ex_2\beta_2}{\vartheta_2},\sigma_2^*\equiv\frac{-\eps_1-\ex_1\beta_1}{\vartheta_1}$.}}
\label{fig:set_valued_pred:MSNE}
\end{figure}

Define the set of possible multinomial distributions over outcomes of the game associated with the selections $\sigma$ of each possible realization of $\eS_{\vartheta}$ as
\begin{equation}
\label{eq:Q-set}
  \eQ_\vartheta=\left\{\eq(\sigma)\equiv \begin{bmatrix}
  (1-\sigma_1)(1-\sigma_2)\\
  \sigma_1(1-\sigma_2)\\
  (1-\sigma_1)\sigma_2\\
  \sigma_1\sigma_2
  \end{bmatrix}:\, \sigma \in \eS_\vartheta\right\}.
\end{equation}
As $\eQ_\vartheta$ is the image of a continuous map applied to the random compact set $\eS_\vartheta$, it is a random compact set.
Its realizations are plotted in Panel (b) of Figure \ref{fig:set_valued_pred:MSNE} as a function of $(\eps_1,\eps_2)$.

The multinomial distribution over outcomes of the game determined by a given $\sigma\in\eS_\vartheta$ is a function of $\eps$.
To obtain the predicted distribution over outcomes of the game conditional on observed payoff shifters only, one needs to integrate out the unobservable payoff shifters $\eps$.
Doing so requires care, as it needs to be done for each $\eq(\sigma)\in\eQ_\vartheta$.
First, observe that all the $\eq(\sigma)\in\eQ_\vartheta$ are contained in the $3$ dimensional unit simplex, and are therefore integrable. 
Next, define the conditional selection expectation (see Definition \ref{def:sel-exp}) of $\eQ_\vartheta$ as
\begin{displaymath}
  \E_{\Phi_r}(\eQ_\vartheta|\ex)=\Big\{\E_{\Phi_r}(\eq(\sigma)|\ex):\;
  \sigma \in \Sel(\eS_\vartheta)\Big\},
\end{displaymath}
where $\Sel(\eS_\vartheta)$ is the set of all measurable selections from $\eS_\vartheta$, see Definition \ref{def:selection}.
By construction, $\E_{\Phi_r}(\eQ_\vartheta|\ex)$ is the set of probability distributions over action profiles conditional on $\ex$ which are
consistent with the maintained modeling assumptions, i.e., with \emph{all} the model's implications (including the assumption that $\eps\sim\Phi_r$). 
If the model is correctly specified, there exists at least one vector $\theta\in\Theta$ such that the observed conditional distribution $\cp(\ex)\equiv[\sP(\ey=y^1|\ex),\dots,\sP(\ey=y^4|\ex)]^\top$ almost surely belongs to the set $\E_{\Phi_\rho}(\eQ_\theta|\ex)$.  
Indeed, by the definition of $\E_{\Phi_\rho}(\eQ_\theta|\ex)$, $\cp(\ex)\in \E_{\Phi_\rho}(\eQ_\theta|\ex)$ almost surely if and only if there exists $\eq\in \Sel(\eQ_\theta)$ such that $\E_{\Phi_\rho}(\eq|\ex)=\cp(\ex)$ almost surely, with $\Sel(\eQ_\theta)$ the set of all measurable selections from $\eQ_\theta$. 
Hence, the collection of parameter vectors $\vartheta\in\Theta$ that are observationally equivalent to the data generating value $\theta$ is given by the ones that satisfy $\cp(\ex)\in \E_{\Phi_r}(\eQ_\vartheta|\ex)$ almost surely.
In turn, observing that by Theorem \ref{thr:exp-supp} the set $\E_{\Phi_r}(\eQ_\vartheta|\ex)$ is convex, we have that $\cp(\ex)\in \E_{\Phi_r}(\eQ_\vartheta|\ex)$ if and only if $u^\top \cp(\ex)\leq h_{\E_{\Phi_r}(\eQ_\vartheta|\ex)}(u)$ for all $u$ in the unit ball \citep[see, e.g.,][Theorem 13.1]{roc70}, where $h_{\E_{\Phi_r}(\eQ_\vartheta|\ex)}(u)$ is the support function of $\E_{\Phi_r}(\eQ_\vartheta|\ex)$, see Definition \ref{def:sup-fun}.
\begin{SIR}[Structural Parameters in Static, Simultaneous Move Finite Games of Complete Information with MSNE]
  \label{SIR:sharpness_mixed}
  Under the assumptions in Identification Problem \ref{IP:entry_game}, allowing for mixed strategies and with the observed outcomes of the game resulting from mixed strategy Nash play, the sharp identification region for $\theta$ is
  \begin{align}
    \idr{\theta} &=\bigg\{\vartheta\in \Theta:\; \max_{u\in\mathbb{B}^{|\cY|}}\left( u^\top \cp(\ex)
    -\E_{\Phi_r}[h_{\eQ_\vartheta}(u)|\ex]\right)=0,\, \ex\text{-a.s.}\bigg\}  
    \label{eq:SIR_sharp_mixed_sup}\\
    &=\bigg\{\vartheta \in \Theta:\; \int_{\mathbb{B}^{|\cY|}} (u^\top \cp(\ex)
    -\E_{\Phi_r}[h_{\eQ_\vartheta}(u)|\ex])_+ \mathrm{d}\mu(u)=0,\, \ex\text{-a.s.}\bigg\}
    \label{eq:SIR_sharp_mixed_int},
  \end{align}
  where $\mu$ is any probability measure on $\mathbb{B}^{|\cY|}$, and $|\cY|=4$ in this case.
\end{SIR}
\begin{proof}
Theorem \ref{thr:exp-supp} (equation \eqref{eq:dom_Aumann:cond}) yields \eqref{eq:SIR_sharp_mixed_sup}, because by the arguments given before the theorem, $\idr{\theta}=\{\vartheta \in \Theta:\;\cp(\ex)\in \E_{\Phi_r}(\eQ_\vartheta|\ex),~\ex\text{-a.s.}\}$.
The result in \eqref{eq:SIR_sharp_mixed_int} follows because the integrand in \eqref{eq:SIR_sharp_mixed_int} is continuous in $u$ and both conditions inside the curly brackets are satisfied if and only if $u^\top \cp(\ex)-\E_{\Phi_r}[h_{\eQ_\vartheta}(u)|\ex]\leq 0$ $\forall u\in \mathbb{B}^{|\cY|}$ $\ex$-a.s.
\end{proof}
For a fixed $u\in\mathbb{B}^4$, the possible realizations of $h_{\eQ_\vartheta}(u)$ are plotted in Panel (c) of Figure \ref{fig:set_valued_pred:MSNE} as a function of $(\eps_1,\eps_2)$.
The expectation of $h_{\eQ_\vartheta}(u)$ is quite straightforward to compute, whereas calculating the set $\E_{\Phi_r}(\eQ_\vartheta|\ex)$ is computationally prohibitive in many cases. 
Hence, the characterization in \eqref{eq:SIR_sharp_mixed_sup} is computationally attractive, because for each $\vartheta\in\Theta$ it requires to maximize an easy-to-compute superlinear, hence concave, function over a convex set, and check if the resulting objective value vanishes.  
Several efficient algorithms in convex programming are available to solve this problem, see for example the MatLab software for disciplined convex programming CVX \citep{gra:boy10}.
Nonetheless, $\idr{\theta}$ itself is not necessarily convex, hence tracing out its boundary is non-trivial.
I return to computational challenges in partial identification in Section \ref{sec:computations}.\medskip

\begin{BI}[Random set theory and partial identification -- continued]
\cite{ber:mol:mol11} provide a general characterization of sharp identification regions for models with \emph{convex moment predictions}.
These are models that for a given $\vartheta\in\Theta$ and realization of observable variables, predict a set of values for a vector of variables
of interest. 
This set is \emph{not} necessarily convex, as exemplified by $\eY_\vartheta$ and $\eQ_\vartheta$, which are finite.
No restriction is placed on the manner in which, in the DGP, a specific model prediction is selected from this set. 
When the researcher takes conditional expectations of the resulting elements of this set, the unrestricted process of selection yields a convex set of moments for the model variables (all possible mixtures).
This is the model's convex set of moment predictions. 
If this set were almost surely single valued, the researcher would learn (features of) $\theta$ by solving moment equality conditions involving the observed variables and predicted ones.
The approach reviewed in this section is a set-valued method of moments that extends the singleton-valued one commonly used in econometrics.
\end{BI}
I conclude this section discussing the case of static, simultaneous move finite games of incomplete information, using the results in \cite[Supplementary Appendix C]{ber:mol:mol11}.\footnote{See \cite[Section 3]{ber:tam06} and \cite{pau13} for a thorough discussion of the literature on identification problems in games of incomplete information with multiple Bayesian Nash equilibria (BNE).
\cite{ber:tam06} explain how to extend the approach proposed by \cite{cil:tam09} to obtain outer regions on $\theta$ when no restrictions are imposed on the equilibrium selection mechanism that chooses among the multiple BNE.}
For clarity, I formalize the maintained assumptions.
\begin{IP}[Structural Parameters in Static, Simultaneous Move Finite Games of Incomplete Information with multiple BNE]
\label{IP:entry_game:incomplete}
Impose the same structure on payoffs, entry decision rule, outcome space, parameter space, and observable variables as in Identification Problem \ref{IP:entry_game}.
Assume that the observed outcome of the game results from simultaneous move, pure strategy Bayesian Nash play.
Both players and the researcher observe $(\ex_1,\ex_2)$. 
However, $\eps_j$ is private information to player $j=1,2$ and unobservable to the researcher, with $\eps_1\independent\eps_2|(\ex_1,\ex_2)$.
Assume that players have correct common prior $\sF_\gamma$ on the distribution of $(\eps_1,\eps_2)$ and the researcher knows this distribution up to $\gamma$, a finite dimensional parameter vector.
Under these assumptions, multiple Bayesian Nash equilibria (BNE) may result.\footnote{Both the independence assumption and the correct common prior assumption are maintained here to simplify exposition.
Both could be relaxed with no conceptual difficulty, though computation of the set of Bayesian Nash equilibria, for example, would become more cumbersome.} 
In the absence of additional information, what can the researcher learn about $\theta=[\delta_1~\delta_2~\beta_1~\beta_2~\gamma]$?
	\qedex
\end{IP}

With incomplete information, players' strategies are decision rules that map the support of $(\eps,\ex)$ into $\{0,1\}$. 
The non-negativity condition on expected payoffs that determines each player's decision to enter the market results in equilibrium mappings (decision rules) that are step functions determined by a threshold: $y_j(\eps_j) =\one(\eps_j\geq t_j), j=1,2$.
As a result, player $j$'s beliefs about player $3-j$'s probability of entry under the common prior assumption is $\int y_{3-j}(\eps_{3-j}) d\sF_\gamma(\eps_{3-j}|\ex) =1-\sF_\gamma(t_{3-j}|\ex)$, and therefore player $j$'s best response cutoff is
\begin{align*}
t_j^b(t_{3-j},\ex;\theta)=-\ex_j\beta_j-\delta_j(1-\sF_\gamma(t_{3-j}|\ex)).
\end{align*}
Hence, the set of equilibria can be defined as the set of cutoff rules:
\begin{equation*}
\eT_{\theta}(\ex)=\left\{(t_1,t_2\right):t_j=t_j^b(t_{3-j},\ex;\theta),~j=1,2\}.
\end{equation*}%
The equilibrium thresholds are functions of $\ex$ and $\theta$ only.
The set $\eT_{\theta}(\ex)$ might contain a finite number of equilibria (e.g., if the common prior is the Normal distribution), or a continuum of equilibria. 
For ease of notation I suppress its dependence on $\ex$ in what follows.

Given the equilibrium decision rules (the selections of the set $\eT_\theta$), it is possible to determine their associated action profiles.
Because in the simple two-player entry game that I consider actions and outcomes coincide, I denote the set of admissible action profiles by $\eY_\theta$:
\begin{align}
\eY_\theta=\left\{
\ey(\et)\equiv 
\begin{bmatrix}
\one(\eps_1<\et_1,\eps_2<\et_2)\\
\one(\eps_1\ge\et_1,\eps_2<\et_2)\\
\one(\eps_1<\et_1,\eps_2\ge\et_2)\\
\one(\eps_1\ge\et_1,\eps_2\ge\et_2)
\end{bmatrix}
:\et\in\Sel(\eT_\theta) 
\right\},\label{eq:q_incomplete}
\end{align}
with $\Sel(\eT_\theta)$ the set of all measurable selections from $\eT_\theta$, see Definition \ref{def:selection}.
To obtain the predicted set of multinomial distributions for the outcomes of the game, one needs to integrate out $\eps$ conditional on $\ex$.
Again this can be done by using the conditional Aumann expectation: 
\begin{equation*}
\E_{\sF_\gamma}(\eY_\theta|\ex)=\{\E_{\sF_\gamma}(\ey(\et)|\ex):\et\in\Sel(\eT_\theta)\}.
\end{equation*}
This set is closed and convex.
Regardless of whether $\eT_\theta$ contains a finite number of equilibria or a continuum, $\eY_\theta$ can take on only a finite number of
realizations corresponding to each of the vertices of the three dimensional simplex, because the vectors $\ey(\et)$ in \eqref{eq:q_incomplete} collect threshold decision rules.
This implies that $\E_{\sF_\gamma}(\eY_\theta|\ex)$ is a closed convex polytope $\ex$-a.s., fully characterized by a finite number of supporting hyperplanes. 
Hence, it is possible to determine whether $\vartheta\in\idr{\theta}$ using efficient algorithms in linear programming.
\begin{SIR}[Structural Parameters in Static, Simultaneous Move Finite Games of Incomplete Information with BNE]
  \label{SIR:incomplete_info}
Under the assumptions in Identification Problem \ref{IP:entry_game:incomplete}, the sharp identification region for $\theta$ is 
  \begin{align}
    \idr{\theta} &=\bigg\{\vartheta\in \Theta:\; \max_{u\in\mathbb{B}^{|\cY|}} u^\top \cp(\ex)
    -\E_{\sF_{\tilde\gamma}}[h_{\eY_\vartheta}(u)|\ex]=0,\, \ex\text{-a.s.}\bigg\} \label{eq:SIR:incomplete_info:1} \\
    &=\bigg\{\vartheta\in \Theta:\; u^\top \cp(\ex)
   \le \E_{\sF_{\tilde\gamma}}[h_{\eY_\vartheta}(u)|\ex],\,\forall u\in D, \ex\text{-a.s.}\bigg\},\label{eq:SIR:incomplete_info:2} \\
   &=  \bigg\{\vartheta\in \Theta:\; \sP(\ey\in K|\ex)\le \sT_{\eY_{\vartheta}(\ex,\eps)}(K;\sF_{\tilde\gamma})\,\forall K\subset\cY,\, \ex\text{-a.s.}\bigg\} \label{eq:SIR:incomplete_info:0}, 
  \end{align}
with $D=\{u=[u_1,\dots,u_{|\cY|}]^\top:u_i\in\{0,1\},i=1,...,|\cY|\} $, $\vartheta=[d_1,d_2,b_1,b_2,\tilde\gamma]$, and $\sT_{\eY_{\vartheta}(\ex,\eps)}(K;\sF_{\tilde\gamma})$ the probability that $\{\eY_\vartheta(\ex,\eps)\cap K\neq \emptyset\}$ implied when $\eps\sim\sF_{\tilde\gamma}$, $\ex$-a.s.
\end{SIR}
\begin{proof}
The result in \eqref{eq:SIR:incomplete_info:1} follows by the same argument as in the proof of Theorem SIR-\ref{SIR:sharpness_mixed}.
Next I show equivalence of the conditions
\begin{align*}
(i)~&u^\top\cp(\ex)\le\E_{\sF_{\tilde\gamma}}[h_{\eY_\vartheta}(u)|\ex]~\forall u\in\mathbb{B}^{|\cY|}, \\
(ii)~&u^\top\cp(\ex)\le\E_{\sF_{\tilde\gamma}}[h_{\eY_\vartheta}(u)|\ex]~\forall u\in D.
\end{align*}%
By the positive homogeneity of the support function, condition $(i)$ is equivalent to $\cp(\ex)\le\E_{\sF_{\tilde\gamma}}[h_{\eY_\vartheta}(u)|\ex]~\forall u\in\R^{|\cY|}$, which implies condition $(ii)$. 
Next I show that condition $(ii)$ implies condition $(i)$.
As explained before, the set $\eY_\theta$, and hence also its convex hull $\conv(\eY_\theta)$, can take on only a finite number of realizations.
Let $Y_1,\dots,Y_m$ be convex compact sets in the simplex of dimension $|\cY|-1$ equal to the possible realizations of $\conv(\eY_\theta)$, and let $\varpi_1(\ex),\dots,\varpi_m(\ex)$ denote the probability of each of these realizations conditional on $\ex$.
Then by Theorem 2.1.34 in \cite{mo1}, $\E_{\sF_{\tilde\gamma}}(\eY_\theta|\ex)=\sum_{j=1}^m Y_j\varpi_j(\ex)$.
By the properties of the support function \citep[see, e.g.,][Theorem 1.7.5]{sch93}, $h_{\E_{\sF_{\tilde\gamma}}(\eY_\theta|\ex)}(u) =\sum_{j=1}^m \varpi_j(\ex)h_{Y_j}(u)$.
For each $j=1,...,m,$ the vertices of $Y_j$ are a subset of the vertices of the $(|\cY|-1)$-dimensional simplex. 
Hence the supporting hyperplanes of $Y_j,j=1,...,m$, are a subset of the supporting hyperplanes of that simplex, which in turn are obtained through its support function evaluated in directions $u\in D$.
Finally, I show equivalence with the result in \eqref{eq:SIR:incomplete_info:0}.
Because the vertices of $Y_j$ are a subset of the vertices of the $(|\cY|-1)$-dimensional simplex, each direction $u\in D$ determines a set $K_u\subset \cY$.
Given the choice of $u$, the value of $u^\top\ey(\et)$ equals one if $\ey(\et)\in K_u$ and zero otherwise.
Hence, condition \eqref{eq:SIR:incomplete_info:2} reduces to
\begin{align*}
\sP(\ey\in K_u|\ex) = u^\top \cp(\ex) &\le \E_{\sF_{\tilde\gamma}}[h_{\eY_\vartheta}(u)|\ex] = \E_{\sF_{\tilde\gamma}}\left[\sup_{\ey(\et)\in\eY_\vartheta}u^\top\ey(\et)|\ex\right] \\
&= \E_{\sF_{\tilde\gamma}}[\one(\eY_\vartheta\cap K_u\neq \emptyset)|\ex]=\sT_{\eY_{\vartheta}(\ex,\eps)}(K_u;\sF_{\tilde\gamma}).
\end{align*}
Observing that the collection $D$ comprises the $2^{|\cY|}$ vectors with entries equal to either 1 or 0, and that these determine all possible subsets $K_u$ of $\cY$, yields condition \eqref{eq:SIR:incomplete_info:0}.
\end{proof}
One can use the same argument as in the proof of Theorem SIR-\ref{SIR:incomplete_info}, to show that the Aumann expectation/support function characterization of the sharp identification region in Theorem SIR-\ref{SIR:sharpness_mixed} coincides with the characterization based on the capacity functional in Theorem SIR-\ref{SIR:entry_game}, when only pure strategies are allowed for.
This shows that in this class of models, the capacity functional based characterization is a special case of the Aumann expectation/support function based one.

\cite{ara:tam08} study what is the identification power of equilibrium also in the case of static entry games with incomplete information.
They show that in the presence of multiple equilibria, assuming Bayesian Nash behavior yields more informative regions for the parameter vector $\theta$ than assuming only rational behavior, but at the price of a higher computational cost. 

\cite{pau:tan12} propose a procedure to test for the sign of the interaction effects (which here I have assumed to be non-positive) in discrete simultaneous games with incomplete information and (possibly) multiple equilibria. 
As a by-product of this procedure, they also provide a test for the presence of multiple equilibria in the DGP.
The test does not require parametric specifications of players' payoffs, the distributions of their private signals, or the equilibrium selection mechanism. 
Rather, the test builds on the commonly invoked assumption that players' private signals are independent conditional on observed states. 

\cite{gri14} introduces an important class of models with flexible information structure.
Each player is assumed to have a vector of payoff shifters unobservable by the researcher composed of elements that are private information to the player, and elements that are known to all players. 
The results of \cite{ber:mol:mol11} reported in this section apply to this set-up as well.

\subsection{Auction Models with Independent Private Values}
\label{subsec:auctions}
\subsubsection{An Inference Approach Robust to Bidding Behavior Assumptions}
\label{subsubsec:HT}
\cite{hai:tam03} study what can be learned about the distribution of valuations in an open outcry English auction where symmetric bidders have independent private values for the object being auctioned.
The standard theoretical model \citep{mil:web82}, called ``button auction" model, posits that each bidder holds down a button while the object's price rises continuously and exogenously, releasing it (in the dominant strategy equilibrium) when it reaches her valuation or all her opponents have left.
In this case, the distribution of bidder's valuation can be learned exactly.
\cite{hai:tam03} show that much can be learned about the distribution of valuations, even allowing for the fact that real-life auctions may depart from this stylized framework, as in the following identification problem.\footnote{Examples of departures from the standard model include the case where active bidding by a player's opponents may eliminate her incentives to bid close to her valuation or at all; the econometrician does not precisely observe the point at which each bidder drops out; there are discrete bid increments; etc.
}
\begin{IP}[Incomplete Auction Model with Independent Private Values]\label{IP:auction}
For a given auction with $n<\infty$ participating bidders, let $\ev_i\sim\sQ,i=1,\dots,n,$ be bidder $i$'s valuation for the object being auctioned and assume that $\ev_i\independent \ev_j$ for all $i\neq j$.
Assume that the support of $\sQ$ is $[\underline{v},\bar{v}]$ and that each bidder knows her own valuation but not that of her opponents.
Let the auctioneer set a minimum bid increment $\delta\in [0,\bar{v})$, and for simplicity suppose there is no reserve price.\footnote{If there is a reserve price $r>\underline{v}$, nothing can be learned about $\sQ(\ev\in [\underline{v},v])$ for any $v<r$. 
In that case, one can learn features of the truncated distribution of valuations using the same insights summarized here.}
Suppose the researcher observes  order statistics of the bids, $\vec{\eb}_n\equiv(\eb_{1:n},\dots,\eb_{n:n})\sim\sP$ in $\R^n_+$, with $\eb_{i:n}$ the $i$-th lowest of the $n$ bids.
Assume that: (1) Bidders do not bid more than they are willing to pay; (2) Bidders do not allow an opponent to win at a price they are willing to beat.
In the absence of additional information, what can the researcher learn about $\sQ$?
	\qedex
\end{IP}
\begin{figure}[tp]
\centering
\includegraphics[scale=1]{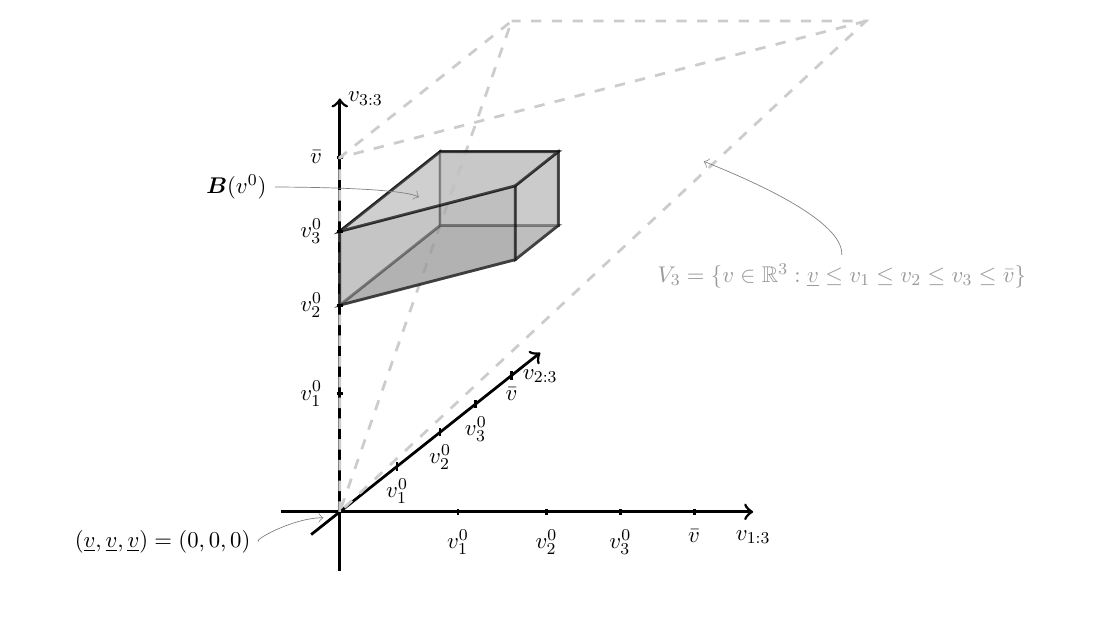}
		\caption{\small{A realization of the model predicted ordered bids $\eB(\vec{\ev}_n)$ in \eqref{eq:RCS_auction} for $n=3,\vec{\ev}_n=v^0,\delta=0$.}}
\label{fig:auction}
\end{figure}
The model in Identification Problem \ref{IP:auction} delivers set valued predictions because given valuations $(\ev_1,\dots,\ev_n)$, the two fundamental assumptions about bidder's behavior yield
\begin{align}
\vec{\eb}_n \in \eB(\vec{\ev}_n)\equiv\left[\left\{\prod_{i=1}^{n-1}[\underline{v},\ev_{i:n}]\right\}\times [\ev_{n-1:n}-\delta,\ev_{n:n}]\right]\cap V_n,\label{eq:RCS_auction}
\end{align}
where $\vec{\ev}_n\equiv(\ev_{1:n},\dots,\ev_{n:n})$ denotes the vector of order statistics of the valuations, and $V_n=\{v\in\R^n:\underline{v}\le v_1\le v_2\le\dots\le v_n\le \bar{v}\}$.\footnote{Using the same convention as for the bids, $\ev_{i:n}$ denotes the $i$-th lowest of the $n$ valuations.}
Figure \ref{fig:auction} provides a stylized depiction of a realization of this set for $\vec{\ev}_n=v^0$ when there are three bidders ($n=3$), $\underline{v}=0$, and $\delta=0$.
In words, $\eB(\vec{\ev}_n)$ collects the model predicted values of ordered bids.
The fact that $\eb_{i:n}\le \ev_{i:n}$ for all $i$ results from assumption (1): since each bidder bids at most an amount equal to her valuation, the $i$-th highest bid cannot exceed the $i$-th highest valuation \citep[][Lemma 1]{hai:tam03}.\footnote{Note that $\eb_{i:n}$ needs not be the bid made by the bidder with valuation $\ev_{i:n}$.}
The fact that $\eb_{n:n}\ge \ev_{n-1,n}-\delta$ follows immediately from assumption (2) \citep[][Lemma 3]{hai:tam03}.
The fact that $\vec{\eb}_n$ has to lie in $V_n$ follows because it is a vector of \emph{ordered} bids.

Why does this set-valued prediction hinder point identification?
The reason is that the distribution of the observable data relates to the model structure in an \emph{incomplete} manner.\footnote{\cite[Appendix D]{hai:tam03} provide the discussion summarized here. Additionally, in their Appendix B, they give a simple example of a two-bidder auction satisfying all assumptions in Identification Problem \ref{IP:auction}, where two different distributions $\sQ$ and $\tilde{\sQ}$ yield the same distribution of ordered bids.} 
Define a bidding rule $\sB(\eb_{1:n},\dots,\eb_{n:n}|\ev_{1:n},\dots,\ev_{n:n})$ to be a conditional joint distribution for the order statistics of the bids conditional on the order statistics of the valuations.
Then, for a given realization of the valuations $\ev_{1:n}=v_1,\dots,\ev_{n:n}=v_n$, the model requires that the support of $\sB(\cdot|v_1,\dots,v_n)$ is in $B(\vec{v})$ as defined in \eqref{eq:RCS_auction} with $\ev_{1:n}=v_1,\dots,\ev_{n:n}=v_n$, but imposes no other restriction on it. 
Hence, the model implied joint distribution of ordered bids is 
\begin{align}
\sM_{1,\dots,n:n}(\cdot;\sB,\sQ)\equiv\int \sB(\cdot|v_1,\dots,v_n)\sQ_{1,\dots,n:n}(dv_1,\dots,dv_n),\label{eq:model:impl_sel_mech_auction} 
\end{align}
where $\sQ_{1,\dots,n:n}$ is the joint distribution of order statistics of the valuations implied by $\sQ$.
Since the bidding rule $\sB$ is left completely unspecified (other than requiring it to be a valid joint conditional probability distribution with support in $\eB$), one can find multiple pairs $(\sB ,\sQ)$ satisfying the assumptions of Identification Problem \ref{IP:auction}, such that $\sM_{1,\dots,n:n}(\cdot;\sB,\sQ)=\sG_{1,\dots,n:n}(\cdot)$, with $\sG_{1,\dots,n:n}$ the observed joint CDF of the order statistics of the bids associated with $\sP$.

\cite{hai:tam03} propose to use simple and tractable implications of the model to learn features of $\sQ$.
Recall that with i.i.d. valuations, the distribution of each order statistic uniquely determines $\sQ(v)$, with $\sQ(v)\equiv\sQ(\ev\le v)$ for any $v\ge\underline{v}$, through:
\begin{align}
\sQ(v)=\sq_{\cB}(\sQ_{i:n}(v);i,n-i+1),\label{eq:HT:beta}
\end{align}
where $\sQ_{i:n}$ is the CDF of $\ev_{i:n}$ and $\sq_{\cB}(\cdot;i,n-i+1)$ is the quantile function of a Beta-distributed random variable with parameters $i$ and $n-i+1$.
Using this, their Lemmas 1 and 3 yield, respectively,
\begin{align}
\sQ(v) &\le \min_{n,i}\sq_{\cB}(\sG_{i:n}(v);i,n-i+1),~\forall v\in[\underline{v},\bar{v}],\label{eq:HT_upper}\\
\sQ(v) &\ge \max_{n}\sq_{\cB}(\sG_{n:n}(v-\delta);i,n-i+1),~\forall v\in[\underline{v},\bar{v}],\label{eq:HT_lower}
\end{align}
where, for any $v\ge\underline{v}$, $\sG_{i:n}(v)\equiv\sP(\eb_{i:n}\le v)$ denotes the observed CDF of $\eb_{i:n}$ for $i=1,\dots,n$.

\begin{BI}
The model and analysis put forward by \cite{hai:tam03} trade point identification of the distribution of valuation under stringent assumptions on the bidding rule, for a \emph{robust} inference approach that yields informative bounds under weak and widely credible assumptions on bidding behavior.
Remarkably, ``nothing is lost" due to the use of their robust approach: point identification is recovered when the standard assumptions of the button auction model hold.\footnote{
The button auction model yields bidding behavior consistent with Identification Problem \ref{IP:auction}.}
This is because in the dominant strategy equilibrium the top losing bidder exits at her valuation, followed immediately by the winning bidder.
Hence, $\eb_{n-1:n}=\ev_{n-1:n}=\eb_{n:n}$ and $\delta=0$, so that the upper and the lower bound in \eqref{eq:HT_upper}-\eqref{eq:HT_lower} coincide and point identify the distribution of valuations.
\end{BI}
\cite{hai:tam03} also provide sharp bounds on the optimal reserve price, which I do not discuss here.
However, they leave open the question of whether the collection of CDFs satisfying \eqref{eq:HT_upper}-\eqref{eq:HT_lower} yields the sharp identification region for $\sQ$.
As discussed in Sections \ref{subsec:missing_data}-\ref{subsec:interval_data}, pointwise bounds on the CDF deliver tubes of admissible CDFs that in general yield outer regions on the CDF of interest.
But in this identification problem, the issue of sharpness is even more subtle, and therefore addressed in the following subsection.\smallskip

Before moving on to that discussion, I note that the work of \cite{hai:tam03} spurred a rich literature applying partial identification analysis to the study of auction models.
\cite{tan11} studies first price sealed bid auctions with equilibrium behavior, where affiliated valuations prevent --in the absence of parametric restrictions on the distribution of the model primitives-- point identification of the model.
He derives bounds on seller revenue under various counterfactual scenarios on reserve prices and auction formats.
\cite{arm13} also studies first price sealed bid auctions with equilibrium behavior, but relaxes the independence assumptions on symmetric valuations by requiring it to hold only conditional on unobserved heterogeneity.
He derives bounds on various functionals of the distributions of interest, including the mean bid and mean valuation.
\cite{ara:gan:qui13} analyze second price auctions with correlated private values.
In this case, the distribution of valuations is not point identified even under the assumptions of the button auction model \citep[][Theorem 4]{ath:hai02}.
Nonetheless, \cite{ara:gan:qui13} show that interesting functionals of it (seller profits and bidder surplus) can be bounded, if one assumes that transaction prices are determined by the second highest valuation and imposes some restrictions on the joint distribution of the number of bidders and distribution of the valuations.
\cite{kom13} studies a related model of second-price ascending auctions with arbitrary dependence in bidders' private values.
She provides partial identification results for the joint distribution of values for any subset of bidders under various assumptions about what data the researcher observes.
While in her framework the highest bid is never observed, she considers the case where only the winner's identity and the winning price are observed, and the case where all the identities and all the bids except for the highest bid are known.
She also investigates the informational content of assuming positive dependence in bidders' values.
\cite{gen:li14} are concerned with nonparametric identification of a two-stage entry and bidding game.
Potential bidders are assumed to have private valuations and observe private signals before deciding whether to enter the auction.
The dependence between signals and valuations is only minimally restricted.
Hence, even with some excluded instruments that affect selection into the auction, the model primitives are only partially identified.
The authors derive bounds on these primitives, and provide conditions under which point identification is restored. 
\cite{syr:tam:zia18} provide partial identification results in private value and common value auctions under weak restrictions on the information available to the bidders.
Their approach leverages a result in \cite{ber:mor16} yielding an equivalence between distributions of valuations that obey the restrictions imposed by a Bayesian Correlated Equilibrium and those that obey the restrictions imposed by Bayesian Nash Equilibrium under some information structure.
Such equivalence is particularly helpful because the set of Bayesian Correlated Equilibria can be characterized through linear programming, so that the sharp identification region provided by \citeauthor{syr:tam:zia18} is given by the collection of parameter vectors $\vartheta$ for which a linear program is feasible.
Related results leveraging the linear structure of correlated equilibria in the context of entry games include \cite{yan06}, \cite[Supplementary Appendix E.2]{ber:mol:mol11}, and \cite{mag:ron17}.

\subsubsection{Characterization of Sharpness through Random Set Theory}\label{subsubsec:sharp:auction}
\possessivecite{hai:tam03} bounds exploit the information contained in the \emph{marginal} CDFs $\sG_{i:n}$ for each $i$ and $n$.
However, in Identification Problem \ref{IP:auction} additional information can be extracted from the \emph{joint} distribution of ordered bids.
\cite{che:ros17} obtain the sharp identification region $\idr{\sQ}$ using random set methods (Artstein's characterization in Theorem \ref{thr:artstein}) applied to a quantile function representation of the order statistics.
Here I provide an equivalent characterization that uses equation \eqref{eq:RCS_auction} directly, and which has not appeared in the literature before.
Let $\cT$ denote the space of probability distributions with support on $[\underline{v},\bar{v}]$, so that $\sQ\in\cT$.
For a candidate distribution $\tilde{\sQ}\in\cT$, let $\tilde{\sQ}_{1,\dots,n:n}$ denote the implied distribution of order statistics of $n$ i.i.d. random variables distributed $\tilde{\sQ}$.
Let $\tilde{\eB}$ be a random closed set defined as in \eqref{eq:RCS_auction} with respect to order statistics of i.i.d. random variables with distribution $\tilde{\sQ}$. 
For a given set $K\in\cK$, with $\cK$ the collection of compact subsets of $\R^n$, let $\sT_{\tilde\eB}(K;\tilde{\sQ})$ denote the probability of the event $\{\tilde\eB\cap K\neq \emptyset\}$ implied by $\tilde{\sQ}$.
\begin{SIR}[Distribution of Valuations in Incomplete Auction Model with Independent Private Values]\label{SIR:auction}
Under the assumptions of Identification Problem \ref{IP:auction}, the sharp identification region for $\sQ$ is
\begin{align}
\label{eq:SIR:auction}
\idr{\sQ}= \left\{\tilde{\sQ}\in\cT: \sP(\vec{\eb}_n\in K) \le \sT_{\tilde\eB}(K;\tilde{\sQ})~ \forall K\in\cK \right\}.
\end{align}
\end{SIR}
\begin{proof}
The sharp identification region for $\sQ$ is given by the collection of probability distributions $\tilde{\sQ}\in\cT$ for which one can find a bidding rule $\sB(\cdot|\cdot)$ with support in $\tilde{\eB}$ a.s. such that $\sG_{1,\dots,n:n}(\cdot)=\sM_{1,\dots,n:n}(\cdot;\sB,\tilde{\sQ})$.
Here $\sM_{1,\dots,n:n}(\cdot;\sB,\tilde{\sQ})$ is defined as in \eqref{eq:model:impl_sel_mech_auction} with $\tilde{\sQ}$ replacing $\sQ$.
Take a distribution $\tilde{\sQ}$ satisfying this definition of sharpness.
Then there exists a selection of $\tilde{\eB}$ determined by the bidding rule associated with $\tilde{\sQ}$, such that its distribution matches that of $\vec{\eb}_n$.
But then Theorem \ref{thr:artstein} implies that the inequalities in \eqref{eq:SIR:auction} hold.
Conversely, take $\tilde{\sQ}$ satisfying the inequalities in \eqref{eq:SIR:auction}.
Then, by Theorem \ref{thr:artstein}, $\vec{\eb}_n$ and $\tilde{\eB}$ can be realized on the same probability space as random elements $\vec{\eb}_n^\prime$ and $\tilde{\eB}^\prime$, $\vec{\eb}_n\edis \vec{\eb}_n^\prime$, $\tilde{\eB}\edis\tilde{\eB}^\prime$, such that $\vec{\eb}_n^\prime \in \tilde{\eB}^\prime$ a.s.
One can then complete the auction model with a bidding rule that picks $\vec{\eb}_n^\prime$ with probability $1$, and the result follows.
\end{proof}
In \eqref{eq:SIR:auction}, $\sP(\vec{\eb}_n\in K)$ is determined by the joint distribution of the ordered bids and hence can be learned from the data.
On the other side, $\sT_{\tilde\eB}(K;\tilde{\sQ})$ is a function of the model and $\tilde{\sQ}\in\cT$. 
Hence, it can be computed using \eqref{eq:RCS_auction}, with $\tilde\eB$ defined with respect to order statistics of i.i.d. random variables with distribution $\tilde{\sQ}\in\sT$.
To gain insights in the characterization of $\idr{\sQ}$, consider for example the set $K=\{\prod_{i=1}^{n-1}(-\infty,+\infty)\}\times(-\infty,v]$.
Plugging it in the inequalities in \eqref{eq:SIR:auction}, one obtains
\begin{align*}
\sG_{n:n}(v) \le \sQ_{n-1,n}(v),~\text{for all } n, 
\end{align*}
which, using \eqref{eq:HT:beta}, yields \eqref{eq:HT_lower}.
Similarly, plugging in the sets $K_j=\{\prod_{i=1}^{j-1}(-\infty,+\infty)\}\times[v,\infty)\times\{\prod_{j+1}^n(-\infty,+\infty)\}$, $j=1,\dots,n$, yields \eqref{eq:HT_upper}.
So the inequalities proposed by \cite{hai:tam03} are a subset of the inequalities yielding the sharp identification region in Theorem SIR-\ref{SIR:auction}.
More information can be obtained by using additional sets $K$.
For instance, the set $K=[v_1,\infty)\times[v_2,\infty)\times\{\prod_{i=1}^{n}(-\infty,+\infty)\}$, $v_2\ge v_1$, yields $\sP(\eb_{1:n}\ge v_1,\eb_{2:n}\ge v_2)\le \sQ_{1,2:n}([v_1,\infty)\times[v_2,\infty))$, which further restricts $\sQ$.
Numerous examples can be given.

Characterization \eqref{eq:SIR:auction} is stated using inequality \eqref{eq:domin-t} for the collection of compact subsets of $\R^n$.
One can instead use the (equivalent) inequality \eqref{eq:dom-c}, and show that in fact it suffices to check it for a much smaller collection of sets, as shown by \cite{che:ros17} \citep[see also][Section 2.2]{mol:mol18}. 
Nonetheless, this collection remains extremely large.

\begin{BI}[Random set theory and partial identification -- continued]
As stated in the Introduction, constructing the (random) set of model predictions delivered by the maintained assumptions is an exercise typically carried out in identification analysis, regardless of whether random set theory is applied.
Indeed, for the problem studied in this section, \cite[equation D1]{hai:tam03} put forward the set of admissible bids in \eqref{eq:RCS_auction}.\footnote{Equations D1 in \citeauthor{hai:tam03} and \eqref{eq:RCS_auction} here differ in that the latter also requires bids to be ordered. This observation was besides the point in \possessivecite{hai:tam03} discussion that led to equation D1.}
With this set in hand, the tools of random set theory (in this case, Theorem \ref{thr:artstein}) immediately deliver the sharp identification region of interest.
\end{BI}
\cite{che:ros17auction} further generalize the analysis in this section by dropping the requirement of independent private values.
This allows them, for example, to consider affiliated private values.
They show that even in this significantly more complex context, the key behavioral restrictions imposed by \cite{hai:tam03} to relate bids to valuations can be coupled with the use of random set theory, to characterize sharp identification regions.

%

\subsection{Network Formation Models}
\label{subsubsec:networks}
Strategic models of network formation generalize the frameworks of single agents and multiple agents discrete choice models reviewed in Sections \ref{subsec:single:ag:RUM} and \ref{subsec:multiple:eq}.
They posit that pairs of agents (nodes) form, maintain, or sever connections (links) according to an explicit equilibrium notion and utility structure.
Each individual's utility depends on the links formed by others (the network) and on utility shifters that may be pair-specific.

One may conjecture that the results reported in Sections \ref{subsec:single:ag:RUM}-\ref{subsec:multiple:eq} apply in this more general context too.
While of course lessons can be carried over, network formation models present challenges that combined cannot be overcome without the development of new tools.
These include the issue of equilibrium existence and the possibility of multiple equilibria when they exist, due to the interdependence in agents' choices (this problem was already discussed in Section \ref{subsec:multiple:eq}).
Another challenge is the degree of correlation between linking decisions, which interacts with how the observable data is generated: one may observe a growing number of independent networks, or a growing number of agents on a single network.
Yet another challenge, which substantially increases the difficulties associated with the previous two, is the combinatoric complexity of network formation problems.
The purpose of this section is exclusively to discuss some recent papers that have made important progress to address these specific challenges and carry out partial identification analysis.
For a thorough treatment of the literature on network formation, I refer to the reviews in \cite{gra15}, \cite{cha16}, \cite{pau17}, and \cite[Chapter XXX in this Volume]{gra19}.\footnote{For a review of the literature on peer group effect analysis, see, e.g., \cite{bro:dur01hoe}, \cite{blu:bro:dur:ioa11}, \cite{pau17}, and \cite{gra19}.}

Depending on whether the researcher observes data from a single network or multiple independent networks, the underlying population of agents may be represented as a continuum or as a countably infinite set in the first case, or as a finite set in the second case.
Henceforth, I denote generic agents as $i$, $j$, $k$, and $m$.
I consider static models of undirected network formation with non-transferable utility.\footnote{\emph{Undirected} means that if a link  from node $i$ to node $j$ exists, then the link from $j$ to $i$ exists.
The discussion that follows can be generalized to the case of models with transferable utility.}
The collection of all links among nodes forms the network, denoted $\ey$.
For any pair $(i,j)$ with $i\neq j$, $\ey_{ij}=1$ if they are linked, and $\ey_{ij}=0$ otherwise ($\ey_{ii}=0$ for all $i$ by convention).
The notation $\ey-\{ij\}$ denotes the network that results if a link present between nodes $i$ and $j$ is deleted, while $\ey+\{ij\}$ denotes the network that results if a link absent between nodes $i$ and $j$ is added.
Denote agent $i$'s payoff by $\bu_i(\ey,\ex,\epsilon)$.
This payoff depends on the network $\ey$ and the payoff shifters $(\ex,\epsilon)$, with $\ex$ observable both to the agents and to the researcher, $\epsilon$ only to the agents, and $(\ex,\epsilon)$ collecting $(\ex_{ij},\epsilon_{ij})$ for all $i$ and $j$.\footnote{Here I consider a framework where the agents have complete information.}

Following much of the literature, I employ \emph{pairwise stability} \citep{jac:wol96} as equilibrium notion: $\ey$ is a pairwise stable network if all linked agents prefer not to sever their links, and all non-existing links are damaging to at least one agent.
Formally,
\begin{align*}
\forall(i,j):\ey_{ij}&=1,~\bu_i(\ey,\ex,\epsilon)\ge \bu_i(\ey-\{ij\},\ex,\epsilon)~\mathrm{and}~\bu_j(\ey,\ex,\epsilon)\ge \bu_j(\ey-\{ij\},\ex,\epsilon),\\
\forall(i,j):\ey_{ij}&=0,~\mathrm{if}~\bu_i(\ey+\{ij\},\ex,\epsilon)> \bu_i(\ey,\ex,\epsilon)~\mathrm{then}~\bu_j(\ey+\{ij\},\ex,\epsilon)< \bu_j(\ey,\ex,\epsilon).
\end{align*}
Under this equilibrium notion, if equilibria exist multiplicity is likely; see, among others, the examples in \cite[p. 475]{gra15}, \cite[p. 301]{pau17}, and \cite[example 3.1]{she18}.
The model is therefore \emph{incomplete}, because it does not specify how an equilibrium is selected in the region of multiplicity.
For the same reasons as discussed in the context of finite games in Section \ref{subsec:multiple:eq}, partial identification results (unless one is willing to impose restrictions on the equilibrium selection mechanism).
However, as I explain below, an immediate application of the identification analysis carried out there presents enormous practical challenges because there are $2^{n(n-1)/2}$ possible network configurations to be checked for stability (and the dimensionality of the space of unobservables is also very large).

In what follows I consider two distinct frameworks that make different assumptions about the utility function and how the data is generated, and discuss what can be learned about the parameters of interest in these cases.

\subsubsection{Data from Multiple Independent Networks}
\label{subsubsec:networks:1}
I first consider the case that the researcher observes data from multiple independent networks.
I follow the set-up put forward by \cite{she18}.
\begin{IP}[Network Formation Model with Multiple Independent Networks]
\label{IP:networks:multiple:indep}
Let there be $n\in\{2,3,\dots\},n<\infty$ agents, and let $(\ex,\ey)\sim\sP$ be observable random variables in $\times_{j=1}^n\R^d\times\{0,1\}^{n(n-1)/2}$, $d<\infty$.
Suppose that $\ey$ is a pairwise stable network.
For each agent $i$, let the utility function be known up to finite dimensional parameter vector $\delta\in\Delta\subset\R^p$, and given by
\begin{multline}
\bu_i(\ey,\ex,\epsilon;\delta)=\sum_{j=1}^n \ey_{ij}(f(\ex_i,\ex_j;\delta_1)+\epsilon_{ij})\\
+\delta_2\frac{\sum_{j=1}^n\sum_{k\neq i,k=1}^n\ey_{ij}\ey_{jk}}{n-2}+\delta_3\frac{\sum_{j=1}^n\sum_{k=j+1}^n\ey_{ij}\ey_{ik}\ey_{jk}}{n-2}\label{eq:utility:network:1}
\end{multline}
with $f(\cdot,\cdot;\cdot)$ a continuous function of its arguments.\footnote{The effects of having friends in common and of friends of friends in \eqref{eq:utility:network:1} are normalized by $n-2$. This enforces that the marginal utility that $i$ receives from linking with $j$ is affected by $j$ having an additional link with $k$ to a smaller degree as $n$ grows. This does not result in diminishing network effects.}
Suppose that $\epsilon_{ij}$ are independent for all $i\neq j$ and identically distributed with CDF known up to parameter vector $\gamma\in\Gamma\subset\R^m$, denoted $\sF_\gamma$.
Assume that the support of $\sF_\gamma$ is $\R$, that $\sF_\gamma$ is absolutely continuous with respect to Lebesgue measure, and continuously differentiable with respect to $\gamma\in\Gamma$.
Let $\Theta=\Delta\times\Gamma$.
Assume that the researcher observes a random sample of networks and observable payoff shifters drawn from $\sP$.
In the absence of additional information, what can the researcher learn about $\theta\equiv[\delta_1~\delta_2~\delta_3~\gamma]$?
	\qedex
\end{IP}
\cite{she18} analyzes this problem.
She establishes equilibrium existence provided that $\delta_2\ge 0$ and $\delta_3\ge 0$ \citep[Proposition 2.2]{she18}.\footnote{With transferable utility, \cite[Proposition 2.1]{she18} establishes existence for any $\delta_2,\delta_3\in\R$.
See \cite{hel13} for an earlier analysis of existence and uniqueness of pairwise stable networks.}
Given payoff shifters $(\ex,\epsilon)$ and parameters $\vartheta\equiv[\tilde\delta_1~\tilde\delta_2~\tilde\delta_3~\tilde\gamma]\in\Theta$, let $\eY_\vartheta(\ex,\epsilon)$ denote the collection of pairwise stable networks implied by the model.
It is easy to show that $\eY_\vartheta(\ex,\epsilon)$ is a random closed set as in Definition \ref{def:rcs}.
The networks in $\eY_\vartheta(\ex,\epsilon)$ are $n\times n$ symmetric adjacency matrices with diagonal elements equal to zero and off diagonal elements in $\{0,1\}$.
To ease notation, I omit $\eY_\vartheta$'s dependence on $(\ex,\epsilon)$ in what follows.
Under the assumption that $\ey$ is a pairwise stable network, at the true data generating value of $\theta\in\Theta$, one has
\begin{align}
\ey\in\eY_\theta~\mathrm{a.s.} \label{eq:y_in_Y_network_multiple}
\end{align}
Equation \eqref{eq:y_in_Y_network_multiple} exhausts the modeling content of Identification Problem \ref{IP:networks:multiple:indep}.
Theorem \ref{thr:artstein} can be leveraged to extract its empirical content from the observed distribution $\sP(\ey,\ex)$.
Let $\cY$ be the collection of $n\times n$ symmetric matrices with diagonal elements equal to zero and all other entries in $\{0,1\}$, so that $|\cY|=2^{n(n-1)/2}$.
For a given set $K\subset\cY$, let $\sT_{\eY_{\vartheta}}(K;\sF_\gamma)$ denote the probability of the event $\{\eY_\vartheta\cap K\neq \emptyset\}$ implied when $\epsilon\sim\sF_\gamma$, $\ex$-a.s.
\begin{SIR}[Structural Parameters in Network Formation Models with Multiple Independent Networks]
\label{SIR:networks:1}
Under the assumptions of Identification Problem \ref{IP:networks:multiple:indep}, the sharp identification region for $\theta$ is
\begin{align}
\idr{\theta}=\{\vartheta\in\Theta:\sP(\ey\in K|\ex)\le \sT_{\eY_{\vartheta}}(K;\sF_{\tilde\gamma})\,\forall K\subset\cY, \, \ex\text{-a.s.}\}.\label{eq:SIR:networks:1}
\end{align} 
\end{SIR}
\begin{proof}
Follows from similar arguments as for the proof of Theorem \ref{SIR:entry_game} on p.~\pageref{SIR:entry_game}.
\end{proof}
The characterization of $\idr{\theta}$ in Theorem SIR-\ref{SIR:networks:1} is new to this chapter.\footnote{
\cite{gua19} has previously used Theorem D.1 in \cite{ber:mol:mol11}, as I do here, to characterize sharp identification regions in unilateral and bilateral directed network formation games.}
While technically it entails a finite number of conditional moment inequalities, in practice their number can be prohibitive as it can be as large as $2^{2^{n(n-1)/2}}-2$.\footnote{This number may be reduced drastically using the notion of \emph{core determining class} of sets, see Definition \ref{def:core-det} and the discussion on p.~\pageref{def:core-det}.
Nonetheless, even with relatively few agents, the number of inequalities in \eqref{eq:SIR:networks:1} may remain overwhelming.}
Even using only a subset of the inequalities in \eqref{eq:SIR:networks:1} to obtain an outer region, for example applying the insights in \cite{cil:tam09}, may not be practical (with $n=20$, $|\cY|\approx 10^{57}$).
Moreover, computation of $\sT_{\eY_{\vartheta}}(K;\sF_\gamma)$ may require (depending on the set $K$) evaluation of rather complex integrals.

To circumvent these challenges, \cite{she18} proposes to analyze network formation through \emph{subnetworks}.
A subnetwork is the restriction of a network to a subset of the agents (i.e., a subset of nodes and the links between them).
For given $A\subseteq\{1,2,\dots,n\}$, let $\ey^A=\{\ey_{ij}\}_{i,j\in A, i\neq j}$ be the submatrix in $\ey$ with rows and columns in $A$, and let $\ey^{-A}$ be the remaining elements of $\ey$ after $\ey^A$ is deleted.
With some abuse of notation, let $(\ey^A,\ey^{-A})$ denote the composition of $\ey^A$ and $\ey^{-A}$ that returns $\ey$.
Recall that $\eY_\vartheta\equiv\eY_\vartheta(\ex,\epsilon)$, and let
\begin{align*}
\eY_{\vartheta}^A=\{\ey^A\in\{0,1\}^{|A|}:\exists \ey^{-A}\in\{0,1\}^{|-A|}~\mathrm{such~that}~(\ey^A,\ey^{-A})\in\eY_{\vartheta}\}
\end{align*}
be the collection of subnetworks with rows and columns in $A$ that can be part of a pairwise stable network in $\eY_\vartheta$.
Let $\ex^A$ denote the subset of $\ex$ collecting $\ex_{ij}$ for $i,j\in A$.
For a given $y^A\in\{0,1\}^{|A|}$, let $\sC_{\eY_{\vartheta}^A}(y^A;\sF_\gamma)$ and $\sT_{\eY_{\vartheta}^A}(y^A;\sF_\gamma)$ denote, respectively, the probability of the events $\{\eY_\vartheta^A=\{y^A\}\}$ and $\{\{y^A\}\in\eY_\vartheta^A\}$ implied when $\epsilon\sim\sF_\gamma$, $\ex$-a.s.
The first event means that only the subnetwork $y^A$ is part of a pairwise stable network, while the second event means that $y^A$ is a possible subnetwork that is part of a pairwise stable network but other subnetworks may be part of it too.
\cite[Proposition 4.1]{she18} provides the following outer region for $\theta$ by adapting the insight in \cite{cil:tam09} to subnetworks.
In the theorem I abuse notation compared to Table \ref{tab:notation} by introducing a superscript, $A$, to make explicit the dependence of the outer region on it.
\begin{OR}[Subnetworks-based Outer Region on Structural Parameters in Network Formation Models with Multiple Independent Networks]
\label{OR:networks:1}
Under the assumptions of Identification Problem \ref{IP:networks:multiple:indep}, for any $A\subseteq\{1,2,\dots,n\}$, an $A$-dependent outer region for $\theta$ is
\begin{align}
\mathcal{O}^A_\sP[\theta]=\{\vartheta\in\Theta:\sC_{\eY_{\vartheta}^A}(y^A;\sF_{\tilde\gamma})\le\sP(\ey^A=y^A|\ex^A)\le \sT_{\eY_{\vartheta}^A}(y^A;\sF_{\tilde\gamma})\,\forall y^A\subset\cY^A, \, \ex^A\text{-a.s.}\},\label{eq:OR:networks:1}
\end{align} 
where $\cY^A$ is the collection of $|A|\times|A|$ symmetric matrices with diagonal elements equal to zero and all other elements in $\{0,1\}$ so that $|\cY^A|=2^{|A|(|A|-1)/2}$. 
\end{OR}
\begin{proof}
Let $\eu(\tilde\ey|\eY_\vartheta)$ be a random variable in the unit simplex in $\R^{n(n-1)/2}$ which assigns to each possible pairwise stable network $\tilde\ey$ that may realize given $(\ex,\epsilon)$ and $\vartheta\in\Theta$ the probability that it is selected from $\eY_\vartheta$.
Given $y\in\cY$, denote by $\sM(y|\ex)$ the model predicted probability that the network realizes equal to $y$.
Then the model yields
\begin{align}
\sM(y|\ex)&=\int\eu(y| Y_\vartheta)d\sF_\gamma=\int_{y\in Y_\vartheta,| Y_\vartheta|=1}d\sF_\gamma+\int_{y\in Y_\vartheta,| Y_\vartheta|\ge 2}\eu( y| Y_\vartheta)d\sF_\gamma.\label{eq:model:distrib:network:1}
\end{align}
The model implied distribution for subnetwork $\tilde\ey^A$ is obtained by taking the marginal of expression \eqref{eq:model:distrib:network:1} with respect to $\tilde\ey^{-A}$
\begin{align}
\sM(y^A|\ex)&=\sum_{y^{-A}}\sM((y^A,y^{-A})|\ex)=
\int_{y^A\in Y_\vartheta^A,| Y_\vartheta^A|=1}d\sF_\gamma+\int_{y^A\in Y_\vartheta^A,| Y_\vartheta^A|\ge 2}\sum_{y^{-A}}\eu((y^A,y^{-A})| Y_\vartheta)d\sF_\gamma.\label{eq:model:distrib:subnetwork:1}
\end{align}
Replacing $\eu$ in \eqref{eq:model:distrib:subnetwork:1} with zero and one yields the bounds in \eqref{eq:OR:networks:1}. 
\end{proof}
\cite[Section 4.2]{she18} further assumes that the selection mechanism $\eu(\tilde\ey|\eY_\vartheta)$ is invariant to permutations of the labels of the players.
Under this condition and the maintained assumptions on $\epsilon$, she shows that the inequalities in \eqref{eq:OR:networks:1} are invariant under permutations of labels, so subnetworks in any two subsets $A,A'\subseteq\{1,2,\dots,n\}$ with $|A|=|A'|$ and $\ex^A=\ex^{A'}$ yield the same inequalities for all $y^A=y^{A'}$.
It is therefore sufficient to consider subnetwork $A$ and the inequalities in \eqref{eq:OR:networks:1} associated with it.
Leveraging this result, \citeauthor{she18} proposes an outer region obtained by looking at unlabeled subnetworks of size $|A|\le\bar{a}$ and given by
\begin{align*}
\outr{\theta}=\bigcap_{|A|\le\bar{a}}\mathcal{O}^A_\sP[\theta].
\end{align*}
As long as the subnetworks are chosen to be small, e.g., $|A|\le 2,3,4$, the inequalities in \eqref{eq:OR:networks:1} can be computed even if the network is large.
\cite{she18} shows that the inequalities in \eqref{eq:OR:networks:1} remain informative as $n$ grows.
This fact highlights the importance of working with subnetworks.
One could have applied the insight of \cite{cil:tam09} directly to the full network by setting $\eu$ equal to zero and to one in \eqref{eq:model:distrib:network:1}.
The resulting bounds, however, would vanish to zero as $n$ grows and become uninformative for $\theta$.
The characterization in Theorem OR-\ref{OR:networks:1} can be refined to obtain a smaller region, adapting the results in \cite[Supplementary Appendix Theorem D.1]{ber:mol:mol11} to subnetworks.
The size of this refined region is weakly decreasing in $|A|$.\footnote{The idea of using random set methods on subnetworks to obtain the refined region was put forward in an earlier version of \cite{she18}. She provided a proof that the refined region's size decreases weakly in $|A|$.}
However, the refinement does not yield $\idr{\theta}$ because it is applied only to subnetworks.
\begin{BI}
At the beginning of this section I highlighted some key challenges to inference in network formation models.
Identification Problem \ref{IP:networks:multiple:indep} bypasses the concern on the dependence among linking decisions through the independence assumption on $\epsilon_{ij}$ and the presumption that the researcher observes data from multiple independent networks, which allows for identification of $\sP(\ey,\ex)$.
\cite{she18} takes on the remaining challenges by formally establishing equilibrium existence and allowing for unrestricted selection among multiple equilibria.
In order to overcome the computational complexity of the problem, she puts forward the important idea of inference based on subnetworks.
While of course information is left on the table, the approach remains feasible even with large networks.
\end{BI}
\cite{miy16} considers a framework similar to the one laid out in Identification Problem \ref{IP:networks:multiple:indep}.
He assumes non-negative externalities, and shows that in this case the set of pairwise stable equilibria is a complete lattice with a smallest and a largest equilibrium.\footnote{This approach exploits supermodularity, and is related to \cite{jia08} and \cite{ech05}.}
He then uses moment functions that are monotone in the pairwise stable network (so that they take their extreme values at the smallest and largest equilibria), to obtain moment conditions that restrict $\theta$.
Examples of the moment functions used include the proportion of pairs with a link, the proportion of links belonging to traingles, and many more \citep[see][Table 1]{miy16}.

\cite{gua19} considers unilateral and bilateral directed network formation games, still under a sampling framework where the researcher observes many independent networks.
The equilibrium notion that she uses is pure strategy Nash.
She assumes that the payoff that player $i$ receives from forming link $ij$ is allowed to depend on the number of additional players forming a link pointing to $j$, but rules out other spillover effects.
Under this assumption and some regularity conditions, \citeauthor{gua19} shows that the network formation game can be decomposed into local games (i.e., games whose sets of players and strategy profiles are subsets of the network formation game's ones), so that the network formation game is in equilibrium if and only if each local game is in equilibrium.
She then obtains a characterization of $\idr{\theta}$ using elements of random set theory.

\subsubsection{Data From a Single Network}
\label{subsubsec:networks:2}
When the researcher observes data from a single network, extra care has to be taken to restrict the dependence among linking decisions.
This can be done in various ways \citep[see, e.g.,][for some examples]{cha16}.
Here I consider a framework proposed by \cite{pau:shu:tam18}.
\begin{IP}[Network Formation Model with a Single Network]
\label{IP:networks:single}
Let there be a continuum of agents $j\in\cI=[0,\mu]$, with $\mu>0$ their total measure, who choose whom to link to based on a utility function specified below.\footnote{This is an approximation to a framework with a large but finite number of agents.
The utility function can be less restrictive than the one considered here \citep[see Assumptions 1 and 2 in][]{pau:shu:tam18}.}
Let $y:\cI\times\cI\to\{0,1\}$ be an adjacency mapping with $y_{jk}=1$ if nodes $j$ and $k$ are linked, and $y_{jk}=0$ otherwise.
Assume that only connections up to distance $\bar{d}$ affect utility and that preferences are such that agents never choose to form more than a total of $\bar{l}$ links.\footnote{The distance measure used here is the shortest path between two nodes.}
To simplify exposition, let $\bar{d}=2$.
Let each agent $j$ be endowed with characteristics $\ex_j\in\cX$, with $\cX$ a finite set in $\R^p$, that are observable to the researcher.
Additionally, let each agent $j$ be endowed with $\bar{l}\times|\cX|$ preference shocks $\epsilon_{j\ell}(x)\in\R,\ell=1,\dots,\bar{l},x\in\cX$, that are unobservable to the researcher and correspond to the possible direct connections and their characteristics.\footnote{Under this assumption, the preference shocks do not depend on the individual identities of the agents.
Hence, it agents $k$ and $m$ have the same observable characteristics, then $j$ is indifferent between them.}
Suppose that the vector of preference shocks is independent of $\ex$ and has a distribution known up to parameter vector $\gamma\in\Gamma\subset\R^m$, denoted $\sQ_\gamma$.
Let $\cI(j)=\{k:y_{jk}=1\}$.
Assume that agents with characteristics and preference shocks $(x,e)$ value links according to the utility function
\begin{multline}
\bu_j(y,x,e)=\sum_{k\in\cI(j)}(f(x_j,x_k)+e_{j\ell(k)}(x_k))\\
+\delta_1\left|\bigcup_{k\in\cI(j)}\cI(k)-\cI(j)-\{j\}\right|
+\delta_2\sum_{k\in\cI(j)}\sum_{m\in\cI(j):m>k}y_{km}-\infty\one(|\cI(k)|>\bar{l})\label{eq:utility:network:2}
\end{multline}
Assume that the network $\ey$ formed by agents with characteristics and shocks $(\ex,\epsilon)$ is pairwise stable.
Let $\Theta\equiv\Upsilon\times\Delta\times\Gamma$, with $\Upsilon$ the parameter space for $\cf\equiv\{f(x,w):x\in\cX,w\in\cX\}$.
In the absence of additional information, what can the researcher learn about $\theta\equiv[\cf~\delta_1~\delta_2~\gamma]$?
	\qedex
\end{IP}
Identification Problem \ref{IP:networks:single} enforces dimension reduction through the restrictions on depth and degree (the bounds $\bar{d}$ and $\bar{l}$), so that it is applicable to frameworks with networks that have limited degree distribution (e.g., close friendships network, but not Facebook network).
It also requires that individual identities are irrelevant.
This substantially reduces the richness of unobserved heterogeneity allowed for and the dimensionality of the space of unobservables.
While the latter feature narrows the domain of applicability of the model, it is very beneficial to obtain a tractable characterization of what can be learned about $\theta$, and yields equilibria that may include isolated nodes, a feature often encountered in networks data.

\cite{pau:shu:tam18} study Identification Problem \ref{IP:networks:single} focusing on the payoff-relevant local subnetworks that result from the maintained assumptions.
These are distinct from the subnetworks used by \cite{she18}: whereas \citeauthor{she18} looks at subnetworks formed by arbitrary individuals and whose size is chosen by the researcher on the base of computational tractability, \citeauthor{pau:shu:tam18} look at subnetworks among individuals that are within a certain distance of each other, as determined by the structure of the preferences.
On the other hand, \possessivecite{she18} analysis does not require that agents have a finite number of types nor bounds the number of links that they may form.

To characterize the local subnetworks relevant for identification analysis in their framework, \cite{pau:shu:tam18} propose the concepts of \emph{network type} and \emph{preference class}.
A network type $t=(a,v)$ describes the local network up to distance $\bar{d}$ from the reference node.
Here $a$ is a square matrix of size $1+\bar{l}\sum_{d=1}^{\bar{d}}(\bar{l}-1)^{d-1}$ that describes the local subnetwork that is utility relevant for an agent of type $t$.
It consists of the reference node, its direct potential neighbors ($\bar{l}$ elements), its second order neighbors ($\bar{l}(\bar{l}-1)$ elements), through its $\bar{d}$-th order neighbors ($\bar{l}(\bar{l}-1)^{\bar{d}-1}$ elements).
The other component of the type, $v$, is a vector of length equal to the size of $a$ that contains the observable characteristics of the reference node and her alters.
The bounds $\bar{d}$ and $\bar{l}$ enforce dimension reduction by bounding the number of network types.
The partial identification approach of \citeauthor{pau:shu:tam18} depends on this number, rather than on the number of agents.
For example, the number of moment inequalities is determined by the number of network types, not by the number of agents.
As such, the approach yields its highest dividends for dimension reduction in large networks.

Let $\cT$ denote the collection of network types generated from a preference structure $\bu$ and set of characteristics $\cX$.
For given realization $(x,e)$ of the observable characteristics and preference shocks of a reference agent, and for given $\vartheta\in\Theta$, define the collection of network types for which no agent wants to drop a link by 
\begin{align*}
H_\vartheta(x,e)=\{(a,v)\in\cT:v_1=x~\mathrm{and}~\bu(a,v,e)\ge \bu(a_{-\ell},v,e)~\forall\ell=1,\dots,\bar{l}\},
\end{align*}
where $a_{-\ell}$ is equal to the local adjacency matrix $a$ but with the $\ell$-th link removed (that is, it sets the $(1,\ell+1)$ and $(\ell+1,1)$ elements of $a$ equal to zero).
Because $(\ex,\epsilon)$ are random vectors, $\eH_\vartheta\equiv H_\vartheta(\ex,\epsilon)$ is a random closed set as per Definition \ref{def:rcs}.
This random set takes on a finite number of realizations (equal to the possible subsets of $\cT$), so that its distribution is completely determined by the probability with which it takes on each of these realizations.
A \emph{preference class} $H\subset\cT$ is one of the possible realizations of $\eH_\vartheta$ for some $\vartheta\in\Theta$. 
The model implied probability that $\eH_\vartheta=H$ is given by
\begin{align}
\sM(H|\ex;\vartheta)\equiv\sQ_{\tilde\gamma}(\epsilon:\eH_\vartheta=H|\ex).\label{eq:model:prediction:network:class}
\end{align}
Observation of data from one network allows the researcher, under suitable restrictions on the sampling process, to learn the distribution of network types in the data (type shares), denoted $\sP(t)$.\footnote{Full observation of the network is not required (and in practice it often does not occur). Sampling uncertainty results from it because in this model there is a continuum of agents.}
For example, in a network of best friends with $\bar{l}=1$ and $\bar{d}=2$, and $\cX=\{x^1,x^2\}$ (e.g., a simplified framework with only two possible races), agents are either isolated or in a pair.
Network types are pairs for the agents' race and the best friend's race (with second element equal zero if the agent is isolated).
Type shares are the fraction of isolated blacks, the fraction of isolated whites, the fraction of blacks with a black best friend, the fraction of whites with a black best friend, and the fraction of whites with a white best friend.
The preference classes for a black agent are $H^1(b,e)=\{(b,0)\}$, $H^2(b,e)=\{(b,0),(b,b)\}$, $H^3(b,e)=\{(b,0),(b,w)\}$, $H^4(b,e)=\{(b,0),(b,w),(b,b)\}$ (and similarly for whites).
In each case, being alone is part of the preference class, as there are no links to sever.
In the second class the agent has a preference for having a black friend, in the third class for a white friend, and in the last class for a friend of either race.
It is easy to see that the model is \emph{incomplete}, as for a given realization of $\epsilon$ it makes multiple predictions on the agent's preference type.

\citeauthor{pau:shu:tam18} propose to map the distribution of preference classes into the observed distribution of preference types in the data through the use of \emph{allocation parameters}, denoted $\alpha_H(t)\in[0,1]$.
These are distinct from but play the same role as a selection mechanism, and they represent a candidate distribution for $t$ given $\eH_\vartheta=H$.
The model, augmented with them, implies a probability that an agent is of network type $t$:
\begin{align}
\sM(t;\vartheta,\alpha)=\frac{1}{\mu}\sum_{H\subset\cT}\mu_{v_1(t)}\sM(H|v_1(t);\vartheta)\alpha_H(t),\label{eq:model:prediction:network:2}
\end{align}
where $\mu_{v_1(t)}$ is the measure of reference agents with characteristics equal to the second component of the preference type $t$, $\ex=v_1(t)$, and $\alpha\equiv\{\alpha_H(t):t\in \cT, H\subset\cT\}$.


\citeauthor{pau:shu:tam18} provide a characterization of an outer region for $\theta$ based on two key implications of pairwise stability that deliver  restrictions on $\alpha$.
They also show that under some additional assumptions, this characterization yields $\idr{\theta}$ \citep[Appendix B]{pau:shu:tam18}.
Here I focus on their more general result.

The first implication that they use is that existing links should not be dropped:
\begin{align}
t\notin H\Rightarrow\alpha_H(t)=0.\label{eq:networks:2:PS1}
\end{align}
The condition in \eqref{eq:networks:2:PS1} is embodied in $\bar\alpha\equiv\{\alpha_H(t):t\in H, H\subset\cT\}$.

The second implication is that it should not be possible to establish mutually beneficial links among nodes that are far from each other.
Let $t^\prime$ and $s^\prime$ denote the network types that are generated if one adds a link in networks of types $t$ and $s$ among two nodes that are at distance at least $2\bar{d}$ from each other and each have less than $\bar{l}$ links.
Then the requirement is
\begin{align}
\left(\sum_{H\subset\cT}\mu_{v_1(t)}\sM(H|v_1(t);\vartheta)\alpha_H(t)\one(t^\prime\in H)\right)\left(\sum_{H\subset\cT}\mu_{v_1(s)}\sM(H|v_1(s);\vartheta)\alpha_H(s)\one(s^\prime\in H)\right)=0\label{eq:networks:2:PS2}
\end{align}
In words, if a positive measure of agents of type $t$ prefer $t^\prime$ (i.e., $\alpha_H(t)>0$ for some $H$ such that $t^\prime\in H$), there must be zero measure of type $s$ individuals who prefer $s^\prime$, because otherwise the network is unstable.
\citeauthor{pau:shu:tam18} show that the conditions in \eqref{eq:networks:2:PS2} can be embodied in a square matrix $q$ of size equal to the length of $\bar{\alpha}$.
The entries of $q$ are constructed as follows. 
Let $H$ and $\tilde{H}$ be two preference classes with $t\in H$ and $s\in\tilde{H}$.
With some abuse of notation, let $q_{\alpha_H(t),\alpha_{\tilde{H}}(s)}$ denote the element of $q$ corresponding to the index of the entry in $\bar\alpha$ equal to $\alpha_H(t)$ for the row, and to $\alpha_{\tilde{H}}(s)$ for the column.
Then set $q_{\alpha_H(t),\alpha_{\tilde{H}}(s)}(\vartheta)=\one(t^\prime\in H)\one(s^\prime\in\tilde{H})$.
It follows that this element yields the term $\big(\alpha_H(t)\one(t^\prime\in H)\big)\big(\alpha_{\tilde{H}}(s)\one(s^\prime\in \tilde{H})\big)$ in the quadratic form $\bar{\alpha}^\top q \bar{\alpha}$.
As long as $\mu_{v_1(\cdot)}$ and $\sM(\cdot|\ex;\vartheta)$ in \eqref{eq:model:prediction:network:class} are strictly positive, this term is equal to zero if and only if condition \eqref{eq:networks:2:PS2} holds for types $t$ and $s$.\footnote{The possibility that $\mu_{v_1(\cdot)}$ or $\sM(\cdot|\ex;\vartheta)$ are equal to zero can be accommodated by setting $q_{\alpha_H(t),\alpha_{\tilde{H}}(s)}(\vartheta)=(\mu_{v_1(t)}\sM(H|v_1(t);\vartheta)\one(t^\prime\in H))(\mu_{v_1(s)}\sM(H|v_1(s);\vartheta)\one(s^\prime\in\tilde{H}))$. However, in that case $q$ depends on $\vartheta$ and its computational cost increases.}

With this background, Theorem OR-\ref{OR:networks:2} below provides an outer region for $\theta$.
The proof of this result follows from the arguments laid out above \citep[see][Theorems 1 and 2, for the full details]{pau:shu:tam18}.
\begin{OR}[Outer Region on Parameters of a Network Formation Model with a Single Network]
\label{OR:networks:2}
Under the assumptions of Identification Problem \ref{IP:networks:single},
\begin{align}
\outr{\theta}=\left\{\vartheta\in\Theta:
\left(\begin{tabular}{rl}
 $\min_{\bar{\alpha}} \bar{\alpha}^\top q \bar{\alpha}$ &  \\ 
 s.t. & $\sM(t;\vartheta,\bar{\alpha})=\sP(t)~\forall~t\in\cT$ \\ 
  & $\sum_{t\in H}\bar\alpha_H(t)=1~\forall H\subset \cT$ \\ 
  & $\alpha_H(t)\ge 0~\forall t\in H,\forall H\subset \cT$  
 \end{tabular}  \right)=0
\right\}.\label{eq:OR:networks:2}
\end{align}
\end{OR}
The set in \eqref{eq:OR:networks:2} does not equal $\idr{\theta}$ in all models allowed for in Identification Problem \ref{IP:networks:single} because condition \eqref{eq:networks:2:PS2} does not embody all implications of pairwise stability on non-existing links.
While the optimization problem in \eqref{eq:OR:networks:2} is quadratic, it is not necessarily convex because $q$ may not be positive definite.
Nonetheless, the simulations reported by \citeauthor{pau:shu:tam18} suggest that $\outr{\theta}$ can be computed rapidly, as least for the examples they considered.

\begin{BI}
At the beginning of this section I highlighted some key challenges to inference in network formation models.
When data is observed from a single network, as in Identification Problem \ref{IP:networks:single}, \possessivecite{pau:shu:tam18} proposal to base inference on local networks achieves two main benefits.
First, it delivers consistently estimable features of the game, namely the probability that an agent belongs to one of a finite collection of network types.
Second, it achieves dimension reduction, so that computation of outer regions on $\theta$ remains feasible even with large networks and allowing for unrestricted selection among multiple equilibria.
\end{BI}

\subsection{Further Theoretical Advances and Empirical Applications}
\label{subsec:applications:struct}
In order to discuss the partial identification approach to learning structural parameters of economic models in some level of detail while keeping this chapter to a manageable length, I have focused on a selection of papers.
In this section I briefly mention several other excellent theoretical contributions that could be discussed more closely, as well as several empirical papers that have applied partial identification analysis of structural models to answer a wide array of questions of substantive economic importance.\medskip

\cite{pak10} and \cite{pak:por:ho:ish15} propose to embed revealed preference-based inequalities into structural models of both demand and supply in markets where firms face discrete choices of product configuration or of location.
Revealed preference arguments are a trademark of the literature on discrete choice analysis.
\cite{pak10} and \cite{pak:por:ho:ish15} use these arguments to leverage a subset of the model's implications to obtain easy-to-compute moment inequalities.
For example, in the context of entry games such as the ones discussed in Section \ref{subsec:multiple:eq}, they propose to base inference on the implication that a player enters the market if and only if (s)he expects to make non-negative profits.
This condition can be exploited even when players have heterogeneous (unobserved to the researcher) information sets, and it implies that the expected profits for entrants should be non-negative.
Nonetheless, the condition does not suffice to obtain moment inequalities that include only observed payoff shifters and preference parameters.
This is because the expected value of unobserved payoff shifters for entrants is not equal to zero, as the group of entrants is selected.
The authors require the availability of valid (monotone) instrumental variables to solve this problem (see Section \ref{subsec:programme:eval} for uses of instrumental variables and monotone instrumental variables in partial identification analysis of treatment effects).
Interesting features of their approach include that the researcher does not need to solve for the set of equilibria, nor to require that the distribution of unobservable payoff shifters is known up to finite dimensional parameter vector.
Moreover, the same basic ideas can be applied to single agent models (with or without heterogeneous information sets).
A shortcoming of the method is that the set of parameter vectors satisfying the moment inequalities may be wider than the sharp identification region under the maintained assumptions.

The breadth of applications of the approach proposed by \cite{pak10} and \cite{pak:por:ho:ish15} is vast.\footnote{Statistical inference in these papers is often carried out using the methods proposed by \cite{che:hon:tam07}, \cite{ber:mol08}, and \cite{and:soa10}. Model specification tests, if carried out, are based on the method proposed by \cite{bug:can:shi15}. See Sections \ref{subsec:CS} and \ref{sec:misspec}, respectively, for a discussion of confidence sets and specification tests.}
For example, \cite{ho09} uses it to model the formation of the hospital networks offered by US health insurers, and \cite{ho:ho:mor12} and \cite{lee13} use it to obtain bounds on firm fixed costs as an input to modeling product choices in the movie industry and in the US video game industry, respectively.
\cite{hol11} estimates the effects of Wal-Mart's strategy of creating a high density network of stores.
While the close proximity of stores implies cannibalization in sales, Wal-Mart is willing to bear it to achieve density economies, which in turn yield savings in distribution costs.
His results suggest that Wal-Mart substantially benefits from high store density.
\cite{ell:hou:tim13} measure the effects of chain economies, business stealing, and heterogeneous firms' comparative advantages in the discount retail industry.
\cite{kaw:wat13} estimate a model of strategic voting and quantify the impact it has on election outcomes. 
As in other models analyzed in this section, the one they study yields multiple predicted outcomes, so that partial identification methods are required to carry out the empirical analysis if one does not assume a specific selection mechanism to resolve the multiplicity.
They estimate their model on Japanese general-election data, and uncover a sizable fraction of strategic voters.
They also estimate that only a small fraction of voters are misaligned (voting for a candidate other than their most preferred one).
\cite{eiz14} studies whether the rapid removal from the market for personal computers of existing central processing units upon creation of new ones through innovation reduces surplus.
He finds that a limited group of price-insensitive consumers enjoys the largest share of the welfare gains from innovation.
A policy that kept older technologies on the shelf would allow for the benefits from innovation to reach price-sensitive consumers thanks to improved access to mobile computing, but total welfare would not increase because consumer welfare gains would be largely offset by producer losses.
\cite{ho:pak14} analyze hospital referrals for labor and birth episodes in California in 2003, for patients enrolled with six health insurers that use, to a different extent, incentives to referring physicians groups to reduce hospital costs (capitation contracts).
The aim is to learn whether enrollees with high-capitation insurers tend to be referred to lower-priced hospitals (ceteris paribus) compared to other patients with same-severity conditions, and whether quality of care was affected.
Their model allows for an insurer-specific preference function that is additively separable in the hospital price paid by the insurer (which is allowed to be measured with error), the distance traveled, and plan and severity-specific hospital fixed effects.
Importantly, unobserved heterogeneity entering the preference function is not assumed to be drawn from a distribution known up to finite dimensional parameter vector.
The results of the empirical analysis indicate that the price paid by insurers to hospitals has an impact on referrals, with higher elasticity to price for insurers whose physicians groups are more highly capitated.
\cite{dic:mor18} study how the information that potential exporters have to predict the profits they will earn when serving a foreign market influences their decisions to export.
They propose a model where the researcher specifies and observes a subset of the variables that agents use to form their expectations, but may not observe other variables that affect firms' expectations heterogeneously (across firms and markets, and over time).
Because only a subset of the variables entering the firms' information set is observed, partial identification results.
They show that, under rational expectations, they can test whether potential exporters know and use specific variables to predict their export profits. 
They also use their model's estimates to quantify the value of information.
\cite{wol18} studies the implications of the \$85 billion automotive industry bailout in 2009 on the commercial vehicle segment.
He finds that had Chrysler and GM been liquidated (or aquired by a major competitor) rather than bailed out, the surviving firms would have experienced a rise in profits high enough to induce them to introduce new products.

A different use of revealed preference arguments appears in the contributions of \cite{blu:bro:cra08}, \cite{blu:kri:mat14}, \cite{hod:sto14,hod:sto15}, \cite{man14}, \cite{bar:mol:tei16}, \cite{hau:new16}, \cite{ada19}, and many others.
For example, \cite{man14} proposes a method to partially identify income-leisure preferences and to evaluate the associated effects of tax policies. 
He starts from basic revealed-preference analysis performed under the assumption that individuals prefer more income and leisure, and no other restriction. 
The analysis shows that observing an individual's time allocation under a status quo tax policy yields bounds on his allocation that may or may not be informative, depending on how the person allocates his time under the status quo policy and on the tax schedules.
He then explores what more can be learned if one additionally imposes restrictions on the distribution of income-leisure preferences, using the method put forward by \cite{man07b}.
One assumption restricts groups of individuals facing different choice sets to have the same distribution of preferences.
The other assumption restricts this distribution to a parametric family.
\cite{kli:tar16} build on and expand \cite{man14}'s framework to evaluate the effect of Connecticut's Jobs First welfare reform experiment on women' labor supply and welfare participation decisions.

\cite{bar:mol:tei16} propose a method to learn features of households' risk preferences in a random utility model that nests expected utility theory plus a range of non-expected utility models.\footnote{Their model is based on the one  put forward by \cite{bar:mol:odo:tei13}. See \cite{bar:mol:odo:tei18} for a review of these and other non-expected utility models in the context of estimation of risk preferences.}
They allow for unobserved heterogeneity in preferences (that may enter the utility function non-separably) and leave completely unspecified their distribution.
The authors use revealed preference arguments to infer, for each household, a set of values for its unobserved heterogeneity terms that are consistent with the household's choices in the three lines of insurance coverage.
As their core restriction, they assume that each household's preferences are \emph{stable} across contexts: the household's utility function is the same when facing distinct but closely related choice problems.
This allows them to use the inferred set valued data to partially identify features of the distribution of preferences, and to classify households into preference types.
They apply their proposed method to analyze data on households' deductible choices across three lines of insurance coverage (home all perils, auto collision, and auto comprehensive).\footnote{Auto collision coverage pays for damage to the insured vehicle caused by a collision with another vehicle or object, without regard to fault. 
Auto comprehensive coverage pays for damage to the insured vehicle from all other causes, without regard to fault.
Home all perils (or simply home) coverage pays for damage to the insured home from all causes, except those that are specifically excluded (e.g., flood, earthquake, or war).} 
Their results show that between 70 and 80 percent of the households make choices that can be rationalized by a model with linear utility and monotone, quadratic, or even linear probability distortions. 
These probability distortions substantially overweight small probabilities.
By contrast, fewer than 40 percent can be rationalized by a model with concave utility but no probability distortions.

\cite{hau:new16} propose a method to carry out demand analysis while allowing for general forms of unobserved heterogeneity.
Preferences and linear budget sets are assumed to be statistically independent (conditional on covariates and control functions).
\citeauthor{hau:new16} show that for continuous demand, average surplus
is generally not identified from the distribution of demand for a given
price and income, and therefore propose a partial identification approach.
They use bounds on income effects to derive bounds on average surplus. They apply the bounds to gasoline demand, using data from the
2001 U.S. National Household Transportation Survey.

Another strand of empirical applications pertains to the analysis of discrete games.
\cite{cil:tam09} use the method they develop, described in Section \ref{subsubsec:tam03:cil:tam09}, to study market structure in the US airline industry and the role that firm heterogeneity plays in shaping it.
Their findings suggest that the competitive effects of each carrier increase in that carrier's airport presence, but also that the competitive effects of large carriers (American, Delta, United) are different from those of low cost ones (Southwest).
They also evaluate the effect of a counterfactual policy repealing the Wright Amendment, and find that doing so would see an increase in the number of markets served out of Dallas Love.

\cite{gri14} proposes a model of static entry that extends the one in Section \ref{subsec:multiple:eq} by allowing individuals to have flexible information structures, where players's payoffs depend on both a common-knowledge unobservable payoff shifter, and a private-information one.
His characterization of $\idr{\theta}$ is based on using an unrestricted selection mechanism, as in \cite{ber:tam06} and \cite{cil:tam09}.
He applies the model to study the impact of supercenters such as Wal-Mart, that sell both food and groceries, on the profitability of rural grocery stores. 
He finds that entry by a supercenter outside, but within 20 miles, of a local monopolist's market has a smaller impact on firm profits than entry by a local grocer.  
Their entrance has a small negative effect on the number of grocery stores in surrounding markets as well as on their profits.
The results suggest that location and format-based differentiation  partially insulate rural stores from competition with supercenters.  

A larger class of information structures is considered in the analysis of static discrete games carried out by \cite{mag:ron17}.
They allow for all information structures consistent with the players knowing their own payoffs and the distribution of opponents' payoffs. 
As solution concept they adopt the Bayes Correlated Equilibrium recently developed by \cite{ber:mor16}.
Also with this solution concept multiple equilibria are possible.
The authors leave completely unspecified the selection mechanism picking the equilibrium played in the regions of multiplicity, so that partial identification attains. 
\citeauthor{mag:ron17} use the random sets approach to characterize $\idr{\theta}$.
They apply the method to estimate a model of entry in the Italian supermarket industry and quantify the effect of large malls on local
grocery stores.
\cite{nor:tan14} provide partial identification results (and Bayesian inference methods) for semiparametric dynamic binary choice models without imposing distributional assumptions on the unobserved state variables.
They carry out an empirical application using \cite{rus87}'s model of bus engine replacement. 
Their results suggest that parametric assumptions about the distribution of the unobserved states can have a considerable effect on the estimates of per-period payoffs, but not a noticeable one on the counterfactual conditional choice probabilities.
\cite{ber:com19} use the random sets approach to partially identify and estimate dynamic discrete choice models with serially correlated unobservables, under instrumental variables restrictions.
They extend two-step dynamic estimation methods to characterize a set of structural parameters that are consistent with the dynamic model, the instrumental variables restrictions, and the data.\footnote{Statistical inference on $\theta$ is carried out using \cite{che:che:kat18}'s method.}
\cite{gua19} uses the random sets approach and a network formation model, to learn about Italian firms' incentives for having their executive directors sitting on the board of their competitors.

\cite{bar:cou:mol:tei18} use the method described in Section \ref{subsubsec:BCMT} to partially identify the distribution of risk preferences using data on deductible choices in auto collision insurance.\footnote{Statistical inference on projections of $\theta$ is carried out using \cite{kai:mol:sto19}'s method.}
They posit an expected utility theory model and allow for unobserved heterogeneity in households' risk aversion and choice sets, with unrestricted dependence between them.
Motivation for why unobserved heterogeneity in choice sets might be an important factor in this empirical framework comes from the earlier analysis of \cite{bar:mol:tei16} and novel findings that are part of \possessivecite{bar:cou:mol:tei18} contribution.
They show that commonly used models that make strong assumptions about choice sets (e.g., the mixed logit model with each individual's choice set assumed equal to the feasible set, and various models of choice set formation) can be rejected in their data.
With regard to risk aversion, their key finding is that their estimated lower bounds are significantly smaller than the point estimates obtained in the related literature. 
This suggests that the data can be explained by expected utility theory with lower and more homogeneous levels of risk aversion than it had been uncovered before. 
This provides new evidence on the importance of developing models that differ in their specification of \emph{which} alternatives agents evaluate (rather than or in addition to models focusing on \emph{how} they evaluate them), and to data collection efforts that seek to directly measure agents' heterogeneous choice sets \citep{cap16}.

\cite{iar:shi:shu18} study the effect of pre-vote deliberation on the decisions of US appellate courts.
The question of interest is weather deliberation increases or reduces the probability of an incorrect decision.
They use a model where communication equilibrium is the solution concept, and only observed heterogeneity in payoffs is allowed for. 
In the model, multiple equilibria are again possible, and the authors leave the selection mechanism completely unspecified.
They characterize $\idr{\theta}$ through an optimization problem, and structurally estimate the model on US Courts of Appeal data.
\citeauthor{iar:shi:shu18} compare the probability of making incorrect decisions under the pre-vote deliberation mechanism, to that in a counterfactual environment where no deliberation occurs.
The results suggest that there is a range of parameters in $\idr{\theta}$, for which judges have ex-ante disagreement of imprecise prior information, for which deliberation is beneficial.
Otherwise deliberation leads to lower effectiveness for the court.

\cite{dha:gai:mau18} propose a test for the hypothesis of rational expectations for the case that one observes only the marginal distributions of realizations and subjective beliefs, but not their joint distribution (e.g., when subjective beliefs are observed in one dataset, and realizations in a different one, and the two cannot be matched).
They establish that the hypothesis of rational expectations can be expressed as testing that a continuum of moment inequalities is satisfied, and they leverage the results in \cite{and:shi17} to provide a simple-to-compute test for this hypothesis.
They apply their method to test for and quantify deviations from rational expectations about future earnings, and examine the consequences of such departures in the context of a life-cycle model of consumption.

\cite{teb:tor:yan19} estimate the demand for health insurance under the Affordable Care Act using data from California.
Methodologically, they use a discrete choice model that allows for endogeneity in insurance premiums (which enter as explanatory variables in the model) and dispenses with parametric assumptions about the unobserved components of utility leveraging the availability of instrumental variables, similarly to the framework presented in Section \ref{subsubsec:CRS}. 
The authors provide a characterization of sharp bounds on the effects of changing premium subsidies on coverage choices, consumer surplus, and government spending, as solutions to linear programming problems, rendering their method computationally attractive.

Another important strand of theoretical literature is concerned with partial identification of panel data models.
\cite{hon:tam06} consider a dynamic random effects probit model, and use partial identification analysis to obtain bounds on the model parameters that circumvent the initial conditions problem.
\cite{ros12} considers a fixed effect panel data model where he imposes a conditional quantile restriction on time varying unobserved heterogeneity.
Differencing out inequalities resulting from the conditional quantile restriction delivers inequalities that depend only on observable variables and parameters to be estimated, but not on the fixed effects, so that they can be used for estimation.
\cite{che:fer:hah:new13} obtain bounds on average and quantile treatment effects in nonparametric and semiparametric nonseparable panel data models.
\cite{kha:pon:tam16} provide partial identification results in linear panel data models when censored outcomes, with unrestricted dependence between censoring and observable and unobservable variables. 
Their results are derived for two classes of models, one where the unobserved heterogeneity terms satisfy a stationarity restriction, and one where they are nonstationary but satisfy a conditional independence restriction. 
\cite{tor19} provides a method to partially identify state dependence in panel data models where individual unobserved heterogeneity needs not be time invariant.
\cite{pak:por16} study semiparametric multinomial choice panel models with fixed effects where the random utility function is assumed additively separable in unobserved heterogeneity, fixed effects, and a linear covariate index.
The key semiparametric assumption is a group stationarity condition on the disturbances which places no restrictions on either the joint distribution of the disturbances across choices or the correlation of disturbances across time. \citeauthor{pak:por16} propose a within-group comparison that delivers a collection of conditional moment inequalities that they use to provide point and partial identification results.
\cite{ari19} proposes a related method, where partial identification relies on the observation of individuals whose outcome changes in two consecutive time periods, and leverages shape restrictions to reduce the number of between alternatives comparisons needed to determine the optimal choice.

\section{Estimation and Inference}
\label{sec:inference}

\subsection{Framework and Scope of the Discussion}
\label{subsec:framework:inference}
The identification analysis carried out in Sections \ref{sec:prob:distr}-\ref{sec:structural} presumes knowledge of the joint distribution $\sP$ of the observable variables.
That is, it presumes that $\sP$ can be learned with certainty from observation of the entire population.
In practice, one observes a sample of size $n$ drawn from $\sP$.
For simplicity I assume it to be a random sample.\footnote{This assumption is often maintained in the literature. See, e.g., \cite{and:soa10} for a treatment of inference with dependent observations. \cite{eps:kai:seo16} study inference in games of complete information as in Identification Problem \ref{IP:entry_game}, imposing the i.i.d. assumption on the unobserved payoff shifters $\{\eps_{i1},\eps_{i2}\}_{i=1}^n$. The authors note that because the selection mechanism picking the equilibrium played in the regions of multiplicity (see Section \ref{subsec:multiple:eq}) is left completely unspecified and may be arbitrarily correlated across markets, the resulting observed variables $\{\ew_i\}_{i=1}^n$ may not be independent and identically distributed, and they propose an inference method to address this issue.}

Statistical inference on $\idr{\theta}$ needs to be conducted using knowledge of $\sP_n$, the empirical distribution of the observable outcomes and covariates.
Because $\idr{\theta}$ is not a singleton, this task is particularly delicate. 
To start, care is required to choose a proper notion of consistency for a set estimator $\idrn{\theta}$ and to obtain palatable conditions under which such consistency attains.
Next, the asymptotic behavior of statistics designed to test hypothesis or build confidence sets for $\idr{\theta}$ or for $\vartheta\in\idr{\theta}$ might change with $\vartheta$, creating technical challenges for the construction of confidence sets that are not encountered when $\theta$ is point identified.
Many of the sharp identification regions derived in Sections \ref{sec:prob:distr}-\ref{sec:structural} can be written as collections of vectors $\vartheta\in\Theta$ that satisfy conditional or unconditional moment (in)equalities.
For simplicity, I assume that $\Theta$ is a compact and convex subset of $\R^d$, and I use the formalization for the case of a finite number of unconditional moment (in)equalities:
\begin{align}
\idr{\theta}=\{\vartheta\in\Theta: \E_\sP(m_j(\ew_i;\vartheta))&\le 0~\forall j\in\cJ_1,~
\E_\sP(m_j(\ew_i;\vartheta))=0~\forall j\in\cJ_2\}.\label{eq:sharp_id_for_inference}
\end{align}
In \eqref{eq:sharp_id_for_inference}, $\ew_i\in\cW\subseteq\R^{d_\cW}$ is a random vector collecting all observable variables, with $\ew\sim\sP$; $m_j:\cW\times\Theta\to\R$, $j\in\cJ\equiv\cJ_1\cup\cJ_2$, are known measurable functions characterizing the model; and $\cJ$ is a finite set equal to $\{1,\dots,|\cJ|\}$.\footnote{Examples where the set $\cJ$ is a compact set (e.g., a unit ball) rather than a finite set include the case of best linear prediction with interval outcome and covariate data, see characterization \eqref{eq:ThetaI:BLP} on p.~\pageref{eq:ThetaI:BLP}, and the case of entry games with multiple mixed strategy Nash equilibria, see characterization \eqref{eq:SIR_sharp_mixed_sup} on p.~\pageref{eq:SIR_sharp_mixed_sup}.
A more general continuum of inequalities is also possible, as in the case of discrete choice with endogenous explanatory variables, see characterization \eqref{eq:SIR:discrete:choice:endogenous} on p.~\pageref{eq:SIR:discrete:choice:endogenous}.
I refer to \cite{and:shi17} and \cite[Supplementary Appendix B]{ber:mol:mol11} for inference methods in the presence of a continuum of conditional moment (in)equalities.}
Instances where $\idr{\theta}$ is characterized through a finite number of conditional moment (in)equalities and the conditioning variables have finite support can easily be recast as in \eqref{eq:sharp_id_for_inference}.\footnote{I refer to \cite{kha:tam09}, \cite{and:shi13}, \cite{che:lee:ros13}, \cite{lee:son:wha13}, \cite{arm14b,arm15}, \cite{arm:cha16}, \cite{che:che:kat18}, and \cite{che18}, for inference methods in the case that the conditioning variables have a continuous distribution.}
Consider, for example, the two player entry game model in Identification Problem \ref{IP:entry_game} on p.~\pageref{IP:entry_game}, where $\ew=(\ey_1,\ey_2,\ex_1,\ex_2)$.
Using (in)equalities \eqref{eq:CT_00}-\eqref{eq:CT_01L} and assuming that the distribution of $(\ex_1,\ex_2)$ has $\bar{k}$ points of support, denoted $(x_{1,k},x_{2,k}),k=1,\dots,\bar{k}$, we have $|\cJ|=4\bar{k}$ and for $k=1,\dots,\bar{k}$,\footnote{In these expressions an index of the form $jk$ not separated by a comma equals the product of $j$ with $k$.}
\begin{align*}
m_{4k-3}(\ew_i;\vartheta)&=[\one((\ey_1,\ey_2)=(0,0))-\Phi((-\infty,-\ex_1b_1),(-\infty,-\ex_2b_2);r)]\one((\ex_1,\ex_2)=(x_{1,k},x_{2,k}))\\
m_{4k-2}(\ew_i;\vartheta)&=[\one((\ey_1,\ey_2)=(1,1))-\Phi([-\ex_1b_1-d_1,\infty),[-\ex_2b_2-d_2,\infty);r)]\one((\ex_1,\ex_2)=(x_{1,k},x_{2,k}))\\
m_{4k-1}(\ew_i;\vartheta)&=[\one((\ey_1,\ey_2)=(0,1))-\Phi((-\infty,-\ex_1b_1-d_1),(-\ex_2b_2,\infty);r)]\one((\ex_1,\ex_2)=(x_{1,k},x_{2,k}))\\
m_{4k}(\ew_i;\vartheta)&=\Big[\one((\ey_1,\ey_2)=(0,1))-\Big\{\Phi((-\infty,-\ex_1b_1-d_1),(-\ex_2b_2,\infty);r)\notag\\
&\quad\quad-\Phi((-\ex_1b_1,-\ex_1b_1-d_1),(-\ex_2b_2,-\ex_2b_2-d_2);r)\Big\}\Big]\one((\ex_1,\ex_2)=(x_{1,k},x_{2,k})).
\end{align*}

In point identified moment equality models it has been common to conduct estimation and inference using a criterion function that aggregates moment violations \citep{han82}.
\cite{man:tam02} adapt this idea to the partially identified case, through a criterion function $\crit_\sP:\Theta\to\R_+$ such that $\crit_\sP(\vartheta)=0$ if and only if $\vartheta\in\idr{\theta}$.
Many criterion functions can be used \citep[see, e.g.][]{man:tam02,che:hon:tam07,rom:sha08,ros08,gal:hen09,and:gug09b,and:soa10,can10,rom:sha10}.
Some simple and commonly employed ones include
\begin{align}
\crit_{\sP,\mathrm{sum}}(\vartheta) &= \sum_{j\in\cJ_1}\left[\frac{\E_\sP(m_j(\ew_i;\vartheta))}{\sigma_{\sP,j}(\vartheta)}\right]_+^2 + \sum_{j\in\cJ_2}\left[\frac{\E_\sP(m_j(\ew_i;\vartheta))}{\sigma_{\sP,j}(\vartheta)}\right]^2,\label{eq:criterion_fn_sum}\\
\crit_{\sP,\mathrm{max}}(\vartheta) &= \max\left\{\max_{j\in\cJ_1}\left[\frac{\E_\sP(m_j(\ew_i;\vartheta))}{\sigma_{\sP,j}(\vartheta)}\right]_+,\max_{j\in\cJ_2}\left|\frac{\E_\sP(m_j(\ew_i;\vartheta))}{\sigma_{\sP,j}(\vartheta)}\right|\right\}^2,\label{eq:criterion_fn_max}
\end{align}
where $[x]_+=\max\{x,0\}$ and $\sigma_{\sP,j}(\vartheta)$ is the population standard deviation of $m_j(\ew_i;\vartheta)$.
In \eqref{eq:criterion_fn_sum}-\eqref{eq:criterion_fn_max} the moment functions are standardized, as doing so is important for statistical power \citep[see, e.g.,][p. 127]{and:soa10}.
To simplify notation, I omit the label and simply use $\crit_\sP(\vartheta)$.
Given the criterion function, one can rewrite \eqref{eq:sharp_id_for_inference} as
\begin{align}
\label{eq:define:idr}
\idr{\theta}=\{\vartheta\in\Theta:\crit_\sP(\vartheta)=0\}.
\end{align}

To keep this chapter to a manageable length, I focus my discussion of statistical inference \emph{exclusively} on consistent estimation and on different notions of coverage that a confidence set may be required to satisfy and that have proven useful in the literature.\footnote{Using the well known duality between tests of hypotheses and confidence sets, the discussion could be re-framed in terms of size of the test.}
The topics of test of hypotheses and construction of confidence sets in partially identified models are covered in \cite{can:sha17}, who provide a comprehensive survey devoted entirely to them in the context of moment inequality models.
\cite[Chapters 4 and 5]{mol:mol18} provide a thorough discussion of related methods based on the use of random set theory.

\subsection{Consistent Estimation}
\label{subsec:consistent}
When the identified object is a set, it is natural that its estimator is also a set.
In order to discuss statistical properties of a set-valued estimator $\idrn{\theta}$ (to be defined below), and in particular its consistency, one needs to specify how to measure the distance between $\idrn{\theta}$ and $\idr{\theta}$.
Several distance measures among sets exist \citep[see, e.g.,][Appendix D]{mo1}.
A natural generalization of the commonly used Euclidean distance is the \emph{Hausdorff distance}, see Definition \ref{def:hausdorff}, which for given $A,B\subset\R^d$ can be written as
\begin{align*}
\dist_H(A,B) = \inf\Big\{r>0:\; A\subseteq B^r,\; B\subseteq A^r\Big\}=\max\left\{\sup_{a \in A} \dist(a,B), \sup_{b \in B} \dist(b,A) \right\},
\end{align*}
with $\dist(a,B)\equiv\inf_{b\in B}\Vert a-b\Vert$.\footnote{The definition of the Hausdorff distance can be generalized to an arbitrary metric space by replacing the Euclidean metric by the metric specified on that space.}
In words, the Hausdorff distance between two sets measures the furthest distance from an arbitrary point in one of the sets to its closest neighbor in the other set.
It is easy to verify that $\dist_H$ metrizes the family of non-empty compact sets; in particular, given non-empty compact sets $A,B\subset\R^d$, $\dist_H(A,B) =0$ if and only if $A=B$.
If either $A$ or $B$ is empty, $\dist_H(A,B) =\infty$.

The use of the Hausdorff distance to conceptualize consistency of set valued estimators in econometrics was proposed by \cite[Section 2.4]{han:hea:lut95} and \cite[Section 3.2]{man:tam02}.\footnote{It was previously used in the mathematical literature on random set theory, for example to formalize laws of large numbers and central limit theorems for random sets such as the ones in Theorems \ref{thr:SLLN-basic} and \ref{thr:clt} \citep{art:vit75,gin:hah:zin83}.}
\begin{definition}[Hausdorff Consistency]
\label{def:consistent_estimator}
An estimator $\idrn{\theta}$ is consistent for $\idr{\theta}$ if 
\begin{align*}
\dist_H(\idrn{\theta},\idr{\theta}) \stackrel{p}{\rightarrow} 0 ~\text{as } n\to \infty.
\end{align*}
\end{definition}
\cite{mol98} establishes Hausdorff consistency of a plug-in estimator of the set $\{\vartheta\in\Theta:g_\sP(\vartheta)\le 0\}$, with $g_\sP:\cW\times\Theta \to \R$ a lower semicontinuous function of $\vartheta\in\Theta$ that can be consistently estimated by a lower semicontinuous function $g_n$ uniformly over $\Theta$.
The set estimator is $\{\vartheta\in\Theta:g_n(\vartheta)\le 0\}$.
The fundamental assumption in \cite{mol98} is that $\{\vartheta\in\Theta:g_\sP(\vartheta)\le 0\}\subseteq\cl(\{\vartheta\in\Theta:g_\sP(\vartheta)< 0\})$, see \cite[Section 5.2]{mol:mol18} for a discussion.
There are important applications where this condition holds.
\cite{che:koc:men15} provide results related to \cite{mol98}, as well as important extensions for the construction of confidence sets, and show that these can be applied to carry out statistical inference on the Hansen–Jagannathan sets of admissible stochastic discount factors \citep{han:jag91}, the Markowitz–Fama mean–variance sets for asset portfolio returns \citep{mar52}, and the set of structural elasticities in \cite{che12b}'s analysis of demand with optimization frictions.
However, these methods are not broadly applicable in the general moment (in)equalities framework of this section, as \citeauthor{mol98}'s key condition generally fails for the set $\idr{\theta}$ in \eqref{eq:define:idr}.\medskip

\subsubsection{Criterion Function Based Estimators}
\cite{man:tam02} extend the standard theory of extremum estimation of point identified parameters to partial identification, and propose to estimate $\idr{\theta}$ using the collection of values $\vartheta\in\Theta$ that approximately minimize a sample analog of $\crit_\sP$:
\begin{align}
\idrn{\theta}=\left\{\vartheta\in\Theta:\crit_n(\vartheta)\le \inf_{\tilde\vartheta\in\Theta}\crit_n(\tilde\vartheta)+\tau_n\right\},\label{eq:define:idrn}
\end{align} 
with $\tau_n$ a sequence of non-negative random variables such that $\tau_n\stackrel{p}{\rightarrow} 0$.
In \eqref{eq:define:idrn}, $\crit_n(\vartheta)$ is a sample analog of $\crit_\sP(\vartheta)$ that replaces $\E_\sP(m_j(\ew_i;\vartheta))$ and $\sigma_{\sP,j}(\vartheta)$ in \eqref{eq:criterion_fn_sum}-\eqref{eq:criterion_fn_max} with properly chosen estimators, e.g.,
\begin{align*}
\bar m_{n,j}(\vartheta) &\equiv {\frac{1}{n}\sum_{i=1}^n m_j(\ew_i,\vartheta)},~j=1,\dots, |\cJ| 
\\
\hat{\sigma}_{n,j}(\vartheta) &\equiv {\left(\frac{1}{n}\sum_{i=1}^n [m_j(\ew_i,\vartheta)]^2-[\bar m_{n,j}(\vartheta)]^2\right)^{1/2}},~j=1,\dots, |\cJ|. 
\end{align*}

It can be shown that as long as $\tau_n=o_p(1)$, under the same assumptions used to prove consistency of extremum estimators of point identified parameters (e.g., with uniform convergence of $\crit_n$ to $\crit_\sP$ and continuity of $\crit_\sP$ on $\Theta$),
\begin{align}
\sup_{\vartheta \in \idrn{\theta}} \inf_{\tilde\vartheta \in \idr{\theta}} \Vert \vartheta-\tilde\vartheta \Vert\stackrel{p}{\rightarrow} 0~\text{as } n\to \infty.\label{eq:inner_consistent}
\end{align}
This yields that asymptotically each point in $\idrn{\theta}$ is arbitrarily close to a point in $\idr{\theta}$, or more formally that $\sP(\idrn{\theta}\subseteq\idr{\theta})\to 1$.
I refer to \eqref{eq:inner_consistent} as \emph{inner consistency} henceforth.\footnote{See \cite[Theorem 1]{ble15} for a pedagogically helpful proof for a semiparametric binary model.}
\cite{red81} provides an early contribution establishing this type of inner consistency for maximum likelihood estimators when the true parameter is not point identified.

However, Hausdorff consistency requires also that
\begin{align*}
\sup_{\vartheta \in \idr{\theta}} \inf_{\tilde\vartheta \in \idrn{\theta}} \Vert \vartheta-\tilde\vartheta \Vert\stackrel{p}{\rightarrow} 0~\text{as } n\to \infty,
\end{align*}
i.e., that each point in $\idr{\theta}$ is arbitrarily close to a point in $\idrn{\theta}$, or more formally that $\sP(\idr{\theta}\subseteq\idrn{\theta})\to 1$.
To establish this result for the sharp identification regions in Theorem SIR-\ref{SIR:man:tam02_param} (parametric regression with interval covariate) and Theorem SIR-\ref{SIR:man:tam02_binary} (semiparametric binary model with interval covariate), \cite[Propositions 3 and 5]{man:tam02} require the rate at which $\tau_n\stackrel{p}{\rightarrow} 0$ to be slower than the rate at which $\crit_n$ converges uniformly to $\crit_\sP$ over $\Theta$.

What might go wrong in the absence of such a restriction?
A simple example can help understand the issue.
Consider a model with linear inequalities of the form
\begin{align*}
\theta_1 &\le \E_\sP(\ew_1),\\
-\theta_1 &\le \E_\sP(\ew_2),\\
\theta_2 &\le \E_\sP(\ew_3)+ \E_\sP(\ew_4)\theta_1,\\
-\theta_2 &\le \E_\sP(\ew_5)+ \E_\sP(\ew_6)\theta_1.
\end{align*}
Suppose $\ew\equiv(\ew_1,\dots,\ew_6)$ is distributed multivariate normal, with $\E_\sP(\ew)=[6~0~2~0~{-2}~0]^\top$ and $\cov_\sP(\ew)$ equal to the identity matrix.
Then $\idr{\theta}=\{\vartheta=[\vartheta_1~\vartheta_2]^\top\in\Theta:\vartheta_1\in[0,6]~\text{and}~\vartheta_2=2\}$.
However, with positive probability in any finite sample $\crit_n(\vartheta)=0$ for $\vartheta$ in a random region (e.g., a triangle if $\crit_n$ is the sample analog of \eqref{eq:criterion_fn_max}) that only includes points that are close to a subset of the points in $\idr{\theta}$.
Hence, with positive probability the minimizer of $\crit_n$ cycles between consistent estimators of subsets of $\idr{\theta}$, but does not estimate the entire set.
Enlarging the estimator to include all points that are close to minimizing $\crit_n$ up to a tolerance that converges to zero sufficiently slowly removes this problem.\medskip

\cite{che:hon:tam07} significantly generalize the consistency results in \cite{man:tam02}.
They work with a normalized criterion function equal to $\crit_n(\vartheta)-\inf_{\tilde\vartheta\in\Theta}\crit_n(\tilde\vartheta)$, but to keep notation light I simply refer to it as $\crit_n$.\footnote{Using this normalized criterion function is especially important in light of possible model misspecification, see Section \ref{sec:misspec}.}
Under suitable regularity conditions, they establish consistency of an estimator that can be a smaller set than the one proposed by \cite{man:tam02}, and derive its convergence rate.
Some of the key conditions required by \cite[Conditions C1 and C2]{che:hon:tam07} to study convergence rates include that $\crit_n$ is lower semicontinuous in $\vartheta$, satisfies various convergence properties among which $\sup_{\vartheta\in\idr{\theta}}\crit_n=O_p(1/a_n)$ for a sequence of normalizing constants $a_n\to\infty$, that $\tau_n\ge \sup_{\vartheta\in\idr{\theta}}\crit_n(\vartheta)$ with probability approaching one, and that $\tau_n\to 0$.
They also require that there exist positive constants $(\delta,\kappa,\gamma)$ such that for any $\epsilon\in(0,1)$ there are $(d_\epsilon,n_\epsilon)$ such that
\begin{align*}
\forall n\ge n_\epsilon, \, \crit_n(\vartheta)\ge\kappa[\min\{\delta,\dist(\vartheta,\idr{\theta})\}]^\gamma
\end{align*} 
uniformly on $\{\vartheta\in\Theta:\dist(\vartheta,\idr{\theta})\ge(d_\epsilon/a_n)^{1/\gamma}\}$ with probability at least $1-\epsilon$.
%
In words, the assumption, referred to as \emph{polynomial minorant} condition, rules out that $\crit_n$ can be arbitrarily close to zero outside $\idr{\theta}$.
It posits that $\crit_n$ changes as at least a polynomial of degree $\gamma$ in the distance of $\vartheta$ from $\idr{\theta}$. 
Under some additional regularity conditions, \cite{che:hon:tam07} establish that 
\begin{align}
\dist_H(\idrn{\theta},\idr{\theta})=O_p(\max\{1/a_n,\tau_n\})^{1/\gamma}.\label{eq:CHT_rate}
\end{align}

What is the role played by the polynomial minorant condition for the result in \eqref{eq:CHT_rate}?
Under the maintained assumptions $\tau_n\ge \sup_{\vartheta\in\idr{\theta}}\crit_n(\vartheta)\ge\kappa[\min\{\delta,\dist(\vartheta,\idr{\theta})\}]^\gamma$, and the latter part of the inequality is used to obtain \eqref{eq:CHT_rate}.
When could the polynomial minorant condition be violated?
In moment (in)equalities models, \citeauthor{che:hon:tam07} require $\gamma=2$.\footnote{\cite[equation (4.1) and equation (4.6)]{che:hon:tam07} set $\gamma=1$ because they report the assumption for a criterion function that does not square the moment violations.}
Consider a simple stylized example with (in)equalities of the form
\begin{align*}
-\theta_1 &\le \E_\sP(\ew_1),\\
-\theta_2 &\le \E_\sP(\ew_2),\\
\theta_1\theta_2 &= \E_\sP(\ew_3),
\end{align*}
with $\E_\sP(\ew_1)=\E_\sP(\ew_2)=\E_\sP(\ew_3)=0$, and note that the sample means $(\bar{\ew}_1,\bar{\ew}_2,\bar{\ew}_3)$ are $\sqrt{n}$-consistent estimators of $(\E_\sP(\ew_1),\E_\sP(\ew_2),\E_\sP(\ew_3))$.
Suppose $(\ew_1,\ew_2,\ew_3)$ are distributed multivariate standard normal.
Consider a sequence $\vartheta_n=[\vartheta_{1n}~\vartheta_{2n}]^\top=[n^{-1/4}~n^{-1/4}]^\top$.
Then $[\dist(\vartheta_n,\idr{\theta})]^\gamma=O_p(n^{-1/2})$.
On the other hand, with positive probability $\crit_n(\vartheta_n)=(\bar{\ew}_3-\vartheta_{1n}\vartheta_{2n})^2=O_p\left(n^{-1}\right)$, so that for $n$ large enough $\crit_n(\vartheta_n)<[\dist(\vartheta_n,\idr{\theta})]^\gamma$, violating the assumption.
This occurs because the gradient of the moment equality vanishes as $\vartheta$ approaches zero, rendering the criterion function flat in a neighborhood of $\idr{\theta}$.
As intuition would suggest, rates of convergence are slower the flatter $\crit_n$ is outside $\idr{\theta}$.

\cite{kai:mol:sto19CQ} show that in moment inequality models with smooth moment conditions, the polynomial minorant assumption with $\gamma=2$ implies the Abadie constraint qualification (ACQ); see, e.g., \cite[Chapter 5]{baz:she:she06} for a definition and discussion of ACQ.
The example just given to discuss failures of the polynomial minorant condition is in fact a known example where ACQ fails at $\vartheta=[0~0]^\top$.

\cite[Condition C.3, referred to as \emph{degeneracy}]{che:hon:tam07} also consider the case that $\crit_n$ vanishes on subsets of $\Theta$ that converge in Hausdorff distance to $\idr{\theta}$ at rate $a_n^{-1/\gamma}$.
While degeneracy might be difficult to verify in practice, \citeauthor{che:hon:tam07} show that if it holds, $\tau_n$ can be set to zero.
\cite{yil12} provides conditions on the moment functions, which are closely related to constraint qualifications \citep[as discussed in][]{kai:mol:sto19CQ} under which it is possible to set $\tau_n=0$.

\cite{men14} studies estimation of $\idr{\theta}$ when the number of moment inequalities is large relative to sample size (possibly infinite).
He provides a consistency result for criterion-based estimators that use a number of unconditional moment inequalities that grows with sample size.
He also considers estimators based on conditional moment inequalities, and derives the fastest possible rate for estimating $\idr{\theta}$ under smoothness conditions on the conditional moment functions. 
He shows that the rates achieved by the procedures in \cite{arm14b,arm15} are (minimax) optimal, and cannot be improved upon.

\begin{BI}
\cite{man:tam02} extend the notion of extremum estimation from point identified to partially identified models.
They do so by putting forward a generalized criterion function whose zero-level set can be used to define $\idr{\theta}$ in partially identified structural semiparametric models.
It is then natural to define the set valued estimator $\idrn{\theta}$ as the collection of approximate minimizers of the sample analog of this criterion function.
\citeauthor{man:tam02}'s analysis of statistical inference focuses exclusively on providing consistent estimators.
\cite{che:hon:tam07} substantially generalize the analysis of consistency of criterion function-based set estimators.
They provide a comprehensive study of convergence rates in partially identified models.
Their work highlights the challenges a researcher faces in this context, and puts forward possible solutions in the form of assumptions under which specific rates of convergence attain.
\end{BI}

\subsubsection{Support Function Based Estimators}
\cite{ber:mol08} introduce to the econometrics literature inference methods for set valued estimators based on random set theory.
They study the class of models where $\idr{\theta}$ is convex and can be written as the Aumann (or selection) expectation of a properly defined random closed set.\footnote{By Theorem \ref{thr:exp-supp}, the Aumann expectation of a random closed set defined on a nonatomic probability space is convex. In this chapter I am assuming nonatomicity of the probability space. Even if I did not make this assumption, however, when working with a random sample the relevant probability space is the product space with $n\to\infty$, hence nonatomic \citep{art:vit75}. If $\idr{\theta}$ is not convex, \citeauthor{ber:mol08}'s analysis applies to its convex hull.}
They propose to carry out estimation and inference leveraging the representation of convex sets through their \emph{support function} (given in Definition \ref{def:sup-fun}), as it is done in random set theory; see \cite[Chapter 3]{mo1} and \cite[Chapter 4]{mol:mol18}.
Because the support function fully characterizes the boundary of $\idr{\theta}$, it allows for a simple sample analog estimator, and for inference procedures with desirable properties.

An example of a framework where the approach of \citeauthor{ber:mol08} can be applied is that of best linear prediction with interval outcome data in Identification Problem \ref{IP:param_pred_interval}.\footnote{\cite[Supplementary Appendix F]{kai:mol:sto19} establish that if $\ex$ has finite support, $\idr{\theta}$ in Theorem SIR-\ref{SIR:BLP_intervalY} can be written as the collection of $\vartheta\in\Theta$ that satisfy a finite number of moment inequalities, as posited in this section.}
Recall that in that case, the researcher observes random variables $(\yL,\yU,\ex)$ and wishes to learn the best linear predictor of $\ey|\ex$, with $\ey$ unobserved and $\sR(\yL\le\ey\le\yU)=1$. 
For simplicity let $\ex$ be a scalar.
Given a random sample $\{\yLi,\yUi,\ex_i\}_{i=1}^n$ from $\sP$, the researcher can construct a random segment $\eG_i$ for each $i$ and a consistent estimator $\hat{\Sigma}_n$ of the random matrix $\Sigma_\sP$ in \eqref{eq:G_and_Sigma} as
\begin{align*}
  \eG_i=\left\{
    \begin{pmatrix}
      \ey_i\\ \ey_i\ex_i 
    \end{pmatrix}
    :\; \ey_i \in \Sel(\eY_i)\right\}\subset\R^2,
   ~~\text{and}~~  \hat\Sigma_n=
  \begin{pmatrix}
    1 & \overline\ex\\ \overline\ex & \overline{\ex^2}
  \end{pmatrix},
\end{align*}
where $\eY_i=[\yLi,\yUi]$ and $\overline\ex,\overline{\ex^2}$ are the sample means of $\ex_i$ and $\ex^2_i$ respectively.
Because in this problem $\idr{\theta}=\Sigma_\sP^{-1}\E_\sP\eG$ (see Theorem SIR-\ref{SIR:BLP_intervalY} on p.~\pageref{SIR:BLP_intervalY}), a natural sample analog estimator replaces $\Sigma_\sP$ with $\hat{\Sigma}_n$, and $\E_\sP\eG$ with a Minkowski average of $\eG_i$ (see Appendix \ref{app:RCS}, p.~\pageref{def:mink:sum} for a formal definition), yielding
\begin{align}
\idrn{\theta}=\hat\Sigma_n^{-1}\frac{1}{n}\sum_{i=1}^n\eG_i.\label{eq:BLP_estimator}
\end{align}
The support function of $\idrn{\theta}$ is the sample analog of that of $\idr{\theta}$ provided in \eqref{eq:supfun:BLP}:
\begin{align*}
h_{\idrn{\theta}}(u)=\frac{1}{n}\sum_{i=1}^n[(\yLi\one(f(\ex_i,u)<0)+\yUi\one(f(\ex_i,u)\ge 0))f(\ex_i,u)],~~u\in\mathbb{S},
\end{align*}
where $f(\ex_i,u)=[1~\ex_i]\hat\Sigma_n^{-1}u$.
\cite{ber:mol08} use the Law of Large Numbers for random sets reported in Theorem \ref{thr:SLLN-basic} to show that $\idrn{\theta}$ in \eqref{eq:BLP_estimator} is $\sqrt{n}$-consistent under standard conditions on the moments of $(\yLi,\yUi,\ex_i)$.

\cite{bon:mag:mau12} and \cite{cha:che:mol:sch18} significantly expand the applicability of \possessivecite{ber:mol08} estimator.
\citeauthor{bon:mag:mau12} show that it can be used in a large class of partially identified linear models, including ones that allow for the availability of instrumental variables.
\citeauthor{cha:che:mol:sch18} show that it can be used for best linear approximation of any function $f(x)$ that is known to lie within two identified bounding functions. 
The lower and upper functions defining the band are allowed to be any functions, including ones carrying an index, and can be estimated parametrically or nonparametrically. 
The method allows for estimation of the parameters of the best linear approximations to the set identified functions in many of the identification problems described in Section \ref{sec:prob:distr}.
It can also be used to estimate the sharp identification region for the parameters of a binary choice model with interval or discrete regressors under the assumptions of \cite{mag:mau08}, characterized in \eqref{eq:SIR:mag:mau} in Section \ref{subsubsec:man:tam02}. 

\cite{kai:san14} develop a theory of efficiency for estimators of sets $\idr{\theta}$ as in \eqref{eq:sharp_id_for_inference} under the additional requirements that the inequalities $\E_\sP(m_j(\ew,\vartheta))$ are convex in $\vartheta\in\Theta$ and smooth as functionals of the distribution of the data.
Because of the convexity of the moment inequalities, $\idr{\theta}$ is convex and can be represented through its support function.  
Using the classic results in \cite{bic:kla:rit:wel93}, \citeauthor{kai:san14} show that under suitable regularity conditions, the support function admits for $\sqrt{n}$-consistent regular estimation. 
They also show that a simple plug-in estimator based on the support function attains the semiparametric efficiency bound, and the corresponding estimator of $\idr{\theta}$ minimizes a wide class of asymptotic loss functions based on the Hausdorff distance. 
As they establish, this efficiency result applies to the estimators proposed by \cite{ber:mol08}, including that in \eqref{eq:BLP_estimator}, and by \cite{bon:mag:mau12}.

\cite{kai16} further enlarges the applicability of the support function approach by establishing its duality with the criterion function approach, for the case that $\crit_\sP$ is a convex function and $\crit_n$ is a convex function almost surely.
This allows one to use the support function approach also when a representation of $\idr{\theta}$ as the Aumann expectation of a random closed set is not readily available.
\citeauthor{kai16} considers $\idr{\theta}$ and its level set estimator $\idrn{\theta}$ as defined, respectively, in \eqref{eq:define:idr} and \eqref{eq:define:idrn}, with $\Theta$ a convex subset of $\R^d$.
Because $\crit_\sP$ and $\crit_n$ are convex functions, $\idr{\theta}$ and $\idrn{\theta}$ are convex sets.
Under the same assumptions as in \cite{che:hon:tam07}, including the polynomial minorant and the degeneracy conditions, one can set $\tau_n=0$ and have $\dist_H(\idrn{\theta},\idr{\theta})=O_p(a_n^{-1/\gamma})$.
Moreover, due to its convexity, $\idr{\theta}$ is fully characterized by its support function, which in turn can be consistently estimated (at the same rate as $\idr{\theta}$) using sample analogs as $h_{\idrn{\theta}}(u)=\max_{a_n\crit_n(\vartheta)\le 0}u^\top\vartheta$.
The latter can be computed via convex programming.\medskip

\cite{kit:gia18} consider consistent estimation of $\idr{\theta}$ in the context of Bayesian inference.
They focus on partially identified models where $\idr{\theta}$ depends on a ``reduced form" parameter $\phi$ (e.g., a vector of moments of observable random variables).
They recognize that while a prior on $\phi$ can be revised in light of the data, a prior on $\theta$ cannot, due to the lack of point identification.
As such they propose to choose a single prior for the revisable parameters, and a set of priors for the unrevisable ones.
The latter is the collection of priors such that the distribution of $\theta|\phi$ places probability one on $\idr{\theta}$. 
A crucial observation in \citeauthor{kit:gia18} is that once $\phi$ is viewed as a random vector, as in the Bayesian paradigm, under mild regularity conditions $\idr{\theta}$ is a random closed set, and Bayesian inference on it can be carried out using elements of random set theory.
In particular, they show that the set of posterior means of $\theta|\ew$ equals the Aumann expectation of $\idr{\theta}$ (with the underlying probability measure of $\phi|\ew$).
They also show that this Aumann expectation converges in Hausdorff distance to the ``true" identified set if the latter is convex, or otherwise to its convex hull. 
They apply their method to analyze impulse-response in set-identified Structural Vector Autoregressions, where standard Bayesian inference is otherwise sensitive to the choice of an unrevisable prior.\footnote{There is a large literature in macro-econometrics, pioneered by \cite{fau98}, \cite{can:den02}, and \cite{uhl05}, concerned with Bayesian inference with a non-informative prior for non-identified parameters. I refer to \cite[Chapter 13]{kil:lut17} for a thorough review. 
Frequentist inference for impulse response functions in Structural Vector Autoregression models is carried out, e.g., in \cite{gra:moo:sch18} and \cite{gaf:mei:mon18}.
}

\begin{BI}
\cite{ber:mol08} show that elements of random set theory can be employed to obtain inference methods for partially identified models that are easy to implement and have desirable statistical properties.
Whereas they apply their findings to a specific class of models based on the Aumann expectation, the ensuing literature demonstrates that random set methods are widely applicable to obtain estimators of sharp identification regions and establish their consistency.
\end{BI}

\cite{che:lee:ros13} propose an alternative to the notion of consistent estimator.
Rather than asking that $\idrn{\theta}$ satisfies the requirement in Definition \ref{def:consistent_estimator}, they propose the notion of \emph{half-median-unbiased} estimator.
This notion is easiest to explain in the case of interval identified scalar parameters.
Take, e.g., the bound in Theorem SIR-\ref{SIR:prob:E:md} for the conditional expectation of selectively observed data.
Then an estimator of that interval is half-median-unbiased if the estimated upper bound exceeds the true upper bound, and the estimated lower bound falls below the true lower bound, each with probability at least $1/2$ asymptotically.
More generally, one can obtain a half-median-unbiased estimator as
\begin{align}
\idrn{\theta}=\left\{\vartheta\in\Theta:a_n\crit_n(\vartheta)\le c_{1/2}(\vartheta)\right\},\label{eq:idrn:half:med:unb}
\end{align}
where $c_{1/2}(\vartheta)$ is a critical value chosen so that $\idrn{\theta}$ asymptotically contains $\idr{\theta}$ (or any fixed element in $\idr{\theta}$; see the discussion in Section \ref{subsub:set:or:point:inference} below) with at least probability $1/2$.
As discussed in the next section, $c_{1/2}(\vartheta)$ can be further chosen so that this probability is uniform over $\sP\in\cP$.

The requirement of half-median unbiasedness has the virtue that, by construction, an estimator such as \eqref{eq:idrn:half:med:unb} is a subset of a $1-\alpha$ confidence set as defined in \eqref{eq:CS} below for any $\alpha<1/2$, provided $c_{1-\alpha}(\vartheta)$ is chosen using the same criterion for all $\alpha\in(0,1)$.
In contrast, a consistent estimator satisfying the requirement in Definition \ref{def:consistent_estimator} needs not be a subset of a confidence set.
This is because the sequence $\tau_n$ in \eqref{eq:define:idrn} may be larger than the critical value used to obtain the confidence set, see equation \eqref{eq:CS} below, unless regularity conditions such as degeneracy or others allow one to set $\tau_n$ equal to zero.
Moreover, choice of the sequence $\tau_n$ is not data driven, and hence can be viewed as arbitrary.
This raises a concern for the scope of consistent estimation in general settings.

However, reporting a set estimator together with a confidence set is arguably important to shed light on how much of the volume of the confidence set is due to statistical uncertainty and how much is due to a large identified set.
One can do so by either using a half-median unbiased estimator as in \eqref{eq:idrn:half:med:unb}, or the set of minimizers of the criterion function in \eqref{eq:define:idrn} with $\tau_n=0$ (which, as previously discussed, satisfies the inner consistency requirement in \eqref{eq:inner_consistent} under weak conditions, and is Hausdorff consistent in some well behaved cases).

\subsection{Confidence Sets Satisfying Various Coverage Notions}
\label{subsec:CS}
\subsubsection{Coverage of $\idr{\theta}$ vs. Coverage of $\theta$}
\label{subsub:set:or:point:inference}
I first discuss confidence sets $\CS\subset\R^d$ defined as level sets of a criterion function.
To simplify notation, henceforth I assume $a_n=n$.
\begin{align}
\CS=\left\{\vartheta\in\Theta:n\crit_n(\vartheta)\le c_{1-\alpha}(\vartheta)\right\}.\label{eq:CS}
\end{align}
In \eqref{eq:CS}, $c_{1-\alpha}(\vartheta)$ may be constant or vary in $\vartheta\in\Theta$.
It is chosen to that $\CS$ satisfies (asymptotically) a certain coverage property with respect to either $\idr{\theta}$ or each $\vartheta\in\idr{\theta}$.
Correspondingly, different appearances of $c_{1-\alpha}(\vartheta)$ may refer to different critical values associated with different coverage notions.
The challenging theoretical aspect of inference in partial identification is the determination of $c_{1-\alpha}$ and of methods to approximate it.\medskip

A first classification of coverage notions pertains to whether the confidence set should cover $\idr{\theta}$ or each of its elements with a prespecified asymptotic probability.
Early on, within the study of interval-identified parameters, \cite{hor:man98,hor:man00} put forward a confidence interval that expands each of the sample analogs of the extreme points of the population bounds by an amount designed so that the confidence interval asymptotically covers the population bounds with prespecified probability.

\cite{che:hon:tam07} study the general problem of inference for a set $\idr{\theta}$ defined as the zero-level set of a criterion function.
The coverage notion that they propose is \emph{pointwise coverage of the set}, whereby $c_{1-\alpha}$ is chosen so that:
\begin{align}
\liminf_{n\to\infty}\sP(\idr{\theta}\subseteq\CS)\ge 1-\alpha~\text{for all}~\sP\in\cP.\label{eq:CS_coverage:set:pw}
\end{align}
\cite{che:hon:tam07} provide conditions under which $\CS$ satisfies \eqref{eq:CS_coverage:set:pw} with $c_{1-\alpha}$ constant in $\vartheta$, yielding the so called \emph{criterion function approach} to statistical inference in partial identification.
Under the same coverage requirement, \cite{bug10} and \cite{gal:hen13} introduce novel bootstrap methods for inference in moment inequality models.
\cite{hen:mea:que15} propose an inference method for finite games of complete information that exploits the structure of these models.

\cite{ber:mol08} propose a method to test hypotheses and build confidence sets satisfying \eqref{eq:CS_coverage:set:pw} based on random set theory, the so called \emph{support function approach}, which yields simple to compute confidence sets with asymptotic coverage equal to $1-\alpha$ when $\idr{\theta}$ is strictly convex.
The reason for the strict convexity requirement is that in its absence, the support function of $\idr{\theta}$ is not fully differentiable, but only directionally differentiable, complicating inference.
Indeed, \cite{fan:san18} show that standard bootstrap methods are consistent if and only if full differentiability holds, and they provide modified bootstrap methods that remain valid when only directional differentiability holds.
\cite{cha:che:mol:sch18} propose a data jittering method that enforces full differentiability at the price of a small conservative distortion.
\cite{kai:san14} extend the applicability of the support function approach to other moment inequality models and establish efficiency results.
\cite{che:koc:men15} show that an Hausdorff distance-based test statistic can be weighted to enforce either exact or first-order equivariance to transformations of parameters.
\cite{adu:ots16} provide empirical likelihood based inference methods for the support function approach.
The test statistics employed in the criterion function approach and in the support function approach are asymptotically equivalent in specific moment inequality models \citep{ber:mol08,kai16}, but the criterion function approach is more broadly applicable.
\medskip

The field's interest changed to a different notion of coverage when \cite{imb:man04} pointed out that often there is one ``true" data generating $\theta$, even if it is only partially identified.
Hence, they proposed confidence sets that cover each $\vartheta\in\idr{\theta}$ with a prespecified probability.
For pointwise coverage, this leads to choosing $c_{1-\alpha}$ so that:
\begin{align}
\liminf_{n\to\infty}\sP(\vartheta\in\CS)\ge 1-\alpha~\text{for all}~\sP\in\cP~\text{and}~\vartheta\in\idr{\theta}.\label{eq:CS_coverage:point:pw}
\end{align}
If $\idr{\theta}$ is a singleton then \eqref{eq:CS_coverage:set:pw} and \eqref{eq:CS_coverage:point:pw} both coincide with the pointwise coverage requirement employed for point identified parameters.
However, as shown in \cite[Lemma 1]{imb:man04}, if $\idr{\theta}$ contains more than one element, the two notions differ, with confidence sets satisfying \eqref{eq:CS_coverage:point:pw} being weakly smaller than ones satisfying \eqref{eq:CS_coverage:set:pw}.
\cite{ros08} provides confidence sets for general moment (in)equalities models that satisfy \eqref{eq:CS_coverage:point:pw} and are easy to compute.

Although confidence sets that take each $\vartheta\in\idr{\theta}$ as the object of interest (and which satisfy the \emph{uniform coverage} requirements described in Section \ref{subsub:uniform:inference} below) have received the most attention in the literature on inference in partially identified models, this choice merits some words of caution.
First, \cite{hen:ona12} point out that if confidence sets are to be used for decision making, a policymaker concerned with robust decisions might prefer ones satisfying \eqref{eq:CS_coverage:set:pw} (respectively, \eqref{eq:CS_coverage:set} below once uniformity is taken into account) to ones satisfying \eqref{eq:CS_coverage:point:pw} (respectively, \eqref{eq:CS_coverage:point} below with uniformity).
Second, while in many applications a ``true" data generating $\theta$ exists, in others it does not.
For example, \cite{man:mol10} and \cite{giu:man:mol19} query survey respondents (in the American Life Panel and in the Health and Retirement Study, respectively) about their subjective beliefs on the probability chance of future events.
A large fraction of these respondents, when given the possibility to do so, report imprecise beliefs in the form of intervals. 
In this case, there is no ``true" point-valued belief: the ``truth" is interval-valued.
If one is interested in (say) average beliefs, the sharp identification region is the (Aumann) expectation of the reported intervals, and the appropriate coverage requirement for a confidence set is that in \eqref{eq:CS_coverage:set:pw} (respectively, \eqref{eq:CS_coverage:set} below with uniformity).

\subsubsection{Pointwise vs. Uniform Coverage}
\label{subsub:uniform:inference}
In the context of interval identified parameters, such as, e.g., the mean with missing data in Theorem SIR-\ref{SIR:prob:E:md} with $\theta\in\R$, \cite{imb:man04} pointed out that extra care should be taken in the construction of confidence sets for partially identified parameters, as otherwise they may be asymptotically valid only pointwise (in the distribution of the observed data) over relevant classes of distributions.\footnote{This discussion draws on many conversations with J\"{o}rg Stoye, as well as on notes that he shared with me, for which I thank him.} 
%
For example, consider a confidence interval that expands each of the sample analogs of the extreme points of the population bounds by a one-sided critical value.
This confidence interval controls the asymptotic coverage probability pointwise for any DGP at which the width of the population bounds is positive.
This is because the sampling variation becomes asymptotically negligible relative to the (fixed) width of the bounds, making the inference problem essentially one-sided. 
However, for every $n$ one can find a distribution $\sP\in\cP$ and a parameter $\vartheta\in\idr{\theta}$ such that the width of the population bounds (under $\sP$) is small relative to $n$ and the coverage probability for $\vartheta$ is below $1-\alpha$.
This happens because the proposed confidence interval does not take into account the fact that for some $\sP\in\cP$ the problem has a two-sided nature. 

This observation naturally leads to a more stringent requirement of \emph{uniform coverage}, whereby \eqref{eq:CS_coverage:set:pw}-\eqref{eq:CS_coverage:point:pw} are replaced, respectively, by
\begin{align}
\liminf_{n\to\infty}\inf_{\sP\in\cP}\sP(\idr{\theta}\subseteq\CS)&\ge 1-\alpha,\label{eq:CS_coverage:set}\\
\liminf_{n\to\infty}\inf_{\sP\in\cP}\inf_{\vartheta\in\idr{\theta}}\sP(\vartheta\in\CS)&\ge 1-\alpha,\label{eq:CS_coverage:point}
\end{align}
and $c_{1-\alpha}$ is chosen accordingly, to obtain either \eqref{eq:CS_coverage:set} or \eqref{eq:CS_coverage:point}.
Sets satisfying \eqref{eq:CS_coverage:set} are referred to as confidence regions for $\idr{\theta}$ that are uniformly consistent in level (over $\sP\in\cP$).
\cite{rom:sha10} propose such confidence regions, study their properties, and provide a step-down procedure to obtain them.

\cite{che:chr:tam18} propose confidence sets that are contour sets of criterion functions using cutoffs that are computed via Monte Carlo simulations from the quasi‐posterior distribution of the criterion and satisfy the coverage requirement in \eqref{eq:CS_coverage:set}.
They recommend the use of a Sequential Monte Carlo algorithm that works well also when the quasi-posterior is irregular and multi-modal.
They establish exact asymptotic coverage, non-trivial local power, and validity of their procedure in point identified and partially identified regular models, and validity in irregular models (e.g., in models where the reduced form parameters are on the boundary of the parameter space).
They also establish efficiency of their procedure in regular models that happen to be point identified.

Sets satisfying \eqref{eq:CS_coverage:point} are referred to as confidence regions for points in $\idr{\theta}$ that are uniformly consistent in level (over $\sP\in\cP$).
Within the framework of \cite{imb:man04}, \cite{sto09} shows that one can obtain a confidence interval satisfying \eqref{eq:CS_coverage:point} by pre-testing whether the lower and upper population bounds are sufficiently close to each other.
If so, the confidence interval expands each of the sample analogs of the extreme points of the population bounds by a two-sided critical value; otherwise, by a one-sided.
\citeauthor{sto09} provides important insights clarifying the connection between superefficient (i.e., faster than $O_p(1/\sqrt{n})$) estimation of the width of the population bounds when it equals zero, and certain challenges in \citeauthor{imb:man04}'s proposed method.\footnote{Indeed, the confidence interval proposed by \cite{sto09} can be thought of as using a Hodges-type shrinkage estimator \citep[see, e.g.,][]{van97} for the width of the population bounds.}
\cite{bon:mag:mau12} leverage \cite{sto09}'s results to obtain confidence sets satisfying \eqref{eq:CS_coverage:point} using the support function approach for set identified linear models.

Obtaining confidence sets that satisfy the requirement in \eqref{eq:CS_coverage:point} becomes substantially more complex in the context of general moment (in)equalities models.
One of the key challenges to uniform inference stems from the fact that the behavior of the limit distribution of the test statistic depends on $\sqrt{n}\E_\sP(m_j(\ew_i;\vartheta)),~j=1,\dots,|\cJ|$, which cannot be consistently estimated.
\cite{rom:sha08,and:gug09b,and:soa10,can10,and:bar12,rom:sha:wol14}, among others, make significant contributions to circumvent these difficulties in the context of a finite number of unconditional moment (in)equalities.
\cite{and:shi13,che:lee:ros13,lee:son:wha13,arm14b,arm15,arm:cha16,che18}, among others, make significant contributions to circumvent these difficulties in the context of a finite number of conditional moment (in)equalities (with continuously distributed conditioning variables).
\cite{che:che:kat18} and \cite{and:shi17} study, respectively, the challenging frameworks where the number of moment inequalities grows with sample size and where there is a continuum of conditional moment inequalities.

I refer to \cite[Section 4]{can:sha17} for a thorough discussion of these methods and a comparison of their relative (de)merits \citep[see also][]{bug:can:gug12,bug16}.

\subsubsection{Coverage of the Vector $\theta$ vs. Coverage of a Component of $\theta$}
\label{subsubsec:proj:inference}
The coverage requirements in \eqref{eq:CS_coverage:set}-\eqref{eq:CS_coverage:point} refer to confidence sets in $\R^d$ for the entire $\theta$ or $\idr{\theta}$.
Often empirical researchers are interested in inference on a specific component or (smooth) function of $\theta$ (e.g., the returns to education; the effect of market size on the probability of entry; the elasticity of demand for insurance to price, etc.).
For simplicity, here I focus on the case of a component of $\theta$, which I represent as $u^\top\theta$, with $u$ a standard basis vector in $\R^d$.
In this case, the (sharp) identification region of interest is
\begin{align*}
\idr{u^\top\theta}=\{s\in[-h_\Theta(-u),h_\Theta(u)]:s=u^\top\vartheta~\text{and}~\vartheta\in\idr{\theta}\}.
\end{align*}
One could report as confidence interval for $u^\top\theta$ the projection of $\CS$ in direction $\pm u$. 
The resulting confidence interval is asymptotically valid but typically conservative. 
The extent of the conservatism increases with the dimension of $\theta$ and is easily appreciated in the case of a point identified parameter.
Consider, for example, a linear regression in $\R^{10}$, and suppose for simplicity that the limiting covariance matrix of the estimator is the identity matrix. 
Then a 95\% confidence interval for $u^\top\theta$ is obtained by adding and subtracting $1.96$ to that component's estimate. 
In contrast, projection of a 95\% confidence ellipsoid for $\theta$ on each component amounts to adding and subtracting $4.28$ to that component's estimate. 

It is therefore desirable to provide confidence intervals $\CI$ specifically designed to cover $u^\top\theta$ rather then the entire $\theta$.
Natural counterparts to \eqref{eq:CS_coverage:set}-\eqref{eq:CS_coverage:point} are
\begin{align}
	\liminf_{n\to\infty}\inf_{\sP\in\cP}\sP(\idr{u^\top\theta} \subseteq \CI)&\ge 1-\alpha,\label{eq:CS_coverage:set:proj}\\
	\liminf_{n\to\infty}\inf_{\sP\in\cP}\inf_{\vartheta\in\idr{\theta}}\sP(u^\top\vartheta\in \CI)&\ge 1-\alpha. \label{eq:CS_coverage:point:proj}
\end{align}
As shown in \cite{ber:mol08} and \cite{kai16} for the case of pointwise coverage, obtaining asymptotically valid confidence intervals is simple if the identified set is convex and one uses the support function approach.
This is because it suffices to base the test statistic on the support function in direction $u$, and it is often possible to easily characterize the limiting distribution of this test statistic.
See \cite[Chapters 4 and 5]{mol:mol18} for details.

The task is significantly more complex in general moment inequality models when $\idr{\theta}$ is non-convex and one wants to satisfy the criterion in \eqref{eq:CS_coverage:set:proj} or that in \eqref{eq:CS_coverage:point:proj}.
\cite{rom:sha08} and \cite{bug:can:shi17} propose confidence intervals of the form
\begin{align}
\CI = \left\{s\in[-h_\Theta(-u),h_\Theta(u)]:\inf_{\vartheta\in\Theta(s)}n\crit_n(\vartheta)\le c_{1-\alpha}(s)\right\},\label{eq:CI:BCS}
\end{align} 
where $\Theta(s)=\{\vartheta\in\Theta:u^\top\vartheta=s\}$ and $c_{1-\alpha}$ is such that \eqref{eq:CS_coverage:point:proj} holds.
An important idea in this proposal is that of \emph{profiling} the test statistic $n\crit_n(\vartheta)$ by minimizing it over all $\vartheta$s such that $u^\top\vartheta=s$.
One then includes in the confidence interval all values $s$ for which the profiled test statistic's value is not too large.
\cite{rom:sha08} propose the use of subsampling to obtain the critical value $c_{1-\alpha}(s)$ and provide high-level conditions ensuring that \eqref{eq:CS_coverage:point:proj} holds.
\cite{bug:can:shi17} substantially extend and improve the \emph{profiling approach} by providing a bootstrap-based method to obtain $c_{1-\alpha}$ so that \eqref{eq:CS_coverage:point:proj} holds.
Their method is more powerful than subsampling (for reasonable choices of subsample size).
\cite{bel:bug:che18} further enlarge the domain of applicability of the profiling approach by proposing a method based on this approach that is asymptotically uniformly valid when the number of moment conditions is large, and can grow with the sample size, possibly at exponential rates.

\cite{kai:mol:sto19} propose a bootstrap-based \emph{calibrated projection approach} where
\begin{align}
\CI= [-h_{\eC_n(c_{1-\alpha})}(-u),h_{\eC_n(c_{1-\alpha})}(u)],\label{eq:def:CI}
\end{align}
with
\begin{align}
h_{\eC_n(c_{1-\alpha})}(u)\equiv\sup_{\vartheta\in\Theta}~u^\top\vartheta~\text{s.t.}~\frac{\sqrt{n}\bar{m}_{n,j}(\vartheta)}{\hat{\sigma}_{n,j}(\vartheta)}\leq c_{1-\alpha}(\vartheta),~j=1,\dots,|\cJ|\label{eq:KMS:proj}
\end{align}
and $c_{1-\alpha}$ a critical level function calibrated so that \eqref{eq:CS_coverage:point:proj} holds.
Compared to the simple projection of $\CS$ mentioned at the beginning of Section \ref{subsubsec:proj:inference}, calibrated projection (weakly) reduces the value of $c_{1-\alpha}$ so that the projection of $\theta$, rather than $\theta$ itself, is asymptotically covered with the desired probability uniformly.

\cite{che:chr:tam18} provide methods to build confidence intervals and confidence sets on projections of $\idr{\theta}$ as contour sets of criterion functions using cutoffs that are computed via Monte Carlo simulations from the quasi‐posterior distribution of the criterion, and that satisfy the coverage requirement in \eqref{eq:CS_coverage:set:proj}.
One of their procedures, designed specifically for scalar projections, delivers a confidence interval as the contour set of a profiled quasi-likelihood ratio with critical value equal to a quantile of the Chi-squared distribution with one degree of freedom.

\subsubsection{A Brief Note on Bayesian Methods}
The confidence sets discussed in this section are based on the frequentist approach to inference.
It is natural to ask whether in partially identified models, as in well behaved point identified models, one can build Bayesian credible sets that at least asymptotically coincide with frequentist confidence sets.
This question was first addressed by \cite{moo:sch12}, with a negative answer for the case that the coverage in \eqref{eq:CS_coverage:point} is sought out.
In particular, they showed that the resulting Bayesian credible sets are a subset of $\idr{\theta}$, and hence too narrow from the frequentist perspective.

This discrepancy can be ameliorated when inference is sought out for $\idr{\theta}$ rather than for each $\vartheta\in\idr{\theta}$.
\cite{nor:tan14}, \cite{kli:tam16}, \cite{kit:gia18}, and \cite{lia:sim19} propose Bayesian credible regions that are valid for frequentist inference in the sense of \eqref{eq:CS_coverage:set:pw}, where the first two build on the criterion function approach and the second two on the support function approach.
All these contributions rely on the model being separable, in the sense that it yields moment inequalities that can be written as the sum of a function of the data only, and a function of the model parameters only (as in, e.g., \eqref{eq:CT_00}-\eqref{eq:CT_01L}).
In these models, the function of the data only (the \emph{reduced form parameter}) is point identified, it is related to the structural parameters $\theta$ through a known mapping, and under standard regularity conditions it can be $\sqrt{n}$-consistently estimated.
The resulting estimator has an asymptotically Normal distribution.
The various approaches place a prior on the reduced form parameter, and standard tools in Bayesian analysis are used to obtain a posterior.
The known mapping from reduced form to structural parameters is then applied to this posterior to obtain a credible set for $\idr{\theta}$.

\section{Misspecification in Partially Identified Models}
\label{sec:misspec}


Although partial identification often results from reducing the number of assumptions maintained in counterpart point identified models, care still needs to be taken in assessing the possible consequences of misspecification.
This section's goal is to discuss the existing literature on the topic, and to provide some additional observations.
To keep the notation light, I refer to the functional of interest as $\theta$ throughout, without explicitly distinguishing whether it belongs to an infinite dimensional parameter space (as in the nonparametric analysis in Section \ref{sec:prob:distr}), or to a finite dimensional one (as in the semiparametric analysis in Section \ref{sec:structural}).

The original nonparametric ``worst-case" bounds proposed by \cite{man89} for the analysis of selectively observed data and discussed in Section \ref{sec:prob:distr} are not subject to the risk of misspecification, because they are based on the empirical evidence alone.
However, often researchers are willing and eager to maintain additional assumptions that can help shrink the bounds, so that one can learn more from the available data.
Indeed, early on \cite{man90} proposed the use of exclusion restrictions in the form of mean independence assumptions.
Section \ref{subsec:programme:eval} discusses related ideas within the context of nonparametric bounds on treatment effects, and \cite[Chapter 2]{man03} provides a thorough treatment of other types of exclusion restriction.
The literature reviewed throughout this chapter provides many more examples of assumptions that have proven useful for empirical research.

Broadly speaking, assumptions can be classified in two types \citep[Chapter 2]{man03}.
The first type is \emph{non-refutable}: it may reduce the size of $\idr{\theta}$, but cannot lead to it being empty.
An example in the context of selectively observed data is that of exogenous selection, or data missing at random conditional on covariates and instruments (see Section \ref{subsec:missing_data}, p.~\pageref{subsec:missing_data}): under this assumption $\idr{\theta}$ is a singleton, but the assumption cannot be refuted because it poses a distributional (independence) assumption on unobservables.

The second type is \emph{refutable}: it may reduce the size of $\idr{\theta}$, and it may result in $\idr{\theta}=\emptyset$ if it does not hold in the DGP.
An example in the context of treatment effects is the assumption of mean independence between response function at treatment $t$ and instrumental variable $\ez$, see \eqref{eq:ass:MI} in Section \ref{subsec:programme:eval}.
There the sharp bounds on $\E_\sQ(\ey(t)|\ex=x)$ are intersection bounds as in \eqref{eq:intersection:bounds}.
If the instrument is invalid, the bounds can be empty.

\cite{pon:tam11} consider the impact of misspecification on semiparametric partially identified models.
One of their examples concerns a linear regression model of the form $\E_\sQ(\ey|\ex)=\theta^\top\ex$ when only interval data is available for $\ey$ (as in Section \ref{subsec:interval_data}).
In this context, $\idr{\theta}=\{\vartheta\in\Theta:\E_\sP(\yL|\ex)\le \vartheta^\top\ex \le\E_\sP(\yU|\ex),~\ex\text{-a.s.}\}$.
The concern is that the conditional expectation might not be linear.
\citeauthor{pon:tam11} make two important observations.
First, they argue that the set $\idr{\theta}$ is of difficult interpretation when the model is misspecified.
When $\ey$ is perfectly observed, if the conditional expectation is not linear, the output of ordinary least squares can be readily interpreted as the best linear approximation to $\E_\sQ(\ey|\ex)$.
This is not the case for $\idr{\theta}$ when only the interval data $[\yL,\yU]$ is observed.
They therefore propose to work with the set of best linear predictors for $\ey|\ex$ even in the partially identified case (rather than fully exploit the linearity assumption).
The resulting set is the one derived by \cite{ber:mol08} and reported in Theorem SIR-\ref{SIR:BLP_intervalY}.
\citeauthor{pon:tam11} work with projections of this set, which coincide with the bounds in \cite{sto07}.

\citeauthor{pon:tam11} also point out that depending on the DGP, misspecification can cause $\idr{\theta}$ to be spuriously tight.
This can happen, for example, if $\E_\sP(\yL|\ex)$ and $\E_\sP(\yU|\ex)$ are sufficiently nonlinear, even if they are relatively far from each other \citep[e.g.,][Figure 1]{pon:tam11}.
Hence, caution should be taken when interpreting very tight partial identification results as indicative of a highly informative model and empirical evidence, as the possibility of model misspecification has to be taken into account.
These observations naturally lead to the questions of how to test for model misspecification in the presence of partial identification, and of what are the consequences of misspecification for the confidence sets discussed in Section \ref{subsec:CS}.

With partial identification, a null hypothesis of correct model specification (and its alternative) can be expressed as
\begin{align*}
H_0:\idr{\theta}\neq\emptyset;\quad H_1:\idr{\theta}=\emptyset.
\end{align*}
Tests for this hypothesis have been proposed both for the case of nonparametric as well as semiparametric partially identified models.
I refer to \cite{san12} for specification tests in a partially identified nonparametric instrumental variable model; to \cite{kit:sto18} for a nonparametric test in random utility models that checks whether a repeated cross section of demand data might have been generated by a population of rational consumers (thereby testing for the Axiom of Revealed Stochastic Preference); and to \cite{gug:hah:kim08} and \cite{bon:mag:mau12} for specification tests in linear moment (in)equality models. 

For the general class of moment inequality models discussed in Section \ref{sec:inference}, \cite{rom:sha08}, \cite{and:gug09b}, \cite{gal:hen09}, and \cite{and:soa10} propose a specification test that rejects the model if $\CS$ in \eqref{eq:CS} is empty, where $\CS$ is defined with $c_{1-\alpha}(\vartheta)$ determined so as to satisfy \eqref{eq:CS_coverage:point} and approximated according to the methods proposed in the respective papers.
The resulting test, commonly referred to as \emph{by-product} test because obtained as a by-product to the construction of a confidence set, takes the form
\begin{align*}
\phi=\one(\CS=\emptyset)=\one\left(\inf_{\vartheta\in\Theta}n\crit_n(\vartheta)> c_{1-\alpha}(\vartheta)\right).
\end{align*}
Denoting by $\cP_0$ the collection of $\sP\in\cP$ such that $\idr{\theta}\neq\emptyset$, one has that the by-product test achieves uniform size control \citep[Theorem C.2]{bug:can:shi15}:
\begin{align}
\limsup_{n\to\infty}\sup_{\sP\in\cP_0}\E_\sP(\phi)\le\alpha.\label{eq:misp:test:uniform:size}
\end{align}

An important feature of the by-product test is that the critical value $c_{1-\alpha}(\vartheta)$ is not obtained to test for model misspecification, but it is obtained to insure the coverage requirement in \eqref{eq:CS_coverage:point}; hence, it is obtained by working with the asymptotic distribution of $n\crit_n(\vartheta)$.
\cite{bug:can:shi15} propose more powerful model specification tests, using a critical value $c_{1-\alpha}$ that they obtain to ensure that \eqref{eq:misp:test:uniform:size}, rather than \eqref{eq:CS_coverage:point}, holds.
In particular, they show that their tests dominate the by-product test in terms of power in any finite sample and in the asymptotic limit.
Their critical value is obtained by working with the asymptotic distribution of $\inf_{\vartheta\in\Theta}n\crit_n(\vartheta)$.
As such, their proposal resembles the classic approach to model specification testing ($J$-test) in point identified generalized method of moments models.\medskip

While it is possible to test for misspecification also in partially identified models, a word of caution is due on what might be the effects of misspecification on confidence sets constructed as in \eqref{eq:CS} with $c_{1-\alpha}$ determined to insure \eqref{eq:CS_coverage:point}, as it is often done in empirical work.
\cite{bug:can:gug12} show that in the presence of local misspecification, confidence sets $\CS$ designed to satisfy \eqref{eq:CS_coverage:point} fail to do so.
In practice, the concern is that when the model is misspecified $\CS$ might be spuriously small.
Indeed, we have seen that it can be empty if the misspecification is sufficiently severe.
If it is less severe but still present, it may lead to inference that is erroneously interpreted as precise.

It is natural to wonder how this compares to the effect of misspecification on inference in point identified models.\footnote{The considerations that I report here are based on conversations with Joachim Freyberger and notes that he shared with me, for which I thank him.}
In that case, the rich set of tools available for inference allows one to avoid this problem.
Consider for example a point identified generalized method of moments model with moment conditions $\E_\sP(m_j(\ew;\theta))=0$, $j=1,\dots,|\cJ|$, and $|\cJ|>d$.
Let $m$ denote the vector that stacks each of the $m_j$ functions, and let the estimator of $\theta$ be
\begin{align}
\hat{\theta}_n=\argmin_{\vartheta\in\Theta}n\bar{m}_n(\vartheta)^\top\hat\Xi^{-1}\bar{m}_n(\vartheta),\label{eq:GMM:estimator}
\end{align}
with $\hat\Xi$ a consistent estimator of $\Xi=\E_\sP[m(\ew;\theta) m(\ew;\theta)^\top]$ and $\bar{m}_n(\vartheta)$ the sample analog of $\E_\sP(m(\ew;\vartheta))$.
As shown by \cite{han82} for correctly specified models, the distribution of $\sqrt{n}(\hat{\theta}_n-\theta)$ converges to a Normal with mean vector equal to zero and covariance matrix $\Sigma$.
\cite{hal:ino03} show that when the model is subject to non-local misspecification, $\sqrt{n}(\hat{\theta}_n-\theta_*)$ converges to a Normal with mean vector equal to zero and covariance matrix $\Sigma_*$, where $\theta_*$ is the pseudo-true vector (the probability limit of \eqref{eq:GMM:estimator}) and where $\Sigma_*$ equals $\Sigma$ if the model is correctly specified, and differs from it otherwise.
Let $\hat{\Sigma}_*$ be a consistent estimator of $\Sigma_*$ as in \cite{hal:ino03}.
Define the Wald-statistic based confidence ellipsoid
\begin{align}
\{\vartheta\in\Theta:n(\hat{\theta}_n-\vartheta)^\top\hat{\Sigma}_*^{-1}(\hat{\theta}_n-\vartheta)\le c_{d,1-\alpha}\},\label{eq:CS:Wald:point:id}
\end{align}
with $c_{d,1-\alpha}$ the $1-\alpha$ critical value of a $\chi_d^2$ (chi-squared random variable with $d$ degrees of freedom).
Under standard regularity conditions \citep[see][]{hal:ino03} the confidence set in \eqref{eq:CS:Wald:point:id} covers with asymptotic probability $1-\alpha$ the true vector $\theta$ if the model is correctly specified, and the pseudo-true vector $\theta_*$ if the model is incorrectly specified.
In either case, \eqref{eq:CS:Wald:point:id} is never empty and its volume depends on $\hat{\Sigma}_*$.\footnote{The effect of misspecification for maximum likelihood, least squares, and GMM estimators in ``point identified" models (by which I mean models where the population criterion function has a unique optimizer) has been studied in the literature; see, e.g., \cite{whi82}, \cite{gal:whi88}, \cite{hal:ino03}, \cite{han:lee19}, and references therein. These estimators have been shown to converge in probability to pseudo-true values, and it has been established that tests of hypotheses and confidence sets based on these estimators have correct asymptotic level with respect to the pseudo-true parameters, provided standard errors are computed appropriately. In the specific case of GMM discussed here, the pseudo-true value $\theta_*$ depends on the choice of weighting matrix in \eqref{eq:GMM:estimator}: I have used $\hat\Xi$, but other choices are possible. I do not discuss this aspect of the problem here, but refer to \cite{hal:ino03}.}

Even in the point identified case a confidence set constructed similarly to \eqref{eq:CS}, i.e.,
\begin{align}
\{\vartheta\in\Theta:n\bar{m}_n(\vartheta)\hat\Xi^{-1}\bar{m}_n(\vartheta)\le c_{|\cJ|,1-\alpha}\},\label{eq:CS:AR:point:id}
\end{align}
where $c_{|\cJ|,1-\alpha}$ is the $1-\alpha$ critical value of a $\chi^2_{|\cJ|}$,
incurs the same problems as its partial identification counterpart.
Under standard regularity conditions, if the model is correctly specified, the confidence set in \eqref{eq:CS:AR:point:id} covers $\theta$ with asymptotic probability $1-\alpha$, because $n\bar{m}_n(\vartheta)\hat\Xi^{-1}\bar{m}_n(\vartheta)\Rightarrow\chi^2_{|\cJ|}$.
However, this confidence set is empty with asymptotic probability $\P(\chi^2_{|\cJ|-d}>c_{|\cJ|,1-\alpha})$, due to the facts that $\P(\CS=\emptyset)=\P(\hat{\theta}_n\notin\CS)$ and that $n\bar{m}_n(\hat{\theta}_n)\hat\Xi^{-1}\bar{m}_n(\hat{\theta}_n)\Rightarrow\chi^2_{|\cJ|-d}$.
Hence, it can be arbitrarily small.

In the very special case of a linear regression model with interval outcome data studied by \cite{pon:tam11}, the procedure proposed by \cite{ber:mol08} yields confidence sets that are always non-empty and whose volume depends on a covariance function that they derive \citep[see][Theorem 4.3]{ber:mol08}.
If the linear regression model is correctly specified, and hence $\{\vartheta\in\Theta:\E_\sP(\yL|\ex)\le \vartheta^\top\ex \le\E_\sP(\yU|\ex),~\ex\text{-a.s.}\}\neq\emptyset$, these confidence sets cover $\{\vartheta\in\Theta:\E_\sP(\yL|\ex)\le \vartheta^\top\ex \le\E_\sP(\yU|\ex),~\ex\text{-a.s.}\}$ with asymptotic probability at least equal to $1-\alpha$, as in \eqref{eq:CS_coverage:set:pw}.
Even if the model is misspecified and $\{\vartheta\in\Theta:\E_\sP(\yL|\ex)\le \vartheta^\top\ex \le\E_\sP(\yU|\ex),~\ex\text{-a.s.}\}=\emptyset$, the confidence sets cover the sharp identification region for the parameters of the best linear predictor of $\ey|\ex$, which can be viewed as a pseudo-true set, with probability exactly equal to $1-\alpha$. 
The test statistic that \citeauthor{ber:mol08} use is based on the Hausdorff distance between the estimator and the hypothesized set, and as such is a generalization of the standard Wald-statistic to the set-valued case.
These considerations can be extended to other models.
For example, \cite{lee:bha19} study empirical measurement of Hicksian consumer welfare with interval data on income.
When the model is misspecified, they provide a best parametric approximation to demand and welfare based on the support function method, and inference procedures for this approximation.
For other moment inequality models, \cite{kai:whi13} propose to build a pseudo-true set $\mathcal{H}_\sP^*[\theta]$ that is obtained through a two-step procedure.
In the first step one obtains a nonparametric estimator of the function(s) for which the researcher wants to impose a parametric structure.
In the second step one obtains the set $\mathcal{H}_\sP^*[\theta]$ as the collection of least squares projections of the set in the first step, on the parametric class imposed.
\citeauthor{kai:whi13} show that under regularity conditions the pseudo-true set can be consistently estimated, and derive rates of convergence for the estimator; however, they do not provide methods to obtain confidence sets.
While conceptually valuable, their construction appears to be computationally difficult.
\cite{mas:poi18} propose that when a model is falsified (in the sense that $\idr{\theta}$ is empty) one should report the \emph{falsification frontier}: the boundary between the set of assumptions which falsify the model and those which do not, obtained through continuous relaxations of the baseline assumptions of concern. 
The researcher can then present the set $\idr{\theta}$ that results if the true model lies somewhere on this frontier. 
This set can be interpreted as a pseudo-true set.
However, \citeauthor{mas:poi18} do not provide methods for inference.

The implications of misspecification in partially identified models remain an open and important question in the literature.
For example, it would be useful to have notions of pseudo-true set that parallel those of pseudo-true value in the point identified case.
It would also be important to provide methods for the construction of confidence sets in general moment inequality models that do not exhibit spurious precision (i.e., are arbitrarily small) when the model is misspecified.
Recent work by \cite{and:kwo19} addresses some of these questions.

\section{Computational Challenges}
\label{sec:computations}
As a rule of thumb, the difficulty in computing estimators of identification regions and confidence sets depends on whether a closed form expression is available for the boundary of the set.
For example, often nonparametric bounds on functionals of a partially identified distribution are known functionals of observed conditional distributions, as in Section \ref{sec:prob:distr}.
Then ``plug in" estimation is possible, and the computational cost is the same as for estimation and construction of confidence intervals (or confidence bands) for point-identified nonparametric regressions (incurred twice, once for the lower bound and once for the upper bound).

Similarly, support function based inference is easy to implement when $\idr{\theta}$ is convex.
Sometimes the extreme points of $\idr{\theta}$ can be expressed as known functionals of observed distributions.
Even if not, level sets of convex functions are easy to compute.

But as it was shown in Section \ref{sec:structural}, many problems of interest yield a set $\idr{\theta}$ that is \emph{not} convex.
In this case, $\idr{\theta}$ is obtained as a level set of a criterion function.
Because $\idr{\theta}$ (or its associated confidence set) is often a subset of $\R^d$ (rather than $\R$), even a moderate value for $d$, e.g., 8 or 10, can lead to extremely challenging computational problems.
This is because if one wants to compute $\idr{\theta}$ or a set that covers it or its elements with a prespecified asymptotic probability (possibly uniformly over $\sP\in\cP$), one has to map out a level set in $\R^d$.
If one is interested in confidence intervals for scalar projections or other smooth functions of $\vartheta\in\idr{\theta}$, one needs to solve complex nonlinear optimization problems, as for example in \eqref{eq:CI:BCS} and \eqref{eq:KMS:proj}.
This can be difficult to do, especially because $c_{1-\alpha}(\vartheta)$ is typically an unknown function of $\vartheta$ for which gradients are not available in closed form.

Mirroring the fact that computation is easier when the boundary of $\idr{\theta}$ is a known function of observed conditional distributions, several portable software packages are available to carry out estimation and inference in this case.
For example, \cite{ber:man00} provide STATA and MatLab packages implementing the methods proposed by \cite{man89,man90,man94,man95,man97:monotone}, \cite{hor:man98,hor:man00}, and \cite{man:pep00}.
\cite{tau14} provides a STATA package to implement the bounds proposed by \cite{lee09}.
\cite{mcc:mil:roy15} provide a STATA package to implement bounds on treatment effects with endogenous and misreported treatment assignment and under the  assumptions of monotone treatment selection, monotone treatment response, and monotone instrumental variables as in \cite{man97:monotone}, \cite{man:pep00}, \cite{kre:pep07}, \cite{gun:kre:pep12}, and  \cite{kre:pep:gun:jol12}.
The code computes the confidence intervals proposed by \cite{imb:man04}.
In the more general context of inference for a one-dimensional parameter defined by intersection bounds, as for example the one in \eqref{eq:intersection:bounds}, \cite{che:kim:lee:ros15} and \cite{and:kim:shi17} provide portable STATA code implementing, respectively, methods to test hypotheses and build confidence intervals in \cite{che:lee:ros13} and in \cite{and:shi13}.

\cite{ber:mol:mor10} provide portable STATA code implementing \cite{ber:mol08}'s method for estimation and inference for best linear prediction with interval outcome data as in Identification Problem \ref{IP:param_pred_interval}.
\cite{cha:che:mol:sch12_code} provide R code implementing \cite{cha:che:mol:sch18}'s method for estimation and inference for best linear approximations of set identified functions.\medskip

On the other hand, there is a paucity of portable software implementing the theoretical methods for inference in structural partially identified models discussed in Section \ref{sec:inference}.
\cite{cil:tam09} compute \cite{che:hon:tam07} confidence sets for a parameter vector in $\R^d$ in an entry game with six players, with $d$ in the order of $20$ and with tens of thousands of inequalities, through a ``guess and verify" algorithm based on simulated annealing (with no cooling) that visits many candidate values $\vartheta\in\Theta$, evaluates $\crit_n(\vartheta)$, and builds $\CS$ by retaining the visited values $\vartheta$ that satisfy $n\crit_n(\vartheta)\le c_{1-\alpha}(\vartheta)$ with $c_{1-\alpha}$ defined to satisfy \eqref{eq:CS_coverage:point:pw}.
Given the computational resources commonly available at this point in time, this is a tremendously hard task, due to the dimension of $\theta$ and the number of moment inequalities employed.

As explained in Section \ref{subsubsec:tam03:cil:tam09}, these inequalities, which in a game of entry with $J$ players and discrete observable payoff shifters are $2^J|\cX|$ (with $\cX$ the support of the observable payoff shifters), yield an outer region $\outr{\theta}$.
It is natural to wonder what are the additional challenges faced to compute $\idr{\theta}$ as described in Section \ref{subsubsec:sharp:games}.
A definitive answer to this question is hard to obtain.
If one employs \emph{all} inequalities listed in Theorem \ref{thr:artstein}, the number of inequalities jumps to $(2^{2^J}-2)|\cX|$, increasing the computational cost.
However, as suggested by \cite{gal:hen06} and extended by other authors \citep[e.g.,][]{ber:mol:mol08,ber:mol:mol11,che:ros:smo13,che:ros17}, often many moment inequalities are redundant, substantially reducing the number of inequalities to be checked.
Specifically, \cite{gal:hen06} propose the notion of \emph{core determining sets}, a collection of compact sets such that if the inequality in Theorem \ref{thr:artstein} holds for these sets, it holds for all sets in $\cK$, see Definition \ref{def:core-det} and the surrounding discussion in Appendix \ref{app:RCS}.
This often yields a number of restrictions similar to the one incurred to obtain outer regions.
For example, \cite[Section 4.2]{ber:mol:mol08} analyze a four player, two type entry game with pure strategy Nash equilibrium as solution concept, originally proposed by \cite{ber:tam06}, and show that while a direct application of Theorem \ref{thr:artstein} entails $512|\cX|$ inequality restrictions, $26|\cX|$ suffice.
In this example, \cite{cil:tam09}'s outer region is based on checking $18|\cX|$ inequalities.

A related but separate question is how to \emph{best} allocate the computational effort.
As one moves from partial identification analysis to finite sample considerations, one may face a trade-off between sharpness of the identification region and statistical efficiency.
This is because inequalities that are redundant from the perspective of identification analysis might nonetheless be estimated with high precision, and hence improve the finite sample statistical properties of a confidence set or of a test of hypothesis.
Recent contributions by \cite{and:shi17}, \cite{che:che:kat18} and \cite{bel:bug:che18}, provide methods to build confidence set, respectively, with a continuum of conditional moment inequalities, and with a number of moment inequalities that may exceed sample size.
These contributions, however, do not yet answer the question of how to optimally select inequalities to yield confidence sets with best finite sample properties according to some specified notion of ``best".

A different approach proposed by \cite{che:chr:tam18} uses directly a quasi-likelihood criterion function. 
In the context, e.g., of entry games, this entails assuming that the selection mechanism depends only on observable payoff shifters, using it to obtain the exact model implied distribution as in \eqref{eq:games_model:pred}, and partially identifying an enlarged parameter vector that includes $\theta$ and the selection mechanism.
In an empirical application with discrete covariates, \cite{che:chr:tam18} apply their method to a two player entry game with correlated errors, where $\theta\in\R^9$ and the selection mechanism is a vector in $\R^8$, for a total of 17 parameters. 
In another application to the analysis of trade flows, their empirical application includes 46 parameters.
\medskip

In terms of general purpose portable code that can be employed in moment inequality models, I am only aware of the MatLab package provided by \cite{kai:mol:sto:thi17} to implement the inference method of \cite{kai:mol:sto19} for projections and smooth functions of parameter vectors in models defined by a finite number of unconditional moment (in)equalities.
More broadly, their method can be used to compute confidence intervals for optimal values of optimization problems with estimated constraints.
Here I summarize their approach to further highlight why the computational task is challenging even in the case of projections.

The confidence interval in \eqref{eq:def:CI}-\eqref{eq:KMS:proj} requires solving two nonlinear programs, each with a linear objective and nonlinear constraints involving a critical value which in general is an unknown function of $\vartheta$, with unknown gradient.
When the dimension of the parameter vector is large, directly solving optimization problems with such constraints can be expensive even if evaluating the critical value at each $\vartheta$ is cheap.\footnote{\cite{kai:mol:sto19} propose a linearization method whereby $c_{1-\alpha}$ is calibrated through repeatedly solving bootstrap linear programs, hence it is reasonably cheap to compute.}
Hence, \citeauthor{kai:mol:sto19} propose to use an algorithm (called E-A-M for Evaluation-Approximation-Maximization) to solve these nonlinear programs, which belongs to the family of \emph{expected improvement algorithms} \citep[see e.g.][and references therein]{jon:sch:wel98,sch:wel:jon98,jon01}.
Given a constrained optimization problem of the form
\begin{align*}
\max_{\vartheta \in \Theta}u^\top\vartheta~\text{s.t. }g_j(\vartheta)\le c(\vartheta),j=1,\dots,J,
\end{align*} 
to which \eqref{eq:KMS:proj} belongs,\footnote{To see this it suffices to set $g_j(\vartheta)=\frac{\sqrt{n}\bar{m}_{n,j}(\vartheta)}{\hat{\sigma}_{n,j}(\vartheta)}$ and $c(\vartheta)= c_{1-\alpha}(\vartheta)$.} the algorithm attempts to solve it by cycling over three steps: 
\begin{enumerate}
\item The \emph{true} critical level function $c$ is evaluated at an initial (uniformly randomly drawn from $\Theta$) set of points $\vartheta^1,\dots,\vartheta^k$.
These values are used to compute a current guess for the optimal value, $u^\top\vartheta^{*,k}=\max\{u^\top\vartheta:~\vartheta\in\{\vartheta^1,\dots,\vartheta^k\}\text{ and }\bar g(\vartheta)\le c(\vartheta)\}$, where $\bar g(\vartheta)=\max_{j=1,\dots,J}g_j(\vartheta)$. 
The ``training data" $(\vartheta^{\ell},c(\vartheta^{\ell})_{\ell=1}^k$ is used to compute an \emph{approximating surface} $c_k$ through a Gaussian-process regression model (kriging), as described in \cite[Section 4.1.3]{san:wil:not13}; 
\item For $L\ge k+1$, with probability $1-\epsilon$ the next evaluation point $\theta^L$ for the \emph{true} critical level function $c$ is chosen by finding the point that maximizes \emph{expected improvement} with respect to the \emph{approximating surface}, $\mathbb{EI}_{L-1}(\vartheta)=(u^\top\vartheta-u^\top\vartheta^{*,L-1})_+\{1-\Phi([\bar g(\vartheta)-c_{L-1}(\vartheta)]/[\hat\varsigma s_{L-1}(\vartheta)])\}$.
Here $c_{L-1}(\vartheta)$ and $\hat\varsigma^2 s_{L-1}^2(\vartheta)$ are estimators of the posterior mean and variance of the approximating surface. 
To aim for global search, with probability $\epsilon$, $\vartheta^L$ is drawn uniformly from $\Theta$. 
The approximating surface is then recomputed using $(\vartheta^{\ell},c(\vartheta^{\ell})_{\ell=1}^L)$. 
Steps 1 and 2 are repeated until a convergence criterion is met.
\item The extreme point of $CI_n$ is reported as the value $u^\top\vartheta^{*,L}$ that maximizes $u^\top\vartheta$ among the evaluation points that satisfy the \emph{true} constraints, i.e. $u^\top\vartheta^{*,L}=\max\{u^\top\vartheta:~\vartheta\in\{\vartheta^1,\dots,\vartheta^L\}\text{ and }\bar g(\vartheta)\le c(\vartheta)\}$. 
\end{enumerate}
The only place where the approximating surface is used is in Step 2, to choose a new evaluation point. 
In particular, the reported extreme points of $\CI$ in \eqref{eq:def:CI} are the extreme values of $u^\top\vartheta$ that are consistent with the true surface where this surface was computed, \emph{not}  with the approximating surface. 
\cite{kai:mol:sto19} establish convergence of their algorithm and obtain a convergence rate, as the number of evaluation points increases, for constrained optimization problems in which the constraints are sufficiently smooth ``black box" functions, building on an earlier contribution of \cite{bul11}.
\citeauthor{bul11} establishes convergence of an expected improvement algorithm for unconstrained optimization problems where the objective is a ``black box" function.
The rate of convergence that \citeauthor{bul11} derives depends on the smoothness of the black box objective function. 
The rate of convergence obtained by \citeauthor{kai:mol:sto19} depends on the smoothness of the black box constraints, and is slightly slower than \citeauthor{bul11}'s rate. 
\citeauthor{kai:mol:sto19}'s Monte Carlo experiments suggest that the E-A-M algorithm is fast and accurate at computing their confidence intervals.
The E-A-M algorithm also allows for very rapid computation of projections of the confidence set proposed by \cite{and:soa10}, and for a substantial improvement in the computational time of the profiling-based confidence intervals proposed by \cite{bug:can:shi17}.\footnote{\cite{bug:can:shi17}'s method does not require solving a nonlinear program such as the one in \eqref{eq:KMS:proj}.
Rather it obtains $\CI$ as in \eqref{eq:CI:BCS}. However, it approximates $c_{1-\alpha}$ by repeatedly solving bootstrap nonlinear programs, thereby incurring a very high computational cost at that stage.}
In all cases, the speed improvement results from a reduced number of evaluation points required to approximate the optimum.
In an application to a point identified setting, \cite[Supplement Section S.3]{fre:rev17} use \cite{kai:mol:sto19}'s E-A-M method to construct uniform confidence bands for an unknown function of interest under (nonparametric) shape restrictions. 
They benchmark it against gridding and find it to be accurate at considerably improved speed.

\section{Conclusions}
\label{sec:future}
This chapter provides a discussion of the econometrics literature on partial identification.
It first reviews what can be learned about (functionals of) probability distributions in the absence of parametric restrictions, under various scenarios of \emph{data incompleteness}.
It then reviews what can be learned about functionals characterizing semiparametric structural economic models, under various scenarios of \emph{model incompleteness}.
Finally, it discusses finite sample inference, the consequences of misspecification, and the computational challenges that a researcher needs to face when implementing partial identification methods.

Taking stock, I argue that several areas emerge where more progress is needed to bring the partial identification approach to empirical research to full fruition.
Whereas the last twenty years have seen the development of a burgeoning  theoretical literature on the topic, empirical applications of the methods still lag behind.
I conjecture that part of the reason for this discrepancy is due to the lack of easy-to-implement procedures for computation of estimators and confidence sets (or intervals) in complex structural models.
While the literature so far has aimed at developing methods that have desirable asymptotic properties for very general classes of models, there is arguably scope for more problem-specific methods that exploit the particularities of a certain model to obtain easy to implement statistical procedures.
It would also seem desirable that portable software accompanies the proposed methodologies, perhaps more in line with the current practice in the Statistics literature.

However, computational concerns cannot be the cause of the relative paucity of applications of partial identification methods as the ones reviewed in Section \ref{sec:prob:distr}, e.g., bounds on treatment effects.
These bounds are extremely easy to estimate and confidence intervals covering them can readily be computed.
I therefore conjecture that the lack of applications might be due to a misconception, whereby nonparametric bounds are perceived as ``always too wide to learn anything".
While it is true that, for example, worst-case nonparametric bounds on the average treatment effect cover zero by construction, the partial identification approach to empirical research proposes a wide array of assumptions that can be brought to bear to augment the empirical evidence and tighten the bounds.
The philosophy of the method is that the systematic reporting of bounds obtained under an increasingly strong set of assumptions illuminates the relative role played by assumptions and data in shaping the conclusions that the researcher draws.
Point identification is the limit of this process, and carefully assessing how this limit is reached is key to learning about the quantities of interest.

In Sections \ref{sec:prob:distr} and \ref{sec:structural}, special attention is devoted to characterizing \emph{sharp} identification regions.
Sharpness often requires \emph{many} moment inequalities, the number of which can exceed the available sample size.
Hence, there is a need of appropriate statistical inference methods.
As briefly mentioned in Sections \ref{sec:inference} and \ref{sec:computations}, methods designed to provide valid test of hypotheses and confidence sets in this scenario already exist.
However, I would argue that there is a need to better understand the trade-off between sharpness of the population identification region, and statistical efficiency, especially in the context of conditional moment inequalities where instrument functions are needed to transform the inequalities in unconditional ones.
Similarly, there is a need of more research on data driven procedures for the choice of tuning parameters for the construction of confidence sets, in particular in the case of projection inference where the question has not yet been addressed.
Another open and arguably important question in the literature, is how to build confidence sets for general moment inequality models that do not exhibit spurious precision (i.e., are arbitrarily small) when the model is misspecified.

\newpage
\appendix
\section{Basic Definitions and Facts from Random Set Theory}
\label{app:RCS}
This appendix provides basic definitions and results from random set theory that are used throughout this chapter.\footnote{The treatment here summarizes a few of the topics presented in \cite{mol:mol18}.}
I refer to \cite{mo1} for a textbook presentation of random set theory, and to \cite{mol:mol18} for a discussion focusing on its applications in econometrics.

The theory of random closed sets generally applies to the space of closed subsets of a locally compact Hausdorff second countable topological space $\carrier$, see \citet{mo1}. 
In this chapter I let $\carrier = \R^d$ to simplify the exposition.
Closedness is a property satisfied by random points (singleton sets), so that the theory of random closed sets includes the classical case of random points or random vectors as a special case. 
A random closed set is a measurable map $\eX:\Omega\mapsto\cF$, where measurability is defined by specifying the family of functionals of $\eX$ that are random variables.
\begin{definition}[Random closed set]
  \label{def:rcs}
  A map $\eX$ from a probability space $(\Omega,\salg,\P)$ to the family $\cF$ of closed subsets of $\R^d$ is   called a \emph{random closed set} if
  \begin{equation}
    \label{eq:X-}
    \eX^-(K)=\{\omega\in\Omega:\; \eX(\omega)\cap K\neq\emptyset\}
  \end{equation}
  belongs to the $\sigma$-algebra $\salg$ on $\Omega$ for each compact set $K$ in $\R^d$.
\end{definition}
A random \emph{compact} set is a random closed set which is compact with probability one, so that almost all values of $\eX$ are compact sets. 
A random \emph{convex} closed set is defined similarly, so that $\eX(\omega)$ is a convex closed set for almost all $\omega$.

Definition~\ref{def:rcs} means that $\eX$ is explored by its hitting events, i.e., the events where $\eX$ hits a compact set $K$. The corresponding hitting probabilities are very important in random set theory, because they uniquely determine the probability distribution of a random closed set $\eX$, see \cite[Section 1.1.3]{mo1}. The formal definition of the hitting probabilities, and the closely related containment probabilities, follows.
\begin{definition}[Capacity functional and containment functional]
  \label{def:capacity}
  {\color{white}0}
  \begin{enumerate}
  \item A functional $\sT_\eX(K):\cK\mapsto[0,1]$ given by 
  \begin{displaymath}
    \sT_\eX(K)=\Prob{\eX\cap K\neq\emptyset},\quad K\in\cK,
  \end{displaymath}
  is called \emph{capacity (or hitting) functional} of $\eX$. 
  \item A functional $\sC_\eX(F):\cF\mapsto[0,1]$ given by
  \begin{displaymath}
    \sC_\eX(F)=\Prob{\eX\subset F},\quad F\in\cF,
    \end{displaymath}
    is called the \emph{containment functional} of $\eX$.
  \end{enumerate}
  I write $\sT(K)$ instead of $\sT_{\eX}(K)$ and $\sC(K)$ instead of $\sC_{\eX}(K)$ where no ambiguity occurs.
\end{definition}

Ever since the seminal work of \citet{aum65}, it has been common to think of random sets as bundles of random variables -- the selections of the random sets.
\begin{definition}[Measurable selection]
   \label{def:selection}
   For any random set $\eX$, a (measurable) \emph{selection} of $\eX$ is a random element $\ex$ with values in $\R^d$ such that $\ex(\omega)\in\eX(\omega)$ almost surely. I denote by $\Sel(\eX)$ the set of all selections from $\eX$.
\end{definition}
The space of closed sets is not linear, which causes substantial difficulties in defining the expectation of a random set. 
One approach, inspired by \citet{aum65} and pioneered by \citet{art:vit75}, relies on representing a random set using the family of its selections, and considering the set formed by their expectations. 
If $\eX$ possesses at least one integrable selection, then $\eX$ is called \emph{integrable}. The family of all integrable selections of $\eX$ is denoted by $\Sel^1(\eX)$.
\begin{definition}[Unconditional and conditional Aumann --or selection-- expectation]
  \label{def:sel-exp}
  The \emph{(selection or) Aumann expectation} of an integrable random closed set $\eX$ is given by
  \[ \E \eX = \cl \left\{ \int_\Omega \ex d\P: \; \ex \in \Sel^1(\eX) \right\}.  \]
  For each sub-$\sigma$-algebra $\ssalg \subset \salg$, the \emph{conditional (selection or) Aumann expectation} of $\eX$ given $\ssalg$ is the $\ssalg$-measurable random closed set $\eY=\E(\eX|\ssalg)$ such that the family of $\ssalg$-measurable integrable selections of $\eY$, denoted $\Sel^1_\ssalg(\eY)$, satisfies
  \begin{equation*}
    \Sel^1_\ssalg(\eY)=\cl\Big\{\E(\ex|\ssalg): \, \ex \in \Sel^1(\eX)\Big\},
  \end{equation*}
  where the closure in the right-hand side is taken in $\mathbf{L}^1$.
\end{definition}
If $\eX$ is almost surely non-empty and its norm $\|\eX\|=\sup\{\|\ex\|:\; \ex\in \eX\}$ is an integrable random variable, then $\eX$ is said to be \emph{integrably bounded} and all its selections are integrable. 
In this case, since $\eX$ takes its realizations in $\R^d$, the family of expectations of these integrable selections is already closed and there is no need to take an additional closure as required in Definition~\ref{def:sel-exp}, see \cite[Theorem~2.1.37]{mo1}.
The selection expectation depends on the probability space used to define $\eX$, see \cite[Section 2.1.2]{mo1} and \cite[Section 3.1]{mol:mol18}. 
In particular, if the probability space is non-atomic and $\eX$ is integrably bounded, the selection expectation $\E \eX$ is a convex set regardless of whether or not $\eX$ might be non-convex itself \cite[Theorem 3.4]{mol:mol18}. 
This convexification property of the selection expectation implies that the expectation of the closed convex hull of $\eX$ equals the closed convex hull of $\E \eX$, which in turn equals $\E \eX$. 
It is then natural to describe the Aumann expectation through its support function, because this function traces out a convex set's boundary and therefore knowing the support function is equivalent to knowing the set itself, see equation (\ref{eq:rocka}) below.
\begin{definition}[Support function]
  \label{def:sup-fun}
  Let $K$ be a convex set. The \emph{support function} of $K$ is
    \begin{displaymath}
     h_K(u)=\sup\{k^\top u:\; k\in K\}\,, \qquad u\in\R^d\,,
    \end{displaymath}
  where $k^\top u$ denotes the scalar product. 
\end{definition}
The support function is finite for all $u$ if $K$ is bounded, and is sublinear (positively homogeneous and subadditive) in $u$. 
Hence, it can be considered only for $u \in \Ball$ or $u \in \Sphere$. 
Moreover, one has
\begin{equation}
\label{eq:rocka}
K=\cap_{u \in \Ball}\{k: k^\top u \leq h_K(u) \} =\cap_{u \in \Sphere}\{k: k^\top u \leq h_K(u)\}.
\end{equation}

Next, I define the Hausdorff metric, a distance on the family $\cK$ of compact sets:
\begin{definition}[Hausdorff metric]
\label{def:hausdorff}
Let $K,L\in\cK$. The \emph{Hausdorff distance} between $K$ and $L$ is
\begin{displaymath}
  \rhoH(K,L)=\inf\Big\{r>0:\; K\subseteq L^r,\; L\subseteq K^r\Big\},
\end{displaymath}
where $K^r=\{x: \dist(x,K)\le r\}$ is the $r$-envelope of $K$. 
\end{definition}
Since $K\subseteq L$ if and only if $h_K(u)\leq h_L(u)$ for all $u\in\Sphere$ and $h_{K^r}(u)=h_K(u)+r$, the uniform metric for support functions on the sphere turns into the Hausdorff distance between compact convex sets. Namely,
\begin{align}
  \rhoH(K,L)=\sup\Big\{|h_K(u)-h_L(u)|:\; \|u\|=1\Big\}.
  \label{eq:Hormander}
\end{align}
It follows that
\begin{displaymath}
  \|K\|=\rhoH(K,\{0\})=\sup\big\{|h_K(u)|:\; \|u\|=1\big\}.
\end{displaymath}
Finally, I define independently and identically distributed random closed sets \citep[see][Proposition 1.1.40 and Theorem 1.3.20, respectively]{mo1}:
\begin{definition}[i.i.d. random closed sets]
Random closed sets $\eX_1,\dots,\eX_n$ in $\R^d$ are independent if and only if $\Prob{\eX_1\cap K_1 \ne \emptyset,\dots,\eX_n\cap K_n \ne \emptyset}=\prod_{i=1}^n\sT_{\eX_i}(K_i)$ for all $K_1,\dots,K_n \in \cK$.
They are identically distributed if and only if for each open set $G$, $\Prob{\eX_1\cap G \ne \emptyset}=\Prob{\eX_2\cap G \ne \emptyset}= \dots =\Prob{\eX_n\cap G \ne \emptyset}$.
\end{definition}

With these definitions in hand, I can state the theorems used throughout the chapter.
The first is a dominance condition due to \citet{art83} \citep[and][]{nor92} that characterizes probability distributions of selections \citep[see][Section 2.2]{mol:mol18}:
\begin{theorem}[Artstein]
  \label{thr:artstein}
  A probability distribution $\mu$ on $\R^d$ is the distribution of a selection of a random closed set $\eX$ in $\R^d$ if and only if
  \begin{equation}
    \label{eq:domin-t}
    \mu(K)\leq \sT(K)=\Prob{\eX\cap K\neq\emptyset}
  \end{equation}
  for all compact sets $K\subseteq\R^d$. Equivalently, if and only if 
  \begin{equation}
    \label{eq:dom-c}
    \mu(F)\geq \sC(F)=\Prob{\eX\subset F}
  \end{equation}
  for all closed sets $F\subset\R^d$. If $\eX$ is a compact random closed set, it suffices to check \eqref{eq:dom-c} for compact sets $F$ only.
\end{theorem}
If $\mu$ from Theorem~\ref{thr:artstein} is the distribution of some random vector $\ex$, then it is not guaranteed that $\ex\in \eX$ a.s., e.g. $\ex$ can be independent of $\eX$.
 Theorem~\ref{thr:artstein} means that for each such $\mu$, it is possible to construct $\ex$ with distribution $\mu$ that belongs to $\eX$ almost surely. In other words, $\ex$ and $\eX$ can be realized on the same probability space (coupled) as random elements $\ex^\prime$ and $\eX^\prime$ such that $\ex\edis\ex^\prime$ and $\eX\edis\eX^\prime$ with $\ex^\prime \in \eX^\prime$ a.s.

The definition of the distribution of a random closed set (Definition \ref{def:capacity}) and the characterization results for its selections in Theorem \ref{thr:artstein} require working with functionals defined on the family of all compact sets, which in general is very rich.  
It is therefore important to reduce the family of all compact sets required to describe the distribution of the random closed set or to characterize its selections.
\begin{definition}
  \label{def:core-det}
  A family of compact sets $\cM$ is said to be a \emph{core determining class} for a random closed set $\eX$ if any probability
  measure $\mu$ satisfying the inequalities
  \begin{equation}
    \label{eq:cdclass}
    \mu(K)\leq \Prob{\eX\cap K\neq\emptyset}
  \end{equation}
  for all $K\in\cM$, is the distribution of a selection of $\eX$, implying that \eqref{eq:cdclass} holds for all compact sets $K$.
\end{definition}
The notion of a core determining class was introduced by \cite{gal:hen06}.
A simple and general, but still mostly too rich, core determining class is obtained as a subfamily of all compact sets that is dense in a certain sense in the family $\cK$. 
For instance, in the Euclidean space, it suffices to consider compact sets obtained as finite unions of closed balls with rational centers and radii \citep[e.g.,][Theorem 3c]{gal:hen06}.
For the case that $\eX$ is a subset of a finite space, \cite[Algorithm 5.1]{ber:mol:mol08} propose a simple algorithm to compute core determining classes.
\cite{che:ros12} provide a related algorithm.
Throughout this chapter, several results are mentioned where the class of sets over which \eqref{eq:domin-t} is verified is reduced from the class of compact subsets of the carrier space, to a (significantly) smaller collection.

The next result characterizes a dominance condition that can be used to verify the existence of selections of $\eX$ with specific properties for their means \citep[see][Sections 3.2-3.3]{mol:mol18}
\begin{theorem}[Convexification in $\R^d$]
  \label{thr:exp-supp}
  Let $\eX$ be an integrable random set. If $\eX$ is defined on a non-atomic probability space, or if $\eX$ is almost surely convex, then $\E \eX=\E \conv\eX$ and
  \begin{equation}
    \E h_\eX(u)=h_{\E \eX}(u),\quad u\in\R^d.
    \label{eq:supf}
  \end{equation}
  If $\P$ is atomless over $\ssalg$,\footnote{An event $A'\in\ssalg$ is called a $\ssalg$-atom if $\Prob{0<\P(A|\ssalg)<\P(A'|\ssalg)}=0$ for all $A\subset A'$ such that $A\in\salg$.} then $\E(\eX|\ssalg)$ is convex and  
  \begin{equation}
    \E(h_\eX(u)|\ssalg)=h_{\E(\eX|\ssalg)}(u),\quad u\in\R^d.
    \label{eq:supf:cond}
  \end{equation}
  Hence, for any vector $b\in\R^d$, it holds that
  \begin{align}
  b \in \E \eX &\Leftrightarrow b^\top u \le \E h_\eX(u)~~\forall u\in\Sphere,\label{eq:dom_Aumann}\\
  b \in \E(\eX|\ssalg) &\Leftrightarrow b^\top u \le \E(h_\eX(u)|\ssalg)~~\forall u\in\Sphere.\label{eq:dom_Aumann:cond}
  \end{align}
\end{theorem}
An important consequence of Theorem \ref{thr:exp-supp} is that it allows one to verify whether $b \in \E \eX$ without having to compute $\E \eX$ but only $\E h_\eX(u)$ (and similarly for the conditional case), a substantially easier task.

Finally, i.i.d. random closed sets satisfy a law of large numbers and a central limit theorem that are similar to the ones for random singletons.
Recall that the \emph{Minkowski sum}\label{def:mink:sum} of two sets $K$ and $L$ in a linear space (which in this chapter I assume to be the Euclidean space $\R^d$) is obtained by adding each point from $K$ to each point from $L$, formally,
\begin{displaymath}
  K+L=\big\{x+y:\; x\in K,\;y\in L\big\}.
\end{displaymath}
Below, $\eX_1+\cdots+\eX_n$ denotes the Minkowski sum of the random closed sets $\eX_1,\dots,\eX_n$, and $(\eX_1+\cdots+\eX_n)/n$ denotes their \emph{Minkowski average}.
\begin{theorem}[Law of large numbers for integrably bounded random
  sets]
  \label{thr:SLLN-basic}
  Let $\eX,\eX_1,\eX_2,\ldots$ be i.i.d. integrably bounded random compact sets. Define $\eS_n=\eX_1+\cdots+\eX_n$. Then
  \begin{align}
  \label{eq:LLN}
    \rhoH\left(\frac{\eS_n}{n},\E \eX\right)\to 0 \quad
    \text{a.s. as }\ n\to\infty.
  \end{align}
\end{theorem}

The support function of a random closed set $\eX$ such that $\E\|\eX\|^2<\infty$, is a random continuous function $h_\eX(u)$ on $\Sphere$ with square integrable values. 
Define its covariance function as 
\begin{align}
  \Gamma_\eX(u,v)\equiv\E\left[(h_\eX(u)-h_{\E \eX}(u))(h_\eX(v)-h_{\E \eX}(v))\right],
  ~~u,v\in\Sphere. 
  \label{eq:cov-Gamma}
\end{align}
Let $\zeta(u)$ be a centered Gaussian random field on $\Sphere$ with the same covariance structure as $\eX$, i.e. $\E\big[\zeta(u)\zeta(v)\big]=\Gamma_\eX(u,v),~u,v\in\Sphere$.
Since the support function of a compact set is Lipschitz, it is easy to show that the random field $\zeta$ has a continuous modification by bounding the moments of $|\zeta(u)-\zeta(v)|$. 
\begin{theorem}[Central limit theorem] 
  \label{thr:clt}
  Let $\eX_1,\eX_2,\dots$ be i.i.d. copies of a random closed set $\eX$ in $\R^d$ such that $\E \|\eX\|^2<\infty$, and let $\eS_n=\eX_1+\cdots+\eX_n$. Then as $n\to\infty$,
  \begin{equation}
    \label{eq:h-weak}
    \sqrt{n}\Big(h_{\frac{\eS_n}{n}}(u)-h_{\E\eX}(u)\Big)\Rightarrow \zeta
  \end{equation}
  in the space of continuous functions on the unit sphere with the uniform metric. Furthermore, 
  \begin{equation}
    \label{eq:clt-basic}
    \sqrt{n}\rhoH\left(\frac{\eS_n}{n},\E \eX\right)\Rightarrow 
    \|\zeta\|_\infty=\sup\big\{|\zeta(u)|:\; u\in\Sphere\big\}.
  \end{equation}
\end{theorem}
\newpage
\bibliography{EwPI_biblio3} 
\label{biblio}
\end{document}